\documentclass[12pt]{article}
\usepackage[utf8]{inputenc}
\usepackage[authoryear,sectionbib]{natbib}
\usepackage[margin=1in]{geometry}
\usepackage{graphicx}
\usepackage{mathrsfs,multirow,amsfonts,amsmath,amssymb,amsthm,mathtools}
\usepackage{comment}
\usepackage{algorithm,algpseudocode}
\usepackage{xcolor}
\usepackage[mathscr]{euscript}
\usepackage{booktabs,lscape,arydshln,soul,array,makecell}
\usepackage{enumitem,setspace,etoc}
\usepackage{microtype}
\usepackage[colorlinks]{hyperref}
\title{\vspace{-2cm}Personalized Federated Vector Autoregression with Personalization Diversity\vspace{-0.2cm}}
\author{Zhiyun Fan\thanks{Zhiyun Fan and Xiaoyu Zhang contribute equally and are joint first authors.}, Xiaoyu Zhang\footnotemark[1], Guodong Li, and Di Wang\thanks{Di Wang is the corresponding author (di.wang@sjtu.edu.cn).}\\
\small{Shanghai Jiao Tong University, Tongji University, and University of Hong Kong}}

\algrenewcommand\algorithmicrequire{\textbf{Input:}}
\algrenewcommand\algorithmicensure{\textbf{Output:}}
\newlength{\comparisonwidth}

\newtheorem{assumption}{Assumption}
\newtheorem{definition}{Definition}
\newtheorem{lemma}{Lemma}
\newtheorem{proposition}{Proposition}
\newtheorem{theorem}{Theorem}
\newtheorem{corollary}{Corollary}
\newtheorem{remark}{Remark}
\newtheorem{example}{Example}

\newtheorem{supplementtheorem}{Theorem}[section]
\newtheorem{supplementlemma}[supplementtheorem]{Lemma}
\newtheorem{supplementproposition}[supplementtheorem]{Proposition}
\newtheorem{supplementremark}[supplementtheorem]{Remark}
\newtheorem{supplementcondition}[supplementtheorem]{Condition}
\newtheorem{supplementcorollary}[supplementtheorem]{Corollary}
\newtheorem{supplementassumption}[supplementtheorem]{Assumption}
\theoremstyle{definition}
\newtheorem{supplementdefinition}[supplementtheorem]{Definition}
\newtheorem*{supplementexample}{Example}

\hypersetup{colorlinks=true,linkcolor=blue,citecolor=blue,filecolor=magenta,urlcolor=magenta,
    pdftitle={Personalized Federated Vector Autoregression with Personalization Diversity},
    pdfauthor={Zhiyun Fan, Xiaoyu Zhang, Guodong Li, Di Wang}}
\DeclareMathOperator*{\argmin}{arg\,min}

\DeclareMathOperator{\supp}{supp}
\DeclareMathOperator{\op}{op}
\newcommand{\bm}{\mathbf}
\newcommand{\bbm}{\boldsymbol}

\newcommand{\inner}[2]{\left\langle #1, #2 \right\rangle}
\newcommand{\Ttot}{T_{\mathrm{tot}}}

\newcommand{\Eall}{\mathcal E_{\mathrm{all}}}
\mathtoolsset{showonlyrefs}

\begin{document}
\etocdepthtag.toc{main}

\maketitle

\vspace{-1cm}
\begin{abstract}
Personalized federated estimation of high-dimensional time series faces two challenges: identifying shared dynamics and retaining the statistical gains from pooling. We address these challenges through personalization diversity in a vector autoregressive model with shared low-rank dynamics, shared sparse links, and personalized sparse departures. A cross-client support-frequency cap yields a sharp majority threshold for shared-personalized identification, while an additional cross-client \(\ell_q\) condition captures approximate diversity and controls approximation under a smaller working personalization frequency cap. We develop PerFeCT-VAR, Personalized Federated Vector Autoregression via Frequency-Capped Thresholding, which jointly enforces clientwise sparsity and cross-client frequency constraints. Under suitable sampling and cross-component compatibility conditions, we establish initialization guarantees and joint local linear convergence for heterogeneous stationary client distributions. Our componentwise error bounds show that diversity attenuates both estimation and approximation errors transmitted from personalized to shared components, and quantify the tradeoff induced by the working cap. These bounds identify regimes in which the shared components attain pooled accuracy despite non-vanishing personalized effects, while personalized estimation achieves local-sample accuracy on average. Simulations and an application to multistore retail revenue data illustrate the estimation, forecasting, and interpretive benefits of the proposed framework.
\end{abstract}

\noindent\textit{Keywords}: Federated learning, heterogeneity, high-dimensional time series, personalization.

\setlength\abovedisplayskip{2pt}
\setlength\belowdisplayskip{2pt}
\newpage

\section{Introduction}
\label{sec:introduction}

High-dimensional multivariate time series are increasingly collected by
networks of related organizations, branches, or operational units
\citep{wilms2018multiclassVAR,lucas2019bank}. These clients often observe
the same variables and share common dynamic mechanisms, yet their raw
data may need to remain decentralized because of data ownership,
regulatory, or operational constraints. Federated learning provides a
natural way to borrow information without pooling the raw data
\citep{tan2022towards}. Suppose there are \(N\) clients, indexed by
\(i=1,\ldots,N\), and client \(i\) follows a \(p\)-dimensional
first-order vector autoregression,
\(\bm y_{i,t}=\bm A_i\bm y_{i,t-1}+\bbm\epsilon_{i,t}\),
where \(\bm A_i\in\mathbb R^{p\times p}\) is its transition matrix.
Estimating each \(\bm A_i\) separately can be inefficient when the local
time series is only moderately long, since each transition matrix
contains \(p^2\) coefficients. Conversely, complete pooling forces
heterogeneous transition mechanisms into a common model, creating an
approximation error that can persist as the pooled sample size grows.
Personalized federated learning therefore seeks to borrow strength
across related clients while retaining genuinely client-specific
dynamics.

Two questions arise when shared and personalized dynamics are estimated
jointly. The first concerns \emph{population identification}. In the
representation \(\bm A_i=\bm G+\bm S_i\), where \(\bm G\) captures shared
dynamics and \(\bm S_i\) is a sparse personalized departure, clientwise
sparsity alone need not distinguish a shared effect from sparse personalized departures assigned across clients. An interpretable decomposition therefore requires restrictions on the cross-client distribution of personalization. 
The second concerns \emph{statistical rate separation}.
Shared components can draw information from all clients,
whereas the personalized estimation depends on local samples. 
Personalized estimation errors enter the shared updates and may limit the gains
from pooling. We study when the decomposition is identifiable and
when shared estimation can attain total-sample-size accuracy while
personalized effects retain non-vanishing magnitudes.

We address these questions through \emph{personalization diversity},
which describes the cross-client distribution of personalized effects.
Consider the three-component transition structure
\[
    \bm A_i
    =
    \underbrace{\bm L\bm R^\top}_{\text{shared low-rank}}
    +
    \underbrace{\bm S_0}_{\text{shared sparse}}
    +
    \underbrace{\bm S_i}_{\text{personalized sparse}},
\]
representing pervasive shared dependence, recurrent
shared links, and client-specific departures, respectively. Exact
personalization diversity bounds the population support frequency
\(m_*\), the largest number of clients with a nonzero personalized
departure at any coordinate. When \(m_*<N/2\), a strict majority of
clients agrees with the shared coefficient at every coordinate.
This identifies the shared transition \(\bm G=\bm L\bm R^\top+\bm S_0\) and
all personalized departures. The threshold is sharp for uniform
shared--personalized identification based on the frequency restriction alone. Separating the low-rank and sparse components within \(\bm G\) additionally requires the local identification conditions. The frequency
restriction allows personalized departures to differ in magnitude and sign.

Support frequency counts every nonzero departure equally,
including coefficients of small magnitude. We therefore
supplement the majority condition with a coordinatewise
cross-client \(\ell_q\) bound, \(0<q\le1\), termed
\emph{approximate personalization diversity}.
This additional condition accounts for both the magnitudes and
the recurrence of personalized departures. Within the identified
population model, it controls the approximation error incurred
by fitting a working frequency cap \(m\) possibly smaller than \(m_*\).
The cap thus controls a tradeoff between the statistical
cost of personalization and approximation error in
shared recovery.

To enforce these restrictions during estimation, we develop
\textsc{PerFeCT-VAR}, short for Personalized Federated Vector
Autoregression via Frequency-Capped Thresholding. Its central operation,
Frequency-Capped Thresholding (FCT), 
selects personalized supports jointly across clients through
a weighted projection enforcing clientwise sparsity and the
working frequency cap. We combine FCT with scaled gradient
updates for the shared low-rank component, greedy
capacity-aware thresholding for the shared sparse component,
and a federated initializer based on local low-rank plus
sparse fits. The procedure maintains the fitted constraints
throughout joint refinement, with raw trajectories retained
at their respective clients.

Our theory establishes joint local linear convergence of the
shared and personalized components to statistical and
approximation error floors. The componentwise bounds quantify how statistical and approximation errors in personalized estimation propagate to the shared components through diversity-dependent interaction terms. The analysis accounts for the effect of heterogeneous stationary covariances through the directional covariance discrepancy measures. Cross-component compatibility conditions control the resulting interactions, accounting for the working cap and sample size imbalance. Initialization guarantees under local sample size
and signal conditions ensure entry into the local convergence basin and complete the analysis of the
federated procedure.

For a balanced illustration, suppose each client has
\(T\) usable transitions and personalized sparsity \(s_*\). Write
\(D_G=pr+s_{0,*}\log p\) and \(D_P=s_*\log p+\log N\), where \(r\)
and \(s_{0,*}\) are the shared rank and sparsity. Let \(\mathcal E_G(m)\)
denote the larger of the squared Frobenius errors of the two shared
matrix components, and let \(\mathcal E_{\mathrm P}(m)\) denote the
average squared Frobenius error of the personalized components.
Under Corollary~\ref{cor:main_equal_sample_rates} and,
for \(q>0\), the sufficient condition
\(D_G+Ns_*\log p\lesssim mT\) in
\eqref{eq:lq_balanced_simplification}, the errors satisfy
\[
\begin{aligned}
    \mathcal E_G(m)
    &\lesssim
    \frac{D_G}{NT}
    +\frac{m}{N}
       \left\{\frac{D_P}{T}+\mathcal A_q(m)\right\},\quad
    \mathcal E_{\mathrm P}(m)
    &\lesssim
    \frac{D_P}{T}
    +\frac{D_G}{NT}
    +\mathcal A_q(m),
\end{aligned}
\]
with high probability once the optimization error is
absorbed. Here regularity and discrepancy factors are uniformly
controlled over the admissible caps being compared, and
\(\mathcal A_q(m)\) is the approximation envelope defined in
Section~\ref{sec:4}. At \(q=0\), this envelope vanishes when
the working cap satisfies the required fitting margin
relative to \(m_*\). 
Diversity attenuates the propagation of both statistical and
approximation errors from personalized to shared components
through the factor \(m/N\).
If \(\mathcal A_q(m)\lesssim D_P/T\) and
\(mD_P\lesssim D_G\), the two shared components attain
the pooled order \(D_G/(NT)\).
When this target rate vanishes, the bounds permit consistent
shared recovery while personalized effects retain fixed
nonzero magnitudes.

Our main contributions are summarized as follows.
\begin{enumerate}[leftmargin=1.2em]

\item
\textbf{Personalization diversity and interpretable decomposition.}
We formulate exact and approximate personalization diversity for the
three-component VAR model. Support-frequency control gives a
sharp uniform condition for identifying the shared
transition and personalized departures. An additional
cross-client \(\ell_q\) condition quantifies approximation
under a smaller working frequency cap.

\item
\textbf{Frequency-capped federated estimation.}
We introduce frequency-capped thresholding (FCT) as a weighted projection
that incorporates personalization diversity into estimation by jointly enforcing clientwise
sparsity and cross-client frequency constraints. Building on FCT, we develop \textsc{PerFeCT-VAR}, which
combines scaled low-rank updates, shared sparse updates, and
a federated initializer for joint estimation of shared and
personalized dynamics.

\item
\textbf{Error propagation and pooled accuracy.}
We establish joint local convergence and componentwise
error bounds under heterogeneous stationary distributions.
The analysis quantifies how diversity attenuates both
personalized sampling error and approximation error in
shared recovery. It describes the working-cap tradeoff
and yields sufficient conditions for pooled shared
accuracy with non-vanishing personalized effects.

\end{enumerate}

\subsection{Related work}
\label{subsec:related_work}

\paragraph*{Structured and multi-client VAR models}
High-dimensional VAR estimation has been studied extensively under
sparse, low-rank, and related structural assumptions
\citep{basu2015regularized,kock2015oracle,basu2019low,wang2022high,huang2025supervised}.
For multiple related time series,
\citet{wilms2018multiclassVAR}, \citet{skripnikov2019joint}, and
\citet{fisher2022penalized} exploit commonality among sparse transition
structures. More recent work learns shared representations:
\citet{xu2025joint} combine a common low-rank basis with
entity-specific rescaling and sparse components, while
\citet{lin2025improvingtimeseriesestimation} transfer shared
multilinear representations across time series.
\citet{chen2026private} study federated VAR estimation with
a shared low-rank transition and client-specific weakly
sparse deviations, including privacy--utility guarantees.
Their deviation magnitude restriction controls residual
shared--personalized ambiguity, which contributes a
heterogeneity term to the error bounds.
Our formulation uses cross-client support-frequency
restrictions to identify a common transition containing
both low-rank and sparse structure, and studies its joint
estimation with the personalized departures.

\paragraph*{Multitask and personalized federated learning}
Federated multitask and personalized learning seek to exploit common
information while accommodating heterogeneous client parameters
\citep{smith2017federated,tan2022towards}. Statistical analyses include
the global--local trade-off of \citet{chen2023minimax} and robust sparse
personalized regression of \citet{liu2025robust}. Particularly related,
\citet{pal2024sample} study shared subspaces with sparse task-specific
deviations under cross-task sparsity conditions.
\citet{huang2025optimal} estimate a shared regression vector
by weighted-median aggregation and bound its coordinatewise
error through the weighted frequency of unaligned tasks.
Their bounds permit pooled accuracy on sufficiently
well-aligned coordinates, even when nonzero personalized
departures do not vanish.
These results connect cross-task sparsity and coordinatewise
alignment with shared--personalized separation and information
sharing.
Our paper studies joint estimation of shared
low-rank, shared sparse, and personalized VAR components under
heterogeneous stationary distributions. Personalization
diversity connects population identification, the constraints
enforced during estimation, and approximation under a working
frequency cap. Our componentwise analysis quantifies how this
cap controls the transmission of personalized sampling and
approximation errors to the two shared components.

Table~\ref{tab:method_comparison} summarizes the selected models
and guarantees for learning shared structure under additive 
personalization. The comparison distinguishes their shared
targets, personalization restrictions, and relevant
statistical conclusions.

\begin{table}[!htp]
\caption{Shared--personalized separation and pooled estimation
with non-vanishing additive deviations.}
\label{tab:method_comparison}
\centering
\footnotesize
\setlength{\tabcolsep}{3pt}
\renewcommand{\arraystretch}{1.16}
\setlength{\comparisonwidth}
  {\dimexpr\linewidth-8\tabcolsep\relax}
\begin{tabular}{@{}
>{\raggedright\arraybackslash}p{0.300\comparisonwidth}
>{\raggedright\arraybackslash}p{0.310\comparisonwidth}
>{\centering\arraybackslash}p{0.140\comparisonwidth}
>{\centering\arraybackslash}p{0.120\comparisonwidth}
>{\centering\arraybackslash}p{0.120\comparisonwidth}
@{}}
\toprule
\makecell[l]{Study and\\shared target}
& \makecell[l]{Basis of\\rate control}
& \makecell{Exact\\identification}
& \makecell{Diversity\\analysis}
& \makecell{Pooled\\shared rate}\\
\midrule
\citet{xu2025joint}\newline
(shared basis)
& Sparsity and magnitude
& \(\checkmark\)
& \(\times\)
& \(\times\)\\[4pt]

\citet{lin2025improvingtimeseriesestimation}\newline
(shared subspaces)
& Deviation magnitude
& \(\times\)
& \(\times\)
& \(\times\)\\[4pt]

\citet{chen2026private}\newline
(shared low-rank VAR)
& Deviation magnitude
& \(\times\)
& \(\times\)
& \(\times\)\\[4pt]

\citet{huang2025optimal}\newline
(shared vector)
& Weighted deviation frequency
& \(\times\)
& Exact
& \(\checkmark\)\\

\textbf{PerFeCT-VAR}\newline(shared low-rank and sparse)
& Deviation frequency
& \(\checkmark\)
& Exact and approximate
& \(\checkmark\)\\
\bottomrule
\end{tabular}

\vspace{3pt}
\begin{minipage}{\linewidth}
\footnotesize
\emph{Note.}
Deviations are additive departures from the shared structure.
Diversity refers to cross-client support-frequency or \(\ell_q\) control.
The pooled rate requires total-sample-size rate to dominate the optimization error, bias, and personalization error for the shared components, when the model admits non-vanishing heterogeneity.
\end{minipage}
\end{table}
\vspace{-0.4cm}

\subsection{Notation and outline}
\label{subsec:notation}

For a positive integer \(n\), let \([n]=\{1,\ldots,n\}\), and let
\(\mathrm{GL}(r)\) denote the set of invertible \(r\times r\) matrices. For a matrix \(\bm M\), let \(\|\bm M\|_{\mathrm F}\), \(\|\bm M\|_*\), \(\|\bm M\|_{\mathrm{op}}\), \(\|\bm M\|_{2,\infty}\), \(\|\bm M\|_1\), and \(\|\bm M\|_{\max}\) denote its Frobenius, nuclear, operator, maximum rowwise Euclidean, entrywise \(\ell_1\), and entrywise maximum norms, respectively. We write \(\langle\bm M,\bm H\rangle=\operatorname{tr}(\bm M^\top\bm H)\) and use \(\supp(\bm M)\) for its entrywise support. The singular values of \(\bm M\) are denoted by \(\sigma_j(\bm M)\), and \(\rho(\bm M)\) denotes the spectral radius of a square matrix. We write \(a\vee b=\max\{a,b\}\), and use \(a\lesssim b\) and \(a\asymp b\) with their usual meanings up to generic positive constants. For symmetric matrices, \(\preceq\) denotes the positive semidefinite
order. Inequalities between vectors are interpreted componentwise.

Section~\ref{sec:2} introduces the model, personalization
diversity, and identification conditions.
Section~\ref{sec:3} presents the federated estimator,
initialization, and communication-efficient implementations.
Section~\ref{sec:4} establishes convergence, statistical
accuracy, and initialization guarantees.
Sections~\ref{sec:simulations} and~\ref{sec:realdata}
present simulations and an application to multistore retail
revenue data. The Supplementary Material contains complete
theoretical statements, proofs, implementation details,
and simulation details.

\section{Model and Structural Identification}
\label{sec:2}

\subsection{The PerFeCT-VAR model}
\label{subsec:2.1}

We consider \(N\) clients observing the same \(p\)-dimensional time series variables over possibly different periods. After re-indexing each local series, client \(i\) contributes \(T_i\) usable transitions from the VAR(1) model
\begin{equation}
    \label{eq:global_local_model}
    \bm y_{i,t}
    =
    \bm A_i \bm y_{i,t-1}
    +
    \bbm\epsilon_{i,t},
    \qquad
    t=1,\ldots,T_i,\quad i=1,\ldots,N,
\end{equation}
where \(\bm A_i\in\mathbb R^{p\times p}\) is the client-specific transition matrix. For each client, \(\{\bbm\epsilon_{i,t}\}_{t\geq1}\) is a mean-zero white-noise process with positive-definite covariance $\bbm\Sigma_{\epsilon,i}$, which may vary across clients. Denote \(\Ttot=\sum_{i=1}^N T_i\) as the total number of usable transitions.

The \textsc{PerFeCT-VAR} model decomposes each transition matrix as
\begin{equation}
\label{eq:decomp_full}
    \bm A_i
    =
    \underbrace{\bm G}_{\text{shared dynamics}}
    +
    \underbrace{\bm S_i}_{\text{personalized dynamics}},
    \qquad
    \bm G
    =
    \underbrace{\bm L\bm R^\top}_{\text{shared low-rank dynamics}}
    +
    \underbrace{\bm S_0}_{\text{shared sparse links}},
\end{equation}
where \(\bm L,\bm R\in\mathbb R^{p\times r}\),
\(\operatorname{rank}(\bm L\bm R^\top)=r\), and typically
\(r\ll p\). The low-rank component captures pervasive predictive
dependence through a common low-dimensional structure. In the
factorization \(\bm L\bm R^\top\bm y_{i,t-1}\), \(\bm R\) specifies
\(r\) predictive linear combinations of the lagged variables, and
\(\bm L\) gives their loadings on the current responses. The shared
sparse component \(\bm S_0\) captures recurrent localized lagged
relationships that complement this low-dimensional predictor.
Together, these components accommodate pervasive and localized
shared dependence \citep{basu2019low}.

The personalized matrix \(\bm S_i\) represents client-specific
departures from the shared transition matrix \(\bm G\). A departure
can introduce a lagged relationship or modify the magnitude or sign
of an existing shared coefficient, including canceling its effect.
The supports of \(\bm S_0\) and \(\bm S_i\) may therefore overlap.
Figure~\ref{fig:perfect_var} illustrates the interpretation of this
three-component structure.

\vspace{-0.3cm}
\begin{figure}[!htb]
    \centering
    \includegraphics[width=0.95\textwidth]{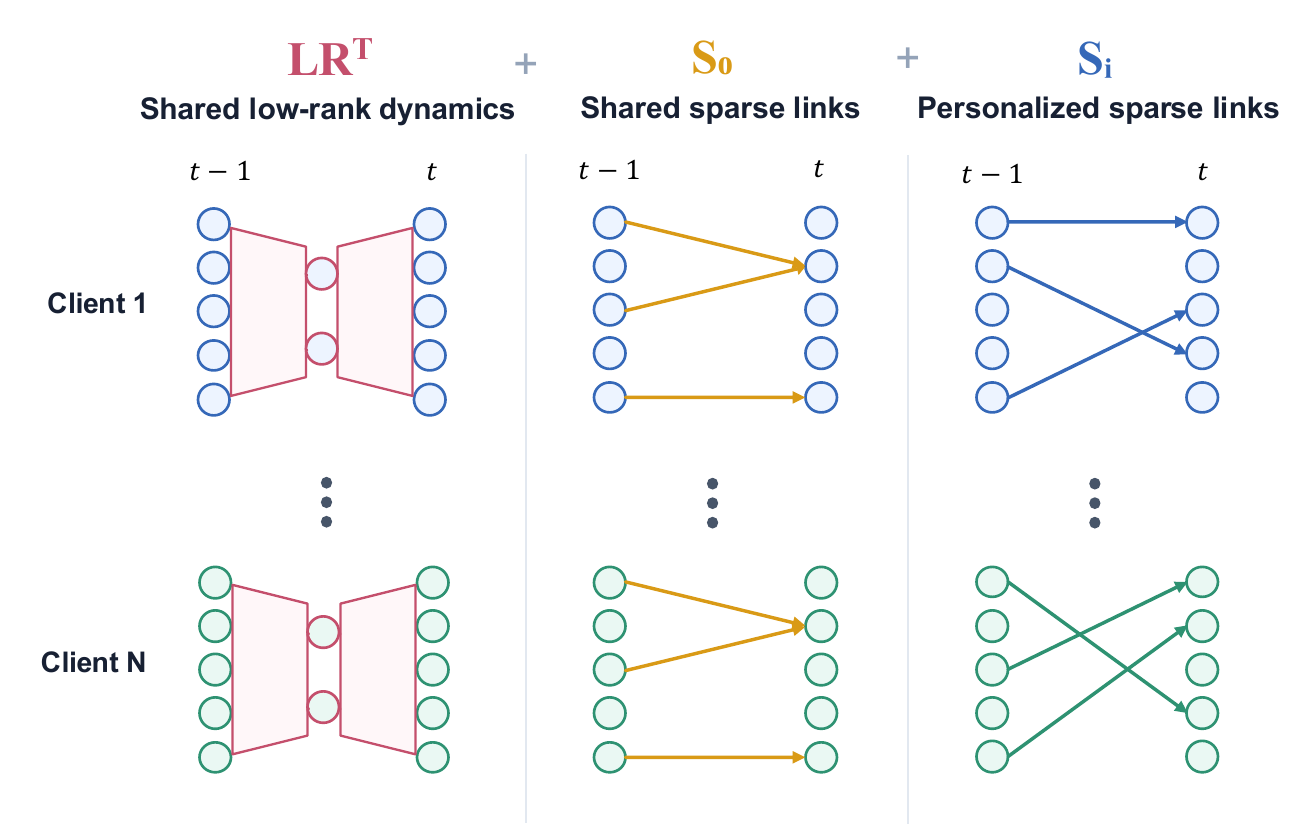}
    \vspace{-0.3cm}
    \caption{Three-component transition structure of \textsc{PerFeCT-VAR}.}
    \label{fig:perfect_var}
\end{figure}

\begin{remark}

Using our notation, \citet{chen2026private} consider the transition
structure \(\bm A_i=\bm L\bm R^\top+\bm S_i\), with a shared
low-rank component and client-specific weakly sparse deviations.
The additional shared sparse component \(\bm S_0\) in
\textsc{PerFeCT-VAR} captures recurrent localized lagged
relationships alongside the pervasive dependence represented by
\(\bm L\bm R^\top\). Its interpretation is illustrated in
Section~\ref{sec:realdata}.

\end{remark}

We study the stationary solutions of the client-specific VAR
processes under the following stability condition. For the remainder of this paper, a subscript \(*\) denotes the true
data-generating parameters.
\begin{assumption}[Stationarity]
\label{assump:stationarity}
For each \(i\in[N]\), the ground truth satisfies $\rho(\bm A_{i,*})<1$.
\end{assumption}

Under Assumption~\ref{assump:stationarity}, the stationary covariance
matrix \(\bbm\Gamma_i:=\operatorname{Var}(\bm y_{i,t})\) satisfies
\(\bbm\Gamma_i=\bm A_{i,*}\bbm\Gamma_i\bm A_{i,*}^\top+
\bbm\Sigma_{\epsilon,i}\). Differences in transition matrices or
innovation covariances can therefore induce heterogeneous stationary
covariances, whose effects on federated estimation are studied in
Section~\ref{sec:4}. We next establish conditions identifying the
shared and personalized components, followed by conditions for
locally separating the shared low-rank and sparse components.

\subsection{Personalization diversity and identification}
\label{subsec:2.2}

Clientwise sparsity alone does not ensure identifiability of the
decomposition \(\bm A_i=\bm G+\bm S_i\). Replacing \(\bm G\) by
\(\bm G+\bm H\) and every \(\bm S_i\) by \(\bm S_i-\bm H\)
leaves all transition matrices unchanged, and a nonzero sparse
\(\bm H\) can yield two decompositions satisfying the same
clientwise sparsity limits. These limits alone therefore do not
determine which sparse effects belong to the shared component.
We formalize \emph{personalization diversity} by restricting
how often a given coordinate can have a personalized
departure across clients.

\begin{definition}[Clientwise-sparse and frequency-capped support class]
\label{def:frequency_capped_class}
For \(s_i\in\{0,\ldots,p^2\}\) and \(m\in\{0,\ldots,N\}\), define
\[
    \mathcal D(\{s_i\}_{i=1}^N,m)
    =
    \left\{
    (\Omega_1,\ldots,\Omega_N):
    |\Omega_i|\leq s_i,\ i\in[N],\
    \max_{1\leq j,k\leq p}
    \sum_{i=1}^N
    \mathbf 1\{(j,k)\in\Omega_i\}
    \leq m
    \right\}.
\]
\end{definition}

Thus, each client has at most \(s_i\) personalized links, while any given transition coordinate can be personalized in at most \(m\) clients. The next proposition shows that the same frequency restriction gives an identification threshold.

\begin{proposition}[Shared--personalized identification]
\label{prop:frequency_capped_identifiability}
Suppose that
\[
    \bm A_i
    =
    \bm G+\bm S_i
    =
    \widetilde{\bm G}+\widetilde{\bm S}_i,
    \quad\text{for }i\in[N],
\]
and both supports belong to \(\mathcal D(\{s_i\}_{i=1}^N,m)\). If \(m<N/2\), then $\widetilde{\bm G}=\bm G$ and $\widetilde{\bm S}_i=\bm S_i$ for $i\in[N]$.
\end{proposition}

Proposition~\ref{prop:frequency_capped_identifiability} follows from
a coordinatewise majority principle. When \(m<N/2\), more than
half of the clients have no personalized departure at each
coordinate, so their transition coefficients agree with the
corresponding shared coefficient. This strict-majority value
uniquely determines that entry of \(\bm G\).

\begin{remark}[Sharpness of the frequency threshold]
\label{rem:frequency_cap_sharpness}
The condition \(m<N/2\) is sharp for uniform
shared--personalized identification based on the frequency
restriction alone. In general, when \(m\ge N/2\), distinct
shared baselines can generate the same client transition matrices
while satisfying the cap; see
Appendix~\ref*{subsec:app-global-local-identifiability}
of the Supplementary Material. The estimation guarantees in
Section~\ref{sec:4} use additional fitting margins and
compatibility conditions.
\end{remark}

Consider the data-generating decomposition
\(\bm A_{i,*}=\bm G_*+\bm S_{i,*}\).
Define the clientwise personalized sparsities and the maximum
cross-client frequency by
\[
    s_{i,*}
    :=
    |\supp(\bm S_{i,*})|,
    \quad\text{and}\quad
    m_*
    :=
    \max_{1\leq j,k\leq p}
    \sum_{i=1}^N
    \mathbf 1\{(\bm S_{i,*})_{jk}\neq0\}.
\]

\begin{assumption}[Exact personalization diversity]
\label{assump:frequency_cap}
The true personalized matrices satisfy \(m_*<N/2\).
\end{assumption}

By Proposition~\ref{prop:frequency_capped_identifiability},
Assumption~\ref{assump:frequency_cap} uniquely determines
\(\bm G_*\) and \(\{\bm S_{i,*}\}_{i=1}^N\) from
\(\{\bm A_{i,*}\}_{i=1}^N\) within the class satisfying the
frequency restriction \(m<N/2\). The quantity \(m_*\) measures the largest cross-client
overlap among the personalized supports. A smaller \(m_*\) restricts
each personalized coordinate to fewer clients and expresses greater
personalization diversity. The identification condition permits
\(m_*\) to increase with \(N\), provided that \(m_*<N/2\).

In statistical modeling, parsimonious models are desirable for
interpretability and estimation efficiency. Although requiring \(m_*\) to remain small or bounded independently of \(N\) offers one way to express parsimony, this could be overly restrictive. The count \(m_*\) treats all nonzero departures equally, so even negligible deviations can
contribute to a large support frequency. In practice, substantial personalized departures at a given coordinate may be concentrated in only a few clients, while other clients exhibit small or no departures. For illustration, suppose 60 of 100 retail stores follow standardized pricing and replenishment policies, with their weekly revenue in a particular product category modeled by a common own-lag coefficient. Minor local adjustments at 35 stores could produce small departures, whereas repeated bulk-purchase promotions at the remaining five might induce larger departures through stockpiling and subsequent sales declines.

To express such parsimony without requiring \(m_*\) to be small, we impose a cross-client weak sparsity condition that accounts for departure magnitudes. We refer to this condition as \emph{approximate personalization diversity}. Define \(\psi_0(x)=\mathbf 1\{x\neq0\}\) and
\(\psi_q(x)=|x|^q\) for \(q\in(0,1]\). 

\begin{assumption}[Approximate personalization diversity]
\label{assump:approx_personalization_diversity}
For a fixed \(q\in[0,1]\) and a budget \(m_{q,*}\geq0\), the
personalized matrices satisfy
\begin{equation}
\label{eq:cross_client_lq_diversity}
    \max_{1\leq j,k\leq p}
    \sum_{i=1}^N
    \psi_q\bigl((\bm S_{i,*})_{jk}\bigr)
    \leq m_{q,*}.
\end{equation}
\end{assumption}

For \(q=0\), \(m_{0,*}:=m_*\). For \(q>0\), Assumption~\ref{assump:approx_personalization_diversity}
provides a magnitude-based description of personalization
complexity. For any threshold \(u>0\), it implies
\[
    \max_{1\leq j,k\leq p}
    \#\left\{
    i\in[N]:|(\bm S_{i,*})_{jk}|>u
    \right\}
    \leq
    \min\left\{m_*,\frac{m_{q,*}}{u^q}\right\}.
\]
Thus, \(m_{q,*}/u^q\) bounds the number of departures exceeding
magnitude \(u\) at each coordinate. Numerous small departures
may contribute little to this budget even when their support
frequency is relatively large; see
Figure~\ref{fig:personalization_diversity}.

\begin{figure}[!htb]
    \centering
    \includegraphics[width=\textwidth]{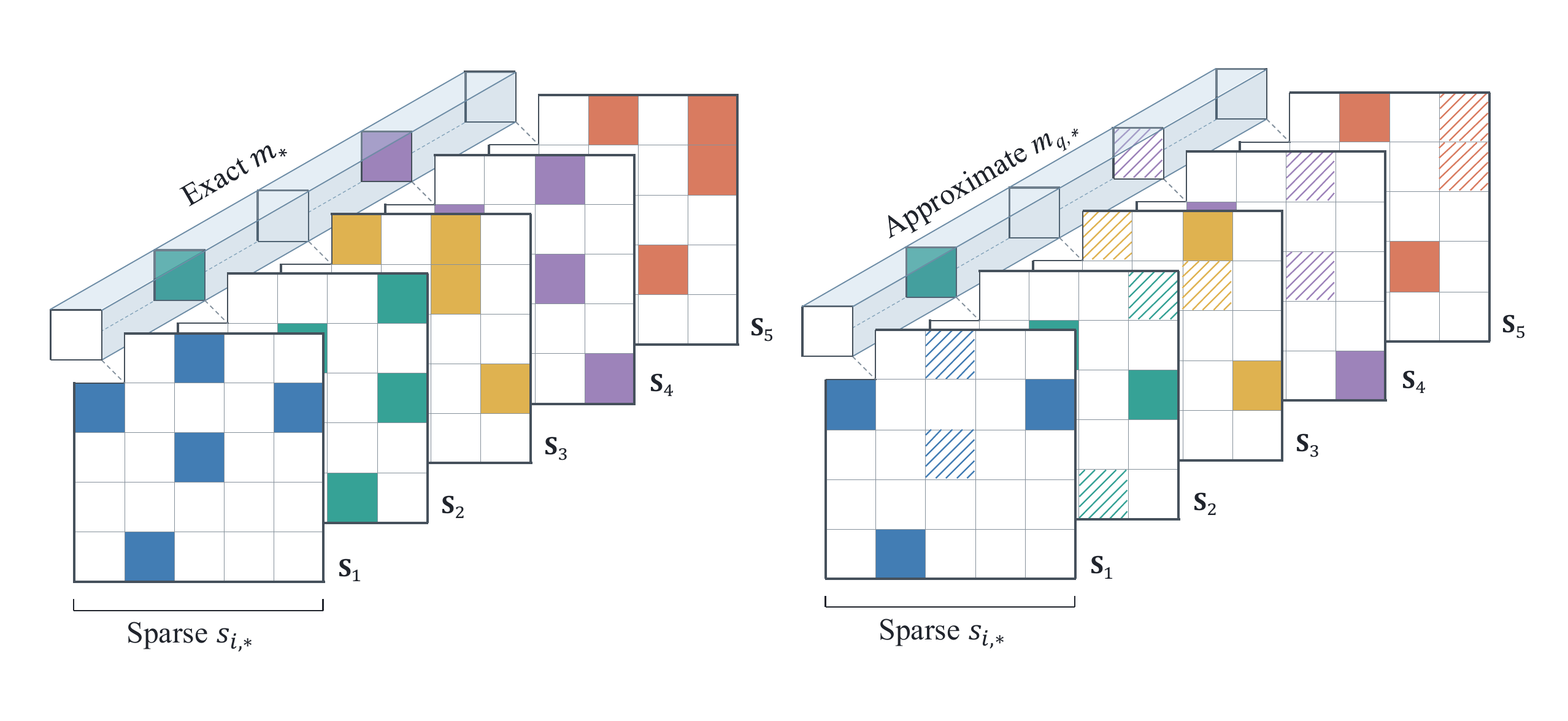}
    \vspace{-1cm}
    \caption{Illustration of exact (left) and approximate (right) personalization diversity.}
    \label{fig:personalization_diversity}
\end{figure}

Approximate personalization diversity also controls approximation
under a working frequency cap. For \(q>0\) and any integer
\(0\leq m<N\), retaining the \(m\) largest departure magnitudes
across clients at each coordinate produces a collection with
support frequency at most \(m\) and clientwise sparsities at most
\(s_{i,*}\). Every omitted coefficient has magnitude at most
\(\{m_{q,*}/(m+1)\}^{1/q}\). Here \(m_*\) and \(m_{q,*}\)
describe the population departures, while \(m\) is chosen for
fitting. Section~\ref{sec:3} restricts the cap to \(m<N/2\);
under approximate diversity, it may be smaller than \(m_*\).
The preceding approximation property underlies the error
\(\mathcal A_q(m)\) derived in Section~\ref{sec:4}.
Throughout, Assumption~\ref{assump:frequency_cap} identifies the
population shared and personalized components, and the additional
cross-client \(\ell_q\) condition controls approximation under
the cap $m$.

\subsection{Separation of the shared low-rank and sparse components}
\label{subsec:2.3}

The identified shared matrix \(\bm G_*\) admits the decomposition
\begin{equation}
\label{eq:global_lowrank_sparse}
    \bm G_*
    =
    \bm B_*
    +
    \bm S_{0,*},\quad \bm B_* := \bm L_*\bm R_*^\top
    \text{ with }
    \operatorname{rank}(\bm B_*)=r.
\end{equation}
As in standard low-rank-plus-sparse decomposition, separation requires the
low-rank component to be sufficiently diffuse and the sparse component to
be sufficiently localized. For a support
\(\Omega\subseteq[p]\times[p]\), let
\(\Omega_{j,\cdot}\) and \(\Omega_{\cdot,k}\) denote its entries in row
\(j\) and column \(k\), respectively.

\begin{definition}[\((s,\gamma)\)-sparse support class]
\label{def:row_column_sparse_class}
For \(s\in\{0,\ldots,p^2\}\) and \(\gamma\in[0,1]\), define
\[
    \mathcal C(s,\gamma)
    =
    \left\{
    \Omega\subseteq[p]\times[p]:
    |\Omega|\leq s,\quad
    \max_j|\Omega_{j,\cdot}|
    \vee
    \max_k|\Omega_{\cdot,k}|
    \leq \gamma p
    \right\}.
\]
A matrix \(\bm S\) is \((s,\gamma)\)-sparse if
\(\supp(\bm S)\in\mathcal C(s,\gamma)\).
\end{definition}

Let \(\bm B_* = \bm U_*\bbm\Sigma_*\bm V_*^\top\) be the compact rank-\(r\) singular value decomposition. The tangent space of the rank-\(r\) matrix manifold at \(\bm B_*\) is
\[
    \mathcal T_*
    =
    \left\{
    \bm U_*\bm X^\top+\bm Y\bm V_*^\top:
    \bm X,\bm Y\in\mathbb R^{p\times r}
    \right\}.
\]
Let \(\Omega_{0,*}:=\supp(\bm S_{0,*})\) and
\(s_{0,*}:=|\Omega_{0,*}|\). Denote by
\(\mathcal P_{\mathcal T_*}\) and
\(\mathcal P_{\Omega_{0,*}}\) the orthogonal projections onto
\(\mathcal T_*\) and the subspace of matrices supported on
\(\Omega_{0,*}\), respectively. We impose the following
standard low-rank--sparse separation condition
\citep{yi2016fast,zhang2018unified,cai2021learned,tong2021accelerating}.

\begin{assumption}[Shared low-rank--sparse separation]
\label{assump:shared_separation}
The singular vectors of \(\bm B_*\) satisfy
\[
    \|\bm U_*\|_{2,\infty}
    \leq
    \sqrt{\frac{\mu r}{p}},
    \qquad
    \|\bm V_*\|_{2,\infty}
    \leq
    \sqrt{\frac{\mu r}{p}},
\]
for some incoherence parameter $\mu \geq 1$. For the overall and row--column sparsity levels $s_{0,*}$ and $\gamma_*$, respectively,
\(
    \Omega_{0,*}\in\mathcal C(s_{0,*},\gamma_*).
\)
Moreover,
\(
    \left\|
    \mathcal P_{\Omega_{0,*}}
    \mathcal P_{\mathcal T_*}
    \right\|_{\mathrm{op}}
    <1.
\)
\end{assumption}

Assumption \ref{assump:shared_separation} prevents the low-rank and sparse components from being excessively concentrated on the same coordinates. The final transversality condition directly ensures that $\mathcal T_* \cap \left\{ \bm H:\supp(\bm H)\subseteq\Omega_{0,*}\right\}=\{\bm 0\}$. Consequently, \(\bm B_*\) and \(\bm S_{0,*}\) are locally identifiable within their respective structural classes.

The locally identified matrix \(\bm B_*\) has factor
representations related by invertible changes of basis. For any
\(\bm Q\in\mathrm{GL}(r)\),
\(
    \bm L_*\bm R_*^\top
    =(\bm L_*\bm Q)(\bm R_*\bm Q^{-\top})^\top
\).
We choose the balanced representative
\(\bm L_*=\bm U_*\bbm\Sigma_*^{1/2}\) and
\(\bm R_*=\bm V_*\bbm\Sigma_*^{1/2}\), and measure factor
estimation error modulo this equivalence in
Section~\ref{subsec:theoretical_setup}.

\section{Federated Estimation}
\label{sec:3}

We develop a federated procedure for estimating the shared and
personalized components in Section~\ref{sec:2}. The central
computational task is to enforce clientwise sparsity and the
cross-client frequency cap jointly. We introduce
\emph{Frequency-Capped Thresholding} (FCT) for this purpose and
combine it with scaled gradient updates for the shared low-rank
component and greedy capacity-aware thresholding for the shared
sparse component. The complete procedure consists of federated
initialization and joint refinement.

\subsection{Federated objective and fitting constraints}
\label{sec:3.1}

For client $i$, define the local loss
\begin{equation}
\label{eq:local_loss}
 F_i(\bm L,\bm R,\bm S_0,\bm S_i)
 =\frac1{2T_i}\sum_{t=1}^{T_i}
 \|\bm y_{i,t}-(\bm L\bm R^\top+\bm S_0+\bm S_i)
                       \bm y_{i,t-1}\|_2^2.
\end{equation}
With $w_i=T_i/\Ttot$, the pooled empirical loss is
$F(\bm L,\bm R,\bm S_0,\{\bm S_i\})
=\sum_{i=1}^Nw_iF_i(\bm L,\bm R,\bm S_0,\bm S_i)$.
Given fitting levels $r,s_0,\gamma,\{s_i\}$ and a working frequency cap
$0\le m<N/2$, we consider
\begin{equation}
\label{eq:federated_estimation_problem}
 \begin{aligned}
  \min_{\bm L,\bm R,\bm S_0,\{\bm S_i\}}\quad&F
   (\bm L,\bm R,\bm S_0,\{\bm S_i\})\\
  \text{subject to}\quad&
  \supp(\bm S_0)\in\mathcal C(s_0,\gamma),\qquad
  (\supp(\bm S_i))_{i=1}^N\in\mathcal D(\{s_i\},m).
 \end{aligned}
\end{equation}
These constraints define the fitted model class. Enforcing both clientwise sparsity and cross-client frequency
cap requires joint support selection, motivating the operator
introduced next.
\subsection{Frequency-capped thresholding}
\label{sec:3.2}
\label{sec:fct}

Separate hard thresholding at each client enforces clientwise
sparsity but may retain the same coordinate in more than \(m\)
clients. FCT enforces both restrictions through a joint weighted
projection. Given proposals \(\{\bm M_i\}_{i=1}^N\), define
\begin{equation}
\label{eq:FCT}
 \operatorname{FCT}_{\{s_i\},m}(\{\bm M_i\})
 \in\argmin_{(\supp(\bm S_i))_i\in\mathcal D(\{s_i\},m)}\sum_{i=1}^Nw_i\|\bm S_i-\bm M_i\|_{\mathrm F}^2.
\end{equation}

For fixed feasible supports \((\Omega_i)_{i=1}^N\), the
minimizer retains the proposal entries on those supports and
sets all remaining entries to zero. Thus, FCT selects supports
maximizing the retained weighted energy,
\(\sum_{i=1}^Nw_i
\sum_{(j,k)\in\Omega_i}|(\bm M_i)_{jk}|^2\).
This is a maximum-weight bipartite \(b\)-matching problem:
client nodes have capacities \(s_i\), and coordinate nodes
have capacity \(m\). Its linear programming relaxation has
an integral optimum, which yields an exact FCT projection.
Computational details are provided in
Appendix~\ref*{subsec:app-fct-computation}
of the Supplementary Material.

The working cap \(m\) controls how many personalized departures
can be retained at each coordinate. A smaller cap can attenuate
error transmitted from the personalized components to the shared
components. However, choosing \(m<m_*\) excludes the true
personalized collection from the fitted class and introduces
approximation error. For \(q>0\), the cross-client \(\ell_q\)
condition in Assumption~\ref{assump:approx_personalization_diversity}
controls approximation under a smaller cap.
Section~\ref{subsec:statistical_rates} quantifies the resulting
tradeoff in shared recovery between estimation and approximation.

\subsection{The PerFeCT-VAR algorithm}
\label{sec:perfect_algorithm}

We now combine FCT with updates for the shared components. At communication round \(l\), client \(i\) has access to the current shared parameters \(\bm L^{(l)},\bm R^{(l)},\bm S_0^{(l)}\) and its personalized component \(\bm S_i^{(l)}\). Define \(\bm A_i^{(l)}=\bm L^{(l)}(\bm R^{(l)})^\top+\bm S_0^{(l)}+\bm S_i^{(l)}\).  Client \(i\) computes the gradient with respect to its transition matrix,
\begin{equation}
\label{eq:local_gradient}
\bm D_i^{(l)}
=
\frac{1}{T_i}
\sum_{t=1}^{T_i}
\left(
\bm A_i^{(l)}\bm y_{i,t-1} - \bm y_{i,t}
\right)
\bm y_{i,t-1}^\top.
\end{equation}
For its personalized component, client \(i\) forms the proposal
\begin{equation}
\label{eq:local_Si}
 \bm S_i^{(l+1/2)}
 =
 \bm S_i^{(l)}-\eta_{\mathrm P}\bm D_i^{(l)}.
\end{equation}
Together with FCT, this gives a projected-gradient step for
the pooled loss in the weighted product norm underlying
\eqref{eq:FCT}.
Client \(i\) computes the scores
\(q_{i,jk}^{(l)}
=|(\bm S_i^{(l+1/2)})_{jk}|^2\)
and uploads \(\bm D_i^{(l)}\) and all \(p^2\) scores to
the server.

The server aggregates $\overline{\bm D}^{(l)}=\sum_{i=1}^Nw_i\bm D_i^{(l)}$
and performs the ScaledGD updates \citep{tong2021accelerating}
\begin{equation}
\label{eq:server_LR}
 \begin{aligned}
  \bm L^{(l+1)}&=\bm L^{(l)}-
  \eta_{\mathrm{LR}}\overline{\bm D}^{(l)}\bm R^{(l)}
       \{\bm R^{(l)\top}\bm R^{(l)}\}^{-1},\\
  \bm R^{(l+1)}&=\bm R^{(l)}-
  \eta_{\mathrm{LR}}\overline{\bm D}^{(l)\top}\bm L^{(l)}
       \{\bm L^{(l)\top}\bm L^{(l)}\}^{-1}.
 \end{aligned}
\end{equation}
The inverse Gram matrices precondition the factor updates.

For the shared sparse component, let
\(\mathcal G_{s_0,\gamma}\) denote the following greedy
thresholding operator. It scans entries in nonincreasing
order of magnitude and retains an entry only if the resulting
support has at most \(s_0\) entries and at most
\(\lfloor\gamma p\rfloor\) entries in each row and column.
The server updates
\begin{equation}
\label{eq:global_sparse_thresholding_update}
\bm S_0^{(l+1)}
=
\mathcal G_{s_0,\gamma}
\left(
\bm S_0^{(l)}
-
\eta_\mathrm{S}\overline{\bm D}^{(l)}
\right).
\end{equation}

Finally, the server solves the FCT allocation using scores
$w_iq_{i,jk}^{(l)}$. It broadcasts the updated shared
parameters and sends the selected support
\(\widehat\Omega_i^{(l+1)}\) to client \(i\). Each client then completes its update by
\begin{equation}
\label{eq:client_HT}
 \bm S_i^{(l+1)}=
 \mathcal P_{\widehat\Omega_i^{(l+1)}}(\bm S_i^{(l+1/2)}).
\end{equation}
All update directions are computed from the round-\(l\)
parameters. Each round consists of one upload and
one download per client.
Algorithm~\ref{alg:perfect_var} summarizes the procedure.

\begin{algorithm}[htb]
\caption{PerFeCT-VAR}
\label{alg:perfect_var}
\label{alg:exact_HT_FCHT}
\small
\begin{algorithmic}[1]
\Require Initial values
$\bm L^{(0)},\bm R^{(0)},\bm S_0^{(0)},\{\bm S_i^{(0)}\}_{i=1}^N$,
fitting levels $r,s_0,\gamma,\{s_i\}_{i=1}^N$ and a working frequency cap $m$.
Step sizes $\eta_{\mathrm{LR}},\eta_\mathrm{S},\eta_{\mathrm P}$ and number of rounds $K$.
\State Set $w_i=T_i/\Ttot$ and broadcast the shared initial parameters.
\For{$l=0,\ldots,K-1$}
 \State Each client computes \eqref{eq:local_gradient} and
 \eqref{eq:local_Si}, then uploads its gradient and squared proposal scores.
 \State The server aggregates the gradients and updates the shared factors
 and sparse component by \eqref{eq:server_LR} and
 \eqref{eq:global_sparse_thresholding_update}.
 \State The server performs FCT support selection for \eqref{eq:FCT}
 using the weighted scores $w_iq_{i,jk}^{(l)}$.
 \State The server sends the shared updates and each client's support.
 \State Each client applies \eqref{eq:client_HT}.
\EndFor
\Ensure $\widehat{\bm B}=\bm L^{(K)}\bm R^{(K)\top}$,
$\widehat{\bm S}_0=\bm S_0^{(K)}$,
and $\widehat{\bm S}_i=\bm S_i^{(K)}$.
\end{algorithmic}
\end{algorithm}

\subsection{Federated initialization}
\label{sec:algorithm_initialization}

We initialize the shared and personalized components by combining
local convex fits with cross-client aggregation and FCT.
Following \citet{basu2019low}, each client first estimates a
low-rank component and a combined sparse component by solving
\begin{equation}
\label{eq:init_local_basu}
 \begin{split}
 (\widetilde{\bm B}_i,\widetilde{\bm C}_i)
 \in\argmin_{\|\bm B\|_{\max}\le\varpi/p}
 \bigg\{&\frac1{2T_i}\sum_{t=1}^{T_i}
       \|\bm y_{i,t}-(\bm B+\bm C)\bm y_{i,t-1}\|_2^2+\lambda_{B,i}\|\bm B\|_*+\lambda_{C,i}\|\bm C\|_1\bigg\}.
 \end{split}
\end{equation}
The local targets are \(\bm B_*\) and
\(\bm C_{i,*}=\bm S_{0,*}+\bm S_{i,*}\).
The constraint \(\|\bm B\|_{\max}\le\varpi/p\) limits
the entrywise concentration of the low-rank fit and helps
control its ambiguity with the sparse component.
Section~\ref{subsec:initializer_theory} specifies sufficient
choices of \(\varpi\), \(\lambda_{B,i}\), and
\(\lambda_{C,i}\) for initialization.

Then, the client takes a rank-$r$ singular value truncation,
\begin{equation}
\label{eq:init_local_rank_compression}
 \operatorname{SVD}_r(\widetilde{\bm B}_i)
 =\widetilde{\bm U}_i\widetilde{\bbm\Sigma}_i\widetilde{\bm V}_i^\top
 =\widetilde{\bm L}_i\widetilde{\bm R}_i^\top,
\end{equation}
where
\(\widetilde{\bm L}_i
=\widetilde{\bm U}_i\widetilde{\bbm\Sigma}_i^{1/2}\)
and
\(\widetilde{\bm R}_i
=\widetilde{\bm V}_i\widetilde{\bbm\Sigma}_i^{1/2}\), and uploads these factors together with
$\widetilde{\bm C}_i$ and $T_i$. Transmitting the factorized low-rank estimate requires \(2pr\) entries rather than \(p^2\). The server averages the reconstructed
matrices, $\overline{\bm B}=\sum_{i=1}^Nw_i
\widetilde{\bm L}_i\widetilde{\bm R}_i^\top$, and sets
$\bm L^{(0)}=\bm U^{(0)}(\bbm\Sigma^{(0)})^{1/2}$ and
$\bm R^{(0)}=\bm V^{(0)}(\bbm\Sigma^{(0)})^{1/2}$ from
$\operatorname{SVD}_r(\overline{\bm B})$.

Under Assumption~\ref{assump:frequency_cap}, a strict majority
of clients satisfies
\((\bm C_{i,*})_{jk}=(\bm S_{0,*})_{jk}\)
at every coordinate. Hence, the coordinatewise
median of \(\{\bm C_{i,*}\}_{i=1}^N\) equals
\(\bm S_{0,*}\). Motivated by this majority property, 
the server forms
\begin{equation}
\label{eq:init_shared_median}
 (\widetilde{\bm S}_0)_{jk}
 =\operatorname{median}_{i\in[N]}(\widetilde{\bm C}_i)_{jk},
 \qquad
 \bm S_0^{(0)}=
       \mathcal G_{s_0,\gamma}(\widetilde{\bm S}_0),
\end{equation}
and initializes personalized components by
\begin{equation}
\label{eq:init_personalized_fct}
 \{\bm S_i^{(0)}\}=
 \operatorname{FCT}_{\{s_i\},m}
          (\{\widetilde{\bm C}_i-\bm S_0^{(0)}\}).
\end{equation}
The server broadcasts the shared parameters and sends each client
its personalized initializer.
The complete procedure therefore has two stages: federated initialization
followed by \(K\) rounds of joint refinement using
Algorithm~\ref{alg:perfect_var}.
Figure~\ref{fig:fed_protocol} summarizes this procedure.

\begin{figure}[!htp]
\centering
\includegraphics[width=0.99\textwidth]{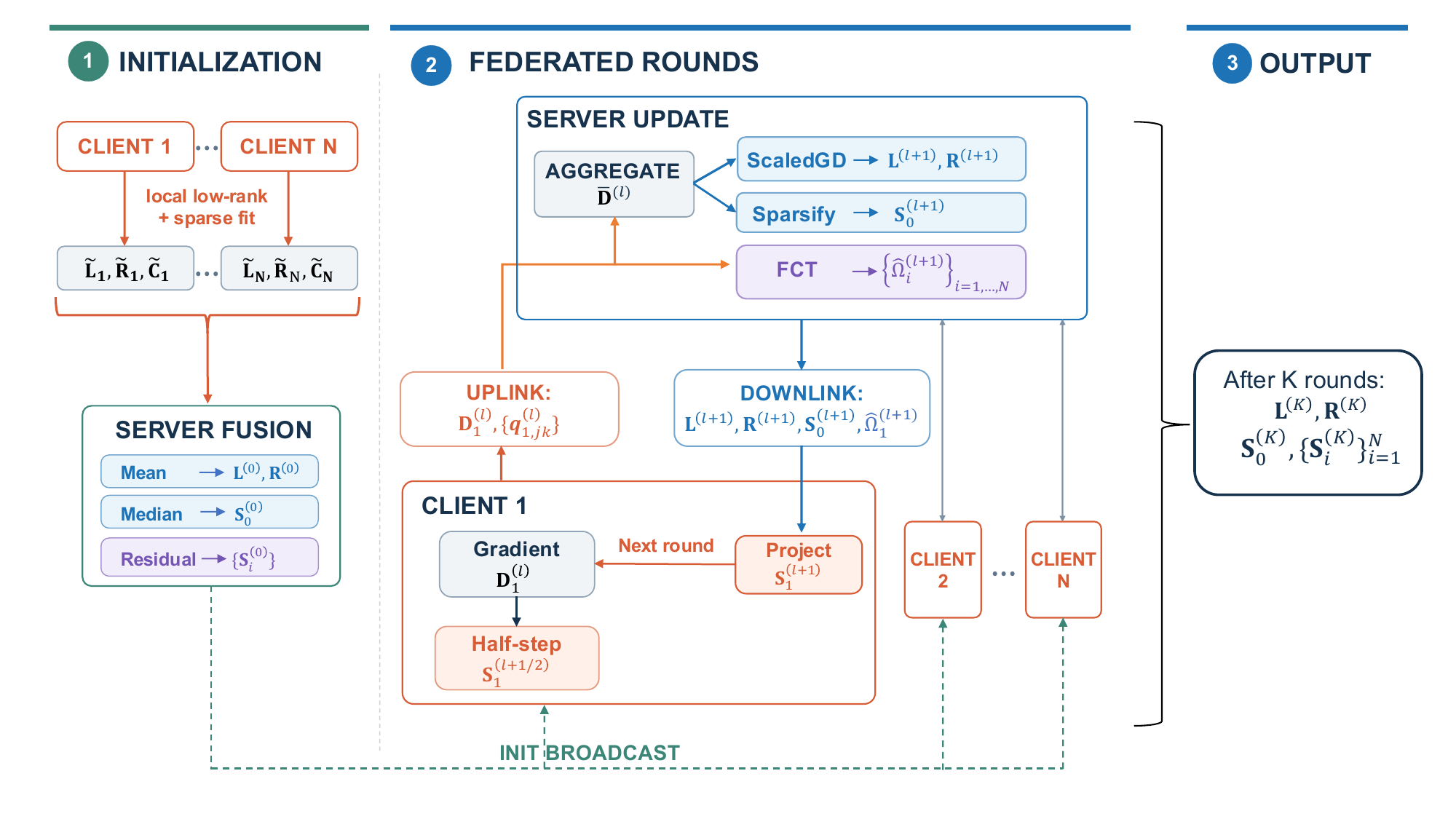}
\vspace{-0.5cm}
\caption{Complete PerFeCT-VAR estimation protocol.}
\label{fig:fed_protocol}
\end{figure}

\subsection{Communication-efficient implementation}
\label{sec:communication}

At each communication round of Algorithm~\ref{alg:perfect_var}, client
\(i\) uploads the dense \(p\times p\) gradient \(\bm D_i^{(l)}\) for the
shared updates and all \(p^2\) FCT scores for personalized support selection.
We now discuss how this uplink communication cost can be reduced.

First, when updating $\bm S_i$, a lossless screening rule can be applied to allow each client to transmit only a subset of its largest FCT scores. When \(s_i=s\), the score transmission can be reduced from \(p^2\) to \(\mathcal O(\min\{Ns/m,p^2\})\) indexed values. Second, for the shared sparse update, we consider sparse-gradient compression with error feedback \citep{stich2018sparsified}, allowing each client to transmit at most \(\widetilde{s}_0\) gradient entries for this update.
Meanwhile, for the shared low-rank update, each client uploads only the two
\(p\times r\) factor gradients \(\bm D_i^{(l)}\bm R^{(l)}\) and
\((\bm D_i^{(l)})^\top\bm L^{(l)}\).
Transmitting the compressed gradient together with these factor gradients
avoids uploading the full \(p\times p\) gradient matrix. Finally, we consider less frequent FCT support updates to further
reduce score uploads.

FCT screening and factor-gradient transmission are exact, whereas gradient compression and less frequent FCT support updates are approximate. Detailed screening results, communication complexities, and the communication-efficient algorithm are provided in Appendix~\ref*{app:reduced-communication} of the Supplementary Material. The theoretical results in Section~\ref{sec:4} apply to the exact \textsc{PerFeCT-VAR} algorithm in
Algorithm~\ref{alg:perfect_var}.

\section{Theoretical Analysis}
\label{sec:4}
\noeqref{eq:scaled_factor_error,eq:directional_heterogeneity_levels}

We establish joint linear convergence, statistical accuracy, and
initialization guarantees for \textsc{PerFeCT-VAR}, focusing on how
personalization diversity and the working cap $m$ determine shared
and personalized recovery under heterogeneous client distributions.
The resulting bounds distinguish the gains from pooling and the
cost of personalization, and quantify the approximation--estimation
tradeoff in shared recovery.

To emphasize these statistical implications, we present simplified
consequences of the full quantitative results in
Appendix~\ref*{app:full-theory} of the Supplementary Material, which
gives detailed parameter dependence and the precise fitting,
compatibility, sample-size, and local-basin conditions.
We absorb detailed curvature factors into the indexed constants
defined below, while the balanced-client rate illustrations
additionally suppress the distributional scales specified there.

\subsection{Sampling, curvature, and error measures}
\label{subsec:theoretical_setup}

Set $w_i=T_i/\Ttot$, $w_{\max}=\max_{i\in[N]}w_i$, and
$\overline T=\Ttot/N$. We adopt the following regularity assumptions.

\begin{assumption}[Sampling and covariance regularity]
\label{assump:gaussian_sampling}
\label{assump:regular_covariance}
The stationary client trajectories are independent. The innovations
are independent across clients and time, with
$\bbm\epsilon_{i,t}\sim N(\bm0,\bbm\Sigma_{\epsilon,i})$.
For every client,
$\phi_{\min}\bm I_p\preceq\bbm\Gamma_i
\preceq\phi_{\max}\bm I_p$, where
$0<\phi_{\min}\le\phi_{\max}<\infty$.
\end{assumption}

\begin{assumption}[Covariance-adjusted incoherence]
\label{assump:cov_adjusted_incoherence}
Let $\bbm\Gamma_w=\sum_{i=1}^Nw_i\bbm\Gamma_i$.
For some $\widetilde\phi_w>0$, $\|\bbm\Gamma_w\bm V_*\|_{2,\infty}
 \le \widetilde\phi_w\sqrt{\mu r/p}$.
\end{assumption}

Assumption~\ref{assump:gaussian_sampling} permits heterogeneous
client distributions while providing common bounds on their
stationary covariance matrices. Gaussianity and independence
across clients allow us to establish finite-sample concentration
bounds for the sample covariance and score matrices.
The estimation procedure itself does not require Gaussianity.
Let
$\tau^2=\max_i\|\bbm\Sigma_{\epsilon,i}\|_{\mathrm{op}}$
denote the innovation scale.
Under Assumptions~\ref{assump:stationarity}
and~\ref{assump:gaussian_sampling}, the stable VAR spectral
representation yields a common bound $M$ on the covariance
operator norm of each stacked predictor trajectory, uniformly
over trajectory lengths. This quantity accounts for temporal dependence
in the concentration bounds, and its explicit expression and the
associated bounds are given in
Appendix~\ref*{subsec:prob-sample-geometry}, following
\citet{basu2015regularized}. Assumption~\ref{assump:cov_adjusted_incoherence} controls the rowwise concentration of the shared right singular vectors after
weighting by the pooled covariance. Together with the incoherence
condition in Assumption~\ref{assump:shared_separation}, it bounds
the interaction between the shared low-rank and sparse components
in the population gradients.

The covariance matrices also determine the blockwise population
curvature and smoothness coefficients
$(\alpha_j,\beta_j)$, $j\in\{\mathrm{LR},\mathrm{S},\mathrm P\}$.
For the low-rank block, set
$\alpha_{\mathrm{LR}}=\lambda_{\min}(\bbm\Gamma_w)$ and
$\beta_{\mathrm{LR}}=\lambda_{\max}(\bbm\Gamma_w)$.
For the shared and personalized sparse blocks, use the
corresponding restricted quadratic-form bounds for
$\bbm\Gamma_w$ and $\{\bbm\Gamma_i\}_{i=1}^N$ over fixed
enlarged support classes, with the personalized bounds taken
uniformly over clients.
The precise definitions are given in
Appendix~\ref*{subsec:app-full-curvature}.
Assumption~\ref{assump:regular_covariance} ensures
$\phi_{\min}\le\alpha_j\le\beta_j\le\phi_{\max}$
for each active block.
Collect these coefficients as
\begin{equation}
\label{eq:curvature_tuple}
 \nu=(\alpha_{\mathrm{LR}},\beta_{\mathrm{LR}},
       \alpha_{\mathrm S},\beta_{\mathrm S},\alpha_{\mathrm P},\beta_{\mathrm P},
       \phi_{\max},
       \phi_{\min}).
\end{equation}
In the following analysis, $c_\nu$, $C_\nu$, and
$\lesssim_\nu$ absorb dependence only on this curvature tuple.
Constants may differ between displays and are fixed throughout
each specified recursion.

Let $\sigma_1\ge\sigma_r>0$ be the extreme singular values of $\bm B_*$.
Using the balanced factors from Section~\ref{subsec:2.3}, define
\begin{equation}
\label{eq:scaled_factor_error}
 e_{\mathrm{LR}}^{(l)}=
 \inf_{\bm Q\in\mathrm{GL}(r)}
 \left\{
 \|(\bm L^{(l)}\bm Q-\bm L_*)\bbm\Sigma_*^{1/2}\|_{\mathrm F}^2
 +\|(\bm R^{(l)}\bm Q^{-\top}-\bm R_*)\bbm\Sigma_*^{1/2}\|_{\mathrm F}^2
 \right\}.
\end{equation}
The alignment accounts for the non-uniqueness of the
factorization, and the scaled factor error controls the
corresponding low-rank matrix error in the local neighborhood.
The shared sparse and personalized errors and the joint
error vector are
\begin{equation}
\label{eq:componentwise_errors}
 \begin{aligned}
  e_{\mathrm S}^{(l)}
  &=
  \|\bm S_0^{(l)}-\bm S_{0,*}\|_{\mathrm F}^2,
  \quad
  e_{\mathrm P}^{(l)}
  =
  \sum_{i=1}^Nw_i\|\bm S_i^{(l)}-\bm S_{i,*}\|_{\mathrm F}^2,
  \quad
  \bm e^{(l)}=
  (e_{\mathrm{LR}}^{(l)},e_{\mathrm S}^{(l)},
   e_{\mathrm P}^{(l)})^\top.
 \end{aligned}
\end{equation}
Define the weighted aggregate error
$V^{(l)}=\bbm\ell_\nu^\top\bm e^{(l)}$, where
$\bbm\ell_\nu=
(1,\beta_{\mathrm S}/\beta_{\mathrm{LR}},
\beta_{\mathrm P}/\beta_{\mathrm{LR}})^\top$.

\subsection{Compatibility and approximation control}
\label{subsec:covariance_compatibility}

A shared perturbation $\bm X$ interacts with personalized
perturbations $\{\bm H_i\}$ through
$\sum_{i=1}^Nw_i\langle\bm H_i,\bm X\bbm\Gamma_i\rangle$
in the population loss. This interaction depends on support
recurrence, covariance heterogeneity, and sample-size weights.

Fix an integer $1\le m<N/2$. Define
$s_0^\dagger=\min\{2s_0+s_{0,*},p^2\}$,
$\gamma^\dagger=\min\{2\gamma+\gamma_*,1\}$, and
$s_i^\dagger=\min\{2s_i+s_{i,*},p^2\}$.
Define the shared sparse direction class by
\[
 \mathcal X_{\mathrm S}
 =
 \left\{
 \bm X\in\mathbb R^{p\times p}:
 \|\bm X\|_{\mathrm F}=1,\quad
 \supp(\bm X)\in\mathcal C(s_0^\dagger,\gamma^\dagger)
 \right\}.
\]
For a sufficiently small curvature-dependent radius $\varepsilon_\nu>0$, let
$\mathcal N_{\varepsilon_\nu}$ be the set of pairs
$(\widetilde{\bm U},\widetilde{\bm V})$ of $p\times r$
matrices with orthonormal columns satisfying
\[
 \|\widetilde{\bm U}\widetilde{\bm U}^{\top}
       -\bm U_*\bm U_*^{\top}\|_{\mathrm{op}}
 \vee
 \|\widetilde{\bm V}\widetilde{\bm V}^{\top}
       -\bm V_*\bm V_*^{\top}\|_{\mathrm{op}}
 \le 2\varepsilon_\nu.
\]
Define the shared low-rank direction class by
\[
 \begin{split}
 \mathcal X_{\mathrm{LR}}
 =\bigl\{\bm X:\;&
 \|\bm X\|_{\mathrm F}=1, 
 \bm X=\bm K_L\widetilde{\bm V}^{\top}
       +\widetilde{\bm U}\bm K_R^{\top},(\widetilde{\bm U},\widetilde{\bm V})
       \in\mathcal N_{\varepsilon_\nu}, 
 \bm K_L,\bm K_R\in\mathbb R^{p\times r}
 \bigr\}.
 \end{split}
\]
In addition, define the personalized direction class as
\[
 \mathcal H_{\mathrm P}(m)
 =
 \left\{(\bm H_i)_i:
 (\supp(\bm H_i))_i\in\mathcal D(\{s_i^\dagger\},m),\quad
 \sum_{i=1}^N\|\bm H_i\|_{\mathrm F}^2=1\right\}.
\]
These classes capture the structured directions accessible to our federated algorithm and are used to quantify cross-component error interactions.

Let $\overline{\bbm\Gamma}=N^{-1}\sum_{i=1}^N\bbm\Gamma_i$ be the
unweighted mean covariance and write
$\bm E_i=\bbm\Gamma_i-\overline{\bbm\Gamma}$.
Define the directional covariance discrepancies by
\begin{equation}
\label{eq:directional_heterogeneity_levels}
 \Lambda_a(m)
 =
 \frac1{\sqrt m}
 \sup_{\substack{\bm X\in\mathcal X_a\\
                 (\bm H_i)_i\in\mathcal H_{\mathrm P}(m)}}
 \left|\sum_{i=1}^N\langle\bm H_i,\bm X\bm E_i\rangle\right|,
 \qquad a\in\{\mathrm{LR},\mathrm S\}.
\end{equation}
These quantities measure covariance variation along the relevant
structured directions. The normalization by $\sqrt m$ separates
this directional covariance effect from the accumulation allowed
by the frequency cap. They vanish when all clients have the same stationary covariance
and can remain uniformly bounded over working caps even when
the covariance matrices differ substantially.
Appendix~\ref*{subsec:app-distributional-discrepancy} of the
Supplementary Material provides several representative examples.

Using the quantities above, we define interaction coefficients
to characterize error propagation across components in the joint
updates. For the shared--personalized interactions, let
\begin{equation}
\label{eq:main_local_to_global_couplings}
 \begin{aligned}
 L_{\mathrm{LR},\mathrm P}(m)
 &=
 m w_{\max}\{\phi_{\max}^2+\Lambda_{\mathrm{LR}}(m)^2\}+
 \frac{M^2}{\Ttot}
 \left(pr+\sum_{i=1}^N s_i^\dagger\log p\right),\\
 L_{\mathrm S,\mathrm P}(m)
 &=
 m w_{\max}\{\phi_{\max}^2+\Lambda_{\mathrm S}(m)^2\}+
 \frac{M^2}{\Ttot}
 \left(s_0^\dagger\log p+\sum_{i=1}^N s_i^\dagger\log p\right).
 \end{aligned}
\end{equation}
The two terms account for cross-component error transmission
induced by the population covariance structure and empirical
covariance fluctuations, respectively. In addition, for the shared low-rank--sparse interaction, define
\begin{equation}
\label{eq:shared_interaction_scale}
 \begin{split}
 L_{\mathrm{LR},\mathrm S}
 ={}&
 \gamma^\dagger\mu r
 \bigl(\widetilde\phi_w^2+\beta_{\mathrm{LR}}^2\bigr)+
 \frac{\varepsilon_\nu^2}{(1-\varepsilon_\nu)^2}
 \beta_{\mathrm{LR}}^2
 +\frac{M\phi_{\max}p}{\Ttot}.
 \end{split}
\end{equation}
The three terms capture population low-rank--sparse interaction,
variation within the local subspace neighborhood, and empirical
covariance fluctuations, respectively.

These interaction coefficients enter the off-diagonal terms of
the contraction matrix in Theorem~\ref{thm:main_fed_iht}.
We impose the following assumption
to control cross-component interactions.
\begin{assumption}[Cross-component compatibility]
\label{cond:population_compatibility}
For a sufficiently small constant $c_\nu>0$ depending only on $\nu$,
require
\[
 \gamma_*
 \le
 \frac{c_\nu\alpha_{\mathrm S}^2}
 {\beta_{\mathrm S}^2\mu r
  (\widetilde\phi_w^2+\beta_{\mathrm{LR}}^2)}\quad \text{and}\quad m w_{\max}
 \left\{\phi_{\max}^2+
 \max_{a\in\{\mathrm{LR},\mathrm S\}}\Lambda_a(m)^2\right\}
 \le c_\nu.
\]
\end{assumption}
With $\gamma\asymp(\beta_{\mathrm S}/\alpha_{\mathrm S})^2\gamma_*$, this assumption,
together with the sample size and initialization conditions of
Theorem~\ref{thm:main_fed_iht}, ensures that
$L_{\mathrm{LR},\mathrm S}$, $L_{\mathrm{LR},\mathrm P}(m)$,
and $L_{\mathrm S,\mathrm P}(m)$ are sufficiently small for joint contraction.

\begin{remark}[Covariance heterogeneity and error propagation]
\label{rem:covariance_heterogeneity}
Through $L_{\mathrm{LR},\mathrm P}(m)$ and
$L_{\mathrm S,\mathrm P}(m)$, the directional discrepancies
$\Lambda_{\mathrm{LR}}(m)$ and $\Lambda_{\mathrm S}(m)$ govern
how covariance heterogeneity affects the transmission of
personalized estimation and approximation errors into the shared
components. Writing $T_{\max}=\max_iT_i$, we have
\[
 m w_{\max}
 =\frac mN\frac{T_{\max}}{\overline T}.
\]
Here $m/N$ measures the fitting personalization diversity, and
$T_{\max}/\overline T$ measures sample size imbalance.
In the interaction terms, a smaller working cap enforces greater personalization diversity, and can therefore  accommodate larger directional covariance discrepancies or greater sample-size imbalance. This characterization accommodates
non-vanishing personalized effects and innovation covariance differences.
The corresponding approximation-transfer bounds are established
in Appendix~\ref*{app:lq-deterministic}.
\end{remark}

The cross-client $\ell_q$ budget controls approximation of the
identified personalized matrices.
Let $C_{\mathrm P}(\nu)>1$ be a sufficiently large fixed projection
factor and set $\overline s_w=\sum_{i=1}^Nw_i s_{i,*}$.
The analysis uses a population comparison retaining the
$\lfloor m/C_{\mathrm P}(\nu)\rfloor$ largest personalized
coefficients in magnitude across clients at each coordinate,
with fixed tie-breaking. The stricter comparison cap supplies
the frequency margin used in the FCT analysis.
Its weighted squared approximation error is bounded by
\begin{equation}
\label{eq:lq_approximation_envelope}
 \mathcal A_q(m)=
 \begin{cases}
 0,&q=0,\quad m\ge C_{\mathrm P}(\nu)m_*,\\[3pt]
 \displaystyle
 \overline s_w
 \left(\dfrac{C_{\mathrm P}(\nu)m_{q,*}}m\right)^{2/q},
 &0<q\le1.
 \end{cases}
\end{equation}
For fixed $q>0$ and $\nu$, this envelope is nonincreasing in $m$.
At $q=0$, we require $C_{\mathrm P}(\nu)m_*\le m<N/2$
together with Assumption~\ref{cond:population_compatibility}, so that the comparison
retains the full personalized support. The same interaction coefficients control the transmission
of approximation error to the shared components; see
Appendix~\ref*{app:lq-deterministic}.

\subsection{Joint error contraction}
\label{subsec:main_theory}

Choose fitting levels $s_0\asymp(\beta_{\mathrm S}/\alpha_{\mathrm S})^2s_{0,*}$,
$\gamma\asymp(\beta_{\mathrm S}/\alpha_{\mathrm S})^2\gamma_*$,
and $s_i\asymp(\beta_{\mathrm P}/\alpha_{\mathrm P})^2s_{i,*}$
for $i\in[N]$. Take $\eta_{\mathrm{LR}}=c/\beta_{\mathrm{LR}}$,
$\eta_{\mathrm S}=c/\beta_{\mathrm S}$, and
$\eta_{\mathrm P}=c/\beta_{\mathrm P}$, where $c>0$ is
sufficiently small.
We impose the following sample-size
requirements for joint error contraction:
\begin{equation}
\label{eq:main_sample_size_conditions}
 \Ttot\gtrsim_\nu M^2(pr+s_{0,*}\log p),\qquad
 T_i\gtrsim_\nu M^2\left(s_{i,*}\log p+\log(pN)\right),
\end{equation}
and
\begin{equation}
\label{eq:main_forcing_sample_condition}
 \Ttot\gtrsim_\nu\frac{\tau^2}{\sigma_r^2}
 \left[
 pr+\left(s_{0,*}+\sum_{i=1}^N s_{i,*}\right)\log p+N\log N
 \right].
\end{equation}

\begin{theorem}[Joint contraction]
\label{thm:main_fed_iht}
Fix $q\in[0,1]$. Suppose Assumptions~\ref{assump:stationarity}--\ref{cond:population_compatibility} hold, and use
the fitting levels, step sizes, and sample sizes specified above.
Choose a working cap $m$ satisfying
\begin{equation}
\label{eq:main_forcing_cap_condition}
 C_{\mathrm P}(\nu)m_{q,*}
 \left\{
 \frac{\overline s_w}{\widetilde c_\nu\sigma_r^2}
 \right\}^{q/2}
 \le m<\frac N2,
\end{equation}
where $\widetilde c_\nu>0$ is sufficiently small and depends only on $\nu$.
Assume a feasible initializer satisfying
$V^{(0)}\le \varepsilon_\nu^2\sigma_r^2$.
There exist $c_\nu,\rho_\nu\in(0,1)$ and $C_\nu>0$,
depending only on $\nu$, such that, with probability at least
$1-Ce^{-cp}-Cp^{-c}-2/(pN)$, simultaneously for all $l\ge0$,
\begin{equation}
\label{eq:main_vector_recursion}
 \bm e^{(l+1)}
 \leq\bm\Theta_\nu\bm e^{(l)}+\bm f,
\end{equation}
where $\leq$ denotes componentwise comparison when applied to
vectors, and the nonnegative matrix $\bm\Theta_\nu$ is given by
\begin{equation}
\label{eq:main_contraction_matrix}
 \bm\Theta_\nu=
 \begin{pmatrix}
  1-c_\nu
  &C_\nu L_{\mathrm{LR},\mathrm S}
  &C_\nu L_{\mathrm{LR},\mathrm P}(m)\\
  C_\nu L_{\mathrm{LR},\mathrm S}
  &1-c_\nu
  &C_\nu L_{\mathrm S,\mathrm P}(m)\\
  C_\nu L_{\mathrm{LR},\mathrm P}(m)
  &C_\nu L_{\mathrm S,\mathrm P}(m)
  &1-c_\nu
 \end{pmatrix},
\end{equation}
and satisfies $\rho(\bm\Theta_\nu)<1$.
Here $\bm f=\bm f_{\mathrm{stat}}+\bm f_{\mathrm{app},q}(m)$, with
\begin{equation}
\label{eq:ap-main-forcing}
 \begin{aligned}
 \bm f_{\mathrm{stat}}
 &=C_\nu\tau^2
 \left(
 pr/\Ttot,\,
 s_{0,*}\log p/\Ttot,\,
 \sum_{i=1}^N w_i(s_{i,*}\log p+\log N)/T_i
 \right)^\top,\\
 \bm f_{\mathrm{app},q}(m)
 &=C_\nu\mathcal A_q(m)
 \left(
 L_{\mathrm{LR},\mathrm P}(m),\,
 L_{\mathrm S,\mathrm P}(m),\,
 \phi_{\max}^2
 \right)^\top.
 \end{aligned}
\end{equation}
Consequently, for all $l\ge0$,
\begin{equation}
\label{eq:main_vector_closed_form}
 \bm e^{(l)}
 \leq
 \bm\Theta_\nu^l\bm e^{(0)}
 +(\bm I_3-\bm\Theta_\nu)^{-1}\bm f.
\end{equation}
\end{theorem}

The lower bound on $m$ in \eqref{eq:main_forcing_cap_condition}
keeps the approximation error sufficiently small to prevent it
from driving the iterates outside the local convergence neighborhood. The diagonal entries of $\bm\Theta_\nu$ describe within-component
contraction, while its off-diagonal entries describe error
propagation across components. The compatibility and sample-size
conditions ensure joint contraction. On the same event,
\[
 V^{(l)}
 \le (1-\rho_\nu)^lV^{(0)}
 +\rho_\nu^{-1}\bbm\ell_\nu^\top\bm f.
\]
Thus, the initialization contribution decays geometrically,
while $(\bm I_3-\bm\Theta_\nu)^{-1}\bm f$ bounds the componentwise
statistical and approximation error floors after accounting
for cross-component propagation. Reducing the initialization
contribution to this accuracy scale requires only logarithmically
many communication rounds. Quantitative statements and proofs are provided in
Appendices~\ref*{subsec:app-unified-main-theorem}
and~\ref*{app:main-theorem-proof}.

\subsection{Statistical accuracy and the diversity tradeoff}
\label{subsec:statistical_rates}

Denote $D_G=pr+s_{0,*}\log p$, $D_P=N^{-1}\sum_{i=1}^N s_{i,*}\log p+\log N$, and $\omega(m)=L_{\mathrm{LR},\mathrm P}(m)\vee L_{\mathrm S,\mathrm P}(m)$. Let $\bm B^{(l)}=\bm L^{(l)}\bm R^{(l)\top}$ and set
$\mathcal E_G^{(l)}
 =\|\bm B^{(l)}-\bm B_*\|_{\mathrm F}^2\vee e_{\mathrm S}^{(l)}$,
$\mathcal E_{\mathrm P}^{(l)}=e_{\mathrm P}^{(l)}$, and
$\mathcal E_A^{(l)}
 =\sum_{i=1}^Nw_i\|\bm A_i^{(l)}-\bm A_{i,*}\|_{\mathrm F}^2$.
On the event in Theorem~\ref{thm:main_fed_iht}, resolving
\eqref{eq:main_vector_closed_form} gives, for all $l\ge0$,
\begin{equation}
\label{eq:global_stat_errors}
 \begin{split}
 \mathcal E_G^{(l)}\lesssim_\nu{}&
 (1-\rho_\nu)^lV^{(0)}
 +\frac{\tau^2D_G}{\Ttot}+\omega(m)\left\{
 \frac{\tau^2D_P}{\overline T}
 +\mathcal A_q(m)\right\},
 \end{split}
\end{equation}
and
\begin{equation}
\label{eq:full_personalized_lq_rate}
 \mathcal E_{\mathrm P}^{(l)}\vee\mathcal E_A^{(l)}
 \lesssim_\nu
 (1-\rho_\nu)^lV^{(0)}
 +\tau^2
   \left(\frac{D_P}{\overline T}+\frac{D_G}{\Ttot}\right)
 +\mathcal A_q(m).
\end{equation}
Shared recovery combines a pooled contribution with personalized finite-sample statistical error and the approximation error attenuated by $\omega(m)$.
The personalized and transition errors retain the direct local statistical and approximation contributions.

For a balanced illustration, suppose $T_i=T$ and $s_{i,*}=s_*$,
so that $D_P=s_*\log p+\log N$.
For comparisons across caps, fix the fitting sparsities and curvature
envelopes.
For the following rate comparisons, consider admissible caps
over which $\Lambda_{\mathrm{LR}}(m)\vee\Lambda_{\mathrm S}(m)
\le\overline\Lambda$ for a fixed constant $\overline\Lambda$,
and write $\zeta=(M,\tau^2,\overline\Lambda)$.
The notation $\lesssim_{\nu,\zeta}$ suppresses dependence
on these quantities.
Then,
\begin{equation}
\label{eq:main_balanced_coupling}
 \omega(m)\lesssim_{\nu,\zeta}
 \frac mN+\frac{D_G}{NT}+\frac{s_*\log p}{T}.
\end{equation}
In the following rate statements, we suppress the iteration index
once the optimization contribution is absorbed into the displayed
bounds, and write $\mathcal E_G(m)$, $\mathcal E_{\mathrm P}(m)$,
and $\mathcal E_A(m)$.

\begin{corollary}[Balanced-client accuracy]
\label{cor:main_equal_sample_rates}
\label{cor:lq_balanced_tradeoff}
Under Theorem~\ref{thm:main_fed_iht}, suppose the clients are balanced
as above.
If
\begin{equation}
\label{eq:main_balanced_interaction_condition}
 Ns_*\log p\,D_P\lesssim T(D_G+mD_P),
\end{equation}
then
\begin{align}
\label{eq:balanced_global_rate}
 \mathcal E_G(m)
 &\lesssim_{\nu,\zeta}
 \frac{D_G+mD_P}{NT}+\omega(m)\mathcal A_q(m),\\
\label{eq:balanced_personalized_rate}
 \mathcal E_{\mathrm P}(m)\vee\mathcal E_A(m)
 &\lesssim_{\nu,\zeta}
 \frac{D_P}{T}+\frac{D_G}{NT}+\mathcal A_q(m).
\end{align}
\end{corollary}

The last two terms in \eqref{eq:main_balanced_coupling}
capture the effect of empirical covariance fluctuations on
error transmission from personalized to shared components.
Condition~\eqref{eq:main_balanced_interaction_condition}
absorbs their contributions to $\omega(m)D_P/T$ into the
pooled term $(D_G+mD_P)/(NT)$.
To make the dependence of statistical error and approximation error on the working cap explicit, fix
$q\in(0,1]$ and suppose that
\begin{equation}
\label{eq:lq_balanced_simplification}
 D_G+Ns_*\log p\lesssim mT.
\end{equation}
This condition gives $\omega(m)\lesssim_{\nu,\zeta}m/N$,
so that the approximation contribution also carries the
explicit transmission factor $m/N$.
Substituting the approximation envelope into
\eqref{eq:global_stat_errors} yields
\begin{equation}
\label{eq:lq_balanced_global_tradeoff}
\begin{aligned}
 \mathcal E_G(m)
 &\lesssim_{\nu,\zeta}
 \underbrace{\frac{D_G}{NT}}_{\substack{\text{pooled}\\\text{contribution}}}
 +\underbrace{\frac mN}_{\substack{\text{transmission}\\\text{factor}}}
 \Biggl\{
 \underbrace{\frac{D_P}{T}}_{\substack{\text{personalized}\\\text{statistical error}}}
 +\underbrace{\mathcal A_q(m)\vphantom{\frac{D_P}{T}}}
             _{\substack{\text{approximation}\\\text{error}}}
 \Biggr\}\\
 &= \frac{D_G}{NT}+\frac{mD_P}{NT}
 +\frac{s_*\{C_{\mathrm P}(\nu)m_{q,*}\}^{2/q}}
       {Nm^{2/q-1}}.
\end{aligned}
\end{equation}

\begin{remark}[Statistical speed-up through personalization diversity]
\label{rem:statistical_speedup}
In \eqref{eq:lq_balanced_global_tradeoff}, the client-level personalized statistical error contribution $D_P/T$ enters the shared
error through the factor $m/N$. Thus, a small ratio $m/N$ reduces the statistical cost of personalization in shared
estimation, yielding the statistical speed-up of federated pooling. For an admissible cap proportional to $N$, the contribution of
personalized uncertainty to shared recovery can remain at the local sample size scale.
\end{remark}

\begin{remark}[Rate separation]
\label{rem:pooled_sample_rates}
Under Corollary~\ref{cor:main_equal_sample_rates}, if an admissible
cap satisfies $\mathcal A_q(m)\lesssim D_P/T$, the approximation
contribution is absorbed into the statistical error, yielding the pooled sample size rate
\(
 \mathcal E_G(m)\lesssim_{\nu,\zeta}(D_G+mD_P)/(NT).
\)
In contrast, the personalized error bound retains the local
statistical contribution $D_P/T$, which does not benefit
from information sharing across clients.
At $q=0$, the approximation condition holds automatically
since $\mathcal A_q(m)=0$.
\end{remark}

\begin{remark}[Statistical and approximation error tradeoff]
\label{rem:statistical_approximation_tradeoff}
For fixed $q\in(0,1]$, the working cap also controls a tradeoff. In the shared error bound, increasing $m$ raises the error transmission from the personalized components and attenuates statistical federated gains, as characterized by the term $mD_P/(NT)$. At the same time, increasing $m$ reduces the contribution of approximation error $m\mathcal A_q(m)/N$, which scales as $m^{1-2/q}$.
Balancing these terms suggests
$m\asymp_{q,\nu,\zeta}m_{q,*}(s_*T/D_P)^{q/2}$,
provided this scale gives an admissible cap satisfying
\eqref{eq:lq_balanced_simplification}.
\end{remark}

\subsection{Initialization and end-to-end recovery}
\label{subsec:initializer_theory}

We now verify that the federated initializer enters the local basin
required by Theorem~\ref{thm:main_fed_iht}.
Let $T_{\min}=\min_iT_i$.

\begin{theorem}[Initialization and basin entry]
\label{thm:initializer_rate}
Fix $q\in[0,1]$. Suppose Assumptions~\ref{assump:stationarity}--\ref{assump:regular_covariance}
hold. Choose the fitting levels $s_0$, $\gamma$, and $\{s_i\}_{i=1}^N$ as specified in Theorem~\ref{thm:main_fed_iht}.
Choose the working cap $m$ to satisfy
\begin{equation}
\label{eq:main_init_cap_condition}
 C_{\mathrm P}(\nu)m_{q,*}\left\{
 \frac{(1+\overline T/T_{\min})\overline s_w}
      {\widetilde c_\nu\sigma_r^2}
 \right\}^{q/2}\le m< \frac{N}{2},\qquad q\in[0,1].
\end{equation}
Set $\varpi\asymp\mu r\sigma_1$ in \eqref{eq:init_local_basu}, and choose
the clientwise penalties as
\begin{equation}
\label{eq:initializer_tuning}
 \lambda_{B,i}\asymp
 \tau\sqrt{\frac{\phi_{\max}(p+\log N)}{T_i}},
 \qquad
 \lambda_{C,i}\asymp
 \tau\sqrt{\frac{\phi_{\max}\log(pN)}{T_i}}
 +\phi_{\min}\frac{\mu r\sigma_1}{p}.
\end{equation}
Suppose that, for $i\in[N]$,
\begin{equation}
\label{eq:main_init_local_condition}
 T_i\gtrsim_\nu
 M(p+\log N)+\frac{\tau^2}{\sigma_r^2}
 \{r(p+\log N)+(s_{0,*}+s_{i,*})\log(pN)\},
 \end{equation}
 and
 \begin{equation}
\label{eq:main_init_global_condition}
 s_{0,*}+\max_{i\in[N]}s_{i,*}
 \lesssim_\nu \frac{ p^2\sigma_r^2}{\mu^2r^2\sigma_1^2}.
\end{equation}
Then, with probability at least
$1-C\exp\{-c(p+\log N)\}-C(pN)^{-c}$, the initialization error satisfies
\begin{equation}
 V^{(0)}\le \varepsilon_\nu^2\sigma_r^2.
\end{equation}
\end{theorem}

The lower bound on $m$ in \eqref{eq:main_init_cap_condition}
ensures that the approximation error in initialization is
sufficiently small under $\ell_q$ diversity.
The detailed initialization error bounds are given in
Theorem~\ref*{thm:app-init-full} of the Supplementary Material.

Theorem~\ref{thm:initializer_rate} establishes the initialization
guarantee required by Theorem~\ref{thm:main_fed_iht}.
Under the remaining conditions of Theorem~\ref{thm:main_fed_iht},
the complete procedure satisfies its joint contraction guarantee
with high probability.
The componentwise bounds in Section~\ref{subsec:statistical_rates}
then hold under their stated conditions.
Since initialization relies on local low-rank estimation and
simultaneous error control across clients, the end-to-end guarantee
retains a per-client sample-size requirement involving $p+\log N$.
Table~\ref{tab:balanced_summary} summarizes sufficient
end-to-end conditions and recovery bounds for balanced clients under the model and fitting
assumptions of Theorems~\ref{thm:main_fed_iht} and~\ref{thm:initializer_rate}.

\begin{table}[htb]
\caption{Sufficient conditions and squared-error bounds for end-to-end recovery with balanced clients. Dependence on  $\nu$ and $\zeta$ is suppressed.}
\label{tab:balanced_summary}
\centering
\small
\setlength{\tabcolsep}{5pt}
\renewcommand{\arraystretch}{1}
\begin{tabular}{@{}
 >{\raggedright\arraybackslash}p{0.28\textwidth}
 >{\raggedright\arraybackslash}p{\dimexpr0.72\textwidth-2\tabcolsep\relax}
 @{}}
\toprule
Requirement & Balanced-client form \\
\midrule

Diversity and working cap
&
$m_*<N/2$ and 
$\displaystyle
 C_{\mathrm P}(\nu)m_{q,*}
 \left\{s_*/(c\sigma_r^2)\right\}^{q/2}
 \le m< N/2$.
\\
\addlinespace

Fitting levels and penalties
&
Fitting levels and step sizes as specified in
Section~\ref{subsec:main_theory};
$\varpi\asymp\mu r\sigma_1$ and penalties as in
\eqref{eq:initializer_tuning}.
\\
\addlinespace

Sample size and signal
&
$\displaystyle
 \begin{aligned}[t]
 &T\gtrsim p+D_P+\frac{D_G}{N}
 +\frac{r(p+\log N)+(s_{0,*}+s_*)\log(pN)}{\sigma_r^2},\\[3pt]
 &s_{0,*}+s_*\lesssim
 \frac{p^2\sigma_r^2}{\mu^2r^2\sigma_1^2}.
 \end{aligned}$
\\
\addlinespace

Compatibility
&
$\displaystyle
 \begin{aligned}[t]
 \gamma_*\mu r(\widetilde\phi_w^2+\beta_{\mathrm{LR}}^2)\le c,\quad \frac mN\left\{\phi_{\max}^2+
 \max_{a\in\{\mathrm{LR},\mathrm S\}}\Lambda_a(m)^2\right\}\le c.
 \end{aligned}$
\\
\addlinespace

Balanced-rate conditions
&
For $q=0$, 
$Ns_*\log p\,D_P\lesssim T(D_G+mD_P)$;
\newline
For $0<q\le1$, $D_G+Ns_*\log p\lesssim mT$.
\\
\midrule
Component & Squared-error bound \\
\midrule

Shared components
&
$\displaystyle
 \mathcal E_G(m)\lesssim
 \frac{D_G+mD_P}{NT}
 +\frac{ms_*}{N}\{C_{\mathrm P}(\nu)m_{q,*}/m\}^{2/q}$.
\\
\addlinespace

Personalized and transition matrices
&
$\displaystyle
 \mathcal E_{\mathrm P}(m)\vee\mathcal E_A(m)
 \lesssim
 \frac{D_P}{T}+\frac{D_G}{NT}
 +s_*\{C_{\mathrm P}(\nu)m_{q,*}/m\}^{2/q}$.
\\
\bottomrule
\end{tabular}
\vspace{3pt}
\begin{minipage}{\textwidth}
\footnotesize
\textit{Note:}
The displayed approximation terms apply to $q>0$
and are replaced by zero at $q=0$.
\end{minipage}
\end{table}

\section{Simulation Studies}
\label{sec:simulations}

We conduct four simulation experiments to examine
(i) finite-sample estimation accuracy and the separation
between shared and personalized rates;
(ii) the effects of personalization diversity and directional
covariance discrepancy on shared recovery;
(iii) the working-cap tradeoff in shared recovery under
approximate personalization diversity; and
(iv) the communication-efficient implementation of
\textsc{PerFeCT-VAR}.
Additional simulation details are provided in
Appendix~\ref*{app:simulation-details} of the Supplementary Material.

Throughout, \(p=40\) and \(T_i=T\). For the
client-specific transition matrices, we report the relative error
$
  \operatorname{RelErr}(\{\widehat{\bm A}_i\})
  :=
  \left\{
  \sum_{i=1}^N
  \|\widehat{\bm A}_i-\bm A_{i,*}\|_{\mathrm F}^2/(\sum_{i=1}^N\|\bm A_{i,*}\|_{\mathrm F}^2)
  \right\}^{1/2},
$
and define \(\operatorname{RelErr}(\{\widehat{\bm S}_i\})\) analogously for the
personalized components.
For the shared components, we use
$\operatorname{RelErr}(\widehat{\bm L}\widehat{\bm R}^\top) := \|\widehat{\bm L}\widehat{\bm R}^\top-\bm L_*\bm R_*^\top\|_{\mathrm F}/\|\bm L_*\bm R_*^\top\|_{\mathrm F}$,
with \(\operatorname{RelErr}(\widehat{\bm S}_0)\) defined similarly.
All reported results are averages over 200 independent replications.

\subsection{Estimation accuracy and rate separation}
\label{subsec:sim_accuracy}

The first experiment compares \textsc{PerFeCT-VAR} with several benchmarks
as \(T\) varies and then examines its componentwise errors as the number of
clients \(N\) increases. The benchmarks are the low-rank-plus-sparse
estimator of \citet{basu2019low} fitted either to the pooled data
(Pooled LRS) or separately to each client (Local LRS), Ditto
\citep{li2021ditto} using a sample-size-weighted pooled OLS estimator for
its shared model, and the proposed federated initializer (Initial).
We set \(p=40\), \(r=2\),
\(s_{0,*}=36\), $\gamma_*p=2$, $m_*=2$, and
\(s_{i,*}=18\). The shared low-rank component is generated as
\(\bm L_*\bm R_*^\top=\bm U\operatorname{diag}(1,0.5)\bm V^\top\), where
\(\bm U\) and \(\bm V\) have random orthonormal columns. The shared and personalized sparse components are generated on random
supports satisfying their prescribed sparsity and frequency constraints.
The innovations are independent \(N(\bm 0,\bm I_{40})\). Starting from the
proposed initializer, \textsc{PerFeCT-VAR} is run for \(K=200\) rounds
with fitting levels \((s_0,s_i,m,\gamma p)=(40,20,3,3)\), and step sizes
\(\eta_{\mathrm{LR}}=\eta_\mathrm{S}=0.1\) and \(\eta_\mathrm{P}=0.3\). Here $m=3>m_*=2$, so the true personalized supports remain
feasible under the working cap. The implementation details are given in
Appendix~\ref*{subsec:sim-exp1-details}.

\begin{figure}[!htp]
  \centering
  \includegraphics[width=0.75\textwidth]{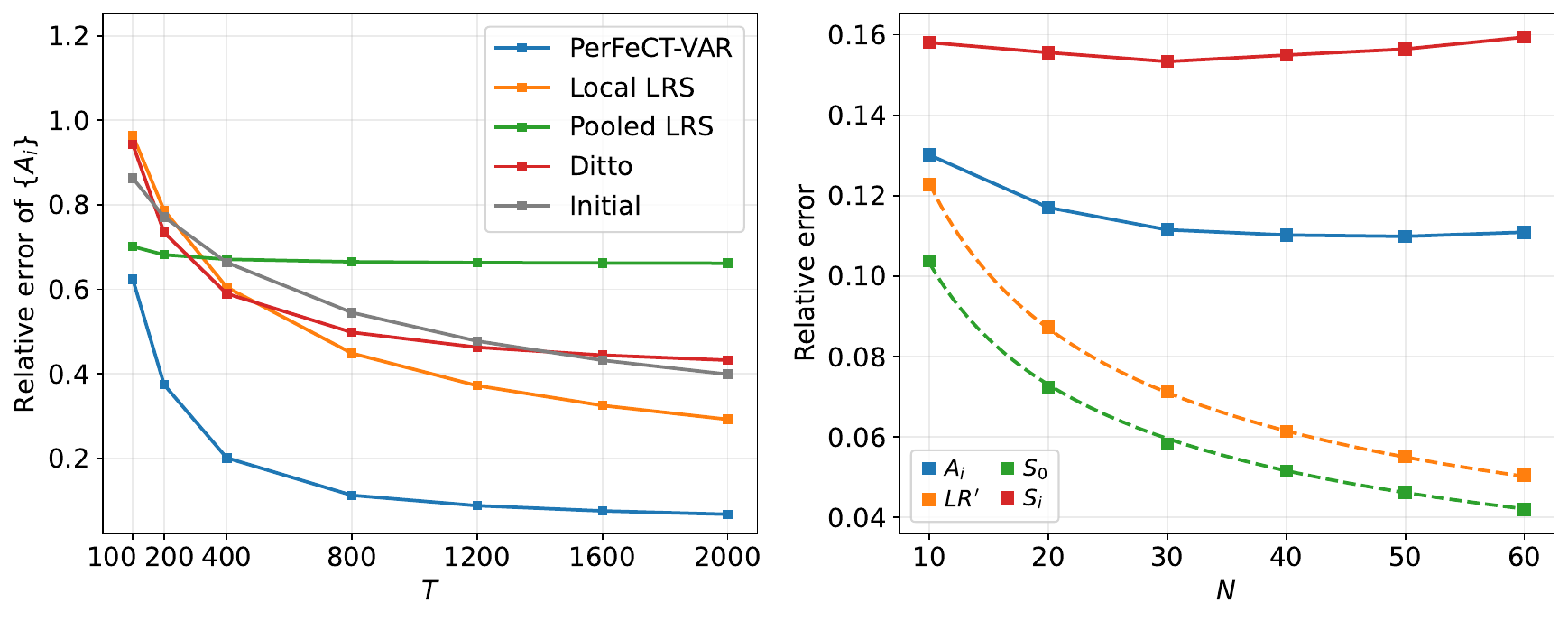}
  \caption{Mean relative errors as \(T\) and \(N\) vary. Left: transition-matrix
errors across \(T\) at \(N=30\). Right:
componentwise errors of \textsc{PerFeCT-VAR} across \(N\) at \(T=800\).}
  \label{fig:sim_vary_TN}
\end{figure}

The left panel of Figure~\ref{fig:sim_vary_TN} shows decreasing estimation
error with \(T\) for the methods considered, with
\textsc{PerFeCT-VAR} attaining the smallest mean relative error across the
displayed settings. The Pooled LRS error decreases much more slowly and
eventually levels off, reflecting the bias from imposing a
single transition matrix on heterogeneous clients. 
The right panel examines the separation between shared and personalized
estimation rates. The dashed curves for \(\bm B_*\) and \(\bm S_{0,*}\)
are fitted functions of the form \(\sqrt{c_{\mathrm{fit}}/N}\).
As \(N\) increases, the errors of both shared components decrease markedly
and track this \(N^{-1/2}\) pattern, whereas the personalized error changes
little. Since \(T\) is fixed in this experiment, this behavior is consistent
with Corollary~\ref{cor:main_equal_sample_rates}: the squared errors of the
shared components can benefit from the pooled \(NT\) sample size, while the
leading personalized error remains at the client-level \(T^{-1}\) scale.

\subsection{Personalization diversity and covariance discrepancy}
\label{subsec:sim_diversity_heterogeneity}

The second experiment examines how personalization diversity
in the true parameters and directional covariance discrepancy
affect shared recovery. The first design varies the degree of
exact personalization diversity.
The second varies directional covariance discrepancy while
keeping the principal curvature and scale quantities unchanged.

\paragraph*{Personalization diversity}
We fix $N=30$, $T=800$, and $p=40$, and set
$\bm L_*\bm R_*^\top=0.5\bm u\bm u^\top$, where
$u_j=15^{-1/2}$ for $j\leq15$ and zero otherwise.
We take $(s_{0,*},\gamma_*p,s_{i,*})=(18,2,6)$ and
restrict the personalized supports to the active top-left
$15\times15$ block.
For $m_*\in\{1,2,3,4,6,9,12\}$, each selected coordinate
is assigned to exactly $m_*$ clients, while each client
retains exactly six personalized entries.
Nonzero personalized magnitudes are drawn from $U(0.1,0.2)$
in all configurations.
Thus, increasing $m_*$ corresponds to weaker
personalization diversity.
We set the working cap to $m=m_*+1$ in each configuration
so that the true personalized supports remain feasible.
The complete construction is given in
Appendix~\ref*{subsec:sim-exp2-details}.

\paragraph*{Directional covariance discrepancy}
We next construct four designs with increasing
\(\Lambda_{\mathrm{LR}}\) and \(\Lambda_{\mathrm S}\). We use
\(N=30\), \(T=800\), \(p=40\), \(m_*=2\), and
\(s_{0,*}=s_{i,*}=18\), with the working cap fixed at $m=3$.
The shared low-rank component is rank one, \(\bm B_*=0.3\bm u\bm v^\top\), with a sparse right singular direction \(\bm v\).
Clients are arranged into 15 pairs with opposite personalized components
and opposite covariance deviations, so that the average covariance remains
fixed. For the zero, low, medium, and high discrepancy settings, respectively,
\(k=0,5,10,15\) client pairs have covariance deviations aligned with the
shared directions. This construction gives
$\Lambda_{\mathrm{LR}}(3)=\Lambda_{\mathrm S}(3)=0.45\sqrt{2k/3}$,
while keeping the principal covariance curvature and innovation-scale
quantities fixed across the four settings. The full construction and its
verification are provided in
Appendix~\ref*{subsec:sim-exp2-details}.

\begin{figure}[!htp]
  \centering
  \includegraphics[width=0.7\textwidth]{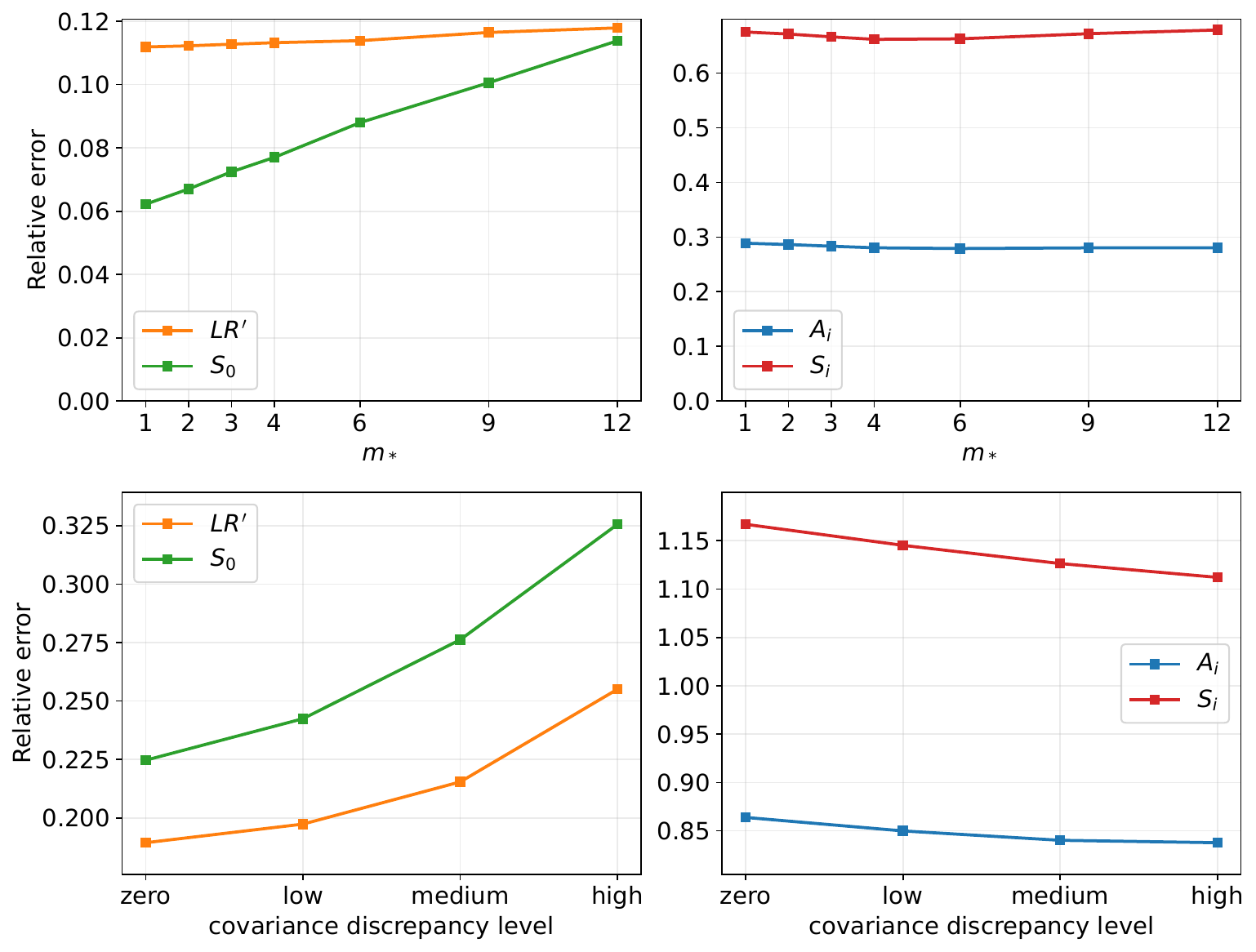}
  \caption{Componentwise relative errors under two configurations.}
  \label{fig:sim_diversity_heterogeneity}
\end{figure}

Figure~\ref{fig:sim_diversity_heterogeneity} shows complementary effects
of these two forms of heterogeneity. In the upper row, increasing
\(m_*\) raises the errors of the shared low-rank and sparse components,
while the personalized error remains stable. The lower row
shows a similar increase in shared-component error as
\(\Lambda_{\mathrm{LR}}\) and \(\Lambda_{\mathrm S}\) increase, though
the principal covariance curvature quantities are unchanged. These patterns
are consistent with the personalized-to-shared coupling terms in
Theorem~\ref{thm:main_fed_iht}: weaker personalization diversity and larger
directional covariance discrepancy both increase the propagation of
personalized estimation error into the shared updates.

\subsection{Working cap tradeoff in shared recovery}
\label{subsec:sim_working_cap}

The third experiment examines the working-cap tradeoff in shared
recovery under approximate personalization diversity. We use the
shared sparse construction and innovation distribution of
Section~\ref{subsec:sim_accuracy}, with $N=30$, $T=400$,
and $(s_{0,*},\gamma_*p,s_{i,*})=(12,2,60)$.
The shared low-rank component is $\bm B_*=0.6\bm u\bm u^\top$, where
$u_j=15^{-1/2}$ for $j\leq15$ and zero otherwise.

We generate $\{\bm S_{i,*}\}_{i=1}^N$ on randomly selected coordinates of the top-left $15\times15$ block.
Each selected coordinate is assigned to at most 14 clients,
with $s_{i,*}=60$ for every client.
At each selected coordinate, we randomly assign the decreasing
magnitudes $a_k=0.4\{1+[(k-1)/2.5]^{1/q}\}^{-1}$, with $q=1/4$,
to the active clients and use a common random sign.
This construction gives $m_*=14$ and
$m_{q,*}=\sum_{k=1}^{14}a_k^q\approx5.57$. Within each replication, we keep the true parameters and observed
data fixed and vary $m\in\{1,2,3,4,5,6,8,10,12,14\}$. The other fitting levels are fixed at $(s_0,s_i,\gamma p)=(16,75,3)$.

\begin{figure}[!htp]
  \centering
  \includegraphics[width=0.8\textwidth]{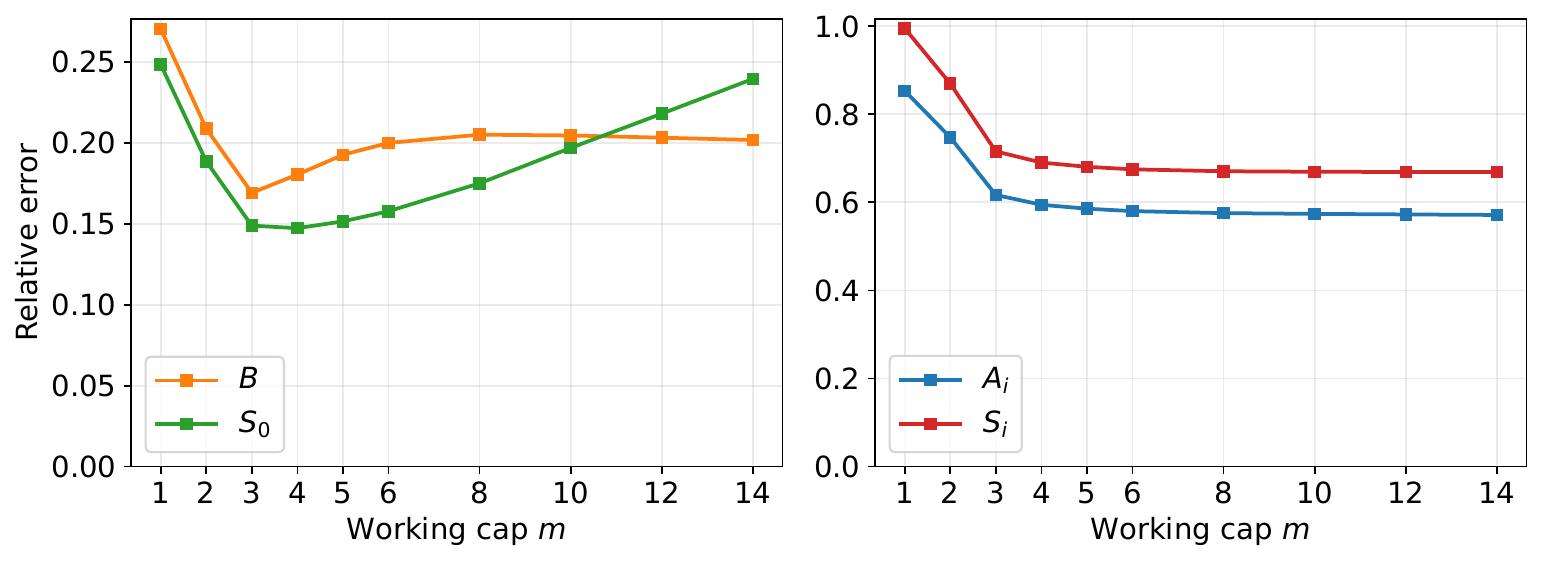}
  \caption{Componentwise relative errors as the working cap $m$
  varies under approximate personalization diversity.}
  \label{fig:sim_working_cap}
\end{figure}
Figure~\ref{fig:sim_working_cap} shows the componentwise relative errors as the working
cap $m$ varies. The shared low-rank and sparse components attain
their smallest mean errors at intermediate caps, $m=3$ and $m=4$,
respectively, with higher errors at both smaller and larger caps.
These patterns are consistent with the tradeoff between
statistical and approximation errors in
Section~\ref{subsec:statistical_rates}.
The transition-matrix and personalized errors decrease rapidly
at first and then become nearly flat, consistent with their
error bounds combining a decreasing approximation contribution
with a local statistical contribution.

\subsection{Communication-efficient implementation}
\label{subsec:sim_communication}

The final experiment examines the proposed communication-efficient implementation
described in Section~\ref{sec:communication}. Using the setting in Section~\ref{subsec:sim_accuracy} with
\(N=30\), \(T=800\), and \(p=40\), we compare the exact
\textsc{PerFeCT-VAR} algorithm with four configurations of the reduced-communication implementation,
formed by combining
\(\widetilde s_0\in\{40,160\}\) with two FCT refresh schedules:
a selective schedule with 62 refreshes and an every-10 schedule with
20 refreshes. We count total post-initialization uplink communication across all clients, assigning eight bytes to each numerical value and
two additional bytes to its coordinate index for sparse transmissions.

Over 200 rounds, the exact implementation of Algorithm \ref{alg:exact_HT_FCHT} uses 146.48 MiB of uplink
communication. The selective schedule uses 13.39 and 20.26 MiB for
\(\widetilde s_0=40\) and \(160\), respectively, while the every-10
schedule uses 10.83 and 17.70 MiB. These values correspond to an \(86.2\%\)--\(92.6\%\) reduction
in uplink communication.

\begin{figure}[!htp]
  \centering
  \includegraphics[width=0.8\textwidth]{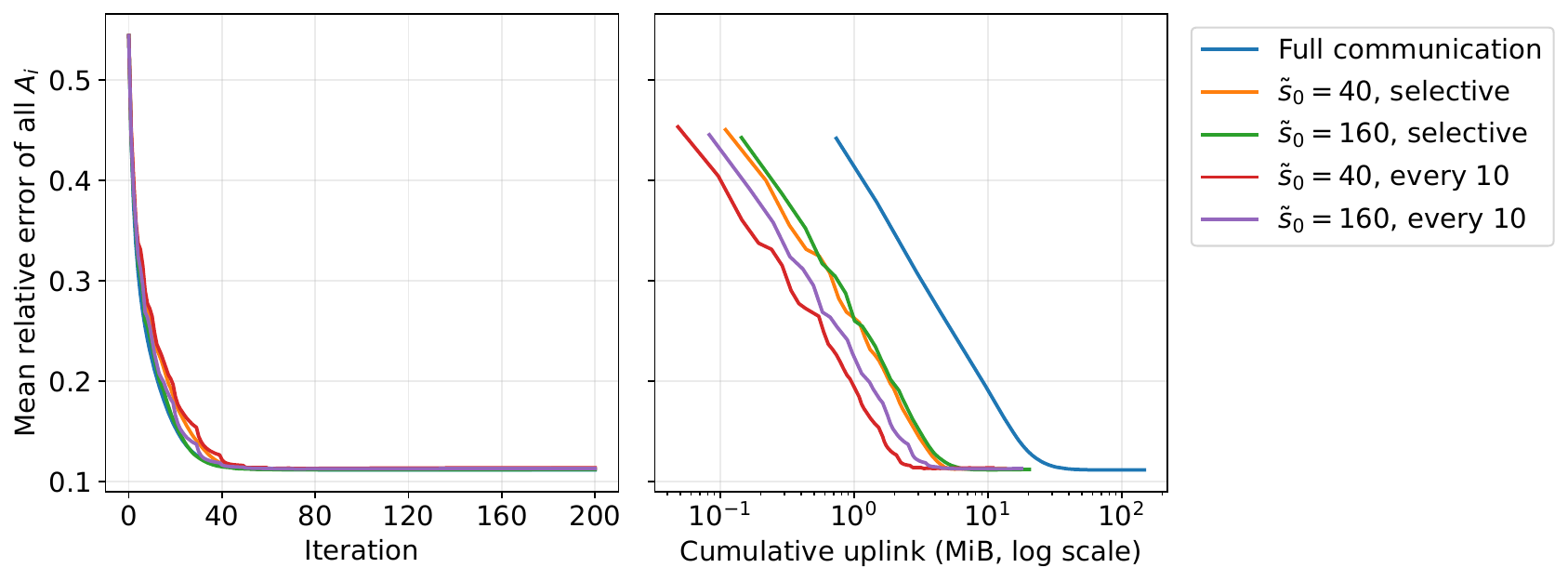}
  \caption{Comparison of Algorithm \ref{alg:exact_HT_FCHT} and the communication-efficient implementation.}
  \label{fig:sim_communication}
\end{figure}

Figure~\ref{fig:sim_communication} compares estimation error against both
communication rounds and cumulative uplink. In this experiment, all four configurations closely track the exact implementation and
reach similar final errors while using substantially less communication.
The results illustrate that the communication reductions in
Section~\ref{sec:communication} can yield substantial practical savings
without appreciable loss of estimation accuracy in the settings considered.

\section{Real Data Analysis: Dominick's Sales Revenue}
\label{sec:realdata}

We apply \textsc{PerFeCT-VAR} to weekly sales data from Dominick's Finer
Foods\footnote{The original data are available at the Kilts Center for Marketing at the University of Chicago Booth School of Business (https://www.chicagobooth.edu/research/kilts/research-data/dominicks).}, a supermarket chain in the Chicago metropolitan area. In the analysis, we aggregate the UPC-level
records to weekly category revenues, yielding \(N=25\) stores
and \(p=17\) product categories observed over 122 weeks from January 1995
to May 1997. The 17 categories are analgesics
(ANA), bottled juices (BJC), cheeses (CHE), cookies (COO), regular crackers
and saltines (CRA), dish detergent (DID), front-end candies (FEC), frozen
dinners (FRD), frozen entrees (FRE), frozen juices (FRJ), fabric softeners
(FSF), grooming products (GRO), laundry detergents (LND), oatmeal (OAT),
snack crackers (SNA), soaps (SOA), and toothpastes (TPA). We analyze
log-transformed revenues after adjusting for trend, seasonality, and holiday
effects. The adjusted series are centered within stores and standardized
using category-specific standard deviations across stores. The 122 weeks are divided chronologically into 82 for initial
training, 20 for model tuning, and 20 for out-of-sample
evaluation. Model tuning and test evaluation both use one-step-ahead
rolling forecasts. Three validation forecast occasions are excluded because of missing-data handling, leaving 17 usable occasions for tuning.
Tuning parameters are selected by minimizing the resulting validation MSE.

We compare \textsc{PerFeCT-VAR} with the Local LRS, Pooled LRS, and Ditto
methods used in Section~\ref{sec:simulations}. Additional benchmarks are
local and pooled VAR(1) models estimated by OLS or LASSO, univariate AR(1) (denoted by iAR(1))
models fitted separately to each store--category series, and a zero-forecast
benchmark on the standardized residual scale. For \textsc{PerFeCT-VAR}, we
fix \(r=2\) and \(\gamma=2/17\); validation selects
$s_0=17$, $s_i=40$, and $m=10$, satisfying \(m<N/2\). Table~\ref{tab:forecast} reports the mean and median store-level test MSE. For each store, the MSE averages squared one-step-ahead forecast errors
over the 20 test weeks and 17 categories. Among the methods considered,
\textsc{PerFeCT-VAR} has the smallest mean and median MSE. The gains over pooled and local benchmarks suggest benefits from sharing information across stores while allowing store-specific dynamics.

\begin{table}[t]
\centering
\small
\renewcommand{\arraystretch}{1.15}
\caption{Out-of-sample rolling-window forecasting performance.}
\label{tab:forecast}
\begin{tabular}{@{}lrr@{\hspace{1em}}|@{\hspace{1em}}lrr@{}}
\Xhline{0.8pt}
Method & Mean MSE & Median MSE & Method & Mean MSE & Median MSE \\
\Xhline{0.4pt}
Pooled OLS & 1.837 & 1.735 & Local OLS & 2.077 & 1.911 \\
Pooled LASSO & 1.800 & 1.729 & Local LASSO & 1.840 & 1.723 \\
Pooled LRS & 1.805 & 1.722 & Local LRS & 1.841 & 1.726 \\
\Xhline{0.4pt}
Ditto & 1.818 & 1.679 & Zero & 2.184 & 2.077 \\
PerFeCT-VAR & \textbf{1.779} & \textbf{1.635} & iAR(1) & 1.861 & 1.795 \\
\Xhline{0.8pt}
\end{tabular}
\end{table}

After the forecasting comparison, we refit \textsc{PerFeCT-VAR} using all
122 weeks to examine the estimated dynamic structure. Write
\(
\widehat{\bm B}
=
\widehat{\bm L}\widehat{\bm R}^{\top}
=
\widehat{\bm U}\widehat{\bbm\Sigma}\widehat{\bm V}^{\top}.
\)
To aid interpretation, we apply varimax rotations separately to $\widehat{\bm U}$ and $\widehat{\bm V}$. Following \citet{huang2025supervised}, we interpret the rotated bases
in terms of response and predictor loading patterns, respectively.
Figure~\ref{fig:real_decomposition}(a) displays these patterns. The first predictor and response patterns, P--1 and R--1, both
place relatively large loadings on the biscuit-related categories CRA,
COO, and SNA; R--1 also has substantial loadings on FEC and FRD. The
second predictor pattern P--2 has positive loadings on LND, FRE, and OAT
and negative loadings on FRJ, CHE, and COO, while R--2 has loadings of the
same sign for most response categories. These patterns provide a
low-dimensional summary of the shared predictive dependence across stores.

\begin{figure}[t]
  \centering
  \includegraphics[width=\textwidth]{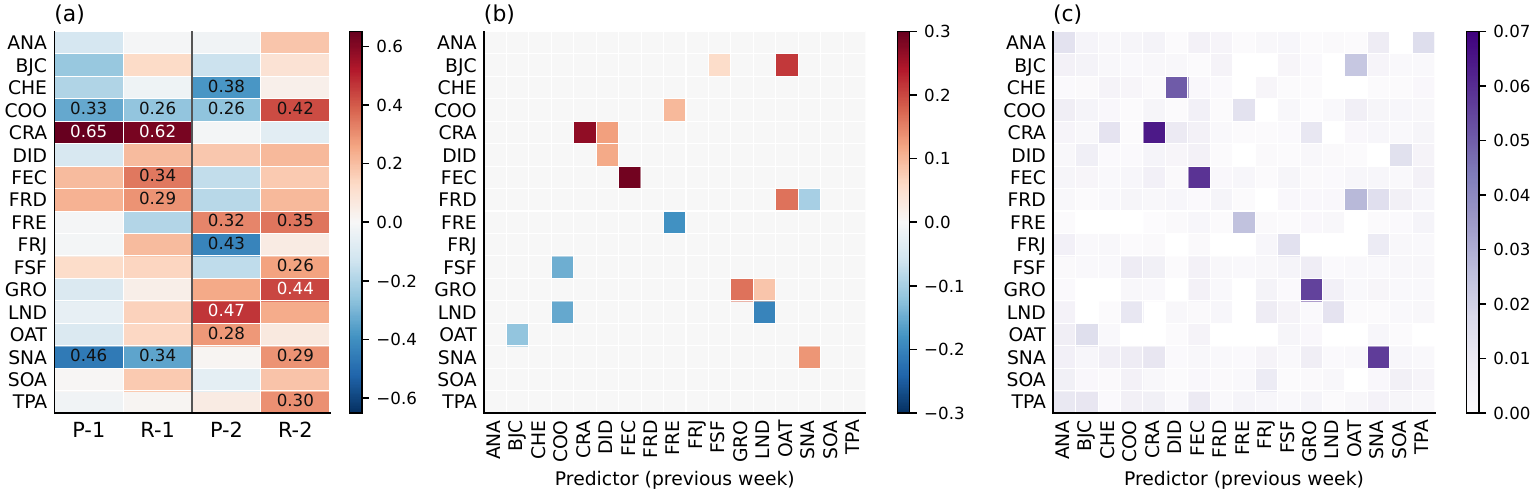}
  \caption{Estimated model components. (a) Rotated loadings of the shared low-rank component; (b) $\widehat{\bm S}_0$; (c) Mean absolute personalized coefficients across stores, $N^{-1}\sum_{i=1}^{N}|(\widehat{\bm S}_i)_{jk}|$.}
  \label{fig:real_decomposition}
\end{figure}

For a transition matrix, the \((j,k)\)-th coefficient represents the
predictive effect of category \(k\) in the previous week on category \(j\)
in the current week. Figure~\ref{fig:real_decomposition}(b) shows the
shared sparse component, which captures sparse predictive relationships
beyond the shared low-rank structure. Among its largest diagonal entries, the FEC and CRA coefficients are positive, indicating shared
persistence, whereas those for FRE and LND are negative, indicating
short-term reversal. The two directed OAT--BJC coefficients have opposite
signs, illustrating that the estimated predictive relationship need not be
symmetric across directions.

Figure~\ref{fig:real_decomposition}(c) reports
\(N^{-1}\sum_{i=1}^N|(\widehat{\bm S}_i)_{jk}|\),
which reflects both the frequency and magnitude of store-specific
adjustments at each directed relationship. The largest values include the CRA, FEC,
SNA, and GRO own-lag coefficients, each of which also has a substantial
shared component. In contrast, the DID own-lag and COO--FSF relationships
are mainly represented by the shared component, with relatively small
personalized adjustments. The estimated shared coefficients for DID--CHE
and FSF--FRJ are close to zero, while
nonzero personalized adjustments occur in eight and nine stores. These
examples illustrate that a relationship can be weak in the shared
component while remaining relevant for a subset of stores.

\begin{figure}[t]
  \centering
  \includegraphics[width=0.85\textwidth]{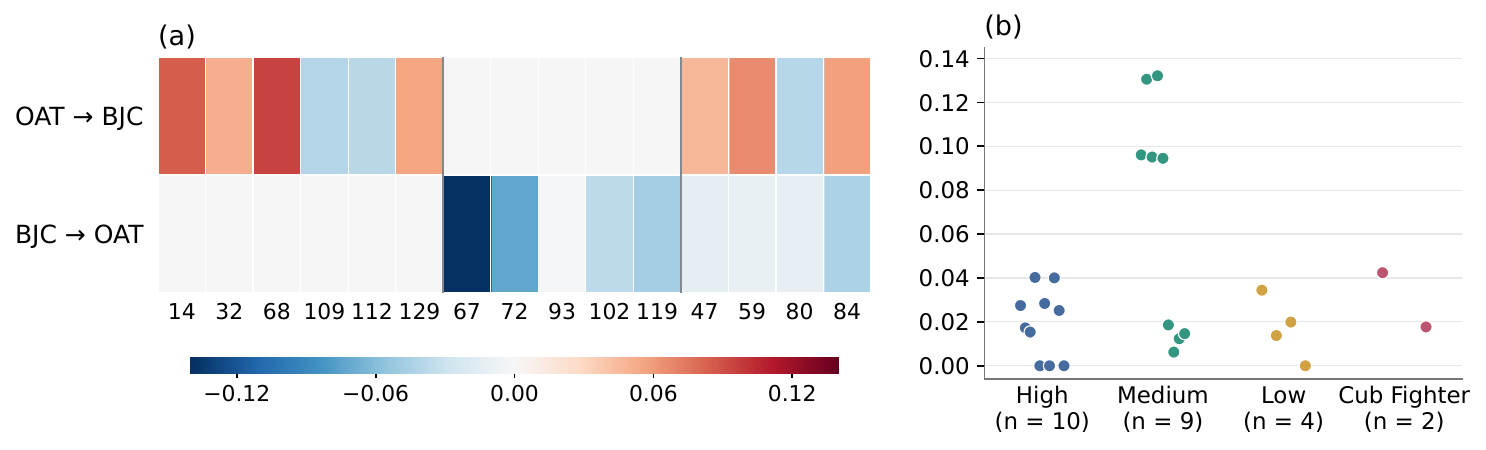}
  \caption{Store-specific personalizations. (a) Signed personalized links between OAT and BJC across 15 stores; (b) Personalization strength for predicting COO by price tier.}
  \label{fig:real_personalization}
\end{figure}

Figure~\ref{fig:real_personalization}(a) shows store-specific
adjustments to the asymmetric relationships between OAT and BJC
observed in $\widehat{\bm S}_0$. Across the two directions, 15 stores
have at least one personalized adjustment: 11 have an adjustment in one
direction and four in both directions. Because the frequency cap applies
to each directed coefficient separately, each direction satisfies the
selected cap \(m=10\). The OAT--BJC adjustments strengthen or weaken the shared relationship,
whereas the BJC--OAT adjustments reinforce the negative shared coefficient. Figure~\ref{fig:real_personalization}(b) summarizes heterogeneity in
predicting COO across stores. According to the data manual, stores
are grouped into four price tiers: High, Medium, Low, and Cub Fighter,
with the last group comprising stores competing directly with Cub Foods.
For store \(i\), define its COO personalization magnitude as
\(
\left\{
\sum_{k=1}^{17}
[(\widehat{\bm S}_i)_{\mathrm{COO},k}]^2
\right\}^{1/2}.
\)
The five stores with the largest values all belong to the Medium tier.
Overall, \textsc{PerFeCT-VAR} achieves competitive forecasting
performance while characterizing shared revenue dynamics and
store-specific departures.

\section{Conclusion and Discussion}
\label{sec:conclusion}

This paper develops \textsc{PerFeCT-VAR} for personalized federated estimation of high-dimensional vector autoregressions with shared low-rank dynamics, shared sparse links, and personalized sparse departures. Personalization diversity yields a sharp majority threshold for shared-personalized identification and motivates frequency-capped thresholding for joint estimation. Within this identified model, approximate diversity controls the error incurred by fitting a smaller frequency cap.

Our theory accounts for heterogeneous client distributions and establishes initialization and joint local convergence guarantees. The componentwise bounds show that diversity attenuates both estimation and approximation errors transmitted from personalized to shared components. These bounds also quantify the working-cap tradeoff
between statistical and approximation errors in
shared recovery. Under suitable conditions, shared estimation attains pooled accuracy despite non-vanishing personalized effects, while personalized estimation retains the local-sample scale on average. Simulations and the retail application illustrate the estimation and interpretive benefits. Future work includes establishing guarantees for
data-driven selection of the working cap and developing
initializers with weaker local sample-size requirements.

\section*{Statement on AI Use} AI is used for proofreading and language editing in this paper.

{\setstretch{1.1}
\bibliography{mybib}

@inproceedings{zhang2018unified,
	author = {Zhang, Xiao and Wang, Lingxiao and Gu, Quanquan},
	booktitle = {International Conference on Artificial Intelligence and Statistics},
	organization = {PMLR},
	pages = {1097--1107},
	title = {A unified framework for nonconvex low-rank plus sparse matrix recovery},
	year = {2018}}

@misc{lin2025improvingtimeseriesestimation,
	archiveprefix = {arXiv},
	author = {Yuchang Lin and Qianqian Zhu and Guodong Li},
	eprint = {2510.25236},
	primaryclass = {stat.ME},
	title = {Improving time series estimation and prediction via transfer learning},
	url = {https://arxiv.org/abs/2510.25236},
	year = {2025}}

@article{liu2025robust,
	author = {Liu, Weidong and Mao, Xiaojun and Zhang, Xiaofei and Zhang, Xin},
	journal = {Journal of the American Statistical Association},
	number = {549},
	pages = {266--277},
	publisher = {Taylor \& Francis},
	title = {Robust personalized federated learning with sparse penalization},
	volume = {120},
	year = {2025}}

@article{chen2023minimax,
	author = {Chen, Shuxiao and Zheng, Qinqing and Long, Qi and Su, Weijie J},
	journal = {Journal of Machine Learning Research},
	number = {262},
	pages = {1--59},
	title = {Minimax estimation for personalized federated learning: An alternative between FedAvg and local training?},
	volume = {24},
	year = {2023}}

@inproceedings{li2021ditto,
	author = {Li, Tian and Hu, Shengyuan and Beirami, Ahmad and Smith, Virginia},
	booktitle = {International Conference on Machine Learning},
	organization = {PMLR},
	pages = {6357--6368},
	title = {Ditto: Fair and robust federated learning through personalization},
	year = {2021}}

@article{smith2017federated,
	author = {Smith, Virginia and Chiang, Chao-Kai and Sanjabi, Maziar and Talwalkar, Ameet S},
	journal = {Advances in neural information processing systems},
	title = {Federated multi-task learning},
	volume = {30},
	year = {2017}}

@article{huang2025optimal,
	author = {Huang, Xinmeng and Xu, Kan and Lee, Donghwan and Hassani, Hamed and Bastani, Hamsa and Dobriban, Edgar},
	journal = {Journal of the American Statistical Association},
	number = {552},
	pages = {2262--2275},
	publisher = {Taylor \& Francis},
	title = {Optimal multitask linear regression and contextual bandits under sparse heterogeneity},
	volume = {120},
	year = {2025}}

@inproceedings{pal2024sample,
	author = {Pal, Soumyabrata and Varshney, Prateek and Madan, Gagan and Jain, Prateek and Thakurta, Abhradeep and Aggarwal, Gaurav and Shenoy, Pradeep and Srivastava, Gaurav},
	booktitle = {International Conference on Artificial Intelligence and Statistics},
	organization = {PMLR},
	pages = {1702--1710},
	title = {Sample-efficient personalization: modeling user parameters as low rank plus sparse components},
	year = {2024}}

@misc{chen2026private,
	archiveprefix = {arXiv},
	author = {Kejun Chen and Qianqian Zhu},
	eprint = {2604.07135},
	primaryclass = {stat.ME},
	title = {Private Federated Learning for High-dimensional Time Series},
	url = {https://arxiv.org/abs/2604.07135},
	year = {2026}}

@article{xu2025joint,
	author = {Xu, Yuchen and Michailidis, George},
	journal = {arXiv preprint arXiv:2509.15402},
	title = {Joint learning of panel VAR models with low rank and sparse structure},
	year = {2025}}

@article{skripnikov2019joint,
	author = {Skripnikov, A and Michailidis, G},
	journal = {Econometrics and Statistics},
	pages = {120--133},
	publisher = {Elsevier},
	title = {Joint estimation of multiple network Granger causal models},
	volume = {10},
	year = {2019}}

@article{fisher2022penalized,
	author = {Fisher, Zachary F and Kim, Younghoon and Fredrickson, Barbara L and Pipiras, Vladas},
	journal = {Psychometrika},
	number = {2},
	pages = {403--431},
	publisher = {Cambridge University Press \& Assessment},
	title = {Penalized estimation and forecasting of multiple subject intensive longitudinal data},
	volume = {87},
	year = {2022}}

@article{wang2022high,
	author = {Wang, Di and Zheng, Yao and Lian, Heng and Li, Guodong},
	journal = {Journal of the American Statistical Association},
	number = {539},
	pages = {1338--1356},
	publisher = {Taylor \& Francis},
	title = {High-dimensional vector autoregressive time series modeling via tensor decomposition},
	volume = {117},
	year = {2022}}

@article{kock2015oracle,
	author = {Kock, Anders Bredahl and Callot, Laurent},
	journal = {Journal of Econometrics},
	number = {2},
	pages = {325--344},
	publisher = {Elsevier},
	title = {Oracle inequalities for high dimensional vector autoregressions},
	volume = {186},
	year = {2015}}

@article{huang2025supervised,
	author = {Huang, Feiqing and Lu, Kexin and Zheng, Yao and Li, Guodong},
	journal = {Journal of Econometrics},
	pages = {105995},
	publisher = {Elsevier},
	title = {Supervised factor modeling for high-dimensional linear time series},
	volume = {249},
	year = {2025}}

@article{tan2022towards,
	author = {Tan, Alysa Ziying and Yu, Han and Cui, Lizhen and Yang, Qiang},
	journal = {IEEE transactions on neural networks and learning systems},
	number = {12},
	pages = {9587--9603},
	publisher = {IEEE},
	title = {Towards personalized federated learning},
	volume = {34},
	year = {2023}}

@article{lucas2019bank,
	author = {Lucas, Andr{\'e} and Schaumburg, Julia and Schwaab, Bernd},
	journal = {Journal of Business \& Economic Statistics},
	number = {3},
	pages = {542--555},
	publisher = {Taylor \& Francis},
	title = {Bank business models at zero interest rates},
	volume = {37},
	year = {2019}}

@article{wilms2018multiclassVAR,
	author = {Wilms, Ines and Barbaglia, Luca and Croux, Christophe},
	doi = {10.1111/rssc.12231},
	eprint = {https://academic.oup.com/jrsssc/article-pdf/67/2/435/49335909/jrsssc_67_2_435.pdf},
	issn = {0035-9254},
	journal = {Journal of the Royal Statistical Society Series C: Applied Statistics},
	month = {02},
	number = {2},
	pages = {435-452},
	title = {Multiclass Vector Auto-Regressive Models for Multistore Sales Data},
	url = {https://doi.org/10.1111/rssc.12231},
	volume = {67},
	year = {2018}}

@book{korte2018combinatorial,
	author = {Korte, Bernhard and Vygen, Jens},
	doi = {10.1007/978-3-662-56039-6},
	edition = {6th},
	publisher = {Springer},
	title = {Combinatorial Optimization: Theory and Algorithms},
	year = {2018}}

@article{basu2019low,
	author = {Basu, Sumanta and Li, Xianqi and Michailidis, George},
	journal = {IEEE Transactions on Signal Processing},
	number = {5},
	pages = {1207--1222},
	publisher = {IEEE},
	title = {Low rank and structured modeling of high-dimensional vector autoregressions},
	volume = {67},
	year = {2019}}

@book{vershynin2018high,
	author = {Vershynin, Roman},
	publisher = {Cambridge University Press},
	title = {High-dimensional probability: An introduction with applications in data science},
	volume = {47},
	year = {2018}}

@article{basu2015regularized,
	author = {Basu, Sumanta and Michailidis, George},
	journal = {The Annals of Statistics},
	number = {4},
	pages = {1535--1567},
	publisher = {JSTOR},
	title = {Regularized estimation in sparse high-dimensional time series models},
	volume = {43},
	year = {2015}}

@article{tong2021accelerating,
	author = {Tong, Tian and Ma, Cong and Chi, Yuejie},
	journal = {Journal of Machine Learning Research},
	number = {150},
	pages = {1--63},
	title = {Accelerating ill-conditioned low-rank matrix estimation via scaled gradient descent},
	volume = {22},
	year = {2021}}

@article{cai2021learned,
	author = {Cai, HanQin and Liu, Jialin and Yin, Wotao},
	journal = {Advances in Neural Information Processing Systems},
	pages = {16977--16989},
	title = {Learned robust PCA: A scalable deep unfolding approach for high-dimensional outlier detection},
	volume = {34},
	year = {2021}}

@article{yi2016fast,
	author = {Yi, Xinyang and Park, Dohyung and Chen, Yudong and Caramanis, Constantine},
	journal = {Advances in neural information processing systems},
	title = {Fast algorithms for robust PCA via gradient descent},
	volume = {29},
	year = {2016}}

@article{jain2014iterative,
	author = {Jain, Prateek and Tewari, Ambuj and Kar, Purushottam},
	journal = {Advances in neural information processing systems},
	title = {On iterative hard thresholding methods for high-dimensional {M}-estimation},
	volume = {27},
	year = {2014}}

@inproceedings{stich2018sparsified,
	author = {Stich, Sebastian U. and Cordonnier, Jean-Baptiste and Jaggi, Martin},
	booktitle = {Advances in Neural Information Processing Systems},
	title = {Sparsified {SGD} with Memory},
	volume = {31},
	year = {2018}}
}

\clearpage
\appendix
\etocdepthtag.toc{appendices}
\setstretch{1.2}
\renewcommand{\thesection}{\Alph{section}}
\numberwithin{equation}{section}
\numberwithin{algorithm}{section}
\numberwithin{figure}{section}
\numberwithin{table}{section}

\let\theorem\supplementtheorem
\let\endtheorem\endsupplementtheorem
\let\lemma\supplementlemma
\let\endlemma\endsupplementlemma
\let\proposition\supplementproposition
\let\endproposition\endsupplementproposition
\let\remark\supplementremark
\let\endremark\endsupplementremark
\let\condition\supplementcondition
\let\endcondition\endsupplementcondition
\let\corollary\supplementcorollary
\let\endcorollary\endsupplementcorollary
\let\assumption\supplementassumption
\let\endassumption\endsupplementassumption
\let\definition\supplementdefinition
\let\enddefinition\endsupplementdefinition
\let\example\supplementexample
\let\endexample\endsupplementexample

\begin{center}
\LARGE Supplementary Material for
Personalized Federated Vector Autoregression
with Personalization Diversity\par
\end{center}
\vspace{0.5em}

\begingroup
\small
\begin{quotation}
This supplementary material provides the complete results, technical proofs, auxiliary results, and implementation details
  for \textsc{PerFeCT-VAR}. Appendix~\ref{app:personalization-diversity}
  establishes shared--personalized identification and formulates the FCT
  projection. Appendix~\ref{app:reduced-communication} gives the
  communication procedures and lossless screening guarantee.
  Appendix~\ref{app:full-theory} gives complete statements of the
convergence, statistical accuracy, and initialization results
in the main article, with explicit conditions and parameter dependence.
  Appendix~\ref{app:main-theorem-proof} proves these results,
  using the concentration and deterministic results in
  Appendices~\ref{app:stochastic-results} and~\ref{sec:app-deterministic}.
  Appendix~\ref{app:simulation-details} specifies the simulation designs.
\end{quotation}
\endgroup

\begingroup
\etocsettagdepth{main}{none}
\etocsettagdepth{appendices}{subsection}
\etocsetnexttocdepth{subsection}
\renewcommand{\contentsname}{Appendix Contents}
\tableofcontents
\endgroup
\clearpage

\section{Personalization Diversity: Identifiability and FCT}
\label{app:personalization-diversity}

This appendix proves shared--personalized identification under exact personalization diversity and formulates the FCT projection used for estimation.

\subsection{Identifiability of the shared--personalized decomposition}
\label{subsec:app-global-local-identifiability}

Under the stationary VAR model, each transition matrix is identified by
\[
 \bm A_{i,*}=\mathbb E[\bm y_{i,t}\bm y_{i,t-1}^{\top}]
 \bbm\Gamma_i^{-1}.
\]
Shared--personalized identifiability concerns whether the transition matrices \(\{\bm A_i\}\) uniquely determine the shared component \(\bm G\) and personalized components \(\{\bm S_i\}\) in \(\bm A_i=\bm G+\bm S_i\). Without any structural restrictions, replacing \((\bm G,\bm S_i)\) by \((\bm G+\bm H,\bm S_i-\bm H)\) for $\bm H\in \mathbb{R}^{p\times p}$ leaves every \(\bm A_i\) unchanged. We therefore seek structural conditions that ensure uniqueness of this decomposition.

\paragraph*{Identification under the frequency cap}
Intuitively, a value shared by a strict majority of clients at a coordinate
should be the shared baseline there. To formalize this intuition, we call
a decomposition admissible if its personalized supports
\((\supp(\bm S_i))_{i=1}^N\) belong to \(\mathcal D(\{s_i\},m)\).
Within this class, \(m<N/2\) guarantees such a majority at every coordinate,
uniquely determining the shared--personalized decomposition. We now prove
Proposition~\ref{prop:frequency_capped_identifiability}.

\begin{proof}[Proof of Proposition~\ref{prop:frequency_capped_identifiability}]
Consider two decompositions whose maximal support frequencies are each strictly below $N/2$. Let $m$ be the larger of those frequencies, so that $2m<N$, and set \(\bm H=\widetilde{\bm G}-\bm G\). Then
\begin{equation*}
  \widetilde{\bm S}_i-\bm S_i=-\bm H,
  \qquad i=1,\ldots,N.
\end{equation*}
Fix a coordinate $(j,k)$.  By the two frequency-cap conditions, the set
\begin{equation*}
  \bigl\{i:(\bm S_i)_{jk}\ne0\bigr\}
  \cup
  \bigl\{i:(\widetilde{\bm S}_i)_{jk}\ne0\bigr\}
\end{equation*}
has cardinality at most $2m<N$.  Hence, there exists a client $i$ for which
both $(\bm S_i)_{jk}=0$ and $(\widetilde{\bm S}_i)_{jk}=0$.  Evaluating the
preceding identity at $(j,k)$ gives $H_{jk}=0$.  Since $(j,k)$ was arbitrary,
$\bm H=\bm 0$, and thus $\widetilde{\bm G}=\bm G$.  The same identity then
implies $\widetilde{\bm S}_i=\bm S_i$ for every $i=1,\ldots,N$.
\end{proof}

\paragraph*{Uniform and pointwise identification}
The proposition establishes \emph{uniform identification}: every admissible
configuration has a unique shared--personalized decomposition. When \(m\ge N/2\),
this guarantee can fail, while a particular configuration may still have a
unique decomposition and thus be \emph{pointwise identifiable}.

To illustrate this distinction, let \(N=5\) and \(m=3\), and fix a
coordinate \((j,k)\). Set all other coordinates of the transition matrices
to zero and assume \(s_i\ge1\) for every client. If
\[
  \bigl((\bm A_i)_{jk}\bigr)_{i=1}^5=(0,0,1,1,1),
\]
then this coordinate admits two decompositions satisfying the frequency
cap: one with
\(\bm G_{jk}=0\) and
\([(\bm S_i)_{jk}]_{i=1}^5=(0,0,1,1,1)\), and another with
\(\bm G_{jk}=1\) and
\([(\bm S_i)_{jk}]_{i=1}^5=(-1,-1,0,0,0)\).
Thus, the shared value is not uniquely determined.

In contrast, if
\([(\bm A_i)_{jk}]_{i=1}^5=(0,0,0,0,1)\), then
\(\bm G_{jk}=0\) is the only admissible shared value: any alternative
would make at least four personalized entries nonzero and hence violate
the cap \(m=3\). All other coordinates have the unique baseline zero,
so this configuration is pointwise identifiable despite \(m\ge N/2\).

More generally, if \(N/2\le m<N\), assigning value \(0\) to \(m\)
clients and value \(1\) to the remaining \(N-m\) clients makes either
value an admissible baseline. For \(m=N\), every \(g\in\mathbb R\) is admissible.
This construction shows that the strict-majority threshold is sharp for uniform identification.

\paragraph*{Recovery of the shared and personalized components}
Suppose the transition matrices are recovered exactly, so that
\(\widehat{\bm A}_i=\bm A_{i,*}\) for all \(i\).
If the full population personalized supports satisfy the exact frequency cap
\(m_*\), and \(s_i\ge s_{i,*}\) and \(m_*\le m<N/2\), any admissible fitted
decomposition \(\widehat{\bm A}_i=\widehat{\bm G}+\widehat{\bm S}_i\)
must coincide with the true decomposition:
\(\widehat{\bm G}=\bm G_*\) and
\(\widehat{\bm S}_i=\bm S_{i,*}\) for all \(i\).
Thus, the estimators also recover \(\bm G_*\) and all \(\bm S_{i,*}\)
exactly.

\subsection{Computing the FCT}
\label{subsec:app-fct-computation}

We formulate the FCT projection in \eqref{eq:FCT} as a maximum-weight bipartite \(b\)-matching problem and describe its solution via linear programming.

\paragraph*{Bipartite matching formulation}
For proposals \(\{\bm M_i\}\), let
\(z_{i,jk}\in\{0,1\}\) indicate whether coordinate \((j,k)\) is selected
for client \(i\). For a fixed feasible support allocation $\{z_{i,jk}\}$, the projection in
\eqref{eq:FCT} retains the selected proposal entries, so the objective of FCT equals
\[
 \sum_{i=1}^N w_i\|\bm M_i\|_{\mathrm F}^2
 -\sum_{i=1}^N w_i\sum_{j,k=1}^p |(\bm M_i)_{jk}|^2z_{i,jk}.
\]
Thus, FCT reduces to the binary support-allocation problem
\begin{equation}
\label{eq:app-fct-support-allocation}
\begin{aligned}
 \max_{\bm z\in\{0,1\}^{N\times p\times p}}\quad
 &\sum_{i=1}^N w_i\sum_{j,k=1}^p |(\bm M_i)_{jk}|^2z_{i,jk}\\
 \text{subject to}\quad
 &\sum_{j,k=1}^p z_{i,jk}\le s_i,\qquad i\in[N],\\
 &\sum_{i=1}^N z_{i,jk}\le m,\qquad (j,k)\in[p]\times[p].
\end{aligned}
\end{equation}
Construct a bipartite graph with client nodes \(u_i\) and coordinate nodes
\(v_{jk}\). Selecting edge \((u_i,v_{jk})\) yields weight
\(w_i|(\bm M_i)_{jk}|^2\) and uses one unit of capacity at each endpoint.
The client capacities are
\(s_i\), and the coordinate capacities are \(m\). Therefore,
\eqref{eq:app-fct-support-allocation} is a maximum-weight bipartite
\(b\)-matching problem with these node capacities
\citep{korte2018combinatorial}.
An optimal allocation \(\widehat{\bm z}\) gives the FCT output through
\((\bm S_i^+)_{jk}=\widehat z_{i,jk}(\bm M_i)_{jk}\).

\paragraph*{Integral relaxation and implementation}
We solve \eqref{eq:app-fct-support-allocation} via its linear programming
(LP) relaxation in \(\bm z\), replacing \(z_{i,jk}\in\{0,1\}\) by
\(0\le z_{i,jk}\le1\) while keeping the objective and capacity constraints
unchanged. The capacity constraint matrix is the totally unimodular
incidence matrix of the bipartite graph, and its right-hand side
\(\bm b=(s_1,\ldots,s_N,m,\ldots,m)^\top\) is integral.
Together with the $[0,1]$ variable bounds, the LP relaxation admits a binary optimal solution that is also optimal for the original problem. We use \texttt{scipy.optimize.linprog} with \texttt{method="highs-ds"},
the HiGHS dual revised simplex solver.

\section{Communication-Efficient Implementation}
\label{app:reduced-communication}

This appendix details the communication-efficient implementation described
in Section~\ref{sec:communication} of the main article. We present the communication and update procedures, followed by the complete algorithm and the proof
of the lossless FCT screening guarantee.

\subsection{Complete procedure}
\label{subsec:app-communication-procedure}

\paragraph*{Lossless FCT screening}
At a refresh round \(l\), client \(i\) forms the proposal
\(\bm S_i^{(l+1/2)}\) in \eqref{eq:local_Si} and the scores
\(q_{i,jk}^{(l)}=|(\bm S_i^{(l+1/2)})_{jk}|^2\). Defining
\begin{equation}
\label{eq:app-fct-candidate-budget}
 L_i:=\min\left\{
 p^2,\ s_i+\left\lfloor\frac{\sum_{h\ne i}s_h}{m}\right\rfloor
 \right\},
\end{equation}
client \(i\) uploads only its \(L_i\) largest scores and their corresponding
index pairs \((j,k)\). The server then solves
\eqref{eq:app-fct-support-allocation} with \(z_{i,jk}=0\) whenever
\((j,k)\) is not reported by client \(i\).
The following proposition establishes that this restriction is lossless.

\begin{proposition}[Lossless FCT screening]
\label{prop:app-lossless-fct-screening}
For each client \(i\), restrict the FCT bipartite graph to the \(L_i\)
candidate edges defined above. The resulting graph contains an optimizer of
the full FCT problem in \eqref{eq:FCT}.
\end{proposition}

The proof is given in Appendix~\ref{subsec:det-lossless-fcht-screening}.

\paragraph*{Communication for the shared updates}
In Algorithm~\ref{alg:perfect_var}, the exact server-side update
of \(\bm S_0\) requires each client to upload a dense \(p\times p\)
gradient matrix. We reduce this cost using Top-\(\widetilde s_0\)
sparsification with error feedback
\citep{stich2018sparsified}.
Specifically, client \(i\) uploads the sparse matrix \(\widetilde{\bm D}_i^{(l)}\) instead of the full gradient \(\bm D_i^{(l)}\). With \(\bm e_i^{(0)}=\bm 0\), client \(i\) computes
\begin{equation}
\label{eq:app-error-feedback}
\begin{aligned}
 \widetilde{\bm D}_i^{(l)}
 &=\operatorname{Top}_{\widetilde s_0}
   (\bm D_i^{(l)}+\bm e_i^{(l)}),\\
 \bm e_i^{(l+1)}
 &=\bm D_i^{(l)}+\bm e_i^{(l)}-\widetilde{\bm D}_i^{(l)},
\end{aligned}
\end{equation}
where \(\operatorname{Top}_k(\cdot)\) retains the \(k\) entries with the largest absolute values and sets all other entries to zero.
The server then updates
\begin{equation}
\label{eq:app-shared-sparse-update}
 \bm S_0^{(l+1)}
 =\mathcal G_{s_0,\gamma}\!\left(
   \bm S_0^{(l)}-\eta_{\mathrm S}\sum_{i=1}^Nw_i\widetilde{\bm D}_i^{(l)}
 \right).
\end{equation}
At the same time, each client uploads the two \(p\times r\) factor gradients
\(\bm D_i^{(l)}\bm R^{(l)}\) and
\((\bm D_i^{(l)})^\top\bm L^{(l)}\) for the low-rank update in
\eqref{eq:server_LR}.
Together, they avoid uploading the full gradient matrix.

\paragraph*{Selective FCT refresh}
When the selected supports stabilize in later iterations, refreshing
them at every round may be unnecessary. This motivates selective FCT
refresh as an additional way to reduce communication.
The server specifies a set of refresh rounds
\(\mathcal R\subseteq\{0,\ldots,K-1\}\). At rounds
\(l\in\mathcal R\), clients upload their indexed scores and the server
applies FCT to update the supports. At other rounds, no scores are
uploaded, and each client updates \(\bm S_i\) on its previous support
according to \eqref{eq:client_HT}.

\paragraph*{Uplink communication cost}
With the three procedures described above, we count post-initialization uplink communication. At each round,
client \(i\) uploads \(2pr\) factor-gradient entries and at most
\(\widetilde s_0\) indexed entries of the compressed gradient
\(\widetilde{\bm D}_i^{(l)}\). At refresh rounds, it additionally
uploads \(L_i\) indexed FCT scores. Its uplink cost is
therefore
\[
 \begin{cases}
 \mathcal O(pr+\widetilde s_0+L_i), & l\in\mathcal R,\\
 \mathcal O(pr+\widetilde s_0), & l\notin\mathcal R,
 \end{cases}
\]
at round \(l\), and
\(\mathcal O\bigl(K(pr+\widetilde s_0)+|\mathcal R|L_i\bigr)\)
over \(K\) rounds. If \(s_i=s\) for all clients and \(1\le m\le N\), then
\[
 L_i=\min\left\{p^2,\ s+\left\lfloor\frac{(N-1)s}{m}\right\rfloor\right\}
 =\mathcal O\!\left(\min\{Ns/m,p^2\}\right),
\]
giving the score-transmission bound stated in the main article. Screening also reduces
the server's allocation problem to a graph with at most \(\sum_iL_i\)
edges instead of \(Np^2\).

We summarize the complete communication-efficient implementation in Algorithm~\ref{alg:efficient_HT_FCHT} below.

\begin{algorithm}[htb]
\caption{Communication-Efficient Implementation of \textsc{PerFeCT-VAR}}
\label{alg:efficient_HT_FCHT}
\small\linespread{1.05}\selectfont
\begin{algorithmic}[1]
\Require Feasible initial values
\(\bm L^{(0)},\bm R^{(0)},\bm S_0^{(0)},\{\bm S_i^{(0)}\}_{i=1}^N\);
sparsity levels \(s_0,\{s_i\}_{i=1}^N\); row--column sparsity level
\(\gamma\); frequency cap \(m\ge1\); step sizes
\(\eta_{\mathrm{LR}},\eta_{\mathrm S},\eta_{\mathrm P}\); gradient sparsification budget
\(\widetilde s_0\); number of communication
rounds \(K\); and refresh set \(\mathcal R\subseteq\{0,\ldots,K-1\}\).
\State Server sets \(w_i=T_i/\Ttot\) and computes \(L_i\) by \eqref{eq:app-fct-candidate-budget} for each $i\in[N]$.
\State Server sends
$\bm L^{(0)},\bm R^{(0)},\bm S_0^{(0)}$, and $L_i$ to client $i$.
\State Each client $i$ stores
$\widehat\Omega_i^{(0)}=\supp(\bm S_i^{(0)})$ and sets
$\bm e_i^{(0)}=\bm 0_{p\times p}$.
\For{$l=0,1,\ldots,K-1$}
  \State Each client $i$ computes $\bm D_i^{(l)}$ by
  \eqref{eq:local_gradient} and $\bm S_i^{(l+1/2)}$ by
  \eqref{eq:local_Si}.
  \State Each client $i$ computes $\widetilde{\bm D}_i^{(l)}$ and
  $\bm e_i^{(l+1)}$ by \eqref{eq:app-error-feedback}, and uploads
  $\bm D_i^{(l)}\bm R^{(l)}$, $(\bm D_i^{(l)})^\top\bm L^{(l)}$, and $\widetilde{\bm D}_i^{(l)}$.
  \State Each client $i$ uploads its $L_i$ largest indexed scores
  $q_{i,jk}^{(l)}=|(\bm S_i^{(l+1/2)})_{jk}|^2$ if $l\in\mathcal R$.
  \State Server updates $\bm L^{(l+1)},\bm R^{(l+1)}$ by
  \eqref{eq:server_LR} using the weighted averages of the uploaded
  factor gradients, and updates $\bm S_0^{(l+1)}$ by
  \eqref{eq:app-shared-sparse-update}.
  \If{$l\in\mathcal R$}
    \State Server solves \eqref{eq:app-fct-support-allocation} with
    $z_{i,jk}=0$ for unreported client--coordinate pairs to obtain
    $\{\widehat\Omega_i^{(l+1)}\}_{i=1}^N$.
  \Else
    \State Each client $i$ retains
    $\widehat\Omega_i^{(l+1)}=\widehat\Omega_i^{(l)}$.
  \EndIf
  \State Server sends $\bm L^{(l+1)},\bm R^{(l+1)},\bm S_0^{(l+1)}$
  to each client $i$, together with $\widehat\Omega_i^{(l+1)}$ if
  $l\in\mathcal R$.
  \State Each client $i$ sets
  $\bm S_i^{(l+1)}=\mathcal P_{\widehat\Omega_i^{(l+1)}}
  (\bm S_i^{(l+1/2)})$.
\EndFor
\Ensure Server returns
$\widehat{\bm L}=\bm L^{(K)}$, $\widehat{\bm R}=\bm R^{(K)}$, and
$\widehat{\bm S}_0=\bm S_0^{(K)}$; each client $i$ retains
$\widehat{\bm S}_i=\bm S_i^{(K)}$.
\end{algorithmic}
\end{algorithm}

\subsection{Proof of Proposition~\ref{prop:app-lossless-fct-screening}}
\label{subsec:det-lossless-fcht-screening}

\begin{proof}
Fix a communication round and suppress the round index. Let
\(\mathcal U_i\subset[p]\times[p]\) be the set of \(L_i\) coordinates
uploaded by client \(i\).  Since these coordinates have the \(L_i\) largest
scores,
\[
  q_{i,cd}\ge q_{i,ab}
  \qquad
  \text{whenever }(c,d)\in\mathcal U_i
  \text{ and }(a,b)\notin\mathcal U_i.
\]

Among all optimal support allocations for the full FCT problem in
\eqref{eq:app-fct-support-allocation}, choose one, denoted by
\((\widehat\Omega_i)_{i=1}^N\), that minimizes
\begin{equation}\label{eq:appC_uplodad_mismatch}
  \sum_{i=1}^N|\widehat\Omega_i\setminus\mathcal U_i|.
\end{equation}
It suffices to show that this minimum is zero.  Suppose otherwise.  Then, for
some client \(i\), there exists
\((a,b)\in\widehat\Omega_i\setminus\mathcal U_i\).
This implies \(L_i<p^2\), and hence
$
  L_i
  =s_i+\left\lfloor\sum_{h\ne i}s_h/m\right\rfloor.
$
Moreover, because \(\widehat\Omega_i\) contains an element outside
\(\mathcal U_i\),
\[
\begin{aligned}
  |\mathcal U_i\setminus\widehat\Omega_i|
  &=L_i-|\mathcal U_i\cap\widehat\Omega_i|\\
  &\ge L_i-(|\widehat\Omega_i|-1)
  \ge L_i-s_i+1\\
  &=\left\lfloor\frac{\sum_{h\ne i}s_h}{m}\right\rfloor+1.
\end{aligned}
\] 
Then, at least one \((c,d)\in\mathcal U_i\setminus\widehat\Omega_i\) must satisfy
\[
  \sum_{h=1}^N\mathbf 1\{(c,d)\in\widehat\Omega_h\}\le m-1.
\]
Otherwise, every coordinate in
\(\mathcal U_i\setminus\widehat\Omega_i\) would attain the frequency cap.
Since none of them belongs to \(\widehat\Omega_i\), this would imply
\[
  m|\mathcal U_i\setminus\widehat\Omega_i|
  =\sum_{(c,d)\in\mathcal U_i\setminus\widehat\Omega_i}
    \sum_{h\ne i}\mathbf 1\{(c,d)\in\widehat\Omega_h\}
    \le\sum_{h\ne i}|\widehat\Omega_h|
  \le\sum_{h\ne i}s_h.
\]
However, the preceding lower bound gives
\[
  m|\mathcal U_i\setminus\widehat\Omega_i|
  \ge
  m\left(
    \left\lfloor\frac{\sum_{h\ne i}s_h}{m}\right\rfloor+1
  \right)
  >\sum_{h\ne i}s_h,
\]
which gives a contradiction.

For such a coordinate \((c,d)\), replace the selected but unreported
coordinate \((a,b)\in\widehat\Omega_i\setminus\mathcal U_i\) by \((c,d)\),
and define
\[
  \widetilde\Omega_i
  :=(\widehat\Omega_i\setminus\{(a,b)\})\cup\{(c,d)\},
  \qquad
  \widetilde\Omega_h:=\widehat\Omega_h\quad(h\ne i).
\]
This replacement is feasible, and changes the
objective value by \(w_i(q_{i,cd}-q_{i,ab})\ge0\).
Since \((\widehat\Omega_i)_{i=1}^N\) is already optimal, the objective cannot
strictly increase; hence \((\widetilde\Omega_i)_{i=1}^N\) is also an optimizer
of the full FCT problem.  Moreover, because
\((a,b)\in\widehat\Omega_i\setminus\mathcal U_i\) and
\((c,d)\in\mathcal U_i\setminus\widehat\Omega_i\),
\[
  \widetilde\Omega_i\setminus\mathcal U_i
  =
  (\widehat\Omega_i\setminus\mathcal U_i)\setminus\{(a,b)\}.
\]
For every \(h\ne i\), we also have
\(\widetilde\Omega_h\setminus\mathcal U_h
=\widehat\Omega_h\setminus\mathcal U_h\).  Therefore,
\[
  \sum_{h=1}^N|\widetilde\Omega_h\setminus\mathcal U_h|
  =
  \sum_{h=1}^N|\widehat\Omega_h\setminus\mathcal U_h|-1.
\]
This contradicts the defining choice of
\((\widehat\Omega_i)_{i=1}^N\) as an optimizer minimizing \eqref{eq:appC_uplodad_mismatch}.
Consequently,
\[
  \sum_{h=1}^N|\widehat\Omega_h\setminus\mathcal U_h|=0.
\]
Hence \(\widehat\Omega_h\subseteq\mathcal U_h\) for every \(h\in[N]\), so
this full-problem optimizer is contained in the restricted graph.
\end{proof}

\section{Complete Statements of the Main Results}
\label{app:full-theory}

This appendix states quantitative convergence, statistical accuracy,
and initialization bounds. Appendix~\ref{subsec:app-full-curvature}
defines the curvature coefficients.
Theorems~\ref{thm:app-unified-main-full} and~\ref{thm:app-init-full}
give the complete versions of Theorems~\ref{thm:main_fed_iht}
and~\ref{thm:initializer_rate}, respectively.
Corollaries~\ref{cor:app-stat-full} and~\ref{cor:app-balanced-full}
establish the accuracy bounds in Section~\ref{subsec:statistical_rates}.
Proofs are given in Appendix~\ref{app:main-theorem-proof}.
Appendix~\ref{subsec:app-distributional-discrepancy} illustrates the
directional covariance discrepancy under several covariance structures.
All update guarantees concern the exact implementation of Algorithm~\ref{alg:perfect_var}.

\subsection{Population curvature}
\label{subsec:app-full-curvature}

We use the notation of the main article and write
\(\bm B:=\bm L\bm R^\top\) and
\(\widehat{\bm B}:=\widehat{\bm L}\widehat{\bm R}^\top\)
for brevity. We also set \(w_{\max}:=\max_{i\in[N]}w_i=T_{\max}/\Ttot\).

For the shared low-rank component, define the pooled curvature and
smoothness coefficients
\begin{equation}
\label{eq:app-full-pooled-curvature}
 \alpha_{\mathrm{LR}}:=\lambda_{\min}(\bbm\Gamma_w),\qquad
 \beta_{\mathrm{LR}}:=\lambda_{\max}(\bbm\Gamma_w),\qquad
 \kappa_{\mathrm{LR}}:=\frac{\beta_{\mathrm{LR}}}{\alpha_{\mathrm{LR}}}.
\end{equation}
To define restricted curvature before choosing the fitting levels, fix
numerical upper bounds on the fitting multiples used below. Choose
\(C\) sufficiently large and define the envelopes as
\begin{equation}
\label{eq:app-full-primitive-envelopes}
\begin{aligned}
 s_0'&:=C\kappa_{\mathrm{LR}}^2s_{0,*},
 &\gamma'&:=C\kappa_{\mathrm{LR}}^2\gamma_*,~~
 s_i'&:=C\left(\frac{\phi_{\max}}{\phi_{\min}}\right)^2s_{i,*},~~
 &\mathcal C_{\mathrm S}'&:=\mathcal C(s_0',\gamma').
\end{aligned}
\end{equation}
Define the shared sparse and personalized sparse coefficients by
\begin{equation}
\label{eq:app-full-population-curvature-coefficients}
\begin{aligned}
 \alpha_{\mathrm S}
 &:=\inf_{\substack{\|\bm H\|_{\mathrm F}=1\\
              \supp(\bm H)\in\mathcal C_{\mathrm S}'}}
       \operatorname{tr}(\bm H\bbm\Gamma_w\bm H^\top),
 &\beta_{\mathrm S}
 &:=\sup_{\substack{\|\bm H\|_{\mathrm F}=1\\
              \supp(\bm H)\in\mathcal C_{\mathrm S}'}}
       \operatorname{tr}(\bm H\bbm\Gamma_w\bm H^\top),\\
 \alpha_i
 &:=\inf_{\substack{\|\bm H\|_{\mathrm F}=1\\
                    |\supp(\bm H)|\le s_i'}}
       \operatorname{tr}(\bm H\bbm\Gamma_i\bm H^\top),
 &\beta_i
 &:=\sup_{\substack{\|\bm H\|_{\mathrm F}=1\\
                    |\supp(\bm H)|\le s_i'}}
       \operatorname{tr}(\bm H\bbm\Gamma_i\bm H^\top).
\end{aligned}
\end{equation}
Set
\begin{equation}
\label{eq:app-full-curvature-ratios}
\begin{aligned}
 \alpha_\mathrm{P}&:=\min_{i\in[N]}\alpha_i,
 &\beta_\mathrm{P}&:=\max_{i\in[N]}\beta_i,\\
 \kappa_{\mathrm S}&:=\frac{\beta_{\mathrm S}}{\alpha_{\mathrm S}},
 &\kappa_{\mathrm{P}}&:=\frac{\beta_\mathrm{P}}{\alpha_\mathrm{P}},
 &\kappa_{\max}&:=\max\{\kappa_{\mathrm{LR}},\kappa_{\mathrm S},\kappa_{\mathrm{P}}\}.
\end{aligned}
\end{equation}
The curvature coefficients satisfy
\begin{equation}
\label{eq:app-full-curvature-order}
\begin{aligned}
 \phi_{\min}&\le\alpha_{\mathrm{LR}}\le\alpha_{\mathrm S}\le\beta_{\mathrm S}\le\beta_{\mathrm{LR}}\le\phi_{\max},\\
 \phi_{\min}&\le\alpha_\mathrm{P}\le\beta_\mathrm{P}\le\phi_{\max}.
\end{aligned}
\end{equation}
Consequently, the condition numbers satisfy
\[
 \kappa_{\mathrm S}\le\kappa_{\mathrm{LR}},
 \qquad
 \kappa_{\mathrm{P}}\le\frac{\phi_{\max}}{\phi_{\min}}.
\]

\subsection{Joint contraction of \textsc{PerFeCT-VAR} (Theorem~\ref{thm:main_fed_iht})}
\label{subsec:app-unified-main-theorem}

The following theorem makes explicit the fitting, sample-size,
signal-separation, and approximation conditions in
Theorem~\ref{thm:main_fed_iht}. We write
\[
 \nu=(\alpha_{\mathrm{LR}},\beta_{\mathrm{LR}},
       \alpha_{\mathrm S},\beta_{\mathrm S},
       \alpha_{\mathrm P},\beta_{\mathrm P},\phi_{\max},\phi_{\min}).
\]
Let $\bm e^{(l)}$ denote the population-error vector defined in
Section~\ref{subsec:theoretical_setup}, and set
\[
 \bbm\ell=
 \left(1,\frac{\beta_{\mathrm S}}{\beta_{\mathrm{LR}}},
          \frac{\beta_{\mathrm P}}{\beta_{\mathrm{LR}}}\right)^\top,
 \qquad V^{(l)}=\bbm\ell^\top\bm e^{(l)}.
\]
For the enlarged support classes, write
\[
 s_0^\dagger=\min\{2s_0+s_{0,*},p^2\},\qquad
 \gamma^\dagger=\min\{2\gamma+\gamma_*,1\},\qquad
 s_i^\dagger=\min\{2s_i+s_{i,*},p^2\}.
\]
The directional discrepancies $\Lambda_a(m)$ and the interaction scales
$L_{\mathrm{LR},\mathrm S}$, $L_{\mathrm{LR},\mathrm P}(m)$, and
$L_{\mathrm S,\mathrm P}(m)$ are defined in
Section~\ref{subsec:covariance_compatibility},
using these support classes and the local radius specified below.

Define the sample-size multipliers
\begin{equation}
\label{eq:app-unified-sample-multipliers}
\begin{aligned}
 \Xi_1
 &:=\max\left\{1,\frac M{\phi_{\max}},
       \max_{j\in\{\mathrm{LR},\mathrm S,\mathrm P\}}
                    \frac{M\phi_{\max}}{\alpha_j^2},
       \frac{M^2}{\alpha_{\mathrm{LR}}\alpha_{\mathrm P}},
       \frac{M^2}{\alpha_{\mathrm S}\alpha_{\mathrm P}}\right\},\\
 \Xi_2
 &:=\frac{\tau^2\phi_{\max}\beta_{\mathrm{LR}}\kappa_{\max}}
              {\alpha_{\mathrm{LR}}\alpha_{\mathrm S}\sigma_r^2}.
\end{aligned}
\end{equation}
For a sufficiently small numerical constant $\varepsilon_0>0$,
set
\begin{equation}
\label{eq:app-unified-local-radius}
 \varepsilon_\nu^2
   =\varepsilon_0\frac{\alpha_{\mathrm{LR}}\alpha_{\mathrm S}}
                         {\beta_{\mathrm{LR}}^2},\qquad
 0<\varepsilon_\nu<1/100.
\end{equation}
Put $\overline s_w=\sum_iw_i s_{i,*}$. For a sufficiently large
fixed numerical constant $C>1$, the approximation envelope is
\begin{equation}
\label{eq:app-unified-approximation-envelope}
 \mathcal A_q(m)=
 \begin{cases}
 0,&q=0,\quad m\ge C\kappa_{\mathrm P}^2m_*,\\[2pt]
 \displaystyle\overline s_w
       \left(\frac{C\kappa_{\mathrm P}^2m_{q,*}}m\right)^{2/q},
       &0<q\le1.
 \end{cases}
\end{equation}
\begin{assumption}[Cross-component compatibility; explicit form of Assumption~\ref{cond:population_compatibility}]
\label{cond:app-unified-compatibility}
For an integer working cap $1\le m<N/2$, require
\begin{equation}
\label{eq:app-unified-compatibility}
\begin{aligned}
 \gamma_*\le c\,
 \frac{\alpha_{\mathrm{LR}}\alpha_{\mathrm S}^3}
 {\beta_{\mathrm S}^2
  (\widetilde\phi_w^2+\beta_{\mathrm{LR}}^2)\mu r},\quad 
 m w_{\max}\{\phi_{\max}^2+\Lambda_a(m)^2\}
 \le c\,\alpha_a\alpha_{\mathrm P},
 \quad a\in\{\mathrm{LR},\mathrm S\}.
\end{aligned}
\end{equation}
Here $c$ is a small constant.
\end{assumption}

\begin{theorem}[Joint contraction of \textsc{PerFeCT-VAR}; full version of Theorem~\ref{thm:main_fed_iht}]
\label{thm:app-unified-main-full}
Fix $q\in[0,1]$ and an integer working cap $1\le m<N/2$.
Suppose Assumptions~\ref{assump:stationarity},
\ref{assump:frequency_cap},
\ref{assump:approx_personalization_diversity},
\ref{assump:shared_separation},
\ref{assump:gaussian_sampling}, and
\ref{assump:cov_adjusted_incoherence} hold, together with
Assumption~\ref{cond:app-unified-compatibility}.
Assume the following conditions, with a sufficiently large numerical
constant $C>0$.

\begin{enumerate}[label=\textup{(\roman*)}]
\item \textbf{Fitting levels and step sizes.}
Choose sufficiently large fixed fitting multiples and a sufficiently
small constant $c>0$ such that
\begin{equation}
\label{eq:app-unified-fitting-levels}
\begin{aligned}
 s_0&\asymp\kappa_{\mathrm S}^2s_{0,*},&
 \gamma&\asymp\kappa_{\mathrm S}^2\gamma_*,&
 s_i&\asymp\kappa_{\mathrm P}^2s_{i,*},\\
 \eta_{\mathrm{LR}}&=c\beta_{\mathrm{LR}}^{-1},&
 \eta_{\mathrm S}&=c\beta_{\mathrm S}^{-1},&
 \eta_{\mathrm P}&=c\beta_{\mathrm P}^{-1}.
\end{aligned}
\end{equation}

\item \textbf{Sample-size and signal-strength conditions.}
Require the pooled and clientwise sample sizes to satisfy
\begin{equation}
\label{eq:app-unified-sample-conditions}
\begin{aligned}
 \Ttot&\ge C\Xi_1(pr+s_0^\dagger\log p),\\
 T_i&\ge C\Xi_1
                   \{s_i^\dagger\log p+\log(pN)\},\qquad i\in[N],
\end{aligned}
\end{equation}
and
\begin{equation}
\label{eq:app-unified-signal-sample-condition}
 \Ttot\ge C\Xi_2
 \left\{\frac{pr}{\alpha_{\mathrm{LR}}}
       +\frac{s_0^\dagger\log p}{\alpha_{\mathrm S}}
       +\frac{\sum_i s_i^\dagger\log p+N\log N}
                   {\alpha_{\mathrm P}}\right\}.
\end{equation}

\item \textbf{Working cap.}
For $C$ as in \eqref{eq:app-unified-approximation-envelope} and a
numerical constant $c>0$ sufficiently small,
require
\begin{equation}
\label{eq:app-unified-cap-condition}
 \begin{cases}
 m\ge C\kappa_{\mathrm P}^2m_*,&q=0,\\[2pt]
 \displaystyle m\ge C\kappa_{\mathrm P}^2m_{q,*}
       \left(\frac{\kappa_{\max}^2\phi_{\max}^2
                         \beta_{\mathrm{LR}}\overline s_w}
                    {c\,\alpha_{\mathrm P}\alpha_{\mathrm{LR}}
                         \alpha_{\mathrm S}\sigma_r^2}\right)^{q/2},
       &0<q\le1.
 \end{cases}
\end{equation}

\item \textbf{Initialization.}
Require a feasible initializer  satisfying
\begin{equation}
\label{eq:app-unified-initial-condition}
 V^{(0)}\le\varepsilon_\nu^2\sigma_r^2.
\end{equation}
\end{enumerate}

Then, with probability at least
\begin{equation}
\label{eq:app-unified-main-probability}
 1-Ce^{-cp}-Cp^{-c}-\frac{2}{pN},
\end{equation}
the errors of Algorithm~\ref{alg:perfect_var} satisfy,
simultaneously for all $l\ge0$,
\begin{equation}
\label{eq:app-unified-vector-recursion}
 \bm e^{(l+1)}\preceq\bbm\Theta\bm e^{(l)}+\bm f,
 \qquad \bm f=\bm f_{\mathrm{stat}}+\bm f_{\mathrm{app},q}(m).
\end{equation}
For universal constants $c_0,C_0>0$, the entrywise nonnegative
contraction matrix can be chosen as
\begin{equation}
\label{eq:app-unified-transition-matrix}
 \bbm\Theta=
 \begin{pmatrix}
 1-c_0\alpha_{\mathrm{LR}}/\beta_{\mathrm{LR}}
 &\dfrac{C_0L_{\mathrm{LR},\mathrm S}}
              {\alpha_{\mathrm{LR}}\beta_{\mathrm{LR}}}
 &\dfrac{C_0L_{\mathrm{LR},\mathrm P}(m)}
              {\alpha_{\mathrm{LR}}\beta_{\mathrm{LR}}}\\[3pt]
 \dfrac{C_0L_{\mathrm{LR},\mathrm S}}
              {\alpha_{\mathrm S}\beta_{\mathrm S}}
 &1-c_0\alpha_{\mathrm S}/\beta_{\mathrm S}
 &\dfrac{C_0L_{\mathrm S,\mathrm P}(m)}
              {\alpha_{\mathrm S}\beta_{\mathrm S}}\\[3pt]
 \dfrac{C_0L_{\mathrm{LR},\mathrm P}(m)}
              {\alpha_{\mathrm P}\beta_{\mathrm P}}
 &\dfrac{C_0L_{\mathrm S,\mathrm P}(m)}
              {\alpha_{\mathrm P}\beta_{\mathrm P}}
 &1-c_0\alpha_{\mathrm P}/\beta_{\mathrm P}
 \end{pmatrix}.
\end{equation}
The statistical and approximation error vectors are
\begin{equation}
\label{eq:app-unified-forcing-vectors}
\begin{aligned}
 \bm f_{\mathrm{stat}}
 &=\frac{C_0\tau^2\phi_{\max}}{\Ttot}
 \begin{pmatrix}
 pr/(\alpha_{\mathrm{LR}}\beta_{\mathrm{LR}})\\
 s_0^\dagger\log p/(\alpha_{\mathrm S}\beta_{\mathrm S})\\
 (\sum_i s_i^\dagger\log p+N\log N)/(\alpha_{\mathrm P}\beta_{\mathrm P})
 \end{pmatrix},\\
 \bm f_{\mathrm{app},q}(m)
 &=C_0\kappa_{\max}\mathcal A_q(m)
 \begin{pmatrix}
 L_{\mathrm{LR},\mathrm P}(m)/(\alpha_{\mathrm{LR}}\beta_{\mathrm{LR}})\\
 L_{\mathrm S,\mathrm P}(m)/(\alpha_{\mathrm S}\beta_{\mathrm S})\\
 \phi_{\max}^2/(\alpha_{\mathrm P}\beta_{\mathrm P})
 \end{pmatrix}.
\end{aligned}
\end{equation}
With $\varrho=c_0/(2\kappa_{\max})\in(0,1)$,
\begin{equation}
\label{eq:app-unified-weighted-margin}
 \rho(\bbm\Theta)\le1-\varrho<1.
\end{equation}
Consequently,
\begin{equation}
\label{eq:app-unified-iterated-bound}
\begin{aligned}
 \bm e^{(l)}
 &\preceq\bbm\Theta^l\bm e^{(0)}
       +(\bm I_3-\bbm\Theta)^{-1}\bm f,\\
 V^{(l)}
 &\le(1-\varrho)^lV^{(0)}
       +\varrho^{-1}\bbm\ell^\top\bm f.
\end{aligned}
\end{equation}
\end{theorem}

Condition~\textup{(i)} requires fitting levels larger than their
true counterparts by factors determined by the squared curvature
ratios, and step sizes inversely proportional to the corresponding
RSS coefficients, agreeing with the IHT analysis of
\citet{jain2014iterative}.
Assumption~\ref{cond:app-unified-compatibility} controls the
interactions among the three components. Its shared low-rank--sparse
separation requirement is analogous to the sparsity conditions in
existing analyses of low-rank-plus-sparse recovery
\citep{yi2016fast,tong2021accelerating}.

Condition~\textup{(ii)} requires sufficient pooled and clientwise
sample sizes, depending on model complexity, temporal dependence,
and the noise and signal levels. Condition~\textup{(iii)} requires
a sufficiently large working cap to control the approximation
error relative to the signal strength, while
condition~\textup{(iv)} requires an initial estimate sufficiently
close to the truth.

The theorem establishes geometric convergence up to statistical
and approximation errors. When $q=0$, condition~\textup{(iii)}
gives $\mathcal A_0(m)=0$, so the approximation contribution
vanishes and the result reduces to the exact-diversity
contraction bound.

\subsection{Statistical rates (Corollary~\ref{cor:main_equal_sample_rates})}
\label{subsec:app-complete-statistical-rates}

We use the errors $\mathcal E_G^{(l)}$, $\mathcal E_{\mathrm P}^{(l)}$,
and $\mathcal E_A^{(l)}$ and the population complexities $D_G,D_P$
from Section~\ref{subsec:statistical_rates}. Write
\[
 D_{\mathrm G}^{\dagger}:=pr+s_0^{\dagger}\log p,\qquad
 D_{\mathrm P}^{\dagger}:=N^{-1}\sum_{i=1}^N s_i^{\dagger}\log p+\log N
\]
for the fitting-level complexities. Recall from
Section~\ref{subsec:covariance_compatibility} that
\begin{equation}
\label{eq:app-rate-omega}
\begin{aligned}
 \omega(m)
 &=L_{\mathrm{LR},\mathrm P}(m)\vee L_{\mathrm S,\mathrm P}(m)\\
 &=m w_{\max}\phi_{\max}^2
   +\frac{M^2}{\Ttot}\sum_{i=1}^N s_i^\dagger\log p\\
 &\quad+\max\left\{
 m w_{\max}\Lambda_{\mathrm{LR}}(m)^2+\frac{M^2pr}{\Ttot},
 \;
 m w_{\max}\Lambda_{\mathrm S}(m)^2
       +\frac{M^2s_0^\dagger\log p}{\Ttot}
 \right\}.
\end{aligned}
\end{equation}
Let $\varrho$ be as in \eqref{eq:app-unified-weighted-margin}.
The following corollary gives the full quantitative bounds underlying
Section~\ref{subsec:statistical_rates}.
\begin{corollary}[Componentwise statistical accuracy]
\label{cor:app-stat-full}
Suppose the conditions of Theorem~\ref{thm:app-unified-main-full} hold.
On the same event and with the same probability, simultaneously for
all $l\ge0$, the shared errors satisfy
\begin{equation}
\label{eq:app-full-general-shared-rate}
 \mathcal E_G^{(l)}
 \lesssim \kappa_{\max}(1-\varrho)^lV^{(0)}+
 \frac{\tau^2\phi_{\max}D_{\mathrm G}^{\dagger}}
      {\alpha_{\mathrm{LR}}^2\Ttot}\quad+
 \frac{\omega(m)}{\alpha_{\mathrm{LR}}^2\alpha_{\mathrm P}^2}
   \left\{\frac{\tau^2\phi_{\max}D_{\mathrm P}^{\dagger}}{\overline T}
          +\kappa_{\max}\phi_{\max}^2\mathcal A_q(m)\right\},
\end{equation}
and the personalized error satisfies
\begin{equation}
\label{eq:app-full-general-personalized-rate}
 \mathcal E_{\mathrm P}^{(l)}
 \lesssim \frac{\beta_{\mathrm{LR}}}{\beta_{\mathrm P}}(1-\varrho)^lV^{(0)}+
 \frac{\tau^2\phi_{\max}}{\alpha_{\mathrm P}\Ttot}
 \left\{
 \frac{ND_{\mathrm P}^{\dagger}}{\alpha_{\mathrm P}}
 +\frac{pr}{\alpha_{\mathrm{LR}}}
 +\frac{s_0^{\dagger}\log p}{\alpha_{\mathrm S}}
 \right\}+\frac{\kappa_{\max}\phi_{\max}^2}{\alpha_{\mathrm P}^2}
                   \mathcal A_q(m).
\end{equation}
\end{corollary}

The fitting levels in \eqref{eq:app-unified-fitting-levels} give
$D_{\mathrm G}^{\dagger}\lesssim_\nu D_G$ and
$D_{\mathrm P}^{\dagger}\lesssim_\nu D_P$, yielding the rates in
Section~\ref{subsec:statistical_rates}.
The transition-error bound in Section~\ref{subsec:statistical_rates}
follows from $\mathcal E_A^{(l)}\le6\mathcal E_G^{(l)}+3\mathcal E_{\mathrm P}^{(l)}$
and $\omega(m)\lesssim_\nu1$.
At $q=0$,
$\mathcal A_0(m)=0$, giving the exact-diversity rates.

For the balanced-client rate comparisons, we use the same admissible
cap range and the notation $\zeta=(M,\tau^2,\overline\Lambda)$ as in
Section~\ref{subsec:statistical_rates}.
The constants in $\lesssim_{\nu,\zeta}$ are uniform over this range.
\begin{corollary}[Full version of Corollary~\ref{cor:main_equal_sample_rates}]
\label{cor:app-balanced-full}
Suppose the conditions of Corollary~\ref{cor:app-stat-full} hold,
\(T_i=T\) and \(s_{i,*}=s_*\) for \(i\in[N]\).
Require
\begin{equation}
\label{eq:app-balanced-interaction-condition}
 Ns_*\log p\,D_P\lesssim T(D_G+mD_P),
 \qquad D_P=s_*\log p+\log N.
\end{equation}
For any integer number of rounds satisfying
\begin{equation}
\label{eq:app-full-statistical-rounds}
 K\ge\frac{\log\max\{1,NTV^{(0)}/D_G\}}{-\log(1-\varrho)},
\end{equation}
the following bounds hold on the same event:
\begin{align}
\label{eq:app-full-balanced-shared-rate}
 \mathcal E_G^{(K)}
 &\lesssim_{\nu,\zeta}
 \frac{D_G+mD_P}{NT}+\omega(m)\mathcal A_q(m),\\
\label{eq:app-full-balanced-personalized-rate}
 \mathcal E_{\mathrm P}^{(K)}\vee\mathcal E_A^{(K)}
 &\lesssim_{\nu,\zeta}
 \frac{D_P}{T}+\frac{D_G}{NT}+\mathcal A_q(m).
\end{align}
\end{corollary}

Under the stronger illustrative condition
$D_G+Ns_*\log p\lesssim mT$, the interaction satisfies
$\omega(m)\lesssim_{\nu,\zeta}m/N$. For $0<q\le1$, this gives
\begin{equation}
\label{eq:lq-balanced-rate}
 \mathcal E_G^{(K)}
 \lesssim_{\nu,\zeta}
 \frac{D_G+mD_P}{NT}
 +\frac{s_*\{C\kappa_{\mathrm P}^2m_{q,*}\}^{2/q}}
             {Nm^{2/q-1}},
\end{equation}
with $C$ as in \eqref{eq:app-unified-approximation-envelope}.
The proofs are given in Appendix~\ref{subsec:app-corollary-proofs}.

\subsection{Validity of the federated initializer (Theorem~\ref{thm:initializer_rate})}
\label{subsec:app-init-full}

The following theorem gives the complete conditions for the initializer
in Section~\ref{sec:algorithm_initialization}. We use $\Xi_2$,
$\varepsilon_0$, and $\mathcal A_q(m)$ from
Appendix~\ref{subsec:app-unified-main-theorem}, and write
$T_{\min}=\min_iT_i$. Define the clientwise and aggregate initialization error envelopes by
\begin{equation}
 R_i^{\mathrm{init}}:=
 \frac{\tau^2\phi_{\max}}{\phi_{\min}^2T_i}
 \left\{r(p+\log N)+(s_{0,*}+s_{i,*})\log(pN)\right\}+\frac{\mu^2r^2\sigma_1^2}{p^2}(s_{0,*}+s_{i,*}),
\end{equation}
and
\begin{equation}
\label{eq:note-initialization}
 \mathcal I_q(m)=\kappa_{\max}
 \left[
 \sum_{i=1}^N\left(w_i+\frac1N\right)R_i^{\mathrm{init}}
 +\left(1+\frac{\overline T}{T_{\min}}\right)\mathcal A_q(m)
 \right].
\end{equation}

\begin{theorem}[Full version of Theorem~\ref{thm:initializer_rate}]
\label{thm:app-init-full}
Fix $q\in[0,1]$. Suppose
Assumptions~\ref{assump:stationarity}--\ref{assump:regular_covariance}
hold. Assume the following conditions.

\begin{enumerate}[label=\textup{(\roman*)}]
\item \textbf{Fitting levels and personalization frequency.}
Choose sufficiently large fixed fitting multiples as in
\eqref{eq:app-unified-fitting-levels}, and an integer working cap such that
\begin{equation}
\label{eq:app-init-fitting-full}
\begin{aligned}
 s_0&\asymp\kappa_{\mathrm S}^2s_{0,*},\qquad
 \gamma\asymp\kappa_{\mathrm S}^2\gamma_*,\qquad
 s_i\asymp\kappa_{\mathrm P}^2s_{i,*},\\
 1&\le m<N/2,\qquad
 m\ge C\kappa_{\mathrm P}^2m_*\quad\text{if }q=0.
\end{aligned}
\end{equation}

\item \textbf{Spikiness radius and penalties.}
Choose the spikiness radius as $\varpi=\mu r\sigma_1$, and penalties as
\begin{equation}
 \lambda_{B,i}\asymp\tau
     \sqrt{\frac{\phi_{\max}(p+\log N)}{T_i}},\quad
 \lambda_{C,i}\asymp 
     \tau\sqrt{\frac{\phi_{\max}\log(pN)}{T_i}}
     +\phi_{\min}\frac{\mu r\sigma_1}{p},\quad i\in[N].
\end{equation}

\item \textbf{Sample-size and dependence conditions.}
For each client \(i\in[N]\), require
\begin{equation}
\label{eq:app-init-rsc-full}
 T_i\ge C M\phi_{\max}\phi_{\min}^{-2}(p+\log N).
\end{equation}
\end{enumerate}
Then, with probability at least
\(1-C\exp\{-c(p+\log N)\}-C(pN)^{-c}\), the federated initializer
satisfies
\begin{equation}
\label{eq:app-init-total-error-rate}
 V^{(0)}\le C\mathcal I_q(m).
\end{equation}

Suppose further that the following condition holds:
\begin{enumerate}[label=\textup{(\roman*)},start=4]
\item \textbf{Local basin entry.}
The sample sizes satisfy
\begin{equation}
\label{eq:app-init-basin-sample-full}
 T_i\ge C\frac{\beta_{\mathrm{LR}}\Xi_2}{\phi_{\min}^2}
 \left\{r(p+\log N)+(s_{0,*}+s_{i,*})\log(pN)\right\},
 \qquad i\in[N],
\end{equation}
and the sparsity and signal levels satisfy
\begin{equation}
\label{eq:app-init-separation-full}
 s_{0,*}+\max_{i\in[N]}s_{i,*}
 \le c\frac{\alpha_{\mathrm{LR}}\alpha_{\mathrm S}\sigma_r^2}
             {\kappa_{\max}\beta_{\mathrm{LR}}^2
              \mu^2r^2\sigma_1^2}\,p^2.
\end{equation}
For $0<q\le1$, the working cap additionally satisfies
\begin{equation}
\label{eq:app-init-cap-full}
 m\ge C\kappa_{\mathrm P}^2m_{q,*}
 \left\{
 \frac{\kappa_{\max}\beta_{\mathrm{LR}}^2
       (1+\overline T/T_{\min})\overline s_w}
      {c\,\alpha_{\mathrm{LR}}\alpha_{\mathrm S}\sigma_r^2}
 \right\}^{q/2}.
\end{equation}
\end{enumerate}
Under conditions \textup{(i)}--\textup{(iv)}, with the same probability,
\[
 V^{(0)}\le
 \frac{\varepsilon_0\alpha_{\mathrm{LR}}\alpha_{\mathrm S}}
      {\beta_{\mathrm{LR}}^2}\sigma_r^2,
\]
so the initializer enters the convergence region of Theorem \ref{thm:app-unified-main-full}.
\end{theorem}

When $q=0$, the approximation error vanishes as $\mathcal A_0(m)=0$, and the result reduces to
the exact-diversity initialization guarantee.
The proof is given in Appendix~\ref{app:initializer-proof}.

\subsection{Directional covariance discrepancy measures}
\label{subsec:app-distributional-discrepancy}

We give four examples to examine
\(\Lambda_{\mathrm{LR}}(m)\) and \(\Lambda_{\mathrm{S}}(m)\) under different covariance
structures. We first record an inequality used below. For any feasible
\(\bm X\) and \((\bm H_i)\) in
\eqref{eq:directional_heterogeneity_levels}, let
\(I_{jk}=\{i:(\bm H_i)_{jk}\ne0\}\).
The frequency cap gives \(|I_{jk}|\le m\). Applying
Cauchy--Schwarz to these nonzero entries yields
\[
 \left(\sum_i|(\bm H_i)_{jk}|\right)^2
 \le |I_{jk}|\sum_i|(\bm H_i)_{jk}|^2
 \le m\sum_i|(\bm H_i)_{jk}|^2.
\]
Summing over \((j,k)\) and using the normalization
\(\sum_i\|\bm H_i\|_{\mathrm F}^2=1\) gives
\begin{equation}
\label{eq:app-discrepancy-frequency-bound}
 \sum_{j,k}\left(\sum_i|(\bm H_i)_{jk}|\right)^2
 \le m\sum_{j,k}\sum_i|(\bm H_i)_{jk}|^2
 =m.
\end{equation}

\begin{enumerate}[label=(\arabic*),leftmargin=*]

\item \textbf{Diagonal deviations.}
Suppose \(\bm E_i\) is diagonal and
\(\max_{i,j}|(\bm E_i)_{jj}|\le\rho\).
Using Cauchy--Schwarz across coordinates and
\eqref{eq:app-discrepancy-frequency-bound}, we obtain
\[
 \left|\sum_i\langle\bm H_i,\bm X\bm E_i\rangle\right|
 \le\rho\sum_{j,k}|X_{jk}|\sum_i|(\bm H_i)_{jk}|\le\rho\|\bm X\|_{\mathrm F}
 \left\{\sum_{j,k}\left(\sum_i|(\bm H_i)_{jk}|\right)^2\right\}^{1/2}
 \le\rho\sqrt{m}.
\]
Dividing by \(\sqrt{m}\) and taking the supremum gives
\[
 \Lambda_{\mathrm{LR}}(m)\vee\Lambda_{\mathrm{S}}(m)\le\rho.
\]

\item \textbf{Banded deviations.}
Suppose
\[
 (\bm E_i)_{jk}=0\quad\text{if }|j-k|>b,
 \qquad
 \max_{i,j,k}|(\bm E_i)_{jk}|\le\rho,
\]
so each row of \(\bbm\Gamma_i\) has at most \(2b+1\) entries that
may differ from the corresponding entries of \(\overline{\bbm\Gamma}\).
Decompose $\bm E_i=\sum_{h=-b}^{b}\bm E_i^{(h)}$,
where \(\bm E_i^{(h)}\) retains only the entries with \(k-j=h\).
Applying the argument in the first example to each diagonal and summing gives
\[
 \left|\sum_i\langle\bm H_i,\bm X\bm E_i\rangle\right|
 \le\sum_{h=-b}^{b}
 \left|\sum_i\langle\bm H_i,\bm X\bm E_i^{(h)}\rangle\right|\le(2b+1)\rho\sqrt{m}.
\]
Dividing by
\(\sqrt{m}\) and taking the supremum yields
\[
 \Lambda_{\mathrm{LR}}(m)\vee\Lambda_{\mathrm{S}}(m)\le(2b+1)\rho.
\]
For fixed bandwidth and $\rho$, $\Lambda_\mathrm{LR}(m)$ and $\Lambda_{\mathrm{S}}(m)$ are uniformly bounded.

\item \textbf{Toeplitz deviations with quadratic decay.}
Let
\[
 (\bm E_i)_{jk}
 =
 \begin{cases}
  0, & j=k,\\
  a_i|j-k|^{-2}, & j\ne k,
 \end{cases}
 \qquad
 |a_i|\le\rho,\qquad \sum_i a_i=0.
\]
Write \(\bm E_i=a_i\bm D\), where \(D_{jj}=0\) and
\(D_{jk}=|j-k|^{-2}\) for \(j\ne k\).
Since \(|a_i|\le\rho\),
\eqref{eq:app-discrepancy-frequency-bound} gives
\[
 \left\|\sum_i a_i\bm H_i\right\|_{\mathrm F}^2
 \le\rho^2\sum_{j,k}\left(\sum_i|(\bm H_i)_{jk}|\right)^2
 \le\rho^2m.
\]
It follows that
\[
 \left|\sum_i\langle\bm H_i,\bm X\bm E_i\rangle\right|
 =\left|\left\langle\sum_i a_i\bm H_i,\bm X\bm D\right\rangle\right|\le\rho\sqrt{m}\|\bm X\bm D\|_{\mathrm F}
 \le\rho\sqrt{m}\|\bm D\|_{\mathrm{op}}.
\]
Because \(\bm D\) is symmetric, its operator norm is bounded by
its maximum absolute row sum:
\[
 \|\bm D\|_{\mathrm{op}}
 \le\max_j\sum_{k\ne j}|j-k|^{-2}
 \le2\sum_{h=1}^{p-1}h^{-2}\le4.
\]
Hence
\[
 \Lambda_{\mathrm{LR}}(m)\vee\Lambda_{\mathrm{S}}(m)
 \le
 2\rho\sum_{h=1}^{p-1}h^{-2}
 \le4\rho.
\]
Thus, summable decay gives a bound independent of \(N\) and \(p\).

\item \textbf{Finite deviation patterns.}
Let
\[
 \bm E_i=\sum_{l=1}^{K}a_{il}\bm D_l,
 \qquad
 |a_{il}|\le1,
 \qquad
 \sum_i a_{il}=0,
\]
where \(K\) is fixed and the patterns \(\bm D_l\) are symmetric.
For each \(l\), \(|a_{il}|\le1\) and
\eqref{eq:app-discrepancy-frequency-bound} imply
\[
 \left\|\sum_i a_{il}\bm H_i\right\|_{\mathrm F}
 \le\sqrt{m}.
\]
Hence,
\[
 \left|\sum_i\langle\bm H_i,\bm X\bm E_i\rangle\right|
 =
 \left|\sum_{l=1}^{K}
 \left\langle
 \sum_i a_{il}\bm H_i,\bm X\bm D_l
 \right\rangle\right|\le
 \sqrt{m}\sum_{l=1}^{K}
 \|\bm D_l\|_{\mathrm{op}},
\]
which gives
\[
 \Lambda_{\mathrm{LR}}(m)\vee\Lambda_{\mathrm{S}}(m)
 \le
 \sum_{l=1}^{K}\|\bm D_l\|_{\mathrm{op}}
 \le
 K\max_{l\in[K]}\|\bm D_l\|_{\mathrm{op}}.
\]
Hence, in this case, if \(K\) and \(\max_l\|\bm D_l\|_{\mathrm{op}}\) are bounded, the directional covariance discrepancy measures are bounded.

\end{enumerate}

\section{Proofs of Main Results}
\label{app:main-theorem-proof}
\label{app:true-errors}

This appendix proves the quantitative statements in
Appendix~\ref{app:full-theory}.
Appendices~\ref{subsec:app-main-setting}--\ref{subsec:app-basin-conclusion}
establish joint contraction and preservation of the local basin.
Appendix~\ref{subsec:app-corollary-proofs} proves the statistical
corollaries, and Appendix~\ref{app:initializer-proof} establishes
the initialization rate and basin entry.

\subsection{Setup and proof outline}
\label{subsec:app-main-setting}
\label{app:ap-joint-proof}
\label{subsec:app-approx-main-setting}
\label{subsec:app-approx-setup}

\paragraph*{Notation and empirical quantities}
We prove Theorem~\ref{thm:app-unified-main-full} for a fixed
$q\in[0,1]$ and working cap $m$. We use the population curvature
coefficients in Appendix~\ref{subsec:app-full-curvature}, and the
weight vector $\bbm\ell$ and population error $V^{(l)}$ in
Appendix~\ref{subsec:app-unified-main-theorem}.
For a one-step calculation, fix a communication round $l$ and suppress the superscript on the current iterates
$\bm L^{(l)}$, $\bm R^{(l)}$, $\bm S_0^{(l)}$, and $\{\bm S_i^{(l)}\}_{i=1}^N$.
Recall $\bm B_*=\bm U_*\bbm\Sigma_*\bm V_*^\top$,
$\bm L_*=\bm U_*\bbm\Sigma_*^{1/2}$, and
$\bm R_*=\bm V_*\bbm\Sigma_*^{1/2}$.
For any factor pair $(\bm L,\bm R)$, write
$\|(\bm L,\bm R)\|_{\mathrm{align},*}^2$ for the alignment error
defined in \eqref{eq:scaled_factor_error}. Thus
$e_{\mathrm{LR}}^{(l)}
=\|(\bm L^{(l)},\bm R^{(l)})\|_{\mathrm{align},*}^2$.

For each client, define the empirical covariance and innovation score by
\begin{equation}
\label{eq:app-covariance-score}
 \widehat{\bbm\Gamma}_i
 :=\frac1{T_i}\sum_{t=1}^{T_i}
 \bm y_{i,t-1}\bm y_{i,t-1}^{\top},
 \qquad
 \widehat{\bm W}_i
 :=\frac1{T_i}\sum_{t=1}^{T_i}
 \bbm\epsilon_{it}\bm y_{i,t-1}^{\top}.
\end{equation}
With $w_i=T_i/\Ttot$, their pooled versions are
\begin{equation}
\label{eq:app-weighted-covariance-score}
 \widehat{\bbm\Gamma}_w:=\sum_iw_i\widehat{\bbm\Gamma}_i,
 \qquad
 \widehat{\bm W}_w:=\sum_iw_i\widehat{\bm W}_i.
\end{equation}
Recall $\mathbb E[\widehat{\bbm\Gamma}_i]=\bbm\Gamma_i$ and
$\bbm\Gamma_w=\sum_iw_i\bbm\Gamma_i$.
For a support set $\Omega\subset[p]\times[p]$, let
$\mathcal P_\Omega$ denote coordinate projection onto $\Omega$.
Use the enlarged support budgets $s_0^\dagger,\gamma^\dagger,s_i^\dagger$
in Appendix~\ref{subsec:app-unified-main-theorem}, and the shared
direction classes and directional discrepancies in
Section~\ref{subsec:covariance_compatibility}.
Write $s_{\min}^\dagger:=\min_i s_i^\dagger$,
$\overline{\bbm\Gamma}:=N^{-1}\sum_i\bbm\Gamma_i$, and
$\bm E_i:=\bbm\Gamma_i-\overline{\bbm\Gamma}$.

\paragraph*{Reference parameters and errors}
Set $\underline m=\lfloor m/(C\kappa_{\mathrm P}^2)\rfloor$, with the
constant $C$ in \eqref{eq:app-unified-approximation-envelope}.
At each coordinate, retain the $\underline m$ largest magnitudes among
$(\bm S_{i,*})_{jk}$, resolving ties by a deterministic rule. Define
\begin{equation}
 \label{eq:ap-population-reference}
 \bm H_{i,*}:=\bm S_{i,*}^{[\underline m]},\qquad
 \bm{\mathfrak R}_i:=\bm S_{i,*}-\bm H_{i,*},\qquad i\in[N].
\end{equation}
Here $\bm{\mathfrak R}_i$ denotes the approximation residual relative to the reference $\bm H_{i,*}$. The supports of $\bm H_{i,*}$ belong to
$\mathcal D(\{s_{i,*}\},\underline m)$. The current, next, and comparison
supports have coordinate frequency at most $2m+\underline m\le3m$. We define the reference personalized error by
\begin{equation}
 e_{\mathrm P,\mathrm{ref}}:=\sum_iw_i\|\bm S_i-\bm H_{i,*}\|_{\mathrm F}^2.
\end{equation}
The reference error vector and weighted total error are
\begin{equation}
 \label{eq:ap-reference-errors}
 \begin{aligned}
 \bm e_{\mathrm{ref}}
   &:=(e_{\mathrm{LR}},e_{\mathrm S},
                              e_{\mathrm P,\mathrm{ref}})^\top,
 &V_{\mathrm{ref}}
   &:=\bbm\ell^\top\bm e_{\mathrm{ref}}.
 \end{aligned}
\end{equation}
Define the weighted squared approximation error by
\[
 A_R:=\sum_iw_i\|\bm{\mathfrak R}_i\|_{\mathrm F}^2.
\]
For $0<q\le1$, define
\[
 \delta_q:=\left(\frac{m_{q,*}}{\underline m+1}\right)^{1/q},
 \qquad A_{\mathrm{ref}}:=\delta_q\sum_iw_i\|\bm{\mathfrak R}_i\|_1,
\]
where $\|\cdot\|_1$ is the entrywise matrix norm.
The approximate-diversity assumption implies that every entry of
$\bm{\mathfrak R}_i$ has magnitude at most $\delta_q$. Consequently,
\[
 A_R
 \le\delta_q\sum_iw_i\|\bm{\mathfrak R}_i\|_1
 =A_{\mathrm{ref}}.
\]
Lemma~\ref{lem:lq-tail} gives $A_{\mathrm{ref}}\le\mathcal A_q(m)$.
At $q=0$, \eqref{eq:app-unified-cap-condition} gives
$\bm H_{i,*}=\bm S_{i,*}$, so $A_R=A_{\mathrm{ref}}=0$ and
$V_{\mathrm{ref}}=V$.

\paragraph*{Gradient identities.}
Define the current componentwise estimation errors by
\begin{equation}
\label{eq:app-current-errors}
 \bbm\Delta_{\bm B}:=\bm L\bm R^\top-\bm B_*,
 \qquad
 \bbm\Delta_0:=\bm S_0-\bm S_{0,*},
 \qquad
 \bbm\Delta_i:=\bm S_i-\bm H_{i,*}.
\end{equation}
Let $\bm D_i$ denote the local matrix gradient:
\begin{equation}
\label{eq:app-local-gradients}
 \bm D_i:=\nabla_{\bm B}F_i
 =\nabla_{\bm S_0}F_i
 =\nabla_{\bm S_i}F_i.
\end{equation}
The full data-generating transition gives
\begin{equation}
 \label{eq:ap-score-decomposition}
 \bm D_i=(\bbm\Delta_{\bm B}+\bbm\Delta_0+\bbm\Delta_i)
                       \widehat{\bbm\Gamma}_i
          -\widehat{\bm W}_i-\bm{\mathfrak R}_i\widehat{\bbm\Gamma}_i.
\end{equation}
The residual term has mean $\bm{\mathfrak R}_i\bbm\Gamma_i$ and is kept
separate from the innovation score. For $F=\sum_iw_iF_i$, write
\begin{equation}
\label{eq:app-global-gradients}
 \overline{\bm D}:=\sum_iw_i\bm D_i
 =\nabla_{\bm B}F=\nabla_{\bm S_0}F,
 \qquad \nabla_{\bm S_i}F=w_i\bm D_i.
\end{equation}
The factor gradients follow by the chain rule:
\begin{equation}
\label{eq:app-factor-gradients}
 \nabla_{\bm L}F=\overline{\bm D}\bm R,
 \qquad
 \nabla_{\bm R}F=\overline{\bm D}^{\top}\bm L,
\end{equation}
with the same identities holding for each $F_i$ with $\bm D_i$
in place of $\overline{\bm D}$. Defining
$$
\bm Z:=\sum_iw_i\bm{\mathfrak R}_i\widehat{\bbm\Gamma}_i,
$$
we have
\begin{equation}
 \overline{\bm D}=\bbm\Delta_{\bm B}\widehat{\bbm\Gamma}_w+\bbm\Delta_0\widehat{\bbm\Gamma}_w+\sum_iw_i\bbm\Delta_i\widehat{\bbm\Gamma}_i
          -\widehat{\bm W}_w-\bm Z.
\end{equation}
The term $\bm Z$ supplies the approximation contribution.

\paragraph*{Concentration events}
Let $\Eall$ denote the intersection of the concentration events in
Lemmas~\ref{lem:prob-M0-geometry},
\ref{lem:prob-pooled-lr-score},
\ref{lem:prob-sparse-geometry}\textup{(i)--(ii)},
\ref{lem:prob-sparse-scores}\textup{(i)--(ii)}, and
\ref{lem:prob-mixed-geometry}\textup{(i)--(ii)}.
Under \eqref{eq:app-unified-sample-conditions}, these events hold
simultaneously with probability at least
\begin{equation}
  \mathbb P(\Eall)
  \ge
  1-C\exp(-cp)
  -C\exp\{-c(s_0^\dagger\vee1)\log p\}
  -C\exp\{-c(s_{\min}^\dagger\vee1)\log p\}.
  \label{eq:app-good-event-probability}
\end{equation}
Under the same sample-size conditions, the local covariance event
$\mathcal E_{\mathrm{loc}}(\log(pN))$ in
Lemma~\ref{lem:lq-local}
has failure probability at most $2/(pN)$.
We work on $\Eall\cap\mathcal E_{\mathrm{loc}}(\log(pN))$
throughout the one-step analysis and basin induction.
These events are uniform over the required directions and supports,
so the bounds apply at every iteration.

\paragraph*{Proof outline for Theorem~\ref{thm:app-unified-main-full}}
We prove Theorem~\ref{thm:app-unified-main-full} by induction on the
communication round. The initialization condition
\eqref{eq:app-unified-initial-condition} provides the base case.
Suppose that the iterate at round $l$ is feasible and lies in the local
basin, i.e., $V^{(l)}\le\varepsilon_\nu^2\sigma_r^2$.
Appendices~\ref{subsec:app-lowrank-update},
\ref{subsec:app-global-sparse-update}, and
\ref{subsec:app-local-sparse-update} analyze one update of the global
low-rank, shared-sparse, and personalized-sparse components, respectively,
and establish the three one-step bounds
\eqref{eq:lowrank_final}, \eqref{eq:app-shared-step}, and
\eqref{eq:app-local-step} relative to the fixed reference parameters.
Appendix~\ref{subsec:app-joint-contraction} combines these bounds into
a reference-error recursion and converts it to the complete error
recursion
\[
 \bm e^{(l+1)}
 \preceq\bbm\Theta\bm e^{(l)}
       +\bm f_{\mathrm{stat}}+\bm f_{\mathrm{app},q}(m).
\]
Here $\bbm\Theta$ is the one-step error propagation matrix, and
$\bm f_{\mathrm{stat}}$ and $\bm f_{\mathrm{app},q}(m)$ are the
statistical and approximation error vectors. The same appendix
verifies weighted contraction under the conditions of the theorem.

Appendix~\ref{subsec:app-basin-conclusion} verifies that the iterates
remain in the local basin, closing the induction, and iterates the
recursion. Appendix~\ref{subsec:app-corollary-proofs} then establishes the
componentwise rates and applies the round condition to obtain the
balanced-client rates, proving
Corollaries~\ref{cor:app-stat-full} and~\ref{cor:app-balanced-full}.

\subsection{Global low-rank update}
\label{subsec:app-lowrank-update}

On the server side, the low-rank factors are updated by
\begin{equation}
\begin{aligned}
  \bm L^{(l+1)}
  &=\bm L-\eta_{\mathrm{LR}}
  \left\{
  \sum_{i=1}^Nw_i
  \nabla_{\bm L}F_i(\bm L,\bm R,\bm S_0,\bm S_i)
  \right\}
  (\bm R^\top\bm R)^{-1}.
\end{aligned}
\label{eq:app-lowrank-L-update}
\end{equation}
\begin{equation}
\begin{aligned}
  \bm R^{(l+1)}
  &=\bm R-\eta_{\mathrm{LR}}
  \left\{
  \sum_{i=1}^Nw_i
  \nabla_{\bm R}F_i(\bm L,\bm R,\bm S_0,\bm S_i)
  \right\}
  (\bm L^\top\bm L)^{-1}.
\end{aligned}
\label{eq:app-lowrank-R-update}
\end{equation}

Let $\bm Q\in\mathrm{GL}(r)$ attain the infimum in
\eqref{eq:scaled_factor_error} for the current $(\bm L,\bm R)$, and its dependence on $l$ is also suppressed. We then define
\begin{equation}
\begin{aligned}
  \bm P_R
  &=\bm R(\bm R^\top\bm R)^{-1}\bm Q\bbm\Sigma_*^{1/2},
  \\
  \bm P_L
  &=\bm L(\bm L^\top\bm L)^{-1}\bm Q^{-\top}\bbm\Sigma_*^{1/2}.
\end{aligned}
\label{eq:app-lowrank-aligned-objects}
\end{equation}
For the estimation error after one-step update, we have
\begin{equation}\label{eq:d_split1}
\begin{aligned}
  \left\|(\bm L^{(l+1)},\bm R^{(l+1)})\right\|_{\mathrm{align},*}^2 \leq &  \left\|(\bm L^{(l+1)}\bm Q-\bm L_*)\bbm\Sigma_*^{1/2}\right\|_{\mathrm F}^2
  +\left\|(\bm R^{(l+1)}\bm Q^{-\top}-\bm R_*)\bbm\Sigma_*^{1/2}\right\|_{\mathrm F}^2\\
  = & \Bigg\|
  (\bm L\bm Q-\bm L_*)\bbm\Sigma_*^{1/2}
  -\eta_{\mathrm{LR}}
  \nabla_{\bm B}F(\bm L,\bm R,\bm S_0,\{\bm S_i\})\bm P_R
  \Bigg\|_{\mathrm F}^2
  \\
  &+
  \Bigg\|
  (\bm R\bm Q^{-\top}-\bm R_*)\bbm\Sigma_*^{1/2}
  -\eta_{\mathrm{LR}}
  \left\{\nabla_{\bm B}F(\bm L,\bm R,
  \bm S_0,\{\bm S_i\})\right\}^\top\bm P_L
  \Bigg\|_{\mathrm F}^2,
\end{aligned}
\end{equation}
where the inequality follows from the definition of the alignment error.
We first split the first term of \eqref{eq:d_split1}. For any $\zeta_1>0$, we have
\begingroup\small\thinmuskip=2mu\medmuskip=2.5mu\thickmuskip=3.5mu
\begin{equation}
\begin{aligned}
  &\left\|
  (\bm L\bm Q-\bm L_*)\bbm\Sigma_*^{1/2}
  -\eta_{\mathrm{LR}}
  \nabla_{\bm B}F(\bm L,\bm R,\bm S_0,\{\bm S_i\})\bm P_R
  \right\|_{\mathrm F}^2\\
  = & \left\|
  (\bm L\bm Q-\bm L_*)\bbm\Sigma_*^{1/2}
  -\eta_{\mathrm{LR}}
  \mathbb E[\nabla_{\bm B}F(\bm L,\bm R,\bm S_{0,*},\{\bm S_{i,*}\})]\bm P_R\right.\\
  & \left. -\eta_{\mathrm{LR}}\Big\{\nabla_{\bm B}F(\bm L,\bm R,\bm S_0,\{\bm S_i\})-\mathbb E[\nabla_{\bm B}F(\bm L,\bm R,\bm S_{0,*},\{\bm S_{i,*}\})]\Big\}\bm P_R\right\|_{\mathrm F}^2\\
  ={}&\left\|
  (\bm L\bm Q-\bm L_*)\bbm\Sigma_*^{1/2}
  -\eta_{\mathrm{LR}}
  \bbm\Delta_{\bm B}\bbm\Gamma_w\bm P_R
  -\eta_{\mathrm{LR}}
  \left\{
  \bbm\Delta_{\bm B}
  (\widehat{\bbm\Gamma}_w-\bbm\Gamma_w)
  +
  \bbm\Delta_0\widehat{\bbm\Gamma}_w
  +\sum_{i=1}^Nw_i\bbm\Delta_i\widehat{\bbm\Gamma}_i
  -\widehat{\bm W}_w-\bm Z
  \right\}\bm P_R
  \right\|_{\mathrm F}^2
  \\
  \le{}&(1+\zeta_1)
  \underbrace{
  \left\|
  (\bm L\bm Q-\bm L_*)\bbm\Sigma_*^{1/2}
  -\eta_{\mathrm{LR}}
  \bbm\Delta_{\bm B}\bbm\Gamma_w\bm P_R
  \right\|_{\mathrm F}^2
  }_{\mathscr T_{\bm L,1}}
  \\
  &+(1+\zeta_1^{-1})\eta_{\mathrm{LR}}^2
  \underbrace{
  \left\|
  \left(
  \bbm\Delta_{\bm B}
  (\widehat{\bbm\Gamma}_w-\bbm\Gamma_w)
  +
  \bbm\Delta_0\widehat{\bbm\Gamma}_w
  +\sum_{i=1}^Nw_i\bbm\Delta_i\widehat{\bbm\Gamma}_i
  -\widehat{\bm W}_w-\bm Z
  \right)\bm P_R
  \right\|_{\mathrm F}^2
  }_{\mathscr T_{\bm L,2}}.
\end{aligned}
\label{eq:app-lowrank-L-decomposition}
\end{equation}
\endgroup
By an analogous argument, we split the second term of \eqref{eq:d_split1}, and the resulting two $\bm R$-factor terms are defined by
\begin{equation}
\begin{aligned}
  \mathscr T_{\bm R,1}
  &:={}
  \left\|
  (\bm R\bm Q^{-\top}-\bm R_*)\bbm\Sigma_*^{1/2}
  -\eta_{\mathrm{LR}}
  (\bbm\Delta_{\bm B}\bbm\Gamma_w)^\top\bm P_L
  \right\|_{\mathrm F}^2,
  \\
  \text{and}\quad \mathscr T_{\bm R,2}
  &:={}
  \left\|
  \left(
  \bbm\Delta_{\bm B}
  (\widehat{\bbm\Gamma}_w-\bbm\Gamma_w)
  +
  \bbm\Delta_0\widehat{\bbm\Gamma}_w
  +\sum_{i=1}^Nw_i\bbm\Delta_i\widehat{\bbm\Gamma}_i
  -\widehat{\bm W}_w-\bm Z
  \right)^\top\bm P_L
  \right\|_{\mathrm F}^2.
\end{aligned}
\label{eq:app-lowrank-R-terms}
\end{equation}
Consequently,
\begin{equation}
  \left\|(\bm L^{(l+1)},\bm R^{(l+1)})\right\|_{\mathrm{align},*}^2
  \le (1+\zeta_1)
  (\mathscr T_{\bm L,1}+\mathscr T_{\bm R,1})+(1+\zeta_1^{-1})\eta_{\mathrm{LR}}^2
  (\mathscr T_{\bm L,2}+\mathscr T_{\bm R,2}).
\end{equation}
Next, we bound the two groups of terms separately.

\paragraph*{Step 1: Upper bound of
$\mathscr T_{\bm L,1}+\mathscr T_{\bm R,1}$}
Directly expanding the sum of the two terms gives
\begin{equation}
\begin{aligned}
  \mathscr T_{\bm L,1}+\mathscr T_{\bm R,1}
  ={}&\left\|(\bm L^{(l)},\bm R^{(l)})\right\|_{\mathrm{align},*}^2
  \\
  &+\eta_{\mathrm{LR}}^2
  \left\{
  \|\bbm\Delta_{\bm B}\bbm\Gamma_w\bm P_R\|_{\mathrm F}^2
  +\|(\bbm\Delta_{\bm B}\bbm\Gamma_w)^\top
  \bm P_L\|_{\mathrm F}^2
  \right\}
  \\
  &-2\eta_{\mathrm{LR}}
  \inner{\bbm\Delta_{\bm B}\bbm\Gamma_w\bm P_R}
  {(\bm L\bm Q-\bm L_*)\bbm\Sigma_*^{1/2}}
  \\
  &-2\eta_{\mathrm{LR}}
  \inner{(\bbm\Delta_{\bm B}\bbm\Gamma_w)^\top\bm P_L}
  {(\bm R\bm Q^{-\top}-\bm R_*)\bbm\Sigma_*^{1/2}}.
\end{aligned}
\label{eq:app-lowrank-T1-direct-expansion}
\end{equation}
For the second term, Lemma~\ref{lem:perturbation} implies
\begin{equation}
  \|\bm P_R\|_{\op}^2
  \vee
  \|\bm P_L\|_{\op}^2
  \le\frac{1}{(1-\varepsilon_\nu)^2},
\end{equation}
Hence,
\begin{equation}
  \|\bbm\Delta_{\bm B}\bbm\Gamma_w\bm P_R\|_{\mathrm F}^2
  +\|(\bbm\Delta_{\bm B}\bbm\Gamma_w)^\top
  \bm P_L\|_{\mathrm F}^2
  \le
  \frac{2}{(1-\varepsilon_\nu)^2}\|\bbm\Delta_{\bm B}\bbm\Gamma_w\|_{\mathrm F}^2.
  \label{eq:app-lowrank-gradient-factor-bound}
\end{equation}
Therefore,
\begin{equation}
\begin{aligned}
  \mathscr T_{\bm L,1}+\mathscr T_{\bm R,1}
  \le{}&\left\|(\bm L^{(l)},\bm R^{(l)})\right\|_{\mathrm{align},*}^2
  +\frac{2\eta_{\mathrm{LR}}^2}{(1-\varepsilon_\nu)^2}
  \|\bbm\Delta_{\bm B}\bbm\Gamma_w\|_{\mathrm F}^2
  \\
  &-2\eta_{\mathrm{LR}} I_1
  +2\eta_{\mathrm{LR}} I_2
  +2\eta_{\mathrm{LR}} I_3
  +2\eta_{\mathrm{LR}} I_4,
\end{aligned}
\label{eq:app-lowrank-T1-decompose}
\end{equation}
where
\begin{equation}
\begin{aligned}
  I_1
  &:={}
  \inner{\bbm\Delta_{\bm B}\bbm\Gamma_w}
  {\bm L\bm R^\top-\bm L_*\bm R_*^\top},
  \\
  I_2
  &:={}
  \inner{\bbm\Delta_{\bm B}\bbm\Gamma_w}
  {(\bm L\bm Q-\bm L_*)
  (\bm R\bm Q^{-\top}-\bm R_*)^\top},
  \\
  I_3
  &:={}
  \left\langle
  \bbm\Delta_{\bm B}\bbm\Gamma_w,
  (\bm L\bm Q-\bm L_*)\bbm\Sigma_*^{1/2}
  \left(
    \bm V_*^\top
    -\bbm\Sigma_*^{1/2}\bm Q^\top
    (\bm R^\top\bm R)^{-1}\bm R^\top
  \right)
  \right\rangle,
  \\
  I_4
  &:={}
  \left\langle
  \bbm\Delta_{\bm B}\bbm\Gamma_w,
  \left(
    \bm U_*
    -\bm L(\bm L^\top\bm L)^{-1}
    \bm Q^{-\top}\bbm\Sigma_*^{1/2}
  \right)
  \bbm\Sigma_*^{1/2}
  (\bm R\bm Q^{-\top}-\bm R_*)^\top
  \right\rangle.
\end{aligned}
\label{eq:app-lowrank-I1-I4}
\end{equation}

We now bound these four terms accordingly. By
\eqref{eq:app-lowrank-population-rcg} and Lemma~\ref{lem:metric_equivalence},
\begin{equation}
\begin{aligned}
  I_1
  \ge{}&
  \frac{\alpha_{\mathrm{LR}}}{2}
  \|\bm L\bm R^\top-\bm L_*\bm R_*^\top\|_{\mathrm F}^2
  +\frac{1}{2\beta_{\mathrm{LR}}}\|\bbm\Delta_{\bm B}\bbm\Gamma_w\|_{\mathrm F}^2
  \\
  \ge{}&
  \frac{(\sqrt{2}-1)\alpha_{\mathrm{LR}}}{2}\left\|(\bm L,\bm R)\right\|_{\mathrm{align},*}^2
  +\frac{1}{2\beta_{\mathrm{LR}}}\|\bbm\Delta_{\bm B}\bbm\Gamma_w\|_{\mathrm F}^2.
\end{aligned}
\label{eq:LR_I1}
\end{equation}
For $I_2$, Young's inequality and the factor-basin bounds imply
\begin{equation}
\begin{aligned}
  I_2
  \le{}&
  \frac{1}{4\beta_{\mathrm{LR}}}\|\bbm\Delta_{\bm B}\bbm\Gamma_w\|_{\mathrm F}^2
  +\beta_{\mathrm{LR}}
  \left\|
    (\bm L\bm Q-\bm L_*)\bbm\Sigma_*^{1/2}
    \left(
      (\bm R\bm Q^{-\top}-\bm R_*)\bbm\Sigma_*^{-1/2}
    \right)^\top
  \right\|_{\mathrm F}^2
  \\
  \le{}&
  \frac{1}{4\beta_{\mathrm{LR}}}\|\bbm\Delta_{\bm B}\bbm\Gamma_w\|_{\mathrm F}^2
  +C\beta_{\mathrm{LR}}\varepsilon_\nu^2\left\|(\bm L,\bm R)\right\|_{\mathrm{align},*}^2.
\end{aligned}
\label{eq:LR_I2}
\end{equation}
For $I_3$, Lemma~\ref{lem:perturbation} gives
\begin{equation}
  \left\|
  \bm V_*^\top
  -\bbm\Sigma_*^{1/2}\bm Q^\top
  (\bm R^\top\bm R)^{-1}\bm R^\top
  \right\|_{\op}
  =  \left\|
  \bm V_*^\top
  -\bm P_R^\top
  \right\|_{\op}
  \le\frac{\sqrt{2}\varepsilon_\nu}{1-\varepsilon_\nu}.
  \label{eq:V_perturb}
\end{equation}
Therefore,
\begin{equation}
\begin{aligned}
  I_3
  \le{}&
  \|\bbm\Delta_{\bm B}\bbm\Gamma_w\|_{\mathrm F}
  \left\|
  (\bm L\bm Q-\bm L_*)\bbm\Sigma_*^{1/2}
  \right\|_{\mathrm F}
  \left\|
  \bm V_*^\top
  -\bbm\Sigma_*^{1/2}\bm Q^\top
  (\bm R^\top\bm R)^{-1}\bm R^\top
  \right\|_{\op}
  \\
  \le{}&
  \frac{1}{16\beta_{\mathrm{LR}}}\|\bbm\Delta_{\bm B}\bbm\Gamma_w\|_{\mathrm F}^2
  +4\beta_{\mathrm{LR}}\left\|
  \bm V_*^\top
  -\bbm\Sigma_*^{1/2}\bm Q^\top
  (\bm R^\top\bm R)^{-1}\bm R^\top
  \right\|_{\op}^2
  \left\|
  (\bm L\bm Q-\bm L_*)\bbm\Sigma_*^{1/2}
  \right\|_{\mathrm F}^2\\
  \leq & \frac{1}{16\beta_{\mathrm{LR}}}\|\bbm\Delta_{\bm B}\bbm\Gamma_w\|_{\mathrm F}^2 + \frac{C\beta_{\mathrm{LR}}\varepsilon_\nu^2}{(1-\varepsilon_\nu)^2}\left\|(\bm L,\bm R)\right\|_{\mathrm{align},*}^2
\end{aligned}
\label{eq:LR_I3}
\end{equation}
In addition, the analogous argument for $I_4$ gives
\begin{equation}
  I_4
  \le
  \frac{1}{16\beta_{\mathrm{LR}}}\|\bbm\Delta_{\bm B}\bbm\Gamma_w\|_{\mathrm F}^2
  +\frac{C\beta_{\mathrm{LR}}\varepsilon_\nu^2}{(1-\varepsilon_\nu)^2}
  \left\|(\bm L,\bm R)\right\|_{\mathrm{align},*}^2.
  \label{eq:LR_I4}
\end{equation}

Therefore, substituting \eqref{eq:LR_I1}, \eqref{eq:LR_I2}, \eqref{eq:LR_I3}, and
\eqref{eq:LR_I4} into
\eqref{eq:app-lowrank-T1-decompose} with $\varepsilon_\nu< 0.01$ yields
\begin{equation}
\begin{aligned}
  \mathscr T_{\bm L,1}+\mathscr T_{\bm R,1}
  \le{}&
  \left(
    1+C\beta_{\mathrm{LR}}\varepsilon_\nu^2
    \eta_{\mathrm{LR}}
    -(\sqrt{2}-1)\alpha_{\mathrm{LR}}\eta_{\mathrm{LR}}
  \right)\left\|(\bm L,\bm R)\right\|_{\mathrm{align},*}^2
  \\
  &+
  \left(
    5\eta_{\mathrm{LR}}
    -\frac{1}{4\beta_{\mathrm{LR}}}
  \right)
  \eta_{\mathrm{LR}}\|\bbm\Delta_{\bm B}\bbm\Gamma_w\|_{\mathrm F}^2.
\end{aligned}
\label{eq:app-lowrank-T1-contraction-raw}
\end{equation}
By \eqref{eq:app-unified-local-radius} and
$\alpha_{\mathrm S}\le\beta_{\mathrm{LR}}$,
$\varepsilon_\nu^2\le\varepsilon_0\alpha_{\mathrm{LR}}/\beta_{\mathrm{LR}}$.
With the sufficiently small $\varepsilon_0$ specified there, take
$\eta_{\mathrm{LR}}=c/\beta_{\mathrm{LR}}$ with
$c>0$ sufficiently small so that
\begin{equation}
  C\beta_{\mathrm{LR}}\varepsilon_\nu^2
  \le \frac{\alpha_{\mathrm{LR}}}{4},
  \quad
  5\eta_{\mathrm{LR}}
  \le\frac{1}{8\beta_{\mathrm{LR}}}.
  \label{eq:app-lowrank-ideal-smallness}
\end{equation}
Then the population low-rank step satisfies
\begin{equation}
  \mathscr T_{\bm L,1}+\mathscr T_{\bm R,1}
  \le
  \left(1-c\frac{\alpha_{\mathrm{LR}}}{\beta_{\mathrm{LR}}}\right)
  \left\|(\bm L^{(l)},\bm R^{(l)})\right\|_{\mathrm{align},*}^2.
\label{eq:TLTR1_lowrank}
\end{equation}
Here the second inequality in
\eqref{eq:app-lowrank-ideal-smallness} makes the coefficient of
$\|\bbm\Delta_{\bm B}\bbm\Gamma_w\|_{\mathrm F}^2$ in
\eqref{eq:app-lowrank-T1-contraction-raw} nonpositive, so that term is
discarded.
\paragraph*{Step 2: Upper bound of
$\mathscr T_{\bm L,2}+\mathscr T_{\bm R,2}$}
We first bound $\mathscr T_{\bm L,2}$. We have
\begin{equation}
\begin{aligned}
  \mathscr T_{\bm L,2}
  \le{}&5\underbrace{
  \left\|
  \bbm\Delta_{\bm B}
  (\widehat{\bbm\Gamma}_w-\bbm\Gamma_w)\bm P_R
  \right\|_{\mathrm F}^2
  }_{\mathscr T_{\bm L,2}^{(0)}}
  +5\underbrace{
  \left\|\bbm\Delta_0\widehat{\bbm\Gamma}_w\bm P_R\right\|_{\mathrm F}^2
  }_{\mathscr T_{\bm L,2}^{(1)}}\\
  &+5\underbrace{\left\|
  \sum_{i=1}^Nw_i\bbm\Delta_i
  \widehat{\bbm\Gamma}_i\bm P_R
  \right\|_{\mathrm F}^2}_{\mathscr T_{\bm L,2}^{(2)}} +5\underbrace{
  \left\|\widehat{\bm W}_w\bm P_R\right\|_{\mathrm F}^2
  }_{\mathscr T_{\bm L,2}^{(3)}}+5\underbrace{
  \left\|\bm Z\bm P_R\right\|_{\mathrm F}^2
  }_{\mathscr T_{\bm L,2}^{(4)}}.
\end{aligned}
\label{eq:app-lowrank-L-residual-decomposition}
\end{equation}
Analogously,
\begingroup\small\thinmuskip=2mu\medmuskip=2.5mu\thickmuskip=3.5mu
\begin{equation}
\begin{aligned}
  \mathscr T_{\bm R,2}
  \le{}&5\underbrace{
  \left\|
  \left\{
  \bbm\Delta_{\bm B}
  (\widehat{\bbm\Gamma}_w-\bbm\Gamma_w)
  \right\}^\top\bm P_L
  \right\|_{\mathrm F}^2
  }_{\mathscr T_{\bm R,2}^{(0)}}
  +5\underbrace{
  \left\|(\bbm\Delta_0\widehat{\bbm\Gamma}_w)^\top
  \bm P_L\right\|_{\mathrm F}^2
  }_{\mathscr T_{\bm R,2}^{(1)}}\\
  &+5\underbrace{\left\|
  \left(\sum_{i=1}^Nw_i\bbm\Delta_i
  \widehat{\bbm\Gamma}_i\right)^\top\bm P_L
  \right\|_{\mathrm F}^2}_{\mathscr T_{\bm R,2}^{(2)}}
  +5\underbrace{
  \left\|\widehat{\bm W}_w^\top\bm P_L\right\|_{\mathrm F}^2
  }_{\mathscr T_{\bm R,2}^{(3)}}+5\underbrace{
  \left\|\bm Z^\top\bm P_L\right\|_{\mathrm F}^2
  }_{\mathscr T_{\bm R,2}^{(4)}}.
\end{aligned}
\label{eq:app-lowrank-R-residual-decomposition}
\end{equation}
\endgroup
In the following, we bound the five terms for $\bm L$ and $\bm R$
accordingly. 

\paragraph*{Step 2.1: Upper bound of $\mathscr T_{\bm L,2}^{(0)}+\mathscr T_{\bm R,2}^{(0)}$}Firstly, on the pooled covariance concentration event in Lemma~\ref{lem:prob-M0-geometry}, we have
\begin{equation}
  \left\|\widehat{\bbm\Gamma}_w-\bbm\Gamma_w\right\|_{\op}^2
  \le
  C M\phi_{\max}\frac{p}{\Ttot}.
  \label{eq:app-pooled-covariance-concentration}
\end{equation}
Then, by Lemmas~\ref{lem:metric_equivalence} and~\ref{lem:perturbation},
\begin{equation}
\begin{aligned}
  \mathscr T_{\bm L,2}^{(0)}+\mathscr T_{\bm R,2}^{(0)}\le &
  4\left\|\bbm\Delta_{\bm B}\right\|_{\mathrm F}^2
  \left\|\widehat{\bbm\Gamma}_w-\bbm\Gamma_w\right\|_{\op}^2
  \left(\left\|\bm P_R\right\|_{\op}^2+\left\|\bm P_L\right\|_{\op}^2\right)\\
  \le &
  \frac{8}{(1-\varepsilon_\nu)^2}
  \left\|\bbm\Delta_{\bm B}\right\|_{\mathrm F}^2
  \left\|\widehat{\bbm\Gamma}_w-\bbm\Gamma_w\right\|_{\op}^2\\
  \le & C\frac{M\phi_{\max}p}{\Ttot}
  \left\|(\bm L,\bm R)\right\|_{\mathrm{align},*}^2.
\end{aligned}
\end{equation}

\paragraph*{Step 2.2: Upper bound of $\mathscr T_{\bm L,2}^{(1)}+\mathscr T_{\bm R,2}^{(1)}$}
The terms $\mathscr T_{\bm L,2}^{(1)}$ and $\mathscr T_{\bm R,2}^{(1)}$ quantify the interference from the
estimation error of $\bm S_0$ to the low-rank factors. We first decompose $\mathscr T_{\bm L,2}^{(1)}$ as
\begin{equation}
  \left\|\bbm\Delta_0\widehat{\bbm\Gamma}_w\bm P_R\right\|_{\mathrm F}^2
  \le{}
  4\left\|\bbm\Delta_0\bbm\Gamma_w\bm V_*\right\|_{\mathrm F}^2
  +4\left\|\bbm\Delta_0\bbm\Gamma_w
  (\bm P_R-\bm V_*)\right\|_{\mathrm F}^2
  +2\left\|\bbm\Delta_0
  \left(\widehat{\bbm\Gamma}_w-\bbm\Gamma_w\right)
  \bm P_R\right\|_{\mathrm F}^2.
\label{eq:TL2_1_decompose}
\end{equation}
For each row and column of
$\bbm\Delta_0$, let
\begin{equation}
  \mathcal J_a
  =\supp\left((\bbm\Delta_0)_{a,\cdot}\right),
  \qquad
  \mathcal K_b
  =\supp\left((\bbm\Delta_0)_{\cdot,b}\right).
  \label{eq:app-M0-support-slices}
\end{equation}
The thresholding step $\mathcal G_{s_0,\gamma}(\cdot)$ ensures
\begin{equation}
  \max_a|\mathcal J_a|
  \vee
  \max_b|\mathcal K_b|
  \le(\gamma+\gamma_*) p.
  \label{eq:app-M0-row-column-cap}
\end{equation}
Then, by Assumption~\ref{assump:cov_adjusted_incoherence}, the first term of
\eqref{eq:TL2_1_decompose} is bounded by
\begin{equation}
\begin{aligned}
  \left\|\bbm\Delta_0\bbm\Gamma_w\bm V_*\right\|_{\mathrm F}^2
  &=\sum_{a=1}^p
  \left\|
    (\bbm\Delta_0)_{a,\mathcal J_a}
    (\bbm\Gamma_w\bm V_*)_{\mathcal J_a,\cdot}
  \right\|_2^2
  \\
  &\le
  \left\|\bbm\Gamma_w\bm V_*\right\|_{2,\infty}^2
  \sum_{a=1}^p
  \left\|(\bbm\Delta_0)_{a,\mathcal J_a}\right\|_1^2
  \\
  &\le
  \widetilde\phi_w^2(\gamma+\gamma_*)\mu r
  \left\|\bbm\Delta_0\right\|_{\mathrm F}^2.
\end{aligned}
\label{eq:app-M0-left-true}
\end{equation}
By Lemma~\ref{lem:perturbation}, the second term of \eqref{eq:TL2_1_decompose} is bounded by
\begin{equation}
\begin{aligned}
  \|\bbm\Delta_0\bbm\Gamma_w
  (\bm P_R-\bm V_*)\|_{\mathrm F}^2
  &\le
  \|\bbm\Delta_0\|_{\mathrm F}^2
  \|\bbm\Gamma_w\|_{\op}^2
  \|\bm P_R-\bm V_*\|_{\op}^2
  \\
  &\le
  \frac{2\varepsilon_\nu^2}{(1-\varepsilon_\nu)^2}
  \|\bbm\Gamma_w\|_{\op}^2
  \|\bbm\Delta_0\|_{\mathrm F}^2\\
  &\le
  \frac{2\varepsilon_\nu^2}{(1-\varepsilon_\nu)^2}
  \beta_{\mathrm{LR}}^2
  \|\bbm\Delta_0\|_{\mathrm F}^2.
\end{aligned}
\label{eq:app-M0-left-perturbation}
\end{equation}
In addition, by Lemmas~\ref{lem:prob-M0-geometry}
and~\ref{lem:perturbation}, the third term of
\eqref{eq:TL2_1_decompose} is bounded by
\begin{equation}
\left\|\bbm\Delta_0
  \left(\widehat{\bbm\Gamma}_w-\bbm\Gamma_w\right)
  \bm P_R\right\|_{\mathrm F}^2
  \le
  \left\|\bbm\Delta_0\right\|_{\mathrm F}^2
  \left\|\widehat{\bbm\Gamma}_w-\bbm\Gamma_w\right\|_{\op}^2\left\|\bm P_R\right\|_{\op}^2
  \lesssim
  M\phi_{\max}\frac{p}{\Ttot}
  \left\|\bbm\Delta_0\right\|_{\mathrm F}^2.
\label{eq:app-M0-left-covariance-perturbation}
\end{equation}

For $\mathscr T_{\bm R,2}^{(1)}$, we decompose analogously as
\begin{equation}
  \left\|(\bbm\Delta_0\widehat{\bbm\Gamma}_w)^\top
  \bm P_L\right\|_{\mathrm F}^2
  \leq
  4\left\|\bbm\Gamma_w\bbm\Delta_0^\top\bm U_*\right\|_{\mathrm F}^2
  +4\left\|\bbm\Gamma_w\bbm\Delta_0^\top
  (\bm P_L-\bm U_*)\right\|_{\mathrm F}^2
  +2\left\|
  (\widehat{\bbm\Gamma}_w-\bbm\Gamma_w)
  \bbm\Delta_0^\top\bm P_L
  \right\|_{\mathrm F}^2.
  \label{eq:TR2_1_decompose}
\end{equation}

For the first term of \eqref{eq:TR2_1_decompose},
\begin{equation}
\begin{aligned}
  \|\bbm\Gamma_w\bbm\Delta_0^\top\bm U_*\|_{\mathrm F}^2
  &\le
  \|\bbm\Gamma_w\|_{\op}^2
  \|\bbm\Delta_0^\top\bm U_*\|_{\mathrm F}^2
  \\
  &=
  \|\bbm\Gamma_w\|_{\op}^2
  \sum_{b=1}^p
  \left\|
    (\bbm\Delta_0)_{\mathcal K_b,b}^\top
    (\bm U_*)_{\mathcal K_b,\cdot}
  \right\|_2^2
  \\
  &\le
  \|\bbm\Gamma_w\|_{\op}^2
  \|\bm U_*\|_{2,\infty}^2
  \sum_{b=1}^p
  \left\|(\bbm\Delta_0)_{\mathcal K_b,b}\right\|_1^2
  \\
  &\le
  \|\bbm\Gamma_w\|_{\op}^2
  \frac{\mu r}{p}
  \sum_{b=1}^p
  |\mathcal K_b|
  \left\|(\bbm\Delta_0)_{\mathcal K_b,b}\right\|_2^2
  \\
  &\le
  (\gamma+\gamma_*)\mu r\,\beta_{\mathrm{LR}}^2
  \|\bbm\Delta_0\|_{\mathrm F}^2.
\end{aligned}
\label{eq:app-M0-right-true-bound}
\end{equation}
For the second term of \eqref{eq:TR2_1_decompose}, similarly
\begin{equation}
  \|\bbm\Gamma_w\bbm\Delta_0^\top
  (\bm P_L-\bm U_*)\|_{\mathrm F}^2
  \le 
  \|\bbm\Gamma_w\|_{\op}^2
  \|\bbm\Delta_0\|_{\mathrm F}^2
  \|\bm P_L-\bm U_*\|_{\op}^2
  \le 
  \frac{2\varepsilon_\nu^2}{(1-\varepsilon_\nu)^2}
  \beta_{\mathrm{LR}}^2
  \|\bbm\Delta_0\|_{\mathrm F}^2.
\label{eq:app-M0-right-factor-perturbation-bound}
\end{equation}
Moreover, the third term of \eqref{eq:TR2_1_decompose} satisfies
\begin{equation}
  \left\|
  (\widehat{\bbm\Gamma}_w-\bbm\Gamma_w)
  \bbm\Delta_0^\top\bm P_L
  \right\|_{\mathrm F}^2
  \le
  \left\|\widehat{\bbm\Gamma}_w-\bbm\Gamma_w\right\|_{\op}^2
  \|\bbm\Delta_0\|_{\mathrm F}^2
  \|\bm P_L\|_{\op}^2
  \lesssim
  M\phi_{\max}\frac{p}{\Ttot}
  \|\bbm\Delta_0\|_{\mathrm F}^2.
\label{eq:app-M0-right-covariance-perturbation}
\end{equation}
Combining the six pieces gives
\begin{equation}
  \mathscr T_{\bm L,2}^{(1)}+\mathscr T_{\bm R,2}^{(1)}
  \lesssim \left[
    \widetilde\phi_w^2(\gamma+\gamma_*)\mu r
    +M\phi_{\max}\frac{p}{\Ttot}
    +\left(
      (\gamma+\gamma_*)\mu r
      +\frac{\varepsilon_\nu^2}{(1-\varepsilon_\nu)^2}
    \right)
    \beta_{\mathrm{LR}}^2
  \right]
  \|\bbm\Delta_0^{(l)}\|_{\mathrm F}^2.
\end{equation}

\paragraph*{Step 2.3: Upper bound of $\mathscr T_{\bm L,2}^{(3)}+\mathscr T_{\bm R,2}^{(3)}$}

$\mathscr T_{\bm L,2}^{(3)}$ and $\mathscr T_{\bm R,2}^{(3)}$ quantify the
statistical error of estimating the low-rank factors using pooled data.
Conditioned on the pooled low-rank score event in
Lemma~\ref{lem:prob-pooled-lr-score}, we have
\begin{equation}
\begin{aligned}
  &\|\widehat{\bm W}_w\bm P_R\|_{\mathrm F}^2
  +\|\widehat{\bm W}_w^\top\bm P_L\|_{\mathrm F}^2
  \\
  \le &
  \|\widehat{\bm W}_w\|_{\op}^2
  (\|\bm P_R\|_{\mathrm F}^2+\|\bm P_L\|_{\mathrm F}^2)
  \\
  \le &
  r\|\widehat{\bm W}_w\|_{\op}^2
  (\|\bm P_R\|_{\op}^2+\|\bm P_L\|_{\op}^2)
  \\
  \le &
  \frac{Cr}{(1-\varepsilon_\nu)^2}\|\widehat{\bm W}_w\|_{\op}^2
  \le C\frac{\tau^2\phi_{\max}pr}{\Ttot}.
\end{aligned}
\label{eq:app-pooled-score-detailed}
\end{equation}

\paragraph*{Step 2.4: Upper bound of $\mathscr T_{\bm L,2}^{(2)}+\mathscr T_{\bm R,2}^{(2)}$}
By the dual representation of the Frobenius norm,
\begin{equation}
  \mathscr T_{\bm L,2}^{(2)}+\mathscr T_{\bm R,2}^{(2)}
  =
  \sup_{\substack{\|\bm K_L\|_{\mathrm F}^2
  +\|\bm K_R\|_{\mathrm F}^2\le1}}
  \left[
    \sum_{i=1}^Nw_i
    \left\langle
      \bbm\Delta_i,
      \bm K_L\bm P_R^\top\widehat{\bbm\Gamma}_i
      +\bm P_L\bm K_R\widehat{\bbm\Gamma}_i
    \right\rangle
  \right]^2,
  \label{eq:app-bilinear-ML}
\end{equation}
where $\bm K_L\in\mathbb R^{p\times r}$ and
$\bm K_R\in\mathbb R^{r\times p}$.
Let $\Omega_i=\supp(\bbm\Delta_i)$. Since
$\bbm\Delta_i=\bm S_i-\bm H_{i,*}$,
\[
 \Omega_i\subseteq\supp(\bm S_i)\cup\supp(\bm H_{i,*}),
 \qquad |\Omega_i|\le\min\{s_i+s_{i,*},p^2\}\le s_i^\dagger.
\]
The coordinate frequency of $(\Omega_i)_i$ is at most
$m+\underline m\le2m$. Weighted Cauchy--Schwarz inequality gives
\begin{equation}
\begin{aligned}
  \mathscr T_{\bm L,2}^{(2)}+\mathscr T_{\bm R,2}^{(2)}
  \le{}&
  \left(\sum_{i=1}^Nw_i\|\bbm\Delta_i\|_{\mathrm F}^2\right)
  \\
  &\quad\times
  \sup_{\substack{\|\bm K_L\|_{\mathrm F}^2
  +\|\bm K_R\|_{\mathrm F}^2\le1}}
  \sum_{i=1}^Nw_i
  \left\|
    \mathcal P_{\Omega_i}
    \left(
      \bm K_L\bm P_R^\top\widehat{\bbm\Gamma}_i
      +\bm P_L\bm K_R\widehat{\bbm\Gamma}_i
    \right)
  \right\|_{\mathrm F}^2.
\end{aligned}
\label{eq:app-ML-weighted-CS}
\end{equation}

We now bound the supremum in \eqref{eq:app-ML-weighted-CS}, which suffices to consider
\(\|\bm K_L\|_{\mathrm F}^2+\|\bm K_R\|_{\mathrm F}^2=1\).
For such a pair, define
\(
  \bm X_{\mathrm{LR}}
  :=\bm K_L\bm P_R^\top+\bm P_L\bm K_R
\).
By Lemma~\ref{lem:orthonormal-lowrank-directions}\textup{(i)},
\(\bm X_{\mathrm{LR}}/\|\bm X_{\mathrm{LR}}\|_{\mathrm F}
\in\mathcal X_{\mathrm{LR}}\) and
\(\|\bm X_{\mathrm{LR}}\|_{\mathrm F}\le C\), with the zero case being immediate. 
The support bounds above imply that
  $(\Omega_i)_{i=1}^N
  \in\mathcal D(\{s_i^\dagger\}_{i=1}^N,\min\{3m,N\})$.

Starting from the three-way covariance decomposition
$
  \widehat{\bbm\Gamma}_i
  =\overline{\bbm\Gamma}+\bm E_i
  +(\widehat{\bbm\Gamma}_i-\bbm\Gamma_i)$, we have
\begin{equation}\label{eq:appA_het_decomp}
\begin{aligned}
  \sum_{i=1}^Nw_i
  \|\mathcal P_{\Omega_i}(\bm X_{\mathrm{LR}}\widehat{\bbm\Gamma}_i)\|_{\mathrm F}^2
  \le{}&
  3\sum_{i=1}^Nw_i
  \|\mathcal P_{\Omega_i}
  (\bm X_{\mathrm{LR}}\overline{\bbm\Gamma})\|_{\mathrm F}^2
  +3\sum_{i=1}^Nw_i
  \|\mathcal P_{\Omega_i}(\bm X_{\mathrm{LR}}\bm E_i)\|_{\mathrm F}^2\\
  &+3\sum_{i=1}^Nw_i
  \|\mathcal P_{\Omega_i}
  \{\bm X_{\mathrm{LR}}(\widehat{\bbm\Gamma}_i-\bbm\Gamma_i)\}\|_{\mathrm F}^2.
\end{aligned}
\end{equation}

For the first term of \eqref{eq:appA_het_decomp}, the frequency cap,
Lemma~\ref{lem:perturbation}, and
\(\|\overline{\bbm\Gamma}\|_{\op}\le\phi_{\max}\) give
\begin{equation*}
\begin{aligned}
  \sum_{i=1}^Nw_i
  \|\mathcal P_{\Omega_i}
  (\bm X_{\mathrm{LR}}\overline{\bbm\Gamma})\|_{\mathrm F}^2
  &=
  \sum_{j,k}|(\bm X_{\mathrm{LR}}\overline{\bbm\Gamma})_{jk}|^2
  \sum_{i=1}^Nw_i\mathbf 1\{(j,k)\in\Omega_i\}
  \\
  &\le
  3m w_{\max}
  \|\bm X_{\mathrm{LR}}\overline{\bbm\Gamma}\|_{\mathrm F}^2\\
  &\le 3m w_{\max}
  \|\bm X_{\mathrm{LR}}\|_{\mathrm F}^2
  \|\overline{\bbm\Gamma}\|_{\op}^2\\
  &\le
  C m w_{\max}\phi_{\max}^2.
\end{aligned}
\end{equation*}

For the second term of \eqref{eq:appA_het_decomp},  since
\(\bm X_{\mathrm{LR}}/\|\bm X_{\mathrm{LR}}\|_{\mathrm F}
\in\mathcal X_{\mathrm{LR}}\), Frobenius duality,
\eqref{eq:directional_heterogeneity_levels}, and
Lemma~\ref{lem:working-cap-interactions} yield
\begin{equation*}
  \sum_{i=1}^Nw_i
  \|\mathcal P_{\Omega_i}(\bm X_{\mathrm{LR}}\bm E_i)\|_{\mathrm F}^2
  \le
  w_{\max}
  \sup_{\substack{
    \supp(\bm H_i)\subseteq\Omega_i\\
    \sum_i\|\bm H_i\|_{\mathrm F}^2=1}}
  \left|
    \sum_{i=1}^N\langle\bm H_i,\bm X_{\mathrm{LR}}\bm E_i\rangle
  \right|^2
  \le
  C m w_{\max}\Lambda_{\mathrm{LR}}^2.
\end{equation*}

For the third term of \eqref{eq:appA_het_decomp}, apply
Lemma~\ref{lem:prob-mixed-geometry}\textup{(i)} to the normalized direction
$\bm X_{\mathrm{LR}}/\|\bm X_{\mathrm{LR}}\|_{\mathrm F}$.
Multiplying the resulting bound by $\|\bm X_{\mathrm{LR}}\|_{\mathrm F}^2$
and using $\|\bm X_{\mathrm{LR}}\|_{\mathrm F}\le C$ gives
\begin{equation*}
  \sum_{i=1}^Nw_i
  \|\mathcal P_{\Omega_i}
  \{\bm X_{\mathrm{LR}}(\widehat{\bbm\Gamma}_i-\bbm\Gamma_i)\}\|_{\mathrm F}^2
  \le
  C M^2
  \left\{
    \frac{pr}{\Ttot}
    +\frac{\log p}{\Ttot}\sum_{i=1}^Ns_i^\dagger
  \right\}
\end{equation*}
on $\Eall$. The same three estimates apply to every normalized
direction in $\mathcal X_{\mathrm{LR}}$ and every support collection
in the enlarged class. Hence
\begin{equation}
 \label{eq:app-lr-support-projection}
 \begin{aligned}
 \sup_{\bm X\in\mathcal X_{\mathrm{LR}}}
  \sup_{\Omega\in\mathcal D(\{s_i^\dagger\},\min\{3m,N\})}
       \sum_iw_i\|\mathcal P_{\Omega_i}
          (\bm X\widehat{\bbm\Gamma}_i)\|_{\mathrm F}^2\le C L_{\mathrm{LR},\mathrm P}(m).
 \end{aligned}
\end{equation}
Applying these bounds in \eqref{eq:app-ML-weighted-CS} gives
\begin{equation}
\begin{aligned}
  \mathscr T_{\bm L,2}^{(2)}+\mathscr T_{\bm R,2}^{(2)}
  \lesssim \left[
    m w_{\max}
    \left(\phi_{\max}^2+\Lambda_{\mathrm{LR}}^2\right)
    +M^2
    \left\{
      \frac{pr}{\Ttot}
      +\frac{\log p}{\Ttot}\sum_{i=1}^Ns_i^\dagger
    \right\}
  \right]
  \sum_{i=1}^Nw_i\|\bbm\Delta_i\|_{\mathrm F}^2.
\end{aligned}
  \label{eq:app-ML-revised-detailed-bound}
\end{equation}

\paragraph*{Step 2.5: Upper bound of $\mathscr T_{\bm L,2}^{(4)}+\mathscr T_{\bm R,2}^{(4)}$}
When $q=0$, $\bm Z=\bm0$ and both terms vanish.
For $0<q\le1$ and $\bm X\in\mathcal X_{\mathrm{LR}}$, apply
Lemma~\ref{lem:lq-fractional-transfer} with
$\bm C_i=\bm X\widehat{\bbm\Gamma}_i$.
Symmetry of the empirical covariances and
\eqref{eq:app-lr-support-projection} give, on $\Eall$,
\begin{equation}
 \label{eq:app-lr-residual-transfer}
 \begin{aligned}
 |\langle\bm X,\bm Z\rangle|^2
 &=\left|\sum_iw_i\langle\bm{\mathfrak R}_i,
                      \bm X\widehat{\bbm\Gamma}_i\rangle\right|^2\\
 &\le A_{\mathrm{ref}}
       \sup_{\Omega\in\mathcal D(\{s_{i,*}\},\underline m+1)}
       \sum_iw_i\|\mathcal P_{\Omega_i}
          (\bm X\widehat{\bbm\Gamma}_i)\|_{\mathrm F}^2\\
 &\le C L_{\mathrm{LR},\mathrm P}(m)A_{\mathrm{ref}}.
 \end{aligned}
\end{equation}
Combining this bound with the Frobenius duality in
\eqref{eq:app-bilinear-ML} and
Lemma~\ref{lem:orthonormal-lowrank-directions}\textup{(i)} yields
\begin{equation}
 \label{eq:ap-factor-residual-duality}
 \begin{aligned}
 \mathscr T_{\bm L,2}^{(4)}+\mathscr T_{\bm R,2}^{(4)}
       &=\|\bm Z\bm P_R\|_{\mathrm F}^2+
         \|\bm Z^\top\bm P_L\|_{\mathrm F}^2\\
       &=\sup_{\substack{\|\bm K_L\|_{\mathrm F}^2+
                          \|\bm K_R\|_{\mathrm F}^2\le1}}
          \left|\left\langle
          \bm K_L\bm P_R^\top+\bm P_L\bm K_R,\bm Z
          \right\rangle\right|^2\\
       &\le 2(1-\varepsilon_\nu)^{-2}
          \sup_{\bm X\in\mathcal X_{\mathrm{LR}}}
                    |\langle\bm X,\bm Z\rangle|^2\\
       &\le C L_{\mathrm{LR},\mathrm P}(m)A_{\mathrm{ref}}.
 \end{aligned}
\end{equation}

\paragraph*{Step 2.6: Combine the bounds}
Therefore, combining Steps 2.1--2.5, we obtain
\begin{equation}
\begin{aligned}
  \mathscr T_{\bm L,2}+\mathscr T_{\bm R,2}
  \lesssim{}&
  \frac{M\phi_{\max}p}{\Ttot}
  \left\|(\bm L,\bm R)\right\|_{\mathrm{align},*}^2
  \\
  +&\Bigg[
    \widetilde\phi_w^2(\gamma+\gamma_*)\mu r
    +M\phi_{\max}\frac{p}{\Ttot}
    +\left\{
      (\gamma+\gamma_*)\mu r
      +\frac{\varepsilon_\nu^2}{(1-\varepsilon_\nu)^2}
    \right\}\beta_{\mathrm{LR}}^2
  \Bigg]
  \|\bbm\Delta_0\|_{\mathrm F}^2
  \\
  +&\Bigg[
    m w_{\max}
    \left(\phi_{\max}^2+\Lambda_{\mathrm{LR}}^2\right)
    +M^2
    \left\{
      \frac{pr}{\Ttot}
      +\frac{\log p}{\Ttot}\sum_{i=1}^Ns_i^\dagger
    \right\}
  \Bigg]
  \sum_{i=1}^Nw_i\|\bbm\Delta_i\|_{\mathrm F}^2
  \\
  +&\frac{\tau^2\phi_{\max}pr}{\Ttot}
   +C L_{\mathrm{LR},\mathrm P}(m)A_{\mathrm{ref}}.
\end{aligned}
  \label{eq:TLTR2_lowrank}
\end{equation}
Finally, take $\zeta_1=c\alpha_{\mathrm{LR}}/(2\beta_{\mathrm{LR}})$.  By
$(1+x/2)(1-x)\le 1-x/2$ for every $x\ge0$, we have
\begin{equation}\label{TLTR1_lowrank_zeta}
  (1+\zeta_1)
  \left(\mathscr T_{\bm L,1}+\mathscr T_{\bm R,1}\right)
  \le
  \left(1-\frac{c}{2}\frac{\alpha_{\mathrm{LR}}}{\beta_{\mathrm{LR}}}\right)
  \left\|(\bm L,\bm R)\right\|_{\mathrm{align},*}^2.
\end{equation}
With the chosen step size
$\eta_{\mathrm{LR}}=c/\beta_{\mathrm{LR}}$ and
$\alpha_{\mathrm{LR}}\le\beta_{\mathrm{LR}}$,
\begin{equation}
  (1+\zeta_1^{-1})\eta_{\mathrm{LR}}^2
  \le \frac{C}{\alpha_{\mathrm{LR}}\beta_{\mathrm{LR}}}.
  \label{eq:app-lowrank-residual-prefactor}
\end{equation}
Since the sample size condition gives
\[
  \frac{M\phi_{\max}p}
  {\alpha_{\mathrm{LR}}\beta_{\mathrm{LR}}\Ttot}
  \leq C\frac{\alpha_{\mathrm{LR}}}{\beta_{\mathrm{LR}}},
\]
combining \eqref{eq:TLTR2_lowrank} and \eqref{TLTR1_lowrank_zeta}, we have
\begin{equation}
\begin{aligned}
  &\left\|(\bm L^{(l+1)},\bm R^{(l+1)})\right\|_{\mathrm{align},*}^2
  \\
  &\le{}
  \left(
    1-C\frac{\alpha_{\mathrm{LR}}}{\beta_{\mathrm{LR}}}
  \right)
  \left\|(\bm L,\bm R)\right\|_{\mathrm{align},*}^2
  \\
  &+\frac{C}{\alpha_{\mathrm{LR}}\beta_{\mathrm{LR}}}
  \Bigg[
    \widetilde\phi_w^2(\gamma+\gamma_*)\mu r
    +M\phi_{\max}\frac{p}{\Ttot}+\left\{
      (\gamma+\gamma_*)\mu r
      +\frac{\varepsilon_\nu^2}{(1-\varepsilon_\nu)^2}
    \right\}\beta_{\mathrm{LR}}^2
  \Bigg]
  \|\bbm\Delta_0\|_{\mathrm F}^2
  \\
  &+\frac{C}{\alpha_{\mathrm{LR}}\beta_{\mathrm{LR}}}
  \Bigg[
    m w_{\max}
    \left(\phi_{\max}^2+\Lambda_{\mathrm{LR}}^2\right)
    +M^2
    \left\{
      \frac{pr}{\Ttot}
      +\frac{\log p}{\Ttot}\sum_{i=1}^Ns_i^\dagger
    \right\}
  \Bigg]
  \sum_{i=1}^Nw_i\|\bbm\Delta_i\|_{\mathrm F}^2
  \\
  &+C\frac{\tau^2\phi_{\max}pr}
  {\alpha_{\mathrm{LR}}\beta_{\mathrm{LR}}\Ttot}+C\frac{L_{\mathrm{LR},\mathrm P}(m)}
          {\alpha_{\mathrm{LR}}\beta_{\mathrm{LR}}}A_{\mathrm{ref}}.
\end{aligned}
\label{eq:lowrank_final}
\end{equation}

\subsection{Global sparse update}
\label{subsec:app-global-sparse-update}

We still suppress the superscript $l$ on the current iterate.  The weighted gradient
proposal is
\begin{equation}
  \widetilde{\bm S}_0^{(l+1)}
  =\bm S_0-\eta_{\mathrm S}\sum_{i=1}^Nw_i
  \nabla_{\bm S_0}F_i(\bm L,\bm R,\bm S_0,\bm S_i).
  \label{eq:global-sparse-gradient-proposal}
\end{equation}
Then, the server applies the capacity-aware greedy thresholding operator
described in Section~\ref{sec:perfect_algorithm} and updates
\begin{equation}
  \bm S_0^{(l+1)}
  =\mathcal G_{s_0,\gamma}
  (\widetilde{\bm S}_0^{(l+1)}).
  \label{eq:global-sparse-capacity-greedy}
\end{equation}
Let the enlarged support be
\begin{equation}
  \Omega_0
  :=\supp(\bm S_0)
  \cup\supp(\bm S_0^{(l+1)})
  \cup\supp(\bm S_{0,*}).
  \label{eq:global-sparse-analysis-support}
\end{equation}
Thus,
\begin{equation}
\label{eq:global-sparse-one-step-support-geometry}
\begin{aligned}
  |\Omega_0|&\le s_0^\dagger,\\
  \max_a|\{b:(a,b)\in\Omega_0\}|
  \vee
  \max_b|\{a:(a,b)\in\Omega_0\}|
  &\le\gamma^\dagger p.
\end{aligned}
\end{equation}
Write $\mathcal C_{\mathrm S}^\dagger:=\mathcal C(s_0^\dagger,\gamma^\dagger)$.
The fitting conditions give
$\Omega_0\in\mathcal C_{\mathrm S}^\dagger\subseteq\mathcal C_{\mathrm S}'$,
so the restricted curvature bounds
$\alpha_{\mathrm S},\beta_{\mathrm S}$ apply on $\Omega_0$.
For any $\zeta_0>0$, we have
\begin{equation}\label{eq:S0_decomp}
\begin{aligned}
  \|\bbm\Delta_0^{(l+1)}\|_{\mathrm F}^2
  \le{}&
  (1+\rho_0)
  \left\|
    \mathcal P_{\Omega_0}
    \left\{
      \bbm\Delta_0-\eta_{\mathrm S}
      \nabla_{\bm S_0}F
      (\bm L,\bm R,\bm S_0,\{\bm S_i\}_{i=1}^N)
    \right\}
  \right\|_{\mathrm F}^2
  \\
  ={}&
  (1+\rho_0)
  \Bigg\|
    \mathcal P_{\Omega_0}
    \Bigg[
      \bbm\Delta_0
      -\eta_{\mathrm S}
      \mathbb E\!\left[
        \nabla_{\bm S_0}F
        (\bm L_*,\bm R_*,\bm S_0,\{\bm S_{i,*}\}_{i=1}^N)
      \right]
      \\
  &
      -\eta_{\mathrm S}
      \Big\{
          \nabla_{\bm S_0}F
          (\bm L,\bm R,\bm S_0,\{\bm S_i\}_{i=1}^N)
          \quad-\mathbb E\!\left[
            \nabla_{\bm S_0}F
            (\bm L_*,\bm R_*,\bm S_0,\{\bm S_{i,*}\}_{i=1}^N)
          \right]
      \Big\}
    \Bigg]
  \Bigg\|_{\mathrm F}^2
  \\
  \le{}&
  (1+\rho_0)(1+\zeta_0)
  \left\|
    \mathcal P_{\Omega_0}
    (\bbm\Delta_0-\eta_{\mathrm S}\bbm\Delta_0\bbm\Gamma_w)
  \right\|_{\mathrm F}^2
  \\
  &+(1+\rho_0)(1+\zeta_0^{-1})\eta_{\mathrm S}^2
  \Bigg\|
    \mathcal P_{\Omega_0}
    \Bigg\{
      \bbm\Delta_{\bm B}\bbm\Gamma_w
      +\bbm\Delta_0
      (\widehat{\bbm\Gamma}_w-\bbm\Gamma_w)
      +\bbm\Delta_{\bm B}
      (\widehat{\bbm\Gamma}_w-\bbm\Gamma_w)
      \\
  &\hspace{6.5cm}
      +\sum_{i=1}^Nw_i\bbm\Delta_i\widehat{\bbm\Gamma}_i
      -\widehat{\bm W}_w-\bm Z
    \Bigg\}
  \Bigg\|_{\mathrm F}^2,
\end{aligned}
\end{equation}
where the first inequality follows from
Lemma~\ref{lem:G_stability}. By the choice
that $s_0\gtrsim \beta_{\mathrm S}^2\alpha_{\mathrm S}^{-2}s_{0,*}$ and
$\gamma\gtrsim \beta_{\mathrm S}^2\alpha_{\mathrm S}^{-2}\gamma_*$, the coefficient $\rho_0$
in Lemma~\ref{lem:G_stability} is bounded by
\[
  \rho_0\leq c\frac{\alpha_{\mathrm S}}{\beta_{\mathrm S}}.
\]

\paragraph*{Step 1: Contraction along the population ideal direction}
For the first term of \eqref{eq:S0_decomp}, since $\supp(\bbm\Delta_0)\subseteq\Omega_0$, direct expansion gives
\begin{equation}
\begin{aligned}
  &\left\|
    \mathcal P_{\Omega_0}
    (\bbm\Delta_0-\eta_{\mathrm S}\bbm\Delta_0\bbm\Gamma_w)
  \right\|_{\mathrm F}^2
  \\
  &=
  \|\bbm\Delta_0\|_{\mathrm F}^2
  +\eta_{\mathrm S}^2
  \left\|
    \mathcal P_{\Omega_0}
    (\bbm\Delta_0\bbm\Gamma_w)
  \right\|_{\mathrm F}^2
  -2\eta_{\mathrm S}
  \inner{\bbm\Delta_0\bbm\Gamma_w}{\bbm\Delta_0}.
\end{aligned}
  \label{eq:global-sparse-population-step-expansion}
\end{equation}
Part~\textup{(ii)} of Lemma~\ref{lem:det-population-rcg} applies with
$\bm H=\bbm\Delta_0$ and $\Omega_0$.
Substituting \eqref{eq:global-sparse-population-RCG} into
\eqref{eq:global-sparse-population-step-expansion}, and taking
$\eta_{\mathrm S}\asymp \beta_{\mathrm S}^{-1}$, yields
\begin{equation}
  \left\|
    \mathcal P_{\Omega_0}
    (\bbm\Delta_0-\eta_{\mathrm S}\bbm\Delta_0\bbm\Gamma_w)
  \right\|_{\mathrm F}^2
  \le
  (1-c\alpha_{\mathrm S}\beta_{\mathrm S}^{-1})
  \|\bbm\Delta_0\|_{\mathrm F}^2.
  \label{eq:global-sparse-population-contraction}
\end{equation}
With $\zeta_0=c\alpha_{\mathrm S}\beta_{\mathrm S}^{-1}/2$, the first term in \eqref{eq:S0_decomp} is bounded by
\begin{equation}
  (1+\rho_0)(1+\zeta_0)
  \left\|
    \mathcal P_{\Omega_0}
    (\bbm\Delta_0-\eta_{\mathrm S}\bbm\Delta_0\bbm\Gamma_w)
  \right\|_{\mathrm F}^2
  \le
  \left(1-c\frac{\alpha_{\mathrm S}}{\beta_{\mathrm S}}\right)
  \|\bbm\Delta_0\|_{\mathrm F}^2.
  \label{eq:global-sparse-amplified-population-contraction}
\end{equation}

\paragraph*{Step 2: Upper bound of the gradient discrepancy}
Expanding the second term in \eqref{eq:S0_decomp} gives
\begin{equation}
\begin{aligned}
  &\left\|
    \mathcal P_{\Omega_0}
    \left\{
      \bbm\Delta_{\bm B}\bbm\Gamma_w
      +\bbm\Delta_0
      (\widehat{\bbm\Gamma}_w-\bbm\Gamma_w)
      +\bbm\Delta_{\bm B}
      (\widehat{\bbm\Gamma}_w-\bbm\Gamma_w)
      +\sum_{i=1}^Nw_i\bbm\Delta_i\widehat{\bbm\Gamma}_i
      -\widehat{\bm W}_w-\bm Z
    \right\}
  \right\|_{\mathrm F}^2
  \\
  \le{}&
  \underbrace{
    6\left\|
      \mathcal P_{\Omega_0}
      (\bbm\Delta_{\bm B}\bbm\Gamma_w)
    \right\|_{\mathrm F}^2
  }_{\mathscr{S}_1}
  +\underbrace{
    6\left\|
      \mathcal P_{\Omega_0}
      \left\{
        \bbm\Delta_0
        (\widehat{\bbm\Gamma}_w-\bbm\Gamma_w)
      \right\}
    \right\|_{\mathrm F}^2
  }_{\mathscr{S}_2}
  \\
  &+\underbrace{
    6\left\|
      \mathcal P_{\Omega_0}
      \left\{
        \bbm\Delta_{\bm B}
        (\widehat{\bbm\Gamma}_w-\bbm\Gamma_w)
      \right\}
    \right\|_{\mathrm F}^2
  }_{\mathscr{S}_3}
  \\
  &+\underbrace{
    6\left\|
      \mathcal P_{\Omega_0}
      \left(
        \sum_{i=1}^Nw_i\bbm\Delta_i\widehat{\bbm\Gamma}_i
      \right)
    \right\|_{\mathrm F}^2
  }_{\mathscr{S}_4}
  +\underbrace{
    6\left\|
      \mathcal P_{\Omega_0}(\widehat{\bm W}_w)
    \right\|_{\mathrm F}^2
  }_{\mathscr{S}_5}
  +\underbrace{
    6\left\|
      \mathcal P_{\Omega_0}\bm Z
    \right\|_{\mathrm F}^2
  }_{\mathscr{S}_6}.
\end{aligned}
  \label{eq:global-sparse-residual-decomposition}
\end{equation}
We bound $\mathscr{S}_1,\ldots,\mathscr{S}_6$ separately.

\paragraph*{Step 2.1: Upper bound of $\mathscr{S}_1$}
For any $\bm A,\bm C\in\mathbb R^{p\times r}$, the row and column caps in
\eqref{eq:global-sparse-one-step-support-geometry} imply
\begin{equation}
\begin{aligned}
  \|\mathcal P_{\Omega_0}(\bm A\bm C^\top)\|_{\mathrm F}^2
  &=
  \sum_{a=1}^p
  \sum_{b:(a,b)\in\Omega_0}
  \left\langle \bm A_{a,\cdot},\bm C_{b,\cdot}\right\rangle^2
  \\
  &\le
  \sum_{a=1}^p
  \|\bm A_{a,\cdot}\|_2^2
  \sum_{b:(a,b)\in\Omega_0}
  \|\bm C_{b,\cdot}\|_2^2
  \\
  &\le
  \max_{a\in[p]}|\{b:(a,b)\in\Omega_0\}|
  \|\bm C\|_{2,\infty}^2\|\bm A\|_{\mathrm F}^2
  \\
  &\le
  (2\gamma+\gamma_*)p
  \|\bm C\|_{2,\infty}^2\|\bm A\|_{\mathrm F}^2,
\end{aligned}
  \label{eq:global-sparse-capped-product-bound}
\end{equation}
Similarly, 
\[
  \|\mathcal P_{\Omega_0}(\bm A\bm C^\top)\|_{\mathrm F}^2
  \le
  (2\gamma+\gamma_*)p
  \|\bm A\|_{2,\infty}^2\|\bm C\|_{\mathrm F}^2.
\]
Decomposing $\bbm\Delta_{\bm B}=(\bm L\bm Q-\bm L_*)(\bm R\bm Q^{-\top})^\top+\bm L_*(\bm R\bm Q^{-\top}-\bm R_*)^\top$, we have
\[
  \mathscr{S}_1
  \le
  12\left\|
    \mathcal P_{\Omega_0}
    \left\{
      (\bm L\bm Q-\bm L_*)
      (\bm R\bm Q^{-\top})^\top\bbm\Gamma_w
    \right\}
  \right\|_{\mathrm F}^2
  +12\left\|
    \mathcal P_{\Omega_0}
    \left\{
      \bm L_*
      (\bm R\bm Q^{-\top}-\bm R_*)^\top\bbm\Gamma_w
    \right\}
  \right\|_{\mathrm F}^2.
\]
For the first term,  we obtain
\begin{equation}
\begin{aligned}
  &\left\|
    \mathcal P_{\Omega_0}
    \left\{
      (\bm L\bm Q-\bm L_*)
      (\bm R\bm Q^{-\top})^\top\bbm\Gamma_w
    \right\}
  \right\|_{\mathrm F}^2
  \\
  \le{}&
  2\left\|
    \mathcal P_{\Omega_0}
    \left\{
      (\bm L\bm Q-\bm L_*)\bbm\Sigma_*^{1/2}
      (\bbm\Gamma_w\bm V_*)^\top
    \right\}
  \right\|_{\mathrm F}^2
  \\
  &+2\left\|
    (\bm L\bm Q-\bm L_*)\bbm\Sigma_*^{1/2}
    \left\{
      \bbm\Gamma_w
      (\bm R\bm Q^{-\top}-\bm R_*)\bbm\Sigma_*^{-1/2}
    \right\}^{\!\top}
  \right\|_{\mathrm F}^2
  \\
  \le{}&
  C\left\{
    \widetilde\phi_w^2(2\gamma+\gamma_*)\mu r
    +\frac{\varepsilon_\nu^2}{(1-\varepsilon_\nu)^2}\beta_{\mathrm{LR}}^2
  \right\}
  \left\|
    (\bm L\bm Q-\bm L_*)\bbm\Sigma_*^{1/2}
  \right\|_{\mathrm F}^2.
\end{aligned}
  \label{eq:global-sparse-aligned-first-piece}
\end{equation}
For the second term,
\begin{equation}
\begin{aligned}
  &\left\|
    \mathcal P_{\Omega_0}
    \left\{
      \bm L_*
      (\bm R\bm Q^{-\top}-\bm R_*)^\top\bbm\Gamma_w
    \right\}
  \right\|_{\mathrm F}^2
  \\
  &=
  \left\|
    \mathcal P_{\Omega_0}
    \left[
      \bm U_*
      \left\{
        \bbm\Gamma_w
        (\bm R\bm Q^{-\top}-\bm R_*)\bbm\Sigma_*^{1/2}
      \right\}^{\!\top}
    \right]
  \right\|_{\mathrm F}^2
  \\
  &\le
  (2\gamma+\gamma_*)\mu r\,\beta_{\mathrm{LR}}^2
  \left\|
    (\bm R\bm Q^{-\top}-\bm R_*)\bbm\Sigma_*^{1/2}
  \right\|_{\mathrm F}^2.
\end{aligned}
  \label{eq:global-sparse-aligned-second-piece}
\end{equation}
Combining the aligned pieces gives
\begin{equation}
\begin{aligned}
  &\left\|
    \mathcal P_{\Omega_0}
    (\bbm\Delta_{\bm B}\bbm\Gamma_w)
  \right\|_{\mathrm F}^2
  \\
  &\quad\le{}
  C\Bigg[
    \widetilde\phi_w^2(2\gamma+\gamma_*)\mu r
    +\left\{
      (2\gamma+\gamma_*)\mu r
      +\frac{\varepsilon_\nu^2}{(1-\varepsilon_\nu)^2}
    \right\}\beta_{\mathrm{LR}}^2
  \Bigg]
  \left\|(\bm L^{(l)},\bm R^{(l)})\right\|_{\mathrm{align},*}^2.
\end{aligned}
  \label{eq:global-sparse-projected-lowrank-bound}
\end{equation}

\paragraph*{Step 2.2: Upper bounds of $\mathscr{S}_2$ and $\mathscr{S}_3$}
Since $\Omega_0$ belongs to $\mathcal C_{\mathrm S}^\dagger$,
Lemma~\ref{lem:prob-sparse-geometry}\textup{(i)} gives, for
$\mathscr{S}_2$,
\begin{equation}
  \left\|
    \mathcal P_{\Omega_0}
    \left\{
      \bbm\Delta_0
      (\widehat{\bbm\Gamma}_w-\bbm\Gamma_w)
    \right\}
  \right\|_{\mathrm F}^2
  \le
  C M\phi_{\max}
  \frac{s_0^\dagger\log p}{\Ttot}
  \|\bbm\Delta_0\|_{\mathrm F}^2.
  \label{eq:global-sparse-centered-shared-projection}
\end{equation}
For $\mathscr{S}_3$,
Lemmas~\ref{lem:prob-M0-geometry} and~\ref{lem:metric_equivalence} give
\begin{equation}
  \left\|
    \mathcal P_{\Omega_0}
    \left\{
      \bbm\Delta_{\bm B}
      (\widehat{\bbm\Gamma}_w-\bbm\Gamma_w)
    \right\}
  \right\|_{\mathrm F}^2
  \le
  \|\bbm\Delta_{\bm B}\|_{\mathrm F}^2
  \|\widehat{\bbm\Gamma}_w-\bbm\Gamma_w\|_{\op}^2
  \le
  C M\phi_{\max}\frac{p}{\Ttot}
  \left\|(\bm L^{(l)},\bm R^{(l)})\right\|_{\mathrm{align},*}^2.
  \label{eq:global-sparse-centered-lowrank-projection}
\end{equation}

\paragraph*{Step 2.3: Upper bound of $\mathscr{S}_4$}
Let $\Omega_i:=\supp(\bbm\Delta_i)$. Since
$\bbm\Delta_i=\bm S_i-\bm H_{i,*}$,
\[
 \Omega_i\subseteq\supp(\bm S_i)\cup\supp(\bm H_{i,*}),
 \qquad |\Omega_i|\le\min\{s_i+s_{i,*},p^2\}\le s_i^\dagger.
\]
The coordinate frequency of $(\Omega_i)_i$ is at most
$m+\underline m\le2m$. Hence
\[
 (\Omega_i)_i\in\mathcal D(\{s_i^\dagger\}_{i=1}^N,\min\{3m,N\}).
\]
Using Frobenius duality and the weighted Cauchy--Schwarz inequality,
we obtain
\begin{equation}
\begin{aligned}
  \left\|
    \mathcal P_{\Omega_0}
    \left(
      \sum_{i=1}^Nw_i\bbm\Delta_i\widehat{\bbm\Gamma}_i
    \right)
  \right\|_{\mathrm F}^2
  &=
  \sup_{\substack{
    \supp(\bm X_{\mathrm S})\subseteq\Omega_0\\
    \|\bm X_{\mathrm S}\|_{\mathrm F}=1}}
  \left[
    \sum_{i=1}^Nw_i
    \left\langle
      \bbm\Delta_i,
      \mathcal P_{\Omega_i}
      (\bm X_{\mathrm S}\widehat{\bbm\Gamma}_i)
    \right\rangle
  \right]^2
  \\
  &\le
  \left(\sum_{i=1}^Nw_i\|\bbm\Delta_i\|_{\mathrm F}^2\right)
  \sup_{\substack{
    \supp(\bm X_{\mathrm S})\subseteq\Omega_0\\
    \|\bm X_{\mathrm S}\|_{\mathrm F}=1}}
  \sum_{i=1}^Nw_i
  \left\|
    \mathcal P_{\Omega_i}
    (\bm X_{\mathrm S}\widehat{\bbm\Gamma}_i)
  \right\|_{\mathrm F}^2.
\end{aligned}
  \label{eq:app-SL-weighted-CS}
\end{equation}
By \eqref{eq:global-sparse-one-step-support-geometry}, every \(\bm X_{\mathrm S}\)
in this supremum belongs to \(\mathcal X_{\mathrm S}\).

The following estimates hold for every $\bm X_{\mathrm S}\in\mathcal X_{\mathrm S}$
and $\Omega\in\mathcal D(\{s_i^\dagger\},\min\{3m,N\})$.
For a fixed such direction and support collection, the decomposition
\(
  \widehat{\bbm\Gamma}_i
  =\overline{\bbm\Gamma}+\bm E_i
  +(\widehat{\bbm\Gamma}_i-\bbm\Gamma_i)
\)
gives
\begin{equation}
\begin{aligned}
  \sum_{i=1}^Nw_i
  \|\mathcal P_{\Omega_i}
  (\bm X_{\mathrm S}\widehat{\bbm\Gamma}_i)\|_{\mathrm F}^2
  \le{}&
  3\sum_{i=1}^Nw_i
  \|\mathcal P_{\Omega_i}
  (\bm X_{\mathrm S}\overline{\bbm\Gamma})\|_{\mathrm F}^2
  +3\sum_{i=1}^Nw_i
  \|\mathcal P_{\Omega_i}(\bm X_{\mathrm S}\bm E_i)\|_{\mathrm F}^2
  \\
  &+3\sum_{i=1}^Nw_i
  \|\mathcal P_{\Omega_i}
  \{\bm X_{\mathrm S}(\widehat{\bbm\Gamma}_i-\bbm\Gamma_i)\}\|_{\mathrm F}^2.
\end{aligned}
  \label{eq:appA-sparse-het-decomp}
\end{equation}

For the first term of \eqref{eq:appA-sparse-het-decomp}, the frequency
cap and \(\|\overline{\bbm\Gamma}\|_{\op}\le\phi_{\max}\) give
\begin{equation*}
\begin{aligned}
  \sum_{i=1}^Nw_i
  \|\mathcal P_{\Omega_i}
  (\bm X_{\mathrm S}\overline{\bbm\Gamma})\|_{\mathrm F}^2
  &=
  \sum_{j,k}|(\bm X_{\mathrm S}\overline{\bbm\Gamma})_{jk}|^2
  \sum_{i=1}^Nw_i\mathbf 1\{(j,k)\in\Omega_i\}
  \\
  &\le
  3m w_{\max}
  \|\bm X_{\mathrm S}\overline{\bbm\Gamma}\|_{\mathrm F}^2
  \\
  &\le
  3m w_{\max}\phi_{\max}^2.
\end{aligned}
\end{equation*}

For the second term of \eqref{eq:appA-sparse-het-decomp}, since
\(\bm X_{\mathrm S}\in\mathcal X_{\mathrm S}\), Frobenius duality,
\eqref{eq:directional_heterogeneity_levels}, and
Lemma~\ref{lem:working-cap-interactions} yield
\begin{equation*}
  \sum_{i=1}^Nw_i
  \|\mathcal P_{\Omega_i}(\bm X_{\mathrm S}\bm E_i)\|_{\mathrm F}^2
  \le
  w_{\max}
  \sup_{\substack{
    \supp(\bm H_i)\subseteq\Omega_i\\
    \sum_i\|\bm H_i\|_{\mathrm F}^2=1}}
  \left|
    \sum_{i=1}^N
    \langle\bm H_i,\bm X_{\mathrm S}\bm E_i\rangle
  \right|^2
  \le
  3m w_{\max}\Lambda_{\mathrm{S}}^2.
\end{equation*}

For the third term of \eqref{eq:appA-sparse-het-decomp}, since
$\bm X_{\mathrm S}\in\mathcal X_{\mathrm S}$,
Lemma~\ref{lem:prob-mixed-geometry}\textup{(ii)} gives
\begin{equation*}
  \sum_{i=1}^Nw_i
  \|\mathcal P_{\Omega_i}
  \{\bm X_{\mathrm S}(\widehat{\bbm\Gamma}_i-\bbm\Gamma_i)\}\|_{\mathrm F}^2
  \le
  C M^2
  \left\{
    \frac{s_0^\dagger\log p}{\Ttot}
    +\frac{\log p}{\Ttot}\sum_{i=1}^Ns_i^\dagger
  \right\}
\end{equation*}
on $\Eall$. Combining the three estimates uniformly gives
\begin{equation}
 \label{eq:app-shared-support-projection}
 \sup_{\bm X\in\mathcal X_{\mathrm S}}
  \sup_{\Omega\in\mathcal D(\{s_i^\dagger\},\min\{3m,N\})}
       \sum_iw_i\|\mathcal P_{\Omega_i}
          (\bm X\widehat{\bbm\Gamma}_i)\|_{\mathrm F}^2\le C L_{\mathrm S,\mathrm P}(m).
\end{equation}
Applying these bounds to the supremum in
\eqref{eq:app-SL-weighted-CS}, and recalling the definition of
\(\mathscr S_4\) gives
\begin{equation}
\begin{aligned}
  \mathscr S_4
  \le{}&
  C\Bigg[
    m w_{\max}
    \left(\phi_{\max}^2+\Lambda_{\mathrm{S}}^2\right)
    +M^2
    \left\{
      \frac{s_0^\dagger\log p}{\Ttot}
      +\frac{\log p}{\Ttot}\sum_{i=1}^Ns_i^\dagger
    \right\}
  \Bigg]
  \sum_{i=1}^Nw_i\|\bbm\Delta_i\|_{\mathrm F}^2.
\end{aligned}
  \label{eq:global-sparse-local-quadratic-adjoint}
\end{equation}

\paragraph*{Step 2.4: Upper bound of $\mathscr{S}_5$}
Since $|\Omega_0|\le s_0^\dagger$, Lemma~\ref{lem:prob-sparse-scores}\textup{(i)} gives
\begin{equation}
  \left\|
    \mathcal P_{\Omega_0}(\widehat{\bm W}_w)
  \right\|_{\mathrm F}^2
  \le
  C\frac{\tau^2\phi_{\max}
  s_0^\dagger\log p}{\Ttot}.
  \label{eq:global-sparse-score-event-bound}
\end{equation}
This is the pooled $\Ttot$-sample statistical error for the shared sparse
component.

\paragraph*{Step 2.5: Upper bound of $\mathscr{S}_6$}
When $q=0$, $\bm Z=\bm0$ and $\mathscr S_6=0$.
For $0<q\le1$ and $\bm X\in\mathcal X_{\mathrm S}$, apply
Lemma~\ref{lem:lq-fractional-transfer} with
$\bm C_i=\bm X\widehat{\bbm\Gamma}_i$.
Since $s_{i,*}\le s_i^\dagger$ and $\underline m+1\le m$,
the resulting support class is contained in that of
\eqref{eq:app-shared-support-projection}.
Symmetry of the empirical covariances and that bound give, on $\Eall$,
\begin{equation}
 \label{eq:app-shared-residual-transfer}
 \begin{aligned}
 |\langle\bm X,\bm Z\rangle|^2
 &=\left|\sum_iw_i\langle\bm{\mathfrak R}_i,
                      \bm X\widehat{\bbm\Gamma}_i\rangle\right|^2\\
 &\le A_{\mathrm{ref}}
       \sup_{\Omega\in\mathcal D(\{s_{i,*}\},\underline m+1)}
       \sum_iw_i\|\mathcal P_{\Omega_i}
          (\bm X\widehat{\bbm\Gamma}_i)\|_{\mathrm F}^2\\
 &\le C L_{\mathrm S,\mathrm P}(m)A_{\mathrm{ref}}.
 \end{aligned}
\end{equation}
Frobenius duality on $\Omega_0$ therefore yields
\begin{equation}
 \mathscr S_6
 =6\|\mathcal P_{\Omega_0}(\bm Z)\|_{\mathrm F}^2=6\sup_{\substack{\supp(\bm X)\subseteq\Omega_0\\
                         \|\bm X\|_{\mathrm F}=1}}
             |\langle\bm X,\bm Z\rangle|^2\le6\sup_{\bm X\in\mathcal X_{\mathrm S}}
             |\langle\bm X,\bm Z\rangle|^2
 \le C L_{\mathrm S,\mathrm P}(m)A_{\mathrm{ref}}.
\end{equation}

\paragraph*{Step 2.6: Combine the bounds}
Substituting the bounds from Steps 2.1--2.5 into
\eqref{eq:global-sparse-residual-decomposition} gives
\begin{equation}
\begin{aligned}
  &\left\|
    \mathcal P_{\Omega_0}
    \left\{
      \nabla_{\bm S_0}F-\bbm\Delta_0\bbm\Gamma_w
    \right\}
  \right\|_{\mathrm F}^2
  \\
  \le{}&
  C\Bigg[
    \widetilde\phi_w^2(2\gamma+\gamma_*)\mu r
    +M\phi_{\max}\frac{p}{\Ttot}
    +\left\{
      (2\gamma+\gamma_*)\mu r
      +\frac{\varepsilon_\nu^2}{(1-\varepsilon_\nu)^2}
    \right\}\beta_{\mathrm{LR}}^2
  \Bigg]
  \left\|(\bm L,\bm R)\right\|_{\mathrm{align},*}^2
  \\
  &+C M\phi_{\max}
  \frac{s_0^\dagger\log p}{\Ttot}
  \|\bbm\Delta_0\|_{\mathrm F}^2
  \\
  &+C\Bigg[
    m w_{\max}
    \left(\phi_{\max}^2+\Lambda_{\mathrm{S}}^2\right)
    +M^2
    \left\{
      \frac{s_0^\dagger\log p}{\Ttot}
      +\frac{\log p}{\Ttot}\sum_{i=1}^Ns_i^\dagger
    \right\}
  \Bigg]
  \sum_{i=1}^Nw_i\|\bbm\Delta_i\|_{\mathrm F}^2
  \\
  &+C\frac{\tau^2\phi_{\max}
  s_0^\dagger\log p}{\Ttot}
  +C L_{\mathrm S,\mathrm P}(m)A_{\mathrm{ref}}.
\end{aligned}
  \label{eq:global-sparse-residual-bound}
\end{equation}
Finally, substitute this bound and the contraction in
\eqref{eq:global-sparse-amplified-population-contraction} into
\eqref{eq:S0_decomp}, with
$\zeta_0=c\alpha_{\mathrm S}\beta_{\mathrm S}^{-1}/2$ and
$\eta_{\mathrm S}=c\beta_{\mathrm S}^{-1}$. Since the sample-size condition gives
\[
  \frac{M\phi_{\max}s_0^\dagger\log p}
  {\alpha_{\mathrm S}\beta_{\mathrm S}\Ttot}
  \leq c\frac{\alpha_{\mathrm S}}{\beta_{\mathrm S}},
\]
we have
\begin{equation}
\begin{aligned}
  &\|\bbm\Delta_0^{(l+1)}\|_{\mathrm F}^2
  \\
  &\le{}
  \left(1-c\frac{\alpha_{\mathrm S}}{\beta_{\mathrm S}}\right)
  \|\bbm\Delta_0\|_{\mathrm F}^2
  \\
  &+ \frac{C}{\alpha_{\mathrm S}\beta_{\mathrm S}}
  \Bigg[
    \widetilde\phi_w^2(2\gamma+\gamma_*)\mu r
    +M\phi_{\max}\frac{p}{\Ttot}
    +\left\{
      (2\gamma+\gamma_*)\mu r
      +\frac{\varepsilon_\nu^2}{(1-\varepsilon_\nu)^2}
    \right\}\beta_{\mathrm{LR}}^2
  \Bigg]\\
  &\qquad{}\times
  \left\|(\bm L,\bm R)\right\|_{\mathrm{align},*}^2
  \\
  &+ \frac{C}{\alpha_{\mathrm S}\beta_{\mathrm S}}
  \Bigg[
    m w_{\max}
    \left(\phi_{\max}^2+\Lambda_{\mathrm{S}}^2\right)
    +M^2
    \left\{
      \frac{s_0^\dagger\log p}{\Ttot}
      +\frac{\log p}{\Ttot}\sum_{i=1}^Ns_i^\dagger
    \right\}
  \Bigg]
  \sum_{i=1}^Nw_i\|\bbm\Delta_i\|_{\mathrm F}^2
  \\
  &+ \frac{C\tau^2\phi_{\max}
  s_0^\dagger\log p}{\alpha_{\mathrm S}\beta_{\mathrm S}\Ttot}
  +\frac{C L_{\mathrm S,\mathrm P}(m)}
        {\alpha_{\mathrm S}\beta_{\mathrm S}}A_{\mathrm{ref}}.
\end{aligned}
  \label{eq:app-shared-step}
\end{equation}

\subsection{Local sparse update}
\label{subsec:app-local-sparse-update}

Recall that at round $l$, client $i$ first computes the local proposal
\begin{equation}
  \widetilde{\bm S}_i^{(l+1)}
  =\bm S_i-\eta_{\mathrm P}
  \nabla_{\bm S_i}F_i(\bm L,\bm R,\bm S_0,\bm S_i),
  \label{eq:app-local-sparse-proposal}
\end{equation}
The server applies the weighted FCT step in
\eqref{eq:FCT} to
$\{\widetilde{\bm S}_i^{(l+1)}\}_{i=1}^N$, and client $i$ obtains $\bm S_i^{(l+1)}$.
Define
\begin{equation}
  \Omega_i
  :=\supp(\bm S_i^{(l)})
  \cup\supp(\bm S_i^{(l+1)})
  \cup\supp(\bm H_{i,*}),
  \label{eq:app-local-witness-support}
\end{equation}
whose coordinate frequency is at most $2m+\underline m\le3m$.
Since $|\supp(\bm H_{i,*})|\le s_{i,*}$, we have
\[
 |\Omega_i|\le s_i^\dagger,
 \quad
 (\Omega_i)_i\in\mathcal D(\{s_i^\dagger\}_{i=1}^N,\min\{3m,N\}).
\]
The fitting conditions give $s_i^\dagger\le s_i'$, so the restricted
curvature bounds $\alpha_i,\beta_i$ apply on $\Omega_i$.
Applying Lemma~\ref{lem:det-weighted-fcht} with comparison matrices
$\bm H_{i,*}$ and reference frequency $\underline m$ gives
\begin{equation}
  \sum_{i=1}^Nw_i
  \|\bbm\Delta_i^{(l+1)}\|_{\mathrm F}^2
  \le
  (1+\rho_L)\sum_{i=1}^Nw_i
  \left\|
  \mathcal P_{\Omega_i}
  \left(\widetilde{\bm S}_i^{(l+1)}-\bm H_{i,*}\right)
  \right\|_{\mathrm F}^2,
  \label{eq:app-local-fcht-amplification}
\end{equation}
where the fitting levels and
$\underline m/m\le(C\kappa_{\mathrm P}^2)^{-1}$ imply
$\rho_L\le c\alpha_{\mathrm P}\beta_{\mathrm P}^{-1}$
for a sufficiently small constant $c$.
Using the gradient identity \eqref{eq:ap-score-decomposition} and
Young's inequality, for every $\zeta_L>0$, we obtain
\begin{equation}\label{eq:Si_decomp}
\begin{aligned}
  \sum_{i=1}^Nw_i
  \|\bbm\Delta_i^{(l+1)}\|_{\mathrm F}^2
  \le{}&
  (1+\rho_L)(1+\zeta_L)
  \sum_{i=1}^Nw_i
  \left\|
    \mathcal P_{\Omega_i}
    \left(
      \bbm\Delta_i-\eta_{\mathrm P}
      \bbm\Delta_i\bbm\Gamma_i
    \right)
  \right\|_{\mathrm F}^2
  \\
  &+(1+\rho_L)(1+\zeta_L^{-1})\eta_{\mathrm P}^2
  \sum_{i=1}^Nw_i
  \Bigg\|
    \mathcal P_{\Omega_i}
    \Big\{
      \bbm\Delta_i
      (\widehat{\bbm\Gamma}_i-\bbm\Gamma_i)
      +\bbm\Delta_{\bm B}\widehat{\bbm\Gamma}_i
      \\
  &\hspace{5.0cm}
      +\bbm\Delta_0\widehat{\bbm\Gamma}_i
      -\widehat{\bm W}_i
      -\bm{\mathfrak R}_i\widehat{\bbm\Gamma}_i
    \Big\}
  \Bigg\|_{\mathrm F}^2.
\end{aligned}
\end{equation}

\paragraph*{Step 1: Contraction of the population ideal direction}
Direct expansion gives
\begin{equation}
  \left\|
    \mathcal P_{\Omega_i}
    \left(
      \bbm\Delta_i-\eta_{\mathrm P}\bbm\Delta_i\bbm\Gamma_i
    \right)
  \right\|_{\mathrm F}^2
  =
  \|\bbm\Delta_i\|_{\mathrm F}^2
  +\eta_{\mathrm P}^2
  \left\|
    \mathcal P_{\Omega_i}
    (\bbm\Delta_i\bbm\Gamma_i)
  \right\|_{\mathrm F}^2
  -2\eta_{\mathrm P}
  \inner{\bbm\Delta_i\bbm\Gamma_i}{\bbm\Delta_i}.
  \label{eq:app-local-population-step-expansion}
\end{equation}
Part~\textup{(iii)} of Lemma~\ref{lem:det-population-rcg} applies with
$\bm H=\bbm\Delta_i$ and $\Omega_i$.
Substituting \eqref{eq:app-local-population-rcg} into
\eqref{eq:app-local-population-step-expansion}, and using
$\eta_{\mathrm P}\le 1/\beta_\mathrm{P}\le1/\beta_i$, yields
\begin{equation}
  \left\|
    \mathcal P_{\Omega_i}
    \left(
      \bbm\Delta_i-\eta_{\mathrm P}\bbm\Delta_i\bbm\Gamma_i
    \right)
  \right\|_{\mathrm F}^2
  \le
  (1-c\frac{\alpha_i}{\beta_\mathrm{P}})\|\bbm\Delta_i\|_{\mathrm F}^2.
  \label{eq:app-local-ideal-contraction}
\end{equation}
Summing with weights $w_i$ therefore gives
\begin{equation}
  \sum_{i=1}^Nw_i
  \left\|
    \mathcal P_{\Omega_i}
    \left(
      \bbm\Delta_i-\eta_{\mathrm P}\bbm\Delta_i\bbm\Gamma_i
    \right)
  \right\|_{\mathrm F}^2
  \le
  (1-c\alpha_\mathrm{P}\beta_\mathrm{P}^{-1})
  \sum_{i=1}^Nw_i\|\bbm\Delta_i\|_{\mathrm F}^2.
  \label{eq:app-local-averaged-ideal-contraction}
\end{equation}
Take $\zeta_L=c\alpha_\mathrm{P}\beta_\mathrm{P}^{-1}$ for a sufficiently
small absolute constant $c$. It follows that
\begin{equation}
  (1+\rho_L)(1+\zeta_L)
  \sum_{i=1}^Nw_i
  \left\|
    \mathcal P_{\Omega_i}
    \left(
      \bbm\Delta_i-\eta_{\mathrm P}\bbm\Delta_i\bbm\Gamma_i
    \right)
  \right\|_{\mathrm F}^2
  \le
  (1-c\alpha_\mathrm{P}\beta_\mathrm{P}^{-1})
  \sum_{i=1}^Nw_i\|\bbm\Delta_i\|_{\mathrm F}^2.
  \label{eq:app-local-amplified-ideal-contraction}
\end{equation}

\paragraph*{Step 2: Upper bound of the local gradient discrepancy}
Expanding the second term of \eqref{eq:Si_decomp} gives
\begin{equation}
\begin{aligned}
  &\sum_{i=1}^Nw_i
  \left\|
    \mathcal P_{\Omega_i}
    \left\{
      \bbm\Delta_i
      (\widehat{\bbm\Gamma}_i-\bbm\Gamma_i)
      +\bbm\Delta_{\bm B}\widehat{\bbm\Gamma}_i
      +\bbm\Delta_0\widehat{\bbm\Gamma}_i
      -\widehat{\bm W}_i
      -\bm{\mathfrak R}_i\widehat{\bbm\Gamma}_i
    \right\}
  \right\|_{\mathrm F}^2
  \\
  \le{}&
  \underbrace{
    5\sum_{i=1}^Nw_i
    \left\|
      \mathcal P_{\Omega_i}
      \left\{
        \bbm\Delta_i
        (\widehat{\bbm\Gamma}_i-\bbm\Gamma_i)
      \right\}
    \right\|_{\mathrm F}^2
  }_{\mathscr{L}_1}
  +\underbrace{
    5\sum_{i=1}^Nw_i
    \left\|
      \mathcal P_{\Omega_i}
      (\bbm\Delta_{\bm B}\widehat{\bbm\Gamma}_i)
    \right\|_{\mathrm F}^2
  }_{\mathscr{L}_2}
  \\
  &+\underbrace{
    5\sum_{i=1}^Nw_i
    \left\|
      \mathcal P_{\Omega_i}
      (\bbm\Delta_0\widehat{\bbm\Gamma}_i)
    \right\|_{\mathrm F}^2
  }_{\mathscr{L}_3}
  +\underbrace{
    5\sum_{i=1}^Nw_i
    \left\|
      \mathcal P_{\Omega_i}(\widehat{\bm W}_i)
    \right\|_{\mathrm F}^2
  }_{\mathscr{L}_4}
  +\underbrace{
    5\sum_{i=1}^Nw_i
    \left\|
      \mathcal P_{\Omega_i}
      (\bm{\mathfrak R}_i\widehat{\bbm\Gamma}_i)
    \right\|_{\mathrm F}^2
  }_{\mathscr{L}_5}.
\end{aligned}
  \label{eq:app-local-residual-decomposition}
\end{equation}
We bound the five terms separately.

\paragraph*{Step 2.1: Upper bound of $\mathscr{L}_1$}
Since $\supp(\bbm\Delta_i)\subseteq\Omega_i$,
Lemma~\ref{lem:prob-sparse-geometry}\textup{(ii)} gives
\begin{equation}
\begin{aligned}
  \sum_{i=1}^Nw_i
  \left\|
    \mathcal P_{\Omega_i}
    \left\{
      \bbm\Delta_i
      (\widehat{\bbm\Gamma}_i-\bbm\Gamma_i)
    \right\}
  \right\|_{\mathrm F}^2
  &\lesssim
  M\phi_{\max}
  \sum_{i=1}^Nw_i
  \frac{s_i^\dagger\log p+\log N}{T_i}
  \|\bbm\Delta_i\|_{\mathrm F}^2
  \\
  &\lesssim
  M\phi_{\max}
  \max_{i\in[N]}
  \frac{s_i^\dagger\log p+\log N}{T_i}
  \sum_{i=1}^Nw_i\|\bbm\Delta_i\|_{\mathrm F}^2.
\end{aligned}
  \label{eq:app-local-centered-covariance}
\end{equation}

\paragraph*{Step 2.2: Upper bound of $\mathscr{L}_2$}
If $\bbm\Delta_{\bm B}=\bm0$, the bound is immediate. Otherwise,
Lemma~\ref{lem:orthonormal-lowrank-directions}\textup{(ii)} gives
\[
  \bm X_{\mathrm{LR}}
  :=\frac{\bbm\Delta_{\bm B}}{\|\bbm\Delta_{\bm B}\|_{\mathrm F}}
  \in\mathcal X_{\mathrm{LR}},
  \qquad
  \|\bbm\Delta_{\bm B}\|_{\mathrm F}\le C\,\left\|(\bm L,\bm R)\right\|_{\mathrm{align},*}.
\]
Frobenius duality and Lemma~\ref{lem:working-cap-interactions} give
\[
 \sum_{i=1}^Nw_i
 \|\mathcal P_{\Omega_i}(\bm X_{\mathrm{LR}}\bm E_i)\|_{\mathrm F}^2
 \le 3m w_{\max}\Lambda_{\mathrm{LR}}^2.
\]

As in \eqref{eq:appA_het_decomp}, decompose
$\widehat{\bbm\Gamma}_i
=\overline{\bbm\Gamma}+\bm E_i
+(\widehat{\bbm\Gamma}_i-\bbm\Gamma_i)$.
The frequency cap and $\|\overline{\bbm\Gamma}\|_{\op}\le\phi_{\max}$
control the common term. Lemmas~\ref{lem:working-cap-interactions}
and~\ref{lem:prob-mixed-geometry}\textup{(i)} control the
covariance discrepancy and empirical-fluctuation terms, respectively. Thus,
\begin{equation}
\begin{aligned}
  \mathscr{L}_2
  ={}&5\|\bbm\Delta_{\bm B}\|_{\mathrm F}^2\sum_{i=1}^Nw_i
  \left\|
    \mathcal P_{\Omega_i}
    (\bm X_{\mathrm{LR}}\widehat{\bbm\Gamma}_i)
  \right\|_{\mathrm F}^2
  \\
  \le{}&
  C\Bigg[
    m w_{\max}
    \left(\phi_{\max}^2+\Lambda_{\mathrm{LR}}^2\right)
    +M^2
    \left\{
      \frac{pr}{\Ttot}
      +\frac{\log p}{\Ttot}\sum_{i=1}^Ns_i^\dagger
    \right\}
  \Bigg]
  \left\|(\bm L,\bm R)\right\|_{\mathrm{align},*}^2.
\end{aligned}
  \label{eq:app-local-lowrank-interaction}
\end{equation}

\paragraph*{Step 2.3: Upper bound of $\mathscr{L}_3$}  If $\bbm\Delta_0=\bm 0$, the claim is
immediate.  Otherwise, define $\bm X_{\mathrm S}:=\bbm\Delta_0/\|\bbm\Delta_0\|_{\mathrm F}$.
Since $\supp(\bbm\Delta_0)\subseteq\Omega_0\in\mathcal C_{\mathrm S}^\dagger$,
we have $\bm X_{\mathrm S}\in\mathcal X_{\mathrm S}$.
Applying \eqref{eq:app-shared-support-projection} gives
\begin{equation}
\begin{aligned}
  \mathscr L_3
  ={}&5\|\bbm\Delta_0\|_{\mathrm F}^2
  \sum_{i=1}^Nw_i
  \left\|
    \mathcal P_{\Omega_i}
    (\bm X_{\mathrm S}\widehat{\bbm\Gamma}_i)
  \right\|_{\mathrm F}^2
  \\
  \le{}&
  C\Bigg[
    m w_{\max}
    \left(\phi_{\max}^2+\Lambda_{\mathrm{S}}^2\right)
    +M^2
    \left\{
      \frac{s_0^\dagger\log p}{\Ttot}
      +\frac{\log p}{\Ttot}\sum_{i=1}^Ns_i^\dagger
    \right\}
  \Bigg]
  \|\bbm\Delta_0\|_{\mathrm F}^2.
\end{aligned}
  \label{eq:app-local-shared-interaction}
\end{equation}

\paragraph*{Step 2.4: Upper bound of $\mathscr{L}_4$} Applying
Lemma~\ref{lem:prob-sparse-scores}\textup{(ii)} gives
\begin{equation}
  \mathscr L_4
  =5\sum_{i=1}^Nw_i
  \left\|
    \mathcal P_{\Omega_i}(\widehat{\bm W}_i)
  \right\|_{\mathrm F}^2
  \le
  C\frac{\tau^2\phi_{\max}}{\Ttot}
  \left\{
    \log p\sum_{i=1}^Ns_i^\dagger+N\log N
  \right\}.
\end{equation}

\paragraph*{Step 2.5: Upper bound of $\mathscr{L}_5$}
When $q=0$, $\bm{\mathfrak R}_i=\bm0$ for every $i$, so
$\mathscr L_5=0$.
For $0<q\le1$, Lemma~\ref{lem:lq-local} gives, on
$\mathcal E_{\mathrm{loc}}(\log(pN))$,
\begin{equation}
 \label{eq:ap-local-residual-bound}
 \begin{aligned}
 \mathscr L_5
 &=5\sum_iw_i\|\mathcal P_{\Omega_i}
       (\bm{\mathfrak R}_i\widehat{\bbm\Gamma}_i)\|_{\mathrm F}^2\\
 &\le5\sup_{(\Omega_i)_i\in\mathcal D(\{s_i^\dagger\},\min\{3m,N\})}
       \sum_iw_i\|\mathcal P_{\Omega_i}
       (\bm{\mathfrak R}_i\widehat{\bbm\Gamma}_i)\|_{\mathrm F}^2\\
 &\le20\phi_{\max}^2 A_R
       \le20\phi_{\max}^2 A_{\mathrm{ref}}.
 \end{aligned}
\end{equation}
The bound is uniform over the output supports and hence applies to
those selected by FCT.

\paragraph*{Step 2.6: Combine the bounds}
Substituting the bounds from Steps 2.1--2.5 into
\eqref{eq:app-local-residual-decomposition} yields
\begin{equation}
\begin{aligned}
  &\sum_{i=1}^Nw_i
  \left\|
    \mathcal P_{\Omega_i}
    \left\{
      \bbm\Delta_i
      (\widehat{\bbm\Gamma}_i-\bbm\Gamma_i)
      +\bbm\Delta_{\bm B}\widehat{\bbm\Gamma}_i
      +\bbm\Delta_0\widehat{\bbm\Gamma}_i
      -\widehat{\bm W}_i
      -\bm{\mathfrak R}_i\widehat{\bbm\Gamma}_i
    \right\}
  \right\|_{\mathrm F}^2
  \\
  \le{}
    &C\Bigg[
    m w_{\max}
    \left(\phi_{\max}^2+\Lambda_{\mathrm{LR}}^2\right)
    +M^2
    \left\{
      \frac{pr}{\Ttot}
      +\frac{\log p}{\Ttot}\sum_{i=1}^Ns_i^\dagger
    \right\}
  \Bigg]
  \left\|(\bm L,\bm R)\right\|_{\mathrm{align},*}^2
  \\
  &+C\Bigg[
    m w_{\max}
    \left(\phi_{\max}^2+\Lambda_{\mathrm{S}}^2\right)
    +M^2
    \left\{
      \frac{s_0^\dagger\log p}{\Ttot}
      +\frac{\log p}{\Ttot}\sum_{i=1}^Ns_i^\dagger
    \right\}
  \Bigg]
  \|\bbm\Delta_0\|_{\mathrm F}^2
  \\
  &+ CM\phi_{\max}
  \max_{i\in[N]}
  \frac{s_i^\dagger\log p+\log N}{T_i}
  \sum_{i=1}^Nw_i\|\bbm\Delta_i\|_{\mathrm F}^2
  \\
  &+C\frac{\tau^2\phi_{\max}}{\Ttot}
  \left\{
    \log p\sum_{i=1}^Ns_i^\dagger+N\log N
  \right\}
  +C\phi_{\max}^2 A_{\mathrm{ref}}.
\end{aligned}
  \label{eq:app-local-residual-bound}
\end{equation}

The local sample-size condition gives
\begin{equation}
  M\phi_{\max}
  \max_{i\in[N]}
  \frac{s_i^\dagger\log p+\log N}{T_i}
  \le c\alpha_\mathrm{P}^2.
  \label{eq:app-local-diagonal-absorption}
\end{equation}
Moreover, the choices of $\eta_{\mathrm P}$, $\zeta_L$, and $\rho_L$
give
\[
  (1+\rho_L)(1+\zeta_L^{-1})\eta_{\mathrm P}^2
  \le\frac{C}{\alpha_\mathrm{P}\,\beta_\mathrm{P}}.
\]
Substituting the perturbation bound and the contraction in
\eqref{eq:app-local-amplified-ideal-contraction} into
\eqref{eq:Si_decomp}, yields
\begin{equation}
\begin{aligned}
  &\sum_{i=1}^Nw_i
  \|\bbm\Delta_i^{(l+1)}\|_{\mathrm F}^2
  \\
  &\le{}
  \left(1-c\frac{\alpha_\mathrm{P}}{\beta_\mathrm{P}}\right)
  \sum_{i=1}^Nw_i
  \|\bbm\Delta_i\|_{\mathrm F}^2
  \\
  &+\frac{C}{\alpha_\mathrm{P}\,\beta_\mathrm{P}}
  \Bigg[
    m w_{\max}
    \left(\phi_{\max}^2+\Lambda_{\mathrm{LR}}^2\right)
    +M^2
    \left\{
      \frac{pr}{\Ttot}
      +\frac{\log p}{\Ttot}\sum_{i=1}^Ns_i^\dagger
    \right\}
  \Bigg]
  \left\|(\bm L,\bm R)\right\|_{\mathrm{align},*}^2
  \\
  &+\frac{C}{\alpha_\mathrm{P}\,\beta_\mathrm{P}}
  \Bigg[
    m w_{\max}
    \left(\phi_{\max}^2+\Lambda_{\mathrm{S}}^2\right)
    +M^2
    \left\{
      \frac{s_0^\dagger\log p}{\Ttot}
      +\frac{\log p}{\Ttot}\sum_{i=1}^Ns_i^\dagger
    \right\}
  \Bigg]
  \|\bbm\Delta_0\|_{\mathrm F}^2
  \\
  &+\frac{C\tau^2\phi_{\max}}
  {\alpha_\mathrm{P}\,\beta_\mathrm{P}\,\Ttot}
  \left\{
    \log p\sum_{i=1}^Ns_i^\dagger+N\log N
  \right\}
  +\frac{C\phi_{\max}^2}{\alpha_{\mathrm P}\beta_{\mathrm P}}
     A_{\mathrm{ref}}.
\end{aligned}
\label{eq:app-local-step}
\end{equation}

\subsection{Joint one-step contraction}
\label{subsec:app-joint-contraction}

Under the one-step conditions in Appendix~\ref{subsec:app-main-setting},
combining \eqref{eq:lowrank_final}, \eqref{eq:app-shared-step}, and
\eqref{eq:app-local-step} gives the entrywise reference-error recursion
\begin{equation}
  \bm e_{\mathrm{ref}}^{(l+1)}
  \leq
  \bbm\Theta_{\mathrm{ref}}\bm e_{\mathrm{ref}}^{(l)}+\bm g,
  \qquad \bm g=\bm g_{\mathrm{stat}}+\bm g_{\mathrm{app}}.
  \label{eq:app-three-block-recursion}
\end{equation}
where $\leq$ denotes the componentwise partial order when applied to vectors or matrices. The reference transition matrix $\bbm\Theta_{\mathrm{ref}}$, the statistical error vector $\bm g_{\mathrm{stat}}$, and the approximation error vector $\bm g_{\mathrm{app}}$ are defined as
\begin{equation}
\begin{aligned}
  \bbm\Theta_{\mathrm{ref}}
  &:={}
  \begin{pmatrix}
    \Theta_{11} & \Theta_{12} & \Theta_{13}\\
    \Theta_{21} & \Theta_{22} & \Theta_{23}\\
    \Theta_{31} & \Theta_{32} & \Theta_{33}
  \end{pmatrix},
  \\
  \bm g_{\mathrm{stat}}
  &:={}
  C\tau^2\phi_{\max}
  \begin{pmatrix}
    \dfrac{pr}{\alpha_{\mathrm{LR}}\beta_{\mathrm{LR}}\Ttot}\\
    \dfrac{s_0^\dagger\log p}
    {\alpha_{\mathrm S}\beta_{\mathrm S}\Ttot}\\
    \dfrac{1}{\alpha_{\mathrm P}\,\beta_{\mathrm P}\Ttot}
    \left(
      \log p\sum_{i=1}^Ns_i^\dagger+N\log N
    \right)
  \end{pmatrix},\\
    \bm g_{\mathrm{app}}
  &:=C A_{\mathrm{ref}}
  \begin{pmatrix}
    \dfrac{L_{\mathrm{LR},\mathrm P}(m)}
          {\alpha_{\mathrm{LR}}\beta_{\mathrm{LR}}}\\
    \dfrac{L_{\mathrm S,\mathrm P}(m)}
          {\alpha_{\mathrm S}\beta_{\mathrm S}}\\
    \dfrac{\phi_{\max}^2}{\alpha_{\mathrm P}\beta_{\mathrm P}}
  \end{pmatrix}.
\end{aligned}
\end{equation}
The first row of the nonnegative recursion matrix $\bbm\Theta_{\mathrm{ref}}$ is defined by
\begin{equation}
\begin{aligned}
  \Theta_{11}
  :={}&
  1-c\alpha_{\mathrm{LR}}\beta_{\mathrm{LR}}^{-1},
  \\
  \Theta_{12}
  :={}&
  \frac{C}{\alpha_{\mathrm{LR}}\beta_{\mathrm{LR}}}
  \Bigg[
    \widetilde\phi_w^2(\gamma+\gamma_*)\mu r
    +\frac{M\phi_{\max}p}{\Ttot}
    +\left\{
      (\gamma+\gamma_*)\mu r
      +\frac{\varepsilon_\nu^2}{(1-\varepsilon_\nu)^2}
    \right\}\beta_{\mathrm{LR}}^2
  \Bigg],
  \\
  \Theta_{13}
  :={}&
  \frac{C}{\alpha_{\mathrm{LR}}\beta_{\mathrm{LR}}}
  \Bigg[
    m w_{\max}
    \left(\phi_{\max}^2+\Lambda_{\mathrm{LR}}(m)^2\right)
    +M^2
    \left\{
      \frac{pr}{\Ttot}
      +\frac{\log p}{\Ttot}\sum_{i=1}^Ns_i^\dagger
    \right\}
  \Bigg].
\end{aligned}
  \label{eq:app-joint-Theta-row-1}
\end{equation}
The second row is
\begin{equation}
\begin{aligned}
  \Theta_{21}
  :={}&
  \frac{C}{\alpha_{\mathrm S}\beta_{\mathrm S}}
  \Bigg[
    \widetilde\phi_w^2(2\gamma+\gamma_*)\mu r
    +\frac{M\phi_{\max}p}{\Ttot}
    +\left\{
      (2\gamma+\gamma_*)\mu r
      +\frac{\varepsilon_\nu^2}{(1-\varepsilon_\nu)^2}
    \right\}\beta_{\mathrm{LR}}^2
  \Bigg],
  \\
  \Theta_{22}
  :={}&1-c\alpha_{\mathrm S}\beta_{\mathrm S}^{-1},
  \\
  \Theta_{23}
  :={}&
  \frac{C}{\alpha_{\mathrm S}\beta_{\mathrm S}}
  \Bigg[
    m w_{\max}
    \left(\phi_{\max}^2+\Lambda_{\mathrm S}(m)^2\right)
    +M^2
    \left\{
      \frac{s_0^\dagger\log p}{\Ttot}
      +\frac{\log p}{\Ttot}\sum_{i=1}^Ns_i^\dagger
    \right\}
  \Bigg].
\end{aligned}
  \label{eq:app-joint-Theta-row-2}
\end{equation}
The third row is
\begin{equation}
\begin{aligned}
  \Theta_{31}
  :={}&
  \frac{C}{\alpha_{\mathrm P}\,\beta_{\mathrm P}}
  \Bigg[
    m w_{\max}
    \left(\phi_{\max}^2+\Lambda_{\mathrm{LR}}(m)^2\right)
    +M^2
    \left\{
      \frac{pr}{\Ttot}
      +\frac{\log p}{\Ttot}\sum_{i=1}^Ns_i^\dagger
    \right\}
  \Bigg],
  \\
  \Theta_{32}
  :={}&
  \frac{C}{\alpha_{\mathrm P}\,\beta_{\mathrm P}}
  \Bigg[
    m w_{\max}
    \left(\phi_{\max}^2+\Lambda_{\mathrm S}(m)^2\right)
    +M^2
    \left\{
      \frac{s_0^\dagger\log p}{\Ttot}
      +\frac{\log p}{\Ttot}\sum_{i=1}^Ns_i^\dagger
    \right\}
  \Bigg],
  \\
  \Theta_{33}
  :={}&1-c\alpha_{\mathrm P}\beta_{\mathrm P}^{-1}.
\end{aligned}
  \label{eq:app-joint-Theta-row-3}
\end{equation}
Here and below, $c>0$ denotes a sufficiently small universal constant,
and $C>0$ denotes a universal constant; their values may change from line to line.
We next verify contraction under the conditions of
Theorem~\ref{thm:app-unified-main-full}.
The radius in \eqref{eq:app-unified-local-radius} and
$\alpha_{\mathrm S}\le\beta_{\mathrm{LR}}$ imply, for sufficiently small
$\varepsilon_0$,
\begin{equation}
  \beta_{\mathrm{LR}}\varepsilon_\nu^2\le c\alpha_{\mathrm{LR}},
  \quad\text{and}\quad
  \frac{\beta_{\mathrm{LR}}^2\varepsilon_\nu^2}{(1-\varepsilon_\nu)^2}
  \le c\alpha_{\mathrm{LR}}\alpha_{\mathrm S}.
  \label{eq:app-joint-epsilon-smallness}
\end{equation}
The fitting levels and sample-size conditions in
Theorem~\ref{thm:app-unified-main-full}, together with
Assumption~\ref{cond:app-unified-compatibility}, then give
\begin{equation}
\begin{aligned}
  &\widetilde\phi_w^2(2\gamma+\gamma_*)\mu r
  +M\phi_{\max}\frac{p}{\Ttot}
  +\left\{
    (2\gamma+\gamma_*)\mu r
    +\frac{\varepsilon_\nu^2}{(1-\varepsilon_\nu)^2}
  \right\}\beta_{\mathrm{LR}}^2
  \le c\alpha_{\mathrm{LR}}\alpha_{\mathrm S},
  \\
  &m w_{\max}
  \left(\phi_{\max}^2+\Lambda_{\mathrm{LR}}(m)^2\right)
  +M^2
  \left\{
    \frac{pr}{\Ttot}
    +\frac{\log p}{\Ttot}\sum_{i=1}^Ns_i^\dagger
  \right\}
  \le c\alpha_{\mathrm{LR}}\alpha_{\mathrm P},
  \\
  &m w_{\max}
  \left(\phi_{\max}^2+\Lambda_{\mathrm S}(m)^2\right)
  +M^2
  \left\{
    \frac{s_0^\dagger\log p}{\Ttot}
    +\frac{\log p}{\Ttot}\sum_{i=1}^Ns_i^\dagger
  \right\}
  \le c\alpha_{\mathrm S}\alpha_{\mathrm P}.
\end{aligned}
  \label{eq:app-readable-cross-block-smallness}
\end{equation}

Next, we verify $\rho(\bbm\Theta_{\mathrm{ref}})<1$. Define
\begin{equation}
\begin{aligned}
  \bbm G
  &:={}
  \begin{pmatrix}
    0 & g_{12} & g_{13}\\
    g_{21} & 0 & g_{23}\\
    g_{31} & g_{32} & 0
  \end{pmatrix},
  &
  g_{jk}
  &:={}
  \frac{\Theta_{jk}}{1-\Theta_{jj}},
  \qquad j,k\in[3],\quad j\ne k.
\end{aligned}
  \label{eq:app-normalized-coupling-matrix}
\end{equation}
By $0\le\Theta_{jj}<1$ and the definition of $\bbm G$,
\begin{equation}
  \bm I-\bbm\Theta_{\mathrm{ref}}
  =
  \operatorname{diag}
  (1-\Theta_{11},1-\Theta_{22},1-\Theta_{33})
  (\bm I-\bbm G),
  \label{eq:app-normalized-matrix-factorization}
\end{equation}
where the diagonal factor is positive. Note that for any entrywise nonnegative matrix $\bm A$, $\rho(\bm A)<1$ if and
only if there exists $\bm v>\bm 0$ such that
$(\bm I-\bm A)\bm v>\bm 0$, with both inequalities understood componentwise.
Consequently,
\begin{equation}
  \rho(\bbm\Theta_{\mathrm{ref}})<1
  \quad\Longleftrightarrow\quad
  \rho(\bbm G)<1.
  \label{eq:app-spectral-radius-equivalence}
\end{equation}
Thus, it remains to control the spectral radius of the zero-diagonal matrix $\bbm G$. Combining \eqref{eq:app-joint-Theta-row-1}--
\eqref{eq:app-joint-Theta-row-3} with
\eqref{eq:app-readable-cross-block-smallness} gives
\begin{equation}
\begin{aligned}
  g_{12}
  &\le c\frac{\alpha_{\mathrm S}}{\alpha_{\mathrm{LR}}},
  &
  g_{21}
  &\le c\frac{\alpha_{\mathrm{LR}}}{\alpha_{\mathrm S}},
  \\
  g_{13}
  &\le c\frac{\alpha_{\mathrm P}}{\alpha_{\mathrm{LR}}},
  &
  g_{31}
  &\le c\frac{\alpha_{\mathrm{LR}}}{\alpha_{\mathrm P}},
  \\
  g_{23}
  &\le c\frac{\alpha_{\mathrm P}}{\alpha_{\mathrm S}},
  &
  g_{32}
  &\le c\frac{\alpha_{\mathrm S}}{\alpha_{\mathrm P}}.
\end{aligned}
  \label{eq:app-normalized-coupling-bounds}
\end{equation}

Letting $\bm v:=
  (\alpha_{\mathrm{LR}}^{-1},\alpha_{\mathrm S}^{-1},\alpha_{\mathrm P}^{-1})^\top
  >\bm 0$, the six bounds in \eqref{eq:app-normalized-coupling-bounds}
directly give $\bbm G\bm v\leq2c\,\bm v$. Since $\bbm G$ is entrywise nonnegative, the Collatz--Wielandt inequality
gives
\[
  \rho(\bbm G)
  \le
  \max_{j\in[3]}\frac{(\bbm G\bm v)_j}{v_j}
  \le2c.
\]
Taking the interaction constants sufficiently small gives
\begin{equation}
  \rho(\bbm G)\le2c\le\frac14,
  \qquad
  \rho(2\bbm G)\le\frac12<1.
  \label{eq:app-verified-spectral-radius}
\end{equation}
Equation~\eqref{eq:app-spectral-radius-equivalence} then gives
$\rho(\bbm\Theta_{\mathrm{ref}})<1$.

Next, we turn to the population errors. Set
\begin{equation}
  \varrho_{\mathrm{ref}}
  :=
  \frac c2
  \min\left\{
    \frac{\alpha_{\mathrm{LR}}}{\beta_{\mathrm{LR}}},
    \frac{\alpha_{\mathrm S}}{\beta_{\mathrm S}},
    \frac{\alpha_{\mathrm P}}{\beta_{\mathrm P}}
  \right\}
  =\frac{c}{2\kappa_{\max}}.
  \label{eq:app-contraction-factor}
\end{equation}
Here $c$ is the common diagonal contraction constant in
$\bbm\Theta_{\mathrm{ref}}$, chosen as a lower bound for the constants
in the three componentwise estimates.
Put
\[
 \eta=\varrho_{\mathrm{ref}}/4,\qquad
 \bm J=\operatorname{diag}(1,1,1+\eta),\qquad
 \bm u=(0,0,1)^\top.
\]
Young's inequality for the personalized component gives
\begin{equation}
\label{eq:app-population-comparison}
 \bm e_{\mathrm{ref}}^{(l)}
 \leq\bm J\bm e^{(l)}+(1+\eta^{-1})A_R\bm u,
 \qquad
 \bm e^{(l)}
 \leq\bm J\bm e_{\mathrm{ref}}^{(l)}+(1+\eta^{-1})A_R\bm u.
\end{equation}
Applying the second inequality at round $l+1$ and substituting
\eqref{eq:app-three-block-recursion} and the first inequality at round $l$
yields
\begin{equation}
\label{eq:app-population-transfer}
 \bm e^{(l+1)}
 \leq\bm J\bbm\Theta_{\mathrm{ref}}\bm J\bm e^{(l)}
       +\bm J\bm g
       +(1+\eta^{-1})A_R
          (\bm J\bbm\Theta_{\mathrm{ref}}\bm u+\bm u).
\end{equation}
The coefficient matrix in \eqref{eq:app-population-transfer} is
\[
 \bm J\bbm\Theta_{\mathrm{ref}}\bm J
 =\begin{pmatrix}
 \Theta_{11}&\Theta_{12}&(1+\eta)\Theta_{13}\\
 \Theta_{21}&\Theta_{22}&(1+\eta)\Theta_{23}\\
 (1+\eta)\Theta_{31}&(1+\eta)\Theta_{32}&(1+\eta)^2\Theta_{33}
 \end{pmatrix}.
\]
Thus the first two diagonal entries are unchanged. For the third,
$\eta\le c\alpha_{\mathrm P}/(8\beta_{\mathrm P})$ implies
\[
 (1+\eta)^2\Theta_{33}
 =(1+\eta)^2
 \left(1-c\frac{\alpha_{\mathrm P}}{\beta_{\mathrm P}}\right)
 \le1-\frac c2\frac{\alpha_{\mathrm P}}{\beta_{\mathrm P}}.
\]
For the off-diagonal entries, multiplication by $\bm J$ introduces at
most the bounded factor $1+\eta$. Since
$\gamma+\gamma_*\le2\gamma+\gamma_*\le3\gamma^\dagger$,
the numerators of $\Theta_{12}$ and $\Theta_{21}$ are bounded by
$C L_{\mathrm{LR},\mathrm S}$.
The pairs $(\Theta_{13},\Theta_{31})$ and $(\Theta_{23},\Theta_{32})$
have numerators $C L_{\mathrm{LR},\mathrm P}(m)$ and
$C L_{\mathrm S,\mathrm P}(m)$, respectively.
The denominator of each off-diagonal entry in row
$a\in\{\mathrm{LR},\mathrm S,\mathrm P\}$ is $\alpha_a\beta_a$.
Thus, for fixed universal constants $0<c_0\le c/2$ and sufficiently
large $C_0$, we take the matrix in
\eqref{eq:app-unified-transition-matrix} as
\[
 \bbm\Theta=
 \begin{pmatrix}
 1-c_0\alpha_{\mathrm{LR}}/\beta_{\mathrm{LR}}
 &\dfrac{C_0L_{\mathrm{LR},\mathrm S}}
              {\alpha_{\mathrm{LR}}\beta_{\mathrm{LR}}}
 &\dfrac{C_0L_{\mathrm{LR},\mathrm P}(m)}
              {\alpha_{\mathrm{LR}}\beta_{\mathrm{LR}}}\\[3pt]
 \dfrac{C_0L_{\mathrm{LR},\mathrm S}}
              {\alpha_{\mathrm S}\beta_{\mathrm S}}
 &1-c_0\alpha_{\mathrm S}/\beta_{\mathrm S}
 &\dfrac{C_0L_{\mathrm S,\mathrm P}(m)}
              {\alpha_{\mathrm S}\beta_{\mathrm S}}\\[3pt]
 \dfrac{C_0L_{\mathrm{LR},\mathrm P}(m)}
              {\alpha_{\mathrm P}\beta_{\mathrm P}}
 &\dfrac{C_0L_{\mathrm S,\mathrm P}(m)}
              {\alpha_{\mathrm P}\beta_{\mathrm P}}
 &1-c_0\alpha_{\mathrm P}/\beta_{\mathrm P}
 \end{pmatrix}.
\]
The preceding bounds give the entrywise comparison
\[
 \bm J\bbm\Theta_{\mathrm{ref}}\bm J\leq\bbm\Theta.
\]
Equation~\eqref{eq:app-readable-cross-block-smallness} gives
\begin{equation}
\label{eq:app-population-weighted-contraction}
\begin{aligned}
 \bbm\ell^\top\bbm\Theta
 \leq\bbm\ell^\top-
   \frac{c_0}{2\beta_{\mathrm{LR}}}
       (\alpha_{\mathrm{LR}},\alpha_{\mathrm S},\alpha_{\mathrm P})\leq(1-\varrho)\bbm\ell^\top,
 \qquad \varrho=\frac{c_0}{2\kappa_{\max}}.
\end{aligned}
\end{equation}
In particular, $\rho(\bbm\Theta)\le1-\varrho<1$.
The matrix $\bm\Theta_\nu$ in main-text
Theorem~\ref{thm:main_fed_iht} is this matrix $\bbm\Theta$,
with $c_0\alpha_a/\beta_a$ and $C_0/(\alpha_a\beta_a)$ in row $a$
abbreviated as $c_\nu$ and $C_\nu$, respectively. These constants
may differ across entries but are fixed throughout the recursion.

To bound the statistical and approximation error terms, note that $1+\eta^{-1}\le C\kappa_{\max}$ and
$A_R\le A_{\mathrm{ref}}\le\mathcal A_q(m)$.
In particular,
\[
 \bm J\bbm\Theta_{\mathrm{ref}}\bm u+\bm u
 =\begin{pmatrix}
   \Theta_{13}\\
   \Theta_{23}\\
   1+(1+\eta)\Theta_{33}
  \end{pmatrix}
 \leq C\begin{pmatrix}
   \dfrac{L_{\mathrm{LR},\mathrm P}(m)}
         {\alpha_{\mathrm{LR}}\beta_{\mathrm{LR}}}\\[3pt]
   \dfrac{L_{\mathrm S,\mathrm P}(m)}
         {\alpha_{\mathrm S}\beta_{\mathrm S}}\\[3pt]
   1
  \end{pmatrix},
\]
where $1+(1+\eta)\Theta_{33}\le2+\eta\le C$.
Since
$\phi_{\max}^2/(\alpha_{\mathrm P}\beta_{\mathrm P})\ge1$,
the same sufficiently large $C_0$ in
\eqref{eq:app-unified-forcing-vectors} gives
\begin{equation}
\label{eq:app-population-forcing-envelope}
 \bm J\bm g+(1+\eta^{-1})A_R
     (\bm J\bbm\Theta_{\mathrm{ref}}\bm u+\bm u)
 \leq\bm f_{\mathrm{stat}}+\bm f_{\mathrm{app},q}(m)
 =\bm f.
\end{equation}
Therefore, the complete errors obey
\begin{equation}
\label{eq:app-population-one-step-contraction}
 \bm e^{(l+1)}\leq\bbm\Theta\bm e^{(l)}+\bm f,
 \qquad
 V^{(l+1)}\le(1-\varrho)V^{(l)}+\bbm\ell^\top\bm f.
\end{equation}
At $q=0$, $A_R=A_{\mathrm{ref}}=0$ and
$\bm e_{\mathrm{ref}}=\bm e$, so the approximation terms vanish.

\subsection{Entry into the local convergence basin}
\label{subsec:app-basin-conclusion}

The feasible initializer satisfies
$V^{(0)}\le\varepsilon_\nu^2\sigma_r^2$ by
\eqref{eq:app-unified-initial-condition}.
Since $e_{\mathrm{LR}}^{(0)}\le V^{(0)}$ and $\varepsilon_\nu<1$,
Lemma~\ref{lem:metric_equivalence} gives an optimal alignment and
full column rank of the initialization factors.

The signal-strength condition
\eqref{eq:app-unified-signal-sample-condition} and the cap condition,
together with \eqref{eq:app-readable-cross-block-smallness}, imply
\begin{align}
 \bbm\ell^\top\bm f_{\mathrm{stat}}
 &\le\frac{\varrho}{16}\varepsilon_\nu^2\sigma_r^2,
 \label{eq:app-statistical-forcing-control}\\
 \bbm\ell^\top\bm f_{\mathrm{app},q}(m)
 &\le\frac{C\kappa_{\max}\phi_{\max}^2}
             {\beta_{\mathrm{LR}}\alpha_{\mathrm P}}
       \mathcal A_q(m)
 \le\frac{\varrho}{16}\varepsilon_\nu^2\sigma_r^2.
 \label{eq:app-approximation-forcing-control}
\end{align}
Hence
\begin{equation}
\label{eq:app-combined-forcing-control}
 \bbm\ell^\top\bm f
 \le\frac{\varrho}{8}\varepsilon_\nu^2\sigma_r^2.
\end{equation}

Suppose the iterates at round $l$ are feasible and satisfy
$V^{(l)}\le\varepsilon_\nu^2\sigma_r^2$.
Then $e_{\mathrm{LR}}^{(l)}\le V^{(l)}$ guarantees the local condition
for the one-step bounds. Applying
\eqref{eq:app-population-one-step-contraction} and
\eqref{eq:app-combined-forcing-control} gives
\[
\begin{aligned}
 V^{(l+1)}
 &\le(1-\varrho)V^{(l)}+\bbm\ell^\top\bm f\\
 &\le\left(1-\varrho+\frac{\varrho}{8}\right)
          \varepsilon_\nu^2\sigma_r^2
 \le\varepsilon_\nu^2\sigma_r^2.
\end{aligned}
\]
This closes the induction.

Finally, since $\bbm\Theta$ is entrywise nonnegative and
$\rho(\bbm\Theta)<1$, iteration gives
\begin{equation}
\label{eq:app-iterated-recursion}
 \bm e^{(l)}
 \leq\bbm\Theta^l\bm e^{(0)}
       +\sum_{k=0}^{l-1}\bbm\Theta^k\bm f
 \leq\bbm\Theta^l\bm e^{(0)}
       +(\bm I_3-\bbm\Theta)^{-1}\bm f.
\end{equation}
Iterating the scalar recursion gives the second bound in
\eqref{eq:app-unified-iterated-bound}. Together with the uniform event
and its probability in Appendix~\ref{subsec:app-main-setting}, this
completes the proof of Theorem~\ref{thm:app-unified-main-full}.
\hfill$\square$

\subsection{Proofs of the statistical corollaries}
\label{subsec:app-corollary-proofs}

\begin{proof}[Proof of Corollary~\ref{cor:app-stat-full}]
We work on the event of Theorem~\ref{thm:app-unified-main-full}.

\paragraph*{Step 1: Reference error recursion}
By \eqref{eq:app-normalized-coupling-bounds},
\[
 (\alpha_{\mathrm{LR}},\alpha_{\mathrm S},\alpha_{\mathrm P})\bbm G
 \leq2c(\alpha_{\mathrm{LR}},\alpha_{\mathrm S},\alpha_{\mathrm P}).
\]
The factorization \eqref{eq:app-normalized-matrix-factorization}
therefore gives
\begin{equation}
\label{eq:app-weighted-spectral-bound}
\begin{aligned}
 \bbm\ell^\top(\bm I_3-\bbm\Theta_{\mathrm{ref}})
 &=\frac c{\beta_{\mathrm{LR}}}
   (\alpha_{\mathrm{LR}},\alpha_{\mathrm S},\alpha_{\mathrm P})
   (\bm I_3-\bbm G)\\
 &\geq\frac c{2\beta_{\mathrm{LR}}}
   (\alpha_{\mathrm{LR}},\alpha_{\mathrm S},\alpha_{\mathrm P})
 \geq\varrho_{\mathrm{ref}}\bbm\ell^\top,
\end{aligned}
\end{equation}
with $\varrho_{\mathrm{ref}}$ as in \eqref{eq:app-contraction-factor}.
Thus
\begin{equation}
\label{eq:app-reference-weighted-contraction}
 \bbm\ell^\top\bbm\Theta_{\mathrm{ref}}
 \leq(1-\varrho_{\mathrm{ref}})\bbm\ell^\top.
\end{equation}
The basin induction in Appendix~\ref{subsec:app-basin-conclusion}
ensures that \eqref{eq:app-three-block-recursion} holds at every round.
Iterating it yields
\begin{equation}
\label{eq:ap-iteration-proof}
 \bm e_{\mathrm{ref}}^{(l)}
 \leq\bbm\Theta_{\mathrm{ref}}^l\bm e_{\mathrm{ref}}^{(0)}
       +(\bm I_3-\bbm\Theta_{\mathrm{ref}})^{-1}
          (\bm g_{\mathrm{stat}}+\bm g_{\mathrm{app}}).
\end{equation}

\paragraph*{Step 2: Componentwise statistical errors}
Write
\[
 \bm x_{\mathrm{stat}}
   :=(\bm I_3-\bbm\Theta_{\mathrm{ref}})^{-1}\bm g_{\mathrm{stat}},
 \qquad
 \bm x_{\mathrm{app}}
   :=(\bm I_3-\bbm\Theta_{\mathrm{ref}})^{-1}\bm g_{\mathrm{app}}.
\]
Normalizing by the diagonal gaps gives
\begin{equation}
\label{eq:app-normalized-statistical-floor}
\begin{aligned}
 \bm x_{\mathrm{stat}}
 &=(\bm I_3-\bbm G)^{-1}
   \operatorname{diag}\left\{
      \frac1{1-\Theta_{11}},\frac1{1-\Theta_{22}},
      \frac1{1-\Theta_{33}}\right\}\bm g_{\mathrm{stat}}\\
 &\leq C\tau^2\phi_{\max}(\bm I_3-\bbm G)^{-1}
 \begin{pmatrix}
  \dfrac{pr}{\alpha_{\mathrm{LR}}^2\Ttot}\\[2mm]
  \dfrac{s_0^\dagger\log p}{\alpha_{\mathrm S}^2\Ttot}\\[2mm]
  \dfrac{ND_{\mathrm P}^\dagger}{\alpha_{\mathrm P}^2\Ttot}
 \end{pmatrix}.
\end{aligned}
\end{equation}
The normalized inverse has the explicit form
\begin{equation}
\label{eq:app-explicit-normalized-inverse}
 (\bm I_3-\bbm G)^{-1}
 =\frac1{\det(\bm I_3-\bbm G)}
 \begin{pmatrix}
  1-g_{23}g_{32}&g_{12}+g_{13}g_{32}&g_{13}+g_{12}g_{23}\\
  g_{21}+g_{23}g_{31}&1-g_{13}g_{31}&g_{23}+g_{21}g_{13}\\
  g_{31}+g_{32}g_{21}&g_{32}+g_{31}g_{12}&1-g_{12}g_{21}
 \end{pmatrix}.
\end{equation}
By \eqref{eq:app-verified-spectral-radius},
$\det(\bm I_3-\bbm G)\ge(1-\rho(\bbm G))^3\ge(3/4)^3$.
The interaction bounds in \eqref{eq:app-readable-cross-block-smallness}
and the normalized entries in \eqref{eq:app-normalized-coupling-bounds}
give
\begin{equation}
\label{eq:app-all-normalized-path-bounds}
\begin{aligned}
 g_{12}+g_{13}g_{32}
 &\lesssim\frac{\alpha_{\mathrm S}}{\alpha_{\mathrm{LR}}},
 \qquad
 g_{21}+g_{23}g_{31}
 \lesssim\frac{\alpha_{\mathrm{LR}}}{\alpha_{\mathrm S}},\\
 g_{13}+g_{12}g_{23}
 &\lesssim\frac{L_{\mathrm{LR},\mathrm P}(m)}{\alpha_{\mathrm{LR}}^2}
   +\frac{L_{\mathrm S,\mathrm P}(m)}
               {\alpha_{\mathrm{LR}}\alpha_{\mathrm S}},\\
 g_{23}+g_{21}g_{13}
 &\lesssim\frac{L_{\mathrm S,\mathrm P}(m)}{\alpha_{\mathrm S}^2}
   +\frac{L_{\mathrm{LR},\mathrm P}(m)}
               {\alpha_{\mathrm{LR}}\alpha_{\mathrm S}},\\
 g_{31}+g_{32}g_{21}
 &\lesssim\frac{\alpha_{\mathrm{LR}}}{\alpha_{\mathrm P}},
 \qquad
 g_{32}+g_{31}g_{12}
 \lesssim\frac{\alpha_{\mathrm S}}{\alpha_{\mathrm P}}.
\end{aligned}
\end{equation}
The third-column bounds retain the interaction coefficients through
which personalized errors enter the shared components.
Using $\alpha_{\mathrm S}\ge\alpha_{\mathrm{LR}}$ and
$L_{\mathrm{LR},\mathrm P}(m)\vee L_{\mathrm S,\mathrm P}(m)=\omega(m)$,
we obtain
\begin{equation}
\label{eq:app-componentwise-statistical-floors}
\begin{aligned}
 [\bm x_{\mathrm{stat}}]_1
 &\lesssim\tau^2\phi_{\max}
 \left\{\frac{D_{\mathrm G}^\dagger}{\alpha_{\mathrm{LR}}^2\Ttot}
 +\frac{\omega(m)D_{\mathrm P}^\dagger}
       {\alpha_{\mathrm{LR}}^2\alpha_{\mathrm P}^2\overline T}\right\},\\
 [\bm x_{\mathrm{stat}}]_2
 &\lesssim\tau^2\phi_{\max}
 \left\{\frac{D_{\mathrm G}^\dagger}
              {\alpha_{\mathrm{LR}}\alpha_{\mathrm S}\Ttot}
 +\frac{\omega(m)D_{\mathrm P}^\dagger}
       {\alpha_{\mathrm{LR}}\alpha_{\mathrm S}\alpha_{\mathrm P}^2\overline T}
 \right\},\\
 [\bm x_{\mathrm{stat}}]_3
 &\lesssim\frac{\tau^2\phi_{\max}}{\alpha_{\mathrm P}\Ttot}
 \left\{\frac{ND_{\mathrm P}^\dagger}{\alpha_{\mathrm P}}
       +\frac{pr}{\alpha_{\mathrm{LR}}}
       +\frac{s_0^\dagger\log p}{\alpha_{\mathrm S}}\right\}.
\end{aligned}
\end{equation}

\paragraph*{Step 3: Approximation errors}
The same normalization gives
\[
 \bm x_{\mathrm{app}}
 \leq C A_{\mathrm{ref}}(\bm I_3-\bbm G)^{-1}
 \begin{pmatrix}
  L_{\mathrm{LR},\mathrm P}(m)/\alpha_{\mathrm{LR}}^2\\
  L_{\mathrm S,\mathrm P}(m)/\alpha_{\mathrm S}^2\\
  \phi_{\max}^2/\alpha_{\mathrm P}^2
 \end{pmatrix}.
\]
Applying \eqref{eq:app-all-normalized-path-bounds}, together with
$L_{a,\mathrm P}(m)\le c\alpha_a\alpha_{\mathrm P}$ for
$a\in\{\mathrm{LR},\mathrm S\}$, yields
\begin{equation}
\label{eq:app-reference-approximation-floors}
\begin{aligned}
 [\bm x_{\mathrm{app}}]_1
 \vee\frac{\alpha_{\mathrm S}}{\alpha_{\mathrm{LR}}}
             [\bm x_{\mathrm{app}}]_2
 &\lesssim\frac{\omega(m)\phi_{\max}^2}
                 {\alpha_{\mathrm{LR}}^2\alpha_{\mathrm P}^2}
                 A_{\mathrm{ref}},\\
 [\bm x_{\mathrm{app}}]_3
 &\lesssim\frac{\phi_{\max}^2}{\alpha_{\mathrm P}^2}A_{\mathrm{ref}}.
\end{aligned}
\end{equation}

\paragraph*{Step 4: Conversion to complete errors}
Applying \eqref{eq:app-population-comparison} at rounds $0$ and $l$
to \eqref{eq:ap-iteration-proof} gives
\begin{equation}
\label{eq:app-rate-complete-conversion}
 \bm e^{(l)}
 \leq{}\bm J\bbm\Theta_{\mathrm{ref}}^l\bm J\bm e^{(0)}
       +\bm J(\bm x_{\mathrm{stat}}+\bm x_{\mathrm{app}})+(1+\eta^{-1})A_R
       (\bm J\bbm\Theta_{\mathrm{ref}}^l\bm u+\bm u).
\end{equation}
Iterating \eqref{eq:app-reference-weighted-contraction} and using
$1+\eta\le C$ and $\varrho_{\mathrm{ref}}\ge\varrho$ yields
$\bbm\ell^\top\bm J\bbm\Theta_{\mathrm{ref}}^l\bm J\bm e^{(0)}
\le C(1-\varrho)^lV^{(0)}$.
For the last term, nonnegativity gives
$\bbm\Theta_{\mathrm{ref}}^l\leq
(\bm I_3-\bbm\Theta_{\mathrm{ref}})^{-1}$, so the third-column
bounds from Step 2 imply
\begin{equation}
\label{eq:app-initial-reference-residual}
 \bm J\bbm\Theta_{\mathrm{ref}}^l\bm u+\bm u
 \leq C\frac{\beta_{\mathrm P}}{\alpha_{\mathrm P}}
 \begin{pmatrix}
  \omega(m)/\alpha_{\mathrm{LR}}^2\\
  \omega(m)/(\alpha_{\mathrm{LR}}\alpha_{\mathrm S})\\
  1
 \end{pmatrix}.
\end{equation}
The stated rates follow by substituting the preceding bounds into
\eqref{eq:app-rate-complete-conversion} and applying
Lemma~\ref{lem:lq-tail} and the factor-product comparison
\eqref{eq:det-factor-product-equivalence}.
At $q=0$, the approximation terms vanish.
\end{proof}

\begin{proof}[Proof of Corollary~\ref{cor:app-balanced-full}]
Under balanced clients, the fitting levels and
\eqref{eq:app-rate-omega} give
\[
 \omega(m)\lesssim_{\nu,\zeta}
 \frac mN+\frac{D_G}{NT}+\frac{s_*\log p}{T}.
\]
The clientwise sample-size condition implies $D_P\lesssim_\nu T$.
Together with \eqref{eq:app-balanced-interaction-condition}, this yields
\begin{equation}
\label{eq:ap-balanced-statistical-product}
\begin{aligned}
 \omega(m)\frac{D_P}{T}
 &\lesssim_{\nu,\zeta}
 \frac{mD_P}{NT}
 +\frac{D_GD_P}{NT^2}
 +\frac{s_*\log p\,D_P}{T^2}\\
 &\lesssim_{\nu,\zeta}\frac{D_G+mD_P}{NT}.
\end{aligned}
\end{equation}
The round condition \eqref{eq:app-full-statistical-rounds} gives
$(1-\varrho)^KV^{(0)}\le D_G/(NT)$.
Substitution in Corollary~\ref{cor:app-stat-full}, using
$D_{\mathrm G}^\dagger\lesssim_\nu D_G$ and
$D_{\mathrm P}^\dagger\lesssim_\nu D_P$, proves the bounds on
$\mathcal E_G^{(K)}$ and $\mathcal E_{\mathrm P}^{(K)}$.
The squared triangle inequality gives
\[
 \mathcal E_A^{(K)}\le6\mathcal E_G^{(K)}+3\mathcal E_{\mathrm P}^{(K)}.
\]
Since $m/N<1/2$ and $\omega(m)\lesssim_\nu1$, the shared bound is
also at most $C_{\nu,\zeta}\{D_P/T+D_G/(NT)+\mathcal A_q(m)\}$.
This proves \eqref{eq:app-full-balanced-shared-rate} and
\eqref{eq:app-full-balanced-personalized-rate}.

Under $D_G+Ns_*\log p\lesssim mT$, the displayed bound on
$\omega(m)$ reduces to $\omega(m)\lesssim_{\nu,\zeta}m/N$.
Substituting \eqref{eq:app-unified-approximation-envelope}, with
$\overline s_w=s_*$, gives \eqref{eq:lq-balanced-rate}.
\end{proof}

\subsection{Proof of Theorem~\ref{thm:app-init-full} (Initialization guarantees)}
\label{app:initializer-proof}

We prove the theorem for a fixed $q\in[0,1]$.
For client $i$, write $\bm C_{i,*}=\bm S_{0,*}+\bm S_{i,*}$,
whose support has size at most $s_{0,*}+s_{i,*}$.
Use the reference $\bm H_{i,*}$ and residual $\bm{\mathfrak R}_i$
in \eqref{eq:ap-population-reference}, with
$\underline m=\lfloor m/(C\kappa_{\mathrm P}^2)\rfloor$.
At $q=0$, \eqref{eq:app-init-fitting-full} gives
$\bm H_{i,*}=\bm S_{i,*}$ and $\bm{\mathfrak R}_i=\bm0$.

\paragraph*{Step 1: Clientwise convex rate}

We first establish the statistical rate for the clientwise estimates by matching the local program \eqref{eq:init_local_basu} to  \citet{basu2019low} and applying their Proposition~1(a). The quantities
\(\mathcal X^\top\mathcal X\) and \(\mathcal X^\top E\) therein correspond
to \(\widehat{\bbm\Gamma}_i\) and \(\widehat{\bm W}_i^\top\), respectively. Meanwhile, Assumption~\ref{assump:shared_separation} gives
\begin{equation*}
  p\|\bm B_*\|_{\max}
  \le
  p\|\bm U_*\|_{2,\infty}\sigma_1
   \|\bm V_*\|_{2,\infty}
  \le
  \mu r\sigma_1
  =\varpi,
\end{equation*}
so the low-rank truth satisfies the spikiness constraint.

We next verify the conditions of Proposition~1(a) of \citet{basu2019low}.  By
Lemmas~\ref{lem:prob-local-geometry} and
\ref{lem:prob-clientwise-scores}, the sample-size and penalty choices in
Theorem~\ref{thm:app-init-full} ensure, simultaneously for all
clients,
\begin{equation}
  \lambda_{\min}(\widehat{\bbm\Gamma}_i)
  \ge\frac{\phi_{\min}}2,
  \quad
  \lambda_{B,i}\ge4\|\widehat{\bm W}_i\|_{\op},
  \quad
  \lambda_{C,i}\ge
  4\|\widehat{\bm W}_i\|_{\max}
  +4\left(\frac{\phi_{\min}}2\right)
   \frac{\mu r\sigma_1}{p},
\end{equation}
with probability at least
\(1-C\exp\{-c(p+\log N)\}-C(pN)^{-c}\).  The first inequality yields the
required RSC condition with \(\zeta=\phi_{\min}/2\) and zero tolerance,
while the remaining two inequalities are their tuning conditions.
Therefore, their Proposition~1(a), followed by the penalty choices and
\((a+b)^2\le2a^2+2b^2\), gives
\begin{equation}
\begin{aligned}
  &\|\widetilde{\bm B}_i-\bm B_*\|_{\mathrm F}^2
  +\|\widetilde{\bm C}_i-\bm C_{i,*}\|_{\mathrm F}^2\\
  \lesssim{} &
  \frac{r\lambda_{B,i}^2+(s_{0,*}+s_{i,*})\lambda_{C,i}^2}{\phi_{\min}^2}
  \\
  \lesssim{} &
  \frac{\tau^2\phi_{\max}}
       {\phi_{\min}^2T_i}
  \left\{
    r(p+\log N)
    +(s_{0,*}+s_{i,*})\log(pN)
  \right\}
  +
  \frac{\mu^2r^2\sigma_1^2}{p^2}(s_{0,*}+s_{i,*})
  =R_i^{\mathrm{init}}.
\end{aligned}
\label{eq:app-initializer-clientwise-rate}
\end{equation}

\paragraph*{Step 2: Aggregation and postprocessing}

Each client first forms
\(\widetilde{\bm B}_{i,r}=\operatorname{SVD}_r(\widetilde{\bm B}_i)\). Since
\(\bm B_*\) has rank \(r\), optimality of the best rank-\(r\) approximation gives
\begin{equation*}
  \|\widetilde{\bm B}_{i,r}-\bm B_*\|_{\mathrm F}
  \le
  \|\widetilde{\bm B}_{i,r}-\widetilde{\bm B}_i\|_{\mathrm F}
  +\|\widetilde{\bm B}_i-\bm B_*\|_{\mathrm F}
  \le 2\|\widetilde{\bm B}_i-\bm B_*\|_{\mathrm F}.
\end{equation*}
Next, for $\bm B^{(0)}$, using Jensen's inequality, followed by the server's second rank truncation, yields
\begin{equation*}
\begin{aligned}
  \|\bm B^{(0)}-\bm B_*\|_{\mathrm F}^2
  &\le4\|\overline{\bm B}-\bm B_*\|_{\mathrm F}^2
  \le16\sum_{i=1}^Nw_i
     \|\widetilde{\bm B}_i-\bm B_*\|_{\mathrm F}^2.
\end{aligned}
\end{equation*}
The global factor-product bound \eqref{eq:det-factor-product-lower} then implies
\begin{equation}
  e_{\mathrm{LR}}^{(0)}=\left\|(\bm L^{(0)},\bm R^{(0)})\right\|_{\mathrm{align},*}^2
  \le C\sum_{i=1}^Nw_i
      \|\widetilde{\bm B}_i-\bm B_*\|_{\mathrm F}^2.
\label{eq:app-initializer-factor-rate}
\end{equation}

Next, we analyze the coordinatewise median \(\widetilde{\bm S}_0\) of
\(\{\widetilde{\bm C}_i\}_{i=1}^N\). Fix a coordinate $(j,k)$. If
\((\widetilde{\bm S}_0)_{jk}\ge(\bm S_{0,*})_{jk}\), the median ensures that at least \(\lceil N/2\rceil\) clients
satisfy $
  (\widetilde{\bm C}_i)_{jk}
  \ge(\widetilde{\bm S}_0)_{jk}
  \ge(\bm S_{0,*})_{jk}$. If \((\widetilde{\bm S}_0)_{jk}<(\bm S_{0,*})_{jk}\), at least
\(\lceil N/2\rceil\) clients instead satisfy $
  (\widetilde{\bm C}_i)_{jk}
  \le(\widetilde{\bm S}_0)_{jk}
  <(\bm S_{0,*})_{jk}$.
Thus, in either case, at least \(\lceil N/2\rceil\) clients satisfy
\begin{equation}\label{eq:appD_median}
  |(\widetilde{\bm C}_i)_{jk}-(\bm S_{0,*})_{jk}|
  \ge
  |(\widetilde{\bm S}_0)_{jk}-(\bm S_{0,*})_{jk}|.
\end{equation}
The reference has at most $\underline m$ nonzero client coefficients
at each coordinate. Since $\underline m<N/4$, at least
$\lceil N/2\rceil-\underline m\ge N/4$ of the clients in
\eqref{eq:appD_median} have $(\bm H_{i,*})_{jk}=0$. For these clients,
\[
 (\widetilde{\bm C}_i-\bm S_{0,*})_{jk}
 = (\widetilde{\bm C}_i-\bm C_{i,*})_{jk}
   +(\bm{\mathfrak R}_i)_{jk}.
\]
Squaring \eqref{eq:appD_median}, using $(a+b)^2\le2a^2+2b^2$,
and summing over these clients and then over coordinates gives
\begin{equation}
\label{eq:app-initializer-median-rate}
 \|\widetilde{\bm S}_0-\bm S_{0,*}\|_{\mathrm F}^2
 \le\frac CN\sum_{i=1}^N
 \left\{\|\widetilde{\bm C}_i-\bm C_{i,*}\|_{\mathrm F}^2
              +\|\bm{\mathfrak R}_i\|_{\mathrm F}^2\right\}.
\end{equation}

Recall $\bm S_0^{(0)}=\mathcal G_{s_0,\gamma}(\widetilde{\bm S}_0)$.
Under the fitting conditions, Lemma~\ref{lem:G_stability} and
\eqref{eq:app-initializer-clientwise-rate} give
\begin{equation}
\label{eq:app-initializer-sparse-rate}
 e_{\mathrm S}^{(0)}
 \le C\|\widetilde{\bm S}_0-\bm S_{0,*}\|_{\mathrm F}^2
 \le\frac CN\sum_{i=1}^N
       \left\{R_i^{\mathrm{init}}+\|\bm{\mathfrak R}_i\|_{\mathrm F}^2\right\}.
\end{equation}

\paragraph*{Step 3: Personalized residual initialization}

The residual proposal satisfies
\[
 \widetilde{\bm S}_i^{(0)}-\bm H_{i,*}
 = (\widetilde{\bm C}_i-\bm C_{i,*})
   -(\bm S_0^{(0)}-\bm S_{0,*})+\bm{\mathfrak R}_i.
\]
The reference supports belong to
$\mathcal D(\{s_{i,*}\},\underline m)$, so the fitting conditions
allow Lemma~\ref{lem:det-weighted-fcht} to be applied with bounded
amplification. Using the squared triangle inequality and
\eqref{eq:app-initializer-clientwise-rate}, we obtain
\begin{equation}
\label{eq:app-initializer-local-rate}
\begin{aligned}
 e_{\mathrm P}^{(0)}
 &\le2e_{\mathrm P,\mathrm{ref}}^{(0)}+2A_R\\
 &\le C\sum_{i=1}^Nw_i
       \|\widetilde{\bm S}_i^{(0)}-\bm H_{i,*}\|_{\mathrm F}^2+2A_R\\
 &\le C\left\{\sum_{i=1}^Nw_iR_i^{\mathrm{init}}
                    +e_{\mathrm S}^{(0)}+A_R\right\}.
\end{aligned}
\end{equation}

\paragraph*{Step 4: Initialization rate and basin entry}

Since $w_i\ge T_{\min}/\Ttot$,
\[
 \frac1N\sum_{i=1}^N\|\bm{\mathfrak R}_i\|_{\mathrm F}^2
 \le\frac{\overline T}{T_{\min}}A_R.
\]
Combining the preceding bounds with
Lemma~\ref{lem:lq-tail} gives
\begin{equation}
\label{eq:app-initializer-weighted-rate}
 V^{(0)}
 \le C\kappa_{\max}
 \left[
 \sum_{i=1}^N\left(w_i+\frac1N\right)R_i^{\mathrm{init}}
 +\left(1+\frac{\overline T}{T_{\min}}\right)\mathcal A_q(m)
 \right]
 =C\mathcal I_q(m).
\end{equation}
This proves \eqref{eq:app-init-total-error-rate}.

For basin entry, \eqref{eq:app-init-basin-sample-full} and
\eqref{eq:app-init-separation-full}, together with the definitions
of $\Xi_2$ and $\varepsilon_\nu$, imply
\begin{equation}
\label{eq:app-initializer-uniform-basin-rate}
 R_i^{\mathrm{init}}
 \le c\frac{\varepsilon_\nu^2\sigma_r^2}{\kappa_{\max}},
 \qquad i\in[N].
\end{equation}
The cap condition \eqref{eq:app-init-cap-full} also gives
\[
 \kappa_{\max}\left(1+\frac{\overline T}{T_{\min}}\right)
 \mathcal A_q(m)\le c\varepsilon_\nu^2\sigma_r^2.
\]
Since $\sum_i(w_i+N^{-1})=2$, the initialization rate yields
\begin{equation}
\label{eq:app-initializer-basin-entry}
 V^{(0)}\le C\mathcal I_q(m)
          \le\varepsilon_\nu^2\sigma_r^2.
\end{equation}
The sparse initializer is feasible by construction. The low-rank
aggregation bound also gives
$\|\bm B^{(0)}-\bm B_*\|_{\mathrm F}\le\sigma_r/2$,
so Weyl's inequality implies $\sigma_r(\bm B^{(0)})\ge\sigma_r/2$.
Thus the initialization factors have full column rank, completing
the basin-entry conclusion on the same local-fit event.
\hfill$\square$

\section{Auxiliary Stochastic Concentration Results}
\label{app:stochastic-results}

This appendix collects the concentration results used in the proofs of the main results.
Throughout, we work under the federated VAR model in
Section~\ref{subsec:2.1}, with Gaussian innovations as specified in
Assumption~\ref{assump:gaussian_sampling}, and retain the notation and
assumptions introduced in the main article and Appendix~\ref{app:main-theorem-proof}.

\subsection{Covariance concentration}
\label{subsec:prob-sample-geometry}

Write the stacked lag observations as
\(\mathcal Y_i=[\bm y_{i,T_i-1},\bm y_{i,T_i-2},\ldots,\bm y_{i,0}]\)
and define their covariance by
\begin{equation}
\label{eq:app-trajectory-envelope}
 \bbm\Upsilon_i
 :=\operatorname{Cov}\!\left(\operatorname{vec}(\mathcal Y_i)\right).
\end{equation}
For \(0\le t,s<T_i\), its \((t,s)\) block is
\(\operatorname{Cov}(\bm y_{i,T_i-1-t},\bm y_{i,T_i-1-s})\).
Let $\mathcal A_i(z)=\bm I_p-\bm A_{i,*}z$ and define
\begin{equation}
\label{eq:app-temporal_dependence_scale}
 M_i=\frac{\lambda_{\max}(\bbm\Sigma_{\epsilon,i})}
 {\min_{|z|=1}\lambda_{\min}\{\mathcal A_i(z)^{\mathrm H}
                                     \mathcal A_i(z)\}},
 \qquad M=\max_i M_i,
\end{equation}
where $\mathrm H$ denotes conjugate transpose. This gives an admissible
choice of the trajectory bound in Section~\ref{subsec:theoretical_setup}.
The VAR spectral bounds of \citet{basu2015regularized} give
\begin{equation}
\label{eq:app-trajectory-spectral-bound}
 \|\bbm\Upsilon_i\|_{\mathrm{op}}\le M_i,
 \qquad
 \|\bbm\Gamma_i\|_{\mathrm{op}}\le M_i,
 \qquad i\in[N].
\end{equation}

\begin{lemma}[Pooled covariance operator concentration]
\label{lem:prob-M0-geometry}
Under Assumption \ref{assump:stationarity}, suppose that $\Ttot
  \gtrsim M\phi_{\max}^{-1}p$.
Then, with probability at least $1-\exp(-Cp)$,
\begin{equation}
  \left\|\widehat{\bbm\Gamma}_w-\bbm\Gamma_w\right\|_{\op}^2
  \lesssim
  \frac{M\phi_{\max}p}{\Ttot}.
  \label{eq:prob-pooled-covariance-concentration}
\end{equation}
\end{lemma}

\begin{proof}
We first control one quadratic direction of the difference and then take a
net over the unit sphere.  Fix $\bm v\in\mathbb S^{p-1}$ and stack the scalar
projections of the lag observations as
\begin{equation}
  \bm z_{\bm v}
  :=
  \begin{bmatrix}
    (\bm I_{T_1}\otimes\bm v^\top)\operatorname{vec}(\mathcal Y_1)
    \\
    \vdots
    \\
    (\bm I_{T_N}\otimes\bm v^\top)\operatorname{vec}(\mathcal Y_N)
  \end{bmatrix}
  \in\mathbb R^{\Ttot}.
  \label{eq:prob-pooled-scalar-trajectory}
\end{equation}
Independence of the client trajectories gives its block-diagonal covariance matrix:
\begin{equation}
  \bbm\Psi_{\bm v}
  :=\operatorname{Cov}(\bm z_{\bm v})
  =
  \begin{bmatrix}
    (\bm I_{T_1}\otimes\bm v^\top)
  \bbm\Upsilon_1
  (\bm I_{T_1}\otimes\bm v)&\cdots&\bm 0
    \\
    \vdots&\ddots&\vdots
    \\
    \bm 0&\cdots&(\bm I_{T_N}\otimes\bm v^\top)
  \bbm\Upsilon_N
  (\bm I_{T_N}\otimes\bm v)
  \end{bmatrix}.
\end{equation}
Since $\|\bm v\|_2=1$ and
$\max_{i\in[N]}\|\bbm\Upsilon_i\|_{\op}\le M$,
\begin{equation}
\begin{aligned}
  \|\bbm\Psi_{\bm v}\|_{\op}
  &\le \max_{i\in [N]} \|\bm I_{T_i}\otimes\bm v^\top\|_{\op}\|\bbm\Upsilon_i\|_{\op}
  \|\bm I_{T_i}\otimes\bm v\|_{\op}\leq M,
  \\
  \operatorname{tr}(\bbm\Psi_{\bm v})
  &=\sum_{i=1}^NT_i\bm v^\top\bbm\Gamma_i\bm v
  \le\Ttot\phi_{\max},
  \\
  \|\bbm\Psi_{\bm v}\|_{\mathrm F}^2
  &=\operatorname{tr}(\bbm\Psi_{\bm v}^2)
  \le
  \|\bbm\Psi_{\bm v}\|_{\op}
  \operatorname{tr}(\bbm\Psi_{\bm v})
  \le \Ttot\phi_{\max}M.
  \label{eq:prob-pooled-scalar-trajectory-norms}
\end{aligned}
\end{equation}

Because the trajectories are jointly Gaussian, we write
\begin{equation}
  \bm z_{\bm v}=
  \bbm\Psi_{\bm v}^{1/2}\bbm\xi,
  \quad
  \bbm \xi\sim N(\bm 0,\bm I_{\Ttot}).
  \label{eq:prob-pooled-scalar-Gaussian-representation}
\end{equation}
Then, we have
\begin{equation}
\begin{aligned}
  \bm v^\top
  (\widehat{\bbm\Gamma}_w-\bbm\Gamma_w)\bm v
  &=\sum_{i=1}^Nw_i
  \bm v^\top
  (\widehat{\bbm\Gamma}_i-\bbm\Gamma_i)\bm v
  \\
  &=\frac1{\Ttot}
  \sum_{i=1}^N\sum_{t=1}^{T_i}
  \Big[
    (\bm v^\top\bm y_{i,t-1})^2
    -\mathbb E
    \big(
      (\bm v^\top\bm y_{i,t-1})^2
    \big)
  \Big]
  \\
  &=\frac1{\Ttot}
  \left\{
    \|\bm z_{\bm v}\|_2^2
    -\mathbb E\|\bm z_{\bm v}\|_2^2
  \right\}
  \\
  &=
  \frac1{\Ttot}
  \left\{
    \bbm\xi^\top\bbm\Psi_{\bm v}\bbm\xi
    -\operatorname{tr}(\bbm\Psi_{\bm v})
  \right\}.
  \label{eq:prob-pooled-covariance-fixed-direction}
\end{aligned}
\end{equation}

Therefore, we apply the Hanson--Wright inequality
\citep[Theorem~6.2.1]{vershynin2018high} with
the sub-Gaussian norms absorbed into the constant $C$. For every $t>0$,
\begin{equation}\label{eq:lemB1_hanson}
\begin{aligned}
  &\mathbb P\left(
    \left|
      \bm v^\top
      (\widehat{\bbm\Gamma}_w-\bbm\Gamma_w)\bm v
    \right|>t
  \right)
  \\
  &\quad\le
  2\exp\left[
    -C\min\left\{
      \frac{\Ttot^2t^2}{\|\bbm\Psi_{\bm v}\|_{\mathrm F}^2},
      \frac{\Ttot t}{\|\bbm\Psi_{\bm v}\|_{\op}}
    \right\}
  \right]
  \\
  &\quad\le
  2\exp\left[
    -C\min\left\{
      \frac{\Ttot t^2}{M\phi_{\max}},
      \frac{\Ttot t}{M}
    \right\}
  \right],
\end{aligned}
\end{equation}

Take $\overline{\mathbb S}$ to be a $1/4$-net of $\mathbb S^{p-1}$. Then, $|\overline{\mathbb S}|\le 9^p$.
For any $\bm v\in\mathbb S^{p-1}$, choose
$\overline{\bm v}\in\overline{\mathbb S}$ such that
$\|\bm v-\overline{\bm v}\|_2\le1/4$.  If
$\bm v=\overline{\bm v}$, the following conclusion is immediate; otherwise,
for every symmetric matrix $\bm M$,
\[
\begin{aligned}
  |\bm v^\top\bm M\bm v|
  &\le
  |\overline{\bm v}^\top\bm M\overline{\bm v}|
  +|(\bm v-\overline{\bm v})^\top\bm M\bm v|
  +|\overline{\bm v}^\top\bm M(\bm v-\overline{\bm v})|
  \\
  &=
  |\overline{\bm v}^\top\bm M\overline{\bm v}|
  +\|\bm v-\overline{\bm v}\|_2
  \left|
    \frac{(\bm v-\overline{\bm v})^\top}
    {\|\bm v-\overline{\bm v}\|_2}
    \bm M\bm v
  \right|
  \\
  &\quad
  +\|\bm v-\overline{\bm v}\|_2
  \left|
    \overline{\bm v}^\top\bm M
    \frac{\bm v-\overline{\bm v}}
    {\|\bm v-\overline{\bm v}\|_2}
  \right|
  \\
  &\le
  \max_{\overline{\bm v}\in\overline{\mathbb S}}
  |\overline{\bm v}^\top\bm M\overline{\bm v}|
  +2\|\bm v-\overline{\bm v}\|_2\|\bm M\|_{\op}
  \\
  &\le
  \max_{\overline{\bm v}\in\overline{\mathbb S}}
  |\overline{\bm v}^\top\bm M\overline{\bm v}|
  +\frac12\|\bm M\|_{\op}.
\end{aligned}
\]
Taking the supremum over $\bm v\in\mathbb S^{p-1}$ and using
$\|\bm M\|_{\op}=\sup_{\|\bm v\|_2=1}|\bm v^\top\bm M\bm v|$ give
\begin{equation}\label{eq:lemB1_discrete}
  \|\bm M\|_{\op}
  \le2\max_{\overline{\bm v}\in\overline{\mathbb S}}
  |\overline{\bm v}^\top\bm M\overline{\bm v}|.  
\end{equation}
Therefore, by \eqref{eq:lemB1_hanson} and the union bound, we have
\begin{equation}\label{eq:lemB1_union}
\begin{aligned}
  \mathbb P\left(\left\|\widehat{\bbm\Gamma}_w-\bbm\Gamma_w\right\|_{\op}\geq t\right) & \le \mathbb P\left(\max_{\overline{\bm v}\in\overline{\mathbb S}}\left|
      \overline{\bm v}^\top
      (\widehat{\bbm\Gamma}_w-\bbm\Gamma_w)\overline{\bm v}
  \right|>\frac{t}{2}\right)\\
  & \le \sum_{\overline{\bm v}\in\overline{\mathbb S}}\mathbb P\left(\left|
      \overline{\bm v}^\top
      (\widehat{\bbm\Gamma}_w-\bbm\Gamma_w)\overline{\bm v}
  \right|>\frac{t}{2}\right)\\
  & \le 9^p\times 2\exp\left[
    -C\min\left\{
      \frac{\Ttot t^2}{M\phi_{\max}},
      \frac{\Ttot t}{M}
    \right\}
  \right]
\end{aligned}
\end{equation}
Setting
\[
  t=C\sqrt{\frac{M\phi_{\max}p}{\Ttot}},
\]
the condition $\Ttot\gtrsim M\phi_{\max}^{-1}p$
gives
\[
  \frac{\Ttot t^2}{M\phi_{\max}}=Cp,
  \quad\text{and}\quad
  \frac{\Ttot t}{M}
  =C\sqrt{\frac{\Ttot\phi_{\max}p}{M}}
  \ge Cp.
\]
Consequently, choosing the absolute constants sufficiently large, we obtain
\begin{equation}
\begin{aligned}
  \mathbb P\left(\left\|\widehat{\bbm\Gamma}_w-\bbm\Gamma_w\right\|_{\op}\geq C\sqrt{\frac{M\phi_{\max}p}{\Ttot}}\right)
  \le \exp\left(
    -Cp
  \right).
\end{aligned}
\end{equation}
\end{proof}

\begin{lemma}[Clientwise covariance bound]
\label{lem:prob-local-geometry}
Under Assumption~\ref{assump:stationarity}, suppose that
\begin{equation*}
  T_i
  \gtrsim
  M\phi_{\max}^{-1}(p+\log N),
  \qquad i\in[N].
\end{equation*}
Then, with probability at least
$1-C\exp\{-c(p+\log N)\}$, simultaneously for all $i\in[N]$,
\begin{equation}
  \left\|\widehat{\bbm\Gamma}_i-\bbm\Gamma_i\right\|_{\op}^2
  \lesssim
  \frac{M\phi_{\max}(p+\log N)}{T_i}.
  \label{eq:prob-local-covariance-concentration}
\end{equation}
Moreover,
\begin{equation*}
  \lambda_{\max}(\widehat{\bbm\Gamma}_i)
  \le
  C\phi_{\max}.
\end{equation*}
If, in addition,
\begin{equation*}
  T_i
  \gtrsim
  M\phi_{\max}\phi_{\min}^{-2}(p+\log N),
  \qquad i\in[N],
\end{equation*}
then, on the same event,
\begin{equation*}
  \frac{\phi_{\min}}{2}
  \le
  \lambda_{\min}(\widehat{\bbm\Gamma}_i)
  \le
  \lambda_{\max}(\widehat{\bbm\Gamma}_i)
  \le
  \phi_{\max}+\frac{\phi_{\min}}{2}
  \le
  \frac{3\phi_{\max}}{2}.
\end{equation*}
\end{lemma}

\begin{proof}
Fix $i\in[N]$ and $\bm v\in\mathbb S^{p-1}$, and define 
\begin{equation*}
  \bm z_{i,\bm v}
  :=
  (\bm I_{T_i}\otimes\bm v^\top)
  \operatorname{vec}(\mathcal Y_i),
  \qquad
  \bbm\Psi_{i,\bm v}
  :=
  \operatorname{Cov}(\bm z_{i,\bm v}).
\end{equation*}
Then,
\begin{equation}
\begin{aligned}
  \|\bbm\Psi_{i,\bm v}\|_{\op}
  &\le M,
  \\
  \operatorname{tr}(\bbm\Psi_{i,\bm v})
  &=T_i\bm v^\top\bbm\Gamma_i\bm v
  \le T_i\phi_{\max},
  \\
  \|\bbm\Psi_{i,\bm v}\|_{\mathrm F}^2
  &\le
  \|\bbm\Psi_{i,\bm v}\|_{\op}
  \operatorname{tr}(\bbm\Psi_{i,\bm v})
  \le T_iM\phi_{\max}.
  \label{eq:prob-local-scalar-trajectory-norms}
\end{aligned}
\end{equation}
Therefore, applying the Hanson--Wright inequality and the covering-net
argument from the proof of Lemma~\ref{lem:prob-M0-geometry} clientwise,
with $\Ttot$ replaced by $T_i$, yields the following: for every $x\ge p$
satisfying
\begin{equation}\label{eq:lemB2_sample_size}
  T_i\gtrsim M\phi_{\max}^{-1}x,
\end{equation}
we have
\begin{equation}
  \mathbb P\left(
    \left\|\widehat{\bbm\Gamma}_i-\bbm\Gamma_i\right\|_{\op}
    >C\sqrt{\frac{M\phi_{\max}x}{T_i}}
  \right)
  \le C\exp(-cx).
  \label{eq:prob-local-covariance-clientwise-tail}
\end{equation}
Taking $x\asymp p+\log N$ and then a union bound over $i\in[N]$ proves
\eqref{eq:prob-local-covariance-concentration} under the first sample size
condition $T_i
  \gtrsim
  M\phi_{\max}^{-1}(p+\log N)$ with the stated probability. It
also gives
\begin{equation*}
  \lambda_{\max}(\widehat{\bbm\Gamma}_i)
  \le
  \lambda_{\max}(\bbm\Gamma_i)
  +\|\widehat{\bbm\Gamma}_i-\bbm\Gamma_i\|_{\op}
  \le C\phi_{\max}.
\end{equation*}

Under the stronger sample-size condition, choosing its implicit constant
sufficiently large gives
\begin{equation}
  \left\|\widehat{\bbm\Gamma}_i-\bbm\Gamma_i\right\|_{\op}^2
  \lesssim
  \frac{M\phi_{\max}(p+\log N)}{T_i}
  \le
  \frac{\phi_{\min}^2}{4}.
\end{equation}
Therefore, Weyl's inequality gives
\begin{equation*}
\begin{aligned}
  \lambda_{\min}(\widehat{\bbm\Gamma}_i)
  &\ge
  \lambda_{\min}(\bbm\Gamma_i)
  -\|\widehat{\bbm\Gamma}_i-\bbm\Gamma_i\|_{\op}
  \ge
  \frac{\phi_{\min}}{2},
  \\
  \lambda_{\max}(\widehat{\bbm\Gamma}_i)
  &\le
  \lambda_{\max}(\bbm\Gamma_i)
  +\|\widehat{\bbm\Gamma}_i-\bbm\Gamma_i\|_{\op}
  \le
  \phi_{\max}+\frac{\phi_{\min}}{2}
  \le
  \frac{3\phi_{\max}}{2}.
\end{aligned}
\end{equation*}
These two inequalities prove the asserted two-sided eigenvalue bound.
\end{proof}

\begin{lemma}[Sparse restricted centered-covariance concentration]
\label{lem:prob-sparse-geometry}
Under Assumption~\ref{assump:stationarity}, the following two conclusions hold.
\begin{enumerate}[label=\textup{(\roman*)}]
\item Suppose that
$\Ttot\gtrsim M\phi_{\max}^{-1}s_0^\dagger\log p$. Then, with
probability at least $1-C\exp\{-c(s_0^\dagger\vee1)\log p\}$, uniformly over every
  $\Omega_0\subset[p]\times[p]$ and $\bm H\in\mathbb R^{p\times p}$ satisfying
\begin{equation*}
  |\Omega_0|\le s_0^\dagger,
  \qquad
  \supp(\bm H)\subseteq\Omega_0,
\end{equation*}
we have
\begin{equation}
  \left\|
    \mathcal P_{\Omega_0}
    \left\{\bm H
    (\widehat{\bbm\Gamma}_w-\bbm\Gamma_w)\right\}
  \right\|_{\mathrm F}^2
  \lesssim 
  \frac{M\phi_{\max}s_0^\dagger\log p}{\Ttot}
  \|\bm H\|_{\mathrm F}^2.
  \label{eq:prob-global-sparse-centered-covariance}
\end{equation}

\item For every $i\in[N]$, suppose that
$T_i\gtrsim M\phi_{\max}^{-1}
(s_i^\dagger\log p+\log N)$. Then, with probability at least
$1-C\exp\{-c((s_{\min}^\dagger\vee1)\log p+\log N)\}$, simultaneously for
every $i\in[N]$, every $\Omega_i\subset[p]\times[p]$ satisfying
$|\Omega_i|\le s_i^\dagger$, and every
$\bm H\in\mathbb R^{p\times p}$ satisfying $\supp(\bm H)\subseteq\Omega_i$, we have
\begin{equation}
  \left\|
    \mathcal P_{\Omega_i}
    \left\{\bm H
    (\widehat{\bbm\Gamma}_i-\bbm\Gamma_i)\right\}
  \right\|_{\mathrm F}^2
  \le
  C M\phi_{\max}
  \frac{s_i^\dagger\log p+\log N}{T_i}
  \|\bm H\|_{\mathrm F}^2.
  \label{eq:prob-local-sparse-centered-covariance}
\end{equation}
\end{enumerate}
\end{lemma}

\begin{proof}
Let
$(\Omega_0)_{a,\cdot}=\{b:(a,b)\in\Omega_0\}$ be the active columns corresponding to the $a$-th row of $\Omega_0$. Then,
\begin{equation}
\begin{aligned}
  &\left\|
    \mathcal P_{\Omega_0}
    \left\{\bm H
    (\widehat{\bbm\Gamma}_w-\bbm\Gamma_w)\right\}
  \right\|_{\mathrm F}^2
  \\
  &=\sum_{a=1}^p
  \left\|
    \bm H_{a,(\Omega_0)_{a,\cdot}}
    [\widehat{\bbm\Gamma}_w-\bbm\Gamma_w]
    _{(\Omega_0)_{a,\cdot},(\Omega_0)_{a,\cdot}}
  \right\|_2^2
  \\
  &\le
  \sum_{a=1}^p
  \|\bm H_{a,(\Omega_0)_{a,\cdot}}\|_2^2
  \left\|
    [\widehat{\bbm\Gamma}_w-\bbm\Gamma_w]
    _{(\Omega_0)_{a,\cdot},(\Omega_0)_{a,\cdot}}
  \right\|_{\op}^2
  \\
  &\le
  \sup_{\substack{J\subset[p],\ |J|\le s_0^\dagger}}
  \left\|
    [\widehat{\bbm\Gamma}_w-\bbm\Gamma_w]_{J,J}
  \right\|_{\op}^2
  \|\bm H\|_{\mathrm F}^2.
\end{aligned}
  \label{eq:prob-pooled-restricted-covariance-to-product}
\end{equation}
The same rowwise calculation gives,
\begin{equation}
  \left\|
    \mathcal P_{\Omega_i}
    \left\{\bm H
    (\widehat{\bbm\Gamma}_i-\bbm\Gamma_i)\right\}
  \right\|_{\mathrm F}^2\le
  \sup_{\substack{J\subset[p],\ |J|\le s_i^\dagger}}
  \left\|
    [\widehat{\bbm\Gamma}_i-\bbm\Gamma_i]_{J,J}
  \right\|_{\op}^2
  \|\bm H\|_{\mathrm F}^2.
  \label{eq:prob-local-restricted-covariance-to-product}
\end{equation}
Thus it remains to control the two restricted covariance norms.

For part~\textup{(i)}, fix $J\subset[p]$ with $|J|\le s_0^\dagger$ and a
unit vector $\bm v$ supported on $J$. Repeating the derivation leading to
\eqref{eq:lemB1_hanson} in Lemma~\ref{lem:prob-M0-geometry}, and using
\eqref{eq:prob-pooled-scalar-trajectory-norms}, gives, for every $t>0$,
\begin{equation}
\begin{aligned}
  \mathbb P\left(
    \left|
      \bm v^\top
      (\widehat{\bbm\Gamma}_w-\bbm\Gamma_w)\bm v
    \right|>t
  \right)
  \le
  2\exp\left[
    -C\min\left\{
      \frac{\Ttot t^2}{M\phi_{\max}},
      \frac{\Ttot t}{M}
    \right\}
  \right].
\end{aligned}
  \label{eq:prob-pooled-restricted-fixed-tail}
\end{equation}
For every such $J$, let $\mathcal N_J$ be a $1/4$-net of its supported unit sphere. The standard volumetric bound gives
\begin{equation}
  |\mathcal N_J|
  \le \left(1+\frac{2}{1/4}\right)^{s_0^\dagger\wedge p}
  =9^{s_0^\dagger\wedge p}.
  \label{eq:prob-supported-sphere-covering}
\end{equation}
Accounting for all possible supports, the logarithm of the total net
cardinality satisfies
\begin{equation}\label{eq:lemB2_sparseNet}
  \log\left\{
    \binom{p}{s_0^\dagger\wedge p}
    9^{s_0^\dagger\wedge p}
  \right\}
  \le
  (s_0^\dagger\wedge p)
  \log\left(
    \frac{9ep}{s_0^\dagger\wedge p}
  \right)
  \le
  C s_0^\dagger\log p.
\end{equation}
For a symmetric matrix $\bm M$, the net argument in
\eqref{eq:lemB1_union}, applied within each supported subspace, yields
\begin{equation}
  \sup_{\substack{J\subset[p],\ |J|\le s_0^\dagger}}
  \|\bm M_{J,J}\|_{\op}
  \le
  2\max_{\substack{J\subset[p],\ |J|=s_0^\dagger\wedge p\\
  \bm v\in\mathcal N_J}}
  |\bm v^\top\bm M\bm v|.
  \label{eq:prob-restricted-operator-net}
\end{equation}
Combining \eqref{eq:prob-pooled-restricted-fixed-tail}--
\eqref{eq:prob-restricted-operator-net} and taking a union bound give
\begin{equation}
  \mathbb P\left(
    \sup_{\substack{J\subset[p],\ |J|\le s_0^\dagger}}
    \left\|
      [\widehat{\bbm\Gamma}_w-\bbm\Gamma_w]_{J,J}
    \right\|_{\op}>2t
  \right)
  \le
  2\exp\left[
    C s_0^\dagger\log p
    -C\min\left\{
      \frac{\Ttot t^2}{M\phi_{\max}},
      \frac{\Ttot t}{M}
    \right\}
  \right].
  \label{eq:prob-pooled-restricted-net-tail}
\end{equation}
As in Lemma~\ref{lem:prob-M0-geometry}, set
\begin{equation*}
  t=C
  \sqrt{\frac{M\phi_{\max}s_0^\dagger\log p}{\Ttot}}.
\end{equation*}
Under
$\Ttot\gtrsim M\phi_{\max}^{-1}s_0^\dagger\log p$,
\begin{equation*}
  \frac{\Ttot t^2}{M\phi_{\max}}
  =C s_0^\dagger\log p,
  \qquad
  \frac{\Ttot t}{M}
  =C\sqrt{
    \frac{\Ttot\phi_{\max}s_0^\dagger\log p}{M}}
  \ge C s_0^\dagger\log p.
\end{equation*}
Consequently, choosing the absolute constant sufficiently large in
\eqref{eq:prob-pooled-restricted-net-tail} yields
\begin{equation}
  \sup_{\substack{J\subset[p],\ |J|\le s_0^\dagger}}
  \left\|
    [\widehat{\bbm\Gamma}_w-\bbm\Gamma_w]_{J,J}
  \right\|_{\op}
  \le
  C\sqrt{
    \frac{M\phi_{\max}s_0^\dagger\log p}{\Ttot}}
  \label{eq:prob-pooled-restricted-covariance-absorbed}
\end{equation}
with probability at least $1-C\exp\{-C(s_0^\dagger\vee1)\log p\}$. Combining
\eqref{eq:prob-pooled-restricted-covariance-to-product} with
\eqref{eq:prob-pooled-restricted-covariance-absorbed} proves
part~\textup{(i)}.

For part~\textup{(ii)}, for each fixed client $i$, repeat the preceding
fixed-direction and sparse-net arguments for
$\widehat{\bbm\Gamma}_i-\bbm\Gamma_i$ and set
\begin{equation*}
  t=C\sqrt{
    \frac{M\phi_{\max}(s_i^\dagger\log p+\log N)}{T_i}}.
\end{equation*}
Under
$T_i\gtrsim M\phi_{\max}^{-1}
(s_i^\dagger\log p+\log N)$, we obtain
\begin{equation}
  \sup_{\substack{J\subset[p],\ |J|\le s_i^\dagger}}
  \left\|
    [\widehat{\bbm\Gamma}_i-\bbm\Gamma_i]_{J,J}
  \right\|_{\op}
  \le
  C\sqrt{
    \frac{M\phi_{\max}(s_i^\dagger\log p+\log N)}{T_i}}.
  \label{eq:prob-local-restricted-covariance-absorbed}
\end{equation}
For each fixed $i$, this holds with failure probability at most
$C\exp\{-C(s_i^\dagger\log p+\log N)\}$. Hence, a union bound over all
clients with $s_i^\dagger>0$, for which
$s_i^\dagger\ge s_{\min}^\dagger\vee1$, gives the
simultaneous probability lower bound $1-C\exp\{-C((s_{\min}^\dagger\vee1)\log p+\log N)\}$. Finally,
\eqref{eq:prob-local-restricted-covariance-to-product} and
\eqref{eq:prob-local-restricted-covariance-absorbed} prove
part~\textup{(ii)}.
\end{proof}

\subsection{Score bounds}
\label{subsec:prob-scores}

\begin{lemma}[Pooled operator score]
\label{lem:prob-pooled-lr-score}
Under Assumption \ref{assump:stationarity}, suppose
$
  \Ttot
  \gtrsim
  M\phi_{\max}^{-1}p
$. 
Then, with probability at least $1-\exp(-Cp)$,
\begin{equation}
  \|\widehat{\bm W}_w\|_{\op}^2
  \le
  C\frac{\tau^2\phi_{\max}p}{\Ttot}.
  \label{eq:prob-pooled-lr-score}
\end{equation}
\end{lemma}

\begin{proof}
For any $\bm u, \bm v\in\mathbb S^{p-1}$, define
\begin{equation*}
  S:=\bm u^\top\widehat{\bm W}_w\bm v
  =
  \sum_{i=1}^N\sum_{t=1}^{T_i}
  \frac{w_i}{T_i}
  (\bm u^\top\bbm\epsilon_{it})
  (\bm v^\top\bm y_{i,t-1}).
\end{equation*}
For client $i$, let
$\mathcal F_{i,t-1}:=\sigma(\bm y_{i,0},\ldots,\bm y_{i,t-1})$.
Since $\bbm\epsilon_{it}$ is independent of $\mathcal F_{i,t-1}$, conditionally on the past, $\bm u^\top\bbm\epsilon_{it}$ remains centered Gaussian and
$\bm v^\top\bm y_{i,t-1}$ is fixed. Then, define the conditional accumulated variance as
\begin{equation*}
  V:=\sum_{i=1}^N\sum_{t=1}^{T_i}
  \operatorname{Var}\left(
    \frac{w_i}{T_i}
    (\bm u^\top\bbm\epsilon_{it})
    (\bm v^\top\bm y_{i,t-1})
    \mathrel{}\middle|\mathrel{}\mathcal F_{i,t-1}
  \right)
  =\sum_{i=1}^N\sum_{t=1}^{T_i}
  \frac{w_i^2}{T_i^2}
  (\bm u^\top\bbm\Sigma_{\epsilon,i}\bm u)
  (\bm v^\top\bm y_{i,t-1})^2.
\end{equation*}

For any $z>0$, $\overline V>0$, and $\lambda>0$, Markov's
inequality gives
\begin{equation}\label{eq:lemB2_Markov}
\begin{aligned}
  \mathbb P(S\ge z,\ V\le\overline V)
  &\le
  \mathbb P\left(
    \exp\left\{
      \lambda S-\frac{\lambda^2}{2}V
    \right\}
    \ge
    \exp\left\{
      \lambda z-\frac{\lambda^2}{2}\overline V
    \right\}
  \right)
  \\
  &\le
  \exp\left(
    -\lambda z+\frac{\lambda^2}{2}\overline V
  \right)
  \mathbb E\exp\left(
    \lambda S-\frac{\lambda^2}{2}V
  \right).
\end{aligned}
\end{equation}
For each term in the summation of $S$, we have
\begin{equation}
  \mathbb E\left[
    \exp\left\{
      \lambda\frac{w_i}{T_i}
      (\bm u^\top\bbm\epsilon_{it})
      (\bm v^\top\bm y_{i,t-1})
    \right\}
    \mathrel{}\middle|\mathrel{}\mathcal F_{i,t-1}
  \right]
  =
  \exp\left\{
    \frac{\lambda^2w_i^2}{2T_i^2}
    (\bm u^\top\bbm\Sigma_{\epsilon,i}\bm u)
    (\bm v^\top\bm y_{i,t-1})^2
  \right\}.
  \label{eq:prob-pooled-score-conditional-mgf}
\end{equation}
Hence, for each client $i$, the tower property yields
\begin{equation*}
\begin{aligned}
  &\mathbb E\exp\left[
    \sum_{t=1}^{T_i}
    \left\{
      \lambda\frac{w_i}{T_i}
      (\bm u^\top\bbm\epsilon_{it})
      (\bm v^\top\bm y_{i,t-1})
      -\frac{\lambda^2w_i^2}{2T_i^2}
      (\bm u^\top\bbm\Sigma_{\epsilon,i}\bm u)
      (\bm v^\top\bm y_{i,t-1})^2
    \right\}
  \right]
  \\
  =&\mathbb E\exp\left[
    \sum_{t=1}^{T_i-1}
    \left\{
      \lambda\frac{w_i}{T_i}
      (\bm u^\top\bbm\epsilon_{it})
      (\bm v^\top\bm y_{i,t-1})
      -\frac{\lambda^2w_i^2}{2T_i^2}
      (\bm u^\top\bbm\Sigma_{\epsilon,i}\bm u)
      (\bm v^\top\bm y_{i,t-1})^2
    \right\}
  \right]
  \\
  =&\cdots=1,
\end{aligned}
\end{equation*}
where the first equality conditions on $\mathcal F_{i,T_i-1}$ and the
remaining equalities repeat the same argument backward in time. Since the client trajectories are independent, we have for every $\lambda\in\mathbb R$,
\begin{equation*}
\begin{aligned}
  &\mathbb E\exp\left[
    \sum_{i=1}^N\sum_{t=1}^{T_i}
    \left\{
      \lambda\frac{w_i}{T_i}
      (\bm u^\top\bbm\epsilon_{it})
      (\bm v^\top\bm y_{i,t-1})
      -\frac{\lambda^2w_i^2}{2T_i^2}
      (\bm u^\top\bbm\Sigma_{\epsilon,i}\bm u)
      (\bm v^\top\bm y_{i,t-1})^2
    \right\}
  \right]
  \\
  =&\prod_{i=1}^N
  \mathbb E\exp\left[
    \sum_{t=1}^{T_i}
    \left\{
      \lambda\frac{w_i}{T_i}
      (\bm u^\top\bbm\epsilon_{it})
      (\bm v^\top\bm y_{i,t-1})
      -\frac{\lambda^2w_i^2}{2T_i^2}
      (\bm u^\top\bbm\Sigma_{\epsilon,i}\bm u)
      (\bm v^\top\bm y_{i,t-1})^2
    \right\}
  \right]
  =1.
\end{aligned}
\end{equation*}
Equivalently,
\begin{equation*}
  \mathbb E\exp\left(
    \lambda S-\frac{\lambda^2}{2}V
  \right)
  =1.
\end{equation*}
Consequently, substituting the preceding identity into the \eqref{eq:lemB2_Markov} gives
\begin{equation*}
  \mathbb P(S\ge z,\ V\le\overline V)
  \le
  \exp\left(
    -\lambda z+\frac{\lambda^2}{2}\overline V
  \right)
  \le
  \exp\left(-\frac{z^2}{2\overline V}\right),
\end{equation*}
where the last line takes $\lambda=z/\overline V$.  Applying the same
argument to $-S$ gives the two-sided bound:
\begin{equation*}
  \mathbb P(|S|\ge z,\ V\le\overline V)
  \le
  2\exp\left(-\frac{z^2}{2\overline V}\right).
\end{equation*}
By the same Hanson--Wright and net argument as in
Lemma~\ref{lem:prob-M0-geometry}, now retaining the sharper bound
$\|\bbm\Psi_{\bm v}\|_{\mathrm F}^2
\le\Ttot M\phi_{\max}$, the condition
$\Ttot\gtrsim M\phi_{\max}^{-1}p$ gives
\begin{equation*}
  \|\widehat{\bbm\Gamma}_w-\bbm\Gamma_w\|_{\op}
  \le C\phi_{\max}
\end{equation*}
with probability at least $1-\exp(-Cp)$.  Since
$\|\bbm\Gamma_w\|_{\op}\le\phi_{\max}$, the same event satisfies
$\|\widehat{\bbm\Gamma}_w\|_{\op}\le C\phi_{\max}$.  On this event,
$V\le C\tau^2\phi_{\max}/\Ttot$, and hence
\begin{equation*}
  \{|S|>z\}
  \cap
  \{\|\widehat{\bbm\Gamma}_w\|_{\op}\le C\phi_{\max}\}
  \subseteq
  \left\{
    |S|>z,\ 
    V\le C\frac{\tau^2\phi_{\max}}{\Ttot}
  \right\}.
\end{equation*}
Taking
$$
  z=C\sqrt{\tau^2\phi_{\max}x/\Ttot}
$$
yields, for every $x>0$,
\begin{equation}
  \mathbb P\left(
    \left\{
      |\bm u^\top\widehat{\bm W}_w\bm v|
      >C\sqrt{
        \frac{\tau^2\phi_{\max}x}{\Ttot}}
    \right\}
    \cap
    \left\{
      \|\widehat{\bbm\Gamma}_w\|_{\op}
      \le C\phi_{\max}
    \right\}
  \right)
  \le 2\exp(-x).
  \label{eq:prob-pooled-score-fixed-tail}
\end{equation}

It remains to pass from fixed directions to the operator norm.  Similarly
to the net argument in Lemma~\ref{lem:prob-M0-geometry}, let
$\mathcal N_u$ and $\mathcal N_v$ be $1/4$-nets of
$\mathbb S^{p-1}$, each with cardinality at most $9^p$. The bilinear version of \eqref{eq:lemB1_discrete} gives, for every $p\times p$ matrix $\bm M$,
\begin{equation*}
  \|\bm M\|_{\op}
  \le
  2\max_{\bm u\in\mathcal N_u,\,\bm v\in\mathcal N_v}
  |\bm u^\top\bm M\bm v|.
\end{equation*}
Set $x=Cp$ in \eqref{eq:prob-pooled-score-fixed-tail}, where $C$ is
sufficiently large.  A union bound over the at most $9^{2p}$ pairs of net
points and the preceding bilinear inequality give
\begin{equation*}
\begin{aligned}
  \mathbb P\left(
    \left\{
      \|\widehat{\bm W}_w\|_{\op}^2
      >C\frac{\tau^2\phi_{\max}p}{\Ttot}
    \right\}
    \cap
    \left\{
      \|\widehat{\bbm\Gamma}_w\|_{\op}\le C\phi_{\max}
    \right\}
  \right)
  &\le
  2\cdot9^{2p}\exp(-Cp)
  \\
  &\le
  \exp(-Cp).
\end{aligned}
\end{equation*}
Finally, the preceding covariance bound and an event decomposition give
\begin{equation*}
\begin{aligned}
  &\mathbb P\left(
    \|\widehat{\bm W}_w\|_{\op}^2
    >C\frac{\tau^2\phi_{\max}p}{\Ttot}
  \right)
  \\
  &\le
  \mathbb P\left(
    \left\{
      \|\widehat{\bm W}_w\|_{\op}^2
      >C\frac{\tau^2\phi_{\max}p}{\Ttot}
    \right\}
    \cap
    \left\{
      \|\widehat{\bbm\Gamma}_w\|_{\op}\le C\phi_{\max}
    \right\}
  \right)
  +
  \mathbb P\left(
    \|\widehat{\bbm\Gamma}_w\|_{\op}>C\phi_{\max}
  \right)
  \\
  &\le
  \exp(-C_1p)+\exp(-C_2p)
  \le
  \exp(-Cp).
\end{aligned}
\end{equation*}
\end{proof}

\begin{lemma}[Clientwise score bounds]
\label{lem:prob-clientwise-scores}
Under Assumption~\ref{assump:stationarity}, the following conclusions
hold.
\begin{enumerate}[label=\textup{(\roman*)}]
\item Suppose that
\begin{equation*}
  T_i
  \gtrsim
  M\phi_{\max}^{-1}(p+\log N),
  \qquad i\in[N].
\end{equation*}
Then, with probability at least
$1-C\exp\{-c(p+\log N)\}$, simultaneously for all $i\in[N]$,
\begin{equation}
  \|\widehat{\bm W}_i\|_{\op}^2
  \le
  C\frac{\tau^2\phi_{\max}(p+\log N)}{T_i}.
  \label{eq:prob-clientwise-operator-score}
\end{equation}

\item Suppose that
\begin{equation*}
  T_i
  \gtrsim
  M\phi_{\max}^{-1}\log(pN),
  \qquad i\in[N].
\end{equation*}
Then, with probability at least $1-C(pN)^{-c}$, simultaneously for all
$i\in[N]$,
\begin{equation}
  \|\widehat{\bm W}_i\|_{\max}^2
  \le
  C\frac{\tau^2\phi_{\max}\log(pN)}{T_i}.
  \label{eq:prob-clientwise-max-score}
\end{equation}
\end{enumerate}
\end{lemma}

\begin{proof}
Fix $i\in[N]$ and $\bm u,\bm v\in\mathbb S^{p-1}$. Similar to the proof of Lemma \ref{lem:prob-pooled-lr-score}, applying the Markov bound \eqref{eq:lemB2_Markov} and the
conditional MGF identity \eqref{eq:prob-pooled-score-conditional-mgf} clientwise, with only client $i$
retained and $w_i/T_i$ replaced by $1/T_i$, the conditional accumulated
variance of $\bm u^\top\widehat{\bm W}_i\bm v$ is
\begin{equation*}
  \frac{\bm u^\top\bbm\Sigma_{\epsilon,i}\bm u}{T_i}
    \bm v^\top\widehat{\bbm\Gamma}_i\bm v.
\end{equation*}
Hence, optimizing the Markov bound as in the derivation of
\eqref{eq:prob-pooled-score-fixed-tail} gives, for every $z>0$,
\begin{equation}
  \mathbb P\left(
    |\bm u^\top\widehat{\bm W}_i\bm v|>z,\
    \bm v^\top\widehat{\bbm\Gamma}_i\bm v\le C\phi_{\max}
  \right)
  \le
  2\exp\left(
    -c\frac{T_i z^2}
    {\tau^2\phi_{\max}}
  \right),
  \label{eq:prob-clientwise-score-fixed-tail}
\end{equation}

For part~\textup{(i)}, Lemma~\ref{lem:prob-local-geometry} gives
$\max_i\|\widehat{\bbm\Gamma}_i\|_{\op}\le C\phi_{\max}$ with probability
at least $1-C\exp\{-c(p+\log N)\}$.  On this event, apply
\eqref{eq:prob-clientwise-score-fixed-tail} to two $1/4$-nets, each of
cardinality at most $9^p$, and use the bilinear net inequality from the
proof of Lemma~\ref{lem:prob-pooled-lr-score}.  Then, setting
$z_i=C\tau
\sqrt{\phi_{\max}(p+\log N)/T_i}$ and taking a union bound over at most
$N\times 9^{2p}$ triples, each consisting of one client and two net points, proves
\eqref{eq:prob-clientwise-operator-score}.

For part~\textup{(ii)}, let $\bm e_j\in\mathbb R^p$ have one in its $j$-th
coordinate and zeros elsewhere.  The fixed-direction covariance tail in
\eqref{eq:lemB1_hanson}, applied clientwise with $\bm v=\bm e_b$, followed
by a union bound over the $Np$ pairs $(i,b)$, gives
\begin{equation*}
  \max_{i\in[N],\,b\in[p]}(\widehat{\bbm\Gamma}_i)_{bb}
  \le2\phi_{\max}
\end{equation*}
with probability at least $1-C(pN)^{-c}$.  On this event, apply
\eqref{eq:prob-clientwise-score-fixed-tail} with $\bm u=\bm e_a$,
$\bm v=\bm e_b$, and
$z_i=C\tau
\sqrt{\phi_{\max}\log(pN)/T_i}$.  A union bound over the $Np^2$ triples
$(i,a,b)$ proves \eqref{eq:prob-clientwise-max-score}.
\end{proof}

\begin{lemma}[Restricted score bounds]
\label{lem:prob-sparse-scores}
Under Assumption~\ref{assump:stationarity}, the following conclusions
hold.
\begin{enumerate}[label=\textup{(\roman*)}]
\item Suppose
$\Ttot\gtrsim M\phi_{\max}^{-1}s_0^\dagger\log p$.
Then, with probability at least
$1-\exp\{-C(s_0^\dagger\vee1)\log p\}$,
\begin{equation}
  \sup_{\substack{\Omega_0\subset[p]\times[p]\\
  |\Omega_0|\le s_0^\dagger}}
  \|\mathcal P_{\Omega_0}(\widehat{\bm W}_w)\|_{\mathrm F}^2
  \lesssim
  \frac{\tau^2\phi_{\max}
  s_0^\dagger\log p}{\Ttot}.
  \label{eq:prob-pooled-sparse-score}
\end{equation}

\item Suppose, for every $i\in[N]$,
$
  T_i
  \gtrsim
  M\phi_{\max}^{-1}
  (s_i^\dagger\log p+\log N).
$
Then, with probability at least
$1-C\exp\{-c((s_{\min}^\dagger\vee1)\log p+\log N)\}$,
uniformly over all support collections $(\Omega_1,\ldots,\Omega_N)$
satisfying $\Omega_i\subset[p]\times[p]$ and
$|\Omega_i|\le s_i^\dagger$ for every $i\in[N]$,
\begin{equation}
  \sum_{i=1}^Nw_i
  \|\mathcal P_{\Omega_i}(\widehat{\bm W}_i)\|_{\mathrm F}^2
  \lesssim
  \frac{\tau^2\phi_{\max}}{\Ttot}
  \left\{
    \log p\sum_{i=1}^Ns_i^\dagger+N\log N
  \right\}.
\end{equation}
\end{enumerate}
\end{lemma}

\begin{proof}
We reuse the conditional MGF, tower property, and Markov argument in the proof of 
Lemma~\ref{lem:prob-pooled-lr-score}.  For any $\bm H$ with $\|\bm H\|_{\mathrm F}=1$, replace
$(\bm u^\top\bbm\epsilon_{it})(\bm v^\top\bm y_{i,t-1})$ there by
$\bbm\epsilon_{it}^\top\bm H\bm y_{i,t-1}$.  The accumulated conditional
variance of $\inner{\widehat{\bm W}_w}{\bm H}$ is
\begin{equation*}
\begin{aligned}
  &\sum_{i=1}^N\sum_{t=1}^{T_i}
  \operatorname{Var}\left(
    \left.
    \frac{w_i}{T_i}
    \bbm\epsilon_{it}^\top\bm H\bm y_{i,t-1}
    \right|\mathcal F_{i,t-1}
  \right)
  \\
  =&\sum_{i=1}^N\sum_{t=1}^{T_i}
  \frac{w_i^2}{T_i^2}
  \bm y_{i,t-1}^\top
  \bm H^\top\bbm\Sigma_{\epsilon,i}\bm H
  \bm y_{i,t-1}
  \\
  \le&
  \frac{\tau^2}{\Ttot}
  \operatorname{tr}
  \left(\bm H\widehat{\bbm\Gamma}_w\bm H^\top\right).
\end{aligned}
\end{equation*}
Hence the argument leading to
\eqref{eq:prob-pooled-score-fixed-tail} gives, for every $z>0$,
\begin{equation}
  \mathbb P\left(
    \left\{
      |\inner{\widehat{\bm W}_w}{\bm H}|>z
    \right\}
    \cap
    \left\{
      \operatorname{tr}
      (\bm H\widehat{\bbm\Gamma}_w\bm H^\top)
      \le C\phi_{\max}
    \right\}
  \right)\le
  2\exp\left(
    -C\frac{\Ttot z^2}
    {\tau^2\phi_{\max}}
  \right).
  \label{eq:prob-pooled-sparse-score-fixed-test}
\end{equation}

For part~\textup{(i)}, if $|\Omega_0|\le s_0^\dagger$ and
$\supp(\bm H)\subseteq\Omega_0$, then
\begin{equation}\label{eq:lemB4_trace}
\begin{aligned}
  \operatorname{tr}
  (\bm H\widehat{\bbm\Gamma}_w\bm H^\top)
  &=
  \operatorname{tr}
  (\bm H\bbm\Gamma_w\bm H^\top)
  +\inner{
    \bm H(\widehat{\bbm\Gamma}_w-\bbm\Gamma_w)
  }{\bm H}
  \\
  &=
  \operatorname{tr}
  (\bm H\bbm\Gamma_w\bm H^\top)
  +\inner{
    \mathcal P_{\Omega_0}
    \left\{\bm H
    (\widehat{\bbm\Gamma}_w-\bbm\Gamma_w)\right\}
  }{\bm H}
  \\
  &\le
  \|\bbm\Gamma_w\|_{\op}\|\bm H\|_{\mathrm F}^2
  +\left\|
    \mathcal P_{\Omega_0}
    \left\{\bm H
    (\widehat{\bbm\Gamma}_w-\bbm\Gamma_w)\right\}
  \right\|_{\mathrm F}\|\bm H\|_{\mathrm F}
  \\
  &\le
  \phi_{\max}
  +\left\|
    \mathcal P_{\Omega_0}
    \left\{\bm H
    (\widehat{\bbm\Gamma}_w-\bbm\Gamma_w)\right\}
  \right\|_{\mathrm F}.
\end{aligned}
\end{equation}
Under
$\Ttot\gtrsim M\phi_{\max}^{-1}s_0^\dagger\log p$,
Lemma~\ref{lem:prob-sparse-geometry}\textup{(i)} gives, uniformly over all
such $\Omega_0$ and supported unit-Frobenius $\bm H$,
\begin{equation*}
  \left\|
    \mathcal P_{\Omega_0}
    \left\{\bm H
    (\widehat{\bbm\Gamma}_w-\bbm\Gamma_w)\right\}
  \right\|_{\mathrm F}
  \le
  C\sqrt{\frac{M\phi_{\max}s_0^\dagger\log p}{\Ttot}}
  \le C\phi_{\max}.
\end{equation*}
Together with \eqref{eq:lemB4_trace}, this shows that
\begin{equation*}
  \operatorname{tr}
  (\bm H\widehat{\bbm\Gamma}_w\bm H^\top)
  \lesssim \phi_{\max}
\end{equation*}
with probability at least $1-\exp\{-C(s_0^\dagger\vee1)\log p\}$.
To pass from the fixed-test bound
\eqref{eq:prob-pooled-sparse-score-fixed-test} to the restricted
Frobenius norm in part~\textup{(i)}, we use Frobenius duality. Since
\begin{equation*}
  \sup_{\substack{\Omega_0\subset[p]\times[p]\\
  |\Omega_0|\le s_0^\dagger}}
  \|\mathcal P_{\Omega_0}(\widehat{\bm W}_w)\|_{\mathrm F}
  =
  \sup_{\substack{\Omega_0\subset[p]\times[p],\;
  |\Omega_0|\le s_0^\dagger\\
  \|\bm H\|_{\mathrm F}=1,\;
  \supp(\bm H)\subseteq\Omega_0}}
  \left|\inner{\widehat{\bm W}_w}{\bm H}\right|,
\end{equation*}
it remains to control these scalar inner products uniformly over
all sparse unit-Frobenius test matrices $\bm H$. Similar to \eqref{eq:lemB2_sparseNet}, the union of $1/4$-nets over these supported unit spheres has
log-cardinality at most $Cs_0^\dagger\log p$.  Applying
\eqref{eq:prob-pooled-sparse-score-fixed-test} on the preceding variance
event and taking
\begin{equation*}
  z=C\sqrt{
    \frac{\tau^2\phi_{\max}
    s_0^\dagger\log p}{\Ttot}}
\end{equation*}
and taking a union bound over the sparse nets, together with the
preceding covariance event, yields
\begin{equation*}
\begin{aligned}
  \mathbb P\left\{
    \sup_{\substack{\Omega_0\subset[p]\times[p]\\
    |\Omega_0|\le s_0^\dagger}}
    \|\mathcal P_{\Omega_0}(\widehat{\bm W}_w)\|_{\mathrm F}^2
    >
    C\frac{\tau^2\phi_{\max}
    s_0^\dagger\log p}{\Ttot}
  \right\}
  &\le
  \exp\{-C(s_0^\dagger\vee1)\log p\}.
\end{aligned}
\end{equation*}
This proves part~\textup{(i)}.

For part~\textup{(ii)}, apply the same argument clientwise.  For each
$i\in[N]$, the accumulated conditional variance of
$\inner{\widehat{\bm W}_i}{\bm H}$ is bounded by
\begin{equation*}
  \frac{\tau^2}{T_i}
  \operatorname{tr}
  \left(\bm H\widehat{\bbm\Gamma}_i\bm H^\top\right).
\end{equation*}
Under the sample-size condition in part~\textup{(ii)},
Lemma~\ref{lem:prob-sparse-geometry}\textup{(ii)} and the clientwise
analogue of \eqref{eq:lemB4_trace} give
\begin{equation*}
  \operatorname{tr}
  \left(\bm H\widehat{\bbm\Gamma}_i\bm H^\top\right)
  \le C\phi_{\max}
\end{equation*}
uniformly over all clients and supported unit-Frobenius matrices $\bm H$.
Repeating the sparse-net argument from part~\textup{(i)} clientwise, with
$s_i^\dagger\log p+\log N$ in place of $s_0^\dagger\log p$, and then
taking a union bound over clients with $s_i^\dagger>0$ gives, with probability at least
$1-C\exp\{-C((s_{\min}^\dagger\vee1)\log p+\log N)\}$, simultaneously for all
$i\in[N]$,
\begin{equation}
  \sup_{\substack{\Omega_i\subset[p]\times[p]\\
  |\Omega_i|\le s_i^\dagger}}
  \|\mathcal P_{\Omega_i}(\widehat{\bm W}_i)\|_{\mathrm F}^2
  \le
  C\frac{\tau^2\phi_{\max}
    (s_i^\dagger\log p+\log N)}{T_i},
  \label{eq:prob-local-score-clientwise-raw}
\end{equation}
Since $w_i=T_i/\Ttot$, summing these bounds over $i$ proves
part~\textup{(ii)}.
\end{proof}

\subsection{Empirical covariance bounds for global--local interactions}
\label{subsec:prob-mixed-geometry}

We use the accessible direction classes
$\mathcal X_{\mathrm{LR}}$ and $\mathcal X_{\mathrm S}$ defined in
Section~\ref{subsec:covariance_compatibility}, with local radius
$\varepsilon_\nu$.

\begin{lemma}[Shared--personalized interactions]
\label{lem:prob-mixed-geometry}
Under Assumption~\ref{assump:stationarity}, suppose that
$T_i\gtrsim s_i^\dagger\log p+\log N$ for every $i\in[N]$.
Then the following two conclusions hold.
\begin{enumerate}[label=\textup{(\roman*)}]
\item With probability at least
$1-C\exp(-Cpr)-C\exp\{-C(s_{\min}^\dagger\vee1)\log p\}$,
uniformly over all $\bm X_{\mathrm{LR}}\in\mathcal X_{\mathrm{LR}}$
and all supports $\Omega_i\subset[p]\times[p]$ satisfying
$|\Omega_i|\le s_i^\dagger$ for every $i\in[N]$,
\begin{equation}
  \sum_{i=1}^Nw_i
  \left\|
    \mathcal P_{\Omega_i}
    \left\{
      \bm X_{\mathrm{LR}}
      (\widehat{\bbm\Gamma}_i-\bbm\Gamma_i)
    \right\}
  \right\|_{\mathrm F}^2
  \lesssim M^2
  \left\{
    \frac{pr}{\Ttot}
    +\frac{\log p}{\Ttot}\sum_{i=1}^Ns_i^\dagger
  \right\}.
  \label{eq:prob-ML-centered-geometry}
\end{equation}

\item With probability at least
$1-C\exp\{-C(s_0^\dagger\vee1)\log p\}
  -C\exp\{-C(s_{\min}^\dagger\vee1)\log p\}$,
uniformly over all $\bm X_{\mathrm S}\in\mathcal X_{\mathrm S}$ and all supports
$\Omega_i\subset[p]\times[p]$ satisfying
$|\Omega_i|\le s_i^\dagger$ for every $i\in[N]$,
\begin{equation}
  \sum_{i=1}^Nw_i
  \left\|
    \mathcal P_{\Omega_i}
    \left\{
      \bm X_{\mathrm S}(\widehat{\bbm\Gamma}_i-\bbm\Gamma_i)
    \right\}
  \right\|_{\mathrm F}^2
  \lesssim
  M^2
  \left\{
    \frac{s_0^\dagger\log p}{\Ttot}
    +\frac{\log p}{\Ttot}\sum_{i=1}^Ns_i^\dagger
  \right\}.
  \label{eq:prob-0L-centered-geometry}
\end{equation}
\end{enumerate}
\end{lemma}

\begin{proof}
If all $s_i^\dagger=0$, both claims are immediate. Otherwise, we first
prove part~\textup{(i)}.

\paragraph*{Step 1: Target supremum and proof outline}
Define the supremum of the quantity in part~\textup{(i)} by
\begin{equation}
\begin{aligned}
  \zeta_{ML}^2
  :={}&
  \sup_{\substack{
    \bm X_{\mathrm{LR}}\in\mathcal X_{\mathrm{LR}}\\
    \Omega_i\subset[p]\times[p],\ |\Omega_i|\le s_i^\dagger}}
  \sum_{i=1}^Nw_i
  \left\|
    \mathcal P_{\Omega_i}
    \left\{
      \bm X_{\mathrm{LR}}
      (\widehat{\bbm\Gamma}_i-\bbm\Gamma_i)
    \right\}
  \right\|_{\mathrm F}^2.
  \label{eq:prob-ML-zeta-target-identity}
\end{aligned}
\end{equation}
Weighted Frobenius duality gives the equivalent representation
\begin{equation}
\begin{aligned}
  \zeta_{ML}
  =\sup_{\substack{
    \bm X_{\mathrm{LR}}\in\mathcal X_{\mathrm{LR}}\\
    \bm D_i\in\mathbb R^{p\times p},\\
    |\supp(\bm D_i)|\le s_i^\dagger,\ i\in[N]\\
    \sum_{i=1}^Nw_i\|\bm D_i\|_{\mathrm F}^2\le1}}
  \left|
    \sum_{i=1}^Nw_i
    \inner{\bm D_i}
    {\bm X_{\mathrm{LR}}
    (\widehat{\bbm\Gamma}_i-\bbm\Gamma_i)}
  \right|.
  \label{eq:prob-ML-centered-block}
\end{aligned}
\end{equation}
Below, $\theta$ denotes the full index
$(\bm X_{\mathrm{LR}},\{\bm D_i\}_{i=1}^N)$ in
\eqref{eq:prob-ML-centered-block}, and every supremum over $\theta$ uses
this same index set.

For a random variable $Z$ and $k\ge1$, write
$\|Z\|_{L_k}:=(\mathbb E|Z|^k)^{1/k}$.
We will prove that, whenever
\begin{equation}
  T_i\ge C\{s_i^\dagger\log p+\log N+k\},
  \qquad i\in[N]\text{ with }s_i^\dagger>0,
  \label{eq:prob-ML-iid-moment-sample}
\end{equation}
the target supremum satisfies
\begin{equation}
  \|\zeta_{ML}\|_{L_k}
  \le CM\left[
    \sqrt{\frac{pr+k}{\Ttot}}
    +\left\{\frac{\log p}{\Ttot}
      \sum_{i=1}^N s_i^\dagger\right\}^{1/2}
  \right].
  \label{eq:prob-ML-target-moment}
\end{equation}
We first show that it suffices to consider iid Gaussian observations.
For this case, we construct a conditional Gaussian process,
bound the conditional mean of its supremum, and use Gaussian
concentration to obtain a conditional $L_k$ bound. We then remove the conditioning to establish
\eqref{eq:prob-ML-target-moment} and apply Markov's inequality
to obtain the claimed probability bound.

\paragraph*{Step 2: Reduction to independent Gaussian observations}
Let $\{\bbm\xi_{it}\}$ be independent $N(\bm0,\bm I_p)$ vectors and define
$\zeta_{ML}^{\mathrm{iid}}$ by replacing
$\widehat{\bbm\Gamma}_i-\bbm\Gamma_i$ in
\eqref{eq:prob-ML-centered-block} with
$T_i^{-1}\sum_t\bbm\xi_{it}\bbm\xi_{it}^{\top}-\bm I_p$.
We show that this replacement costs only a factor $CM$ in the
$L_k$ bound.

By independence across clients and
\eqref{eq:app-trajectory-spectral-bound}, the stacked observations have
covariance
\begin{equation}
  \bbm\Upsilon
  =\operatorname{diag}(\bbm\Upsilon_1,\ldots,\bbm\Upsilon_N),
  \qquad \|\bbm\Upsilon\|_{\mathrm{op}}\le M.
  \label{eq:prob-ML-stacked-trajectory-covariance}
\end{equation}
Their joint Gaussian distribution therefore allows us to represent them
as $\bbm\Upsilon^{1/2}\bbm\xi$, where
$\bbm\xi\sim N(\bm0,\bm I_{p\Ttot})$.
Write $\operatorname{sym}(\bm M)=(\bm M+\bm M^\top)/2$ and set
\[
  \bm Q_\theta:=\frac1{\Ttot}
  \operatorname{diag}\!\left(
    \left\{\bm I_{T_i}\otimes
    \operatorname{sym}(\bm D_i^\top\bm X_{\mathrm{LR}})
    \right\}_{i=1}^N\right).
\]
Since $w_i/T_i=1/\Ttot$, the dependent and iid suprema have the
quadratic-form representations
\[
\begin{aligned}
  \zeta_{ML}
  &=\sup_\theta\left|
    \bbm\xi^\top\bbm\Upsilon^{1/2}\bm Q_\theta
      \bbm\Upsilon^{1/2}\bbm\xi
    -\operatorname{tr}(\bbm\Upsilon\bm Q_\theta)
  \right|,\\
  \zeta_{ML}^{\mathrm{iid}}
  &=\sup_\theta\left|
    \bbm\xi^\top\bm Q_\theta\bbm\xi
    -\operatorname{tr}(\bm Q_\theta)
  \right|.
\end{aligned}
\]

We first replace each centered quadratic form by a bilinear form in
independent Gaussian vectors. For any indexed family $\{Y_\theta\}$ and
any sigma-field $\mathcal F$, conditional Jensen's inequality gives
\begin{equation}
  \left\|\sup_\theta
    \left|\mathbb E(Y_\theta\mid\mathcal F)\right|
  \right\|_{L_k}
  \le\left\|\sup_\theta|Y_\theta|\right\|_{L_k}.
  \label{eq:prob-ML-conditional-Jensen}
\end{equation}
Let $\bbm\xi'$ be an independent copy of $\bbm\xi$, and set
$\bm a=(\bbm\xi+\bbm\xi')/\sqrt2$ and
$\bm b=(\bbm\xi-\bbm\xi')/\sqrt2$.
Then $\bm a,\bm b$ are independent standard Gaussian vectors.
For every symmetric matrix $\bm M$,
\[
\begin{aligned}
  \bbm\xi^\top\bm M\bbm\xi-\operatorname{tr}(\bm M)
  &=\mathbb E_{\bbm\xi'}\!\left[
    \bbm\xi^\top\bm M\bbm\xi
    -(\bbm\xi')^\top\bm M\bbm\xi'
    \,\middle|\,\bbm\xi\right],\\
  \bbm\xi^\top\bm M\bbm\xi
    -(\bbm\xi')^\top\bm M\bbm\xi'
  &=2\bm a^\top\bm M\bm b.
\end{aligned}
\]
Applying these identities and \eqref{eq:prob-ML-conditional-Jensen}
with $\bm M=\bbm\Upsilon^{1/2}\bm Q_\theta\bbm\Upsilon^{1/2}$ yields
\[
  \|\zeta_{ML}\|_{L_k}
  \le2\left\|\sup_\theta\left|
    \bm a^\top\bbm\Upsilon^{1/2}\bm Q_\theta
      \bbm\Upsilon^{1/2}\bm b
  \right|\right\|_{L_k}.
\]

To remove the covariance factors, write
$\bbm\Upsilon^{1/2}=M^{1/2}\bm C_0$, so that
$\bm C_0\bm C_0^\top\preceq\bm I$.
Let $\bbm\eta,\bbm\eta'$ be independent
$N(\bm0,\bm I-\bm C_0\bm C_0^\top)$ vectors, independent of
$\bm a,\bm b$, and define
\[
  \bm g=\bm C_0\bm a+\bbm\eta,
  \qquad
  \bm h=\bm C_0\bm b+\bbm\eta'.
\]
The vectors $\bm g,\bm h$ are independent standard Gaussians, and
\[
  (\bm C_0\bm a)^\top\bm Q_\theta(\bm C_0\bm b)
  =\mathbb E_{\bbm\eta,\bbm\eta'}\!\left[
    \bm g^\top\bm Q_\theta\bm h\,\middle|\,\bm a,\bm b
  \right].
\]
Another application of \eqref{eq:prob-ML-conditional-Jensen} gives
\[
  \|\zeta_{ML}\|_{L_k}
  \le2M\left\|\sup_\theta
    |\bm g^\top\bm Q_\theta\bm h|
  \right\|_{L_k}.
\]
Finally, set $\bm g_+=(\bm g+\bm h)/\sqrt2$ and
$\bm g_-=(\bm g-\bm h)/\sqrt2$. They are independent standard
Gaussians, and symmetry of $\bm Q_\theta$ gives
\[
\begin{aligned}
  \bm g^\top\bm Q_\theta\bm h
  ={}&\frac12\left\{
    \bm g_+^\top\bm Q_\theta\bm g_+
    -\operatorname{tr}(\bm Q_\theta)\right\}\\
  &-\frac12\left\{
    \bm g_-^\top\bm Q_\theta\bm g_-
    -\operatorname{tr}(\bm Q_\theta)\right\}.
\end{aligned}
\]
Taking the supremum and using the triangle inequality in $L_k$ gives
\[
\begin{aligned}
  \left\|\sup_\theta
    |\bm g^\top\bm Q_\theta\bm h|\right\|_{L_k}
  &\le\frac12\left\|\sup_\theta
    \left|\bm g_+^\top\bm Q_\theta\bm g_+
    -\operatorname{tr}(\bm Q_\theta)\right|\right\|_{L_k}\\
  &\quad+\frac12\left\|\sup_\theta
    \left|\bm g_-^\top\bm Q_\theta\bm g_-
    -\operatorname{tr}(\bm Q_\theta)\right|\right\|_{L_k}\\
  &=\|\zeta_{ML}^{\mathrm{iid}}\|_{L_k}.
\end{aligned}
\]
The equality holds because each centered supremum has the same
distribution as $\zeta_{ML}^{\mathrm{iid}}$. Combining this with the
preceding bound yields
\begin{equation}
  \|\zeta_{ML}\|_{L_k}
  \le CM\|\zeta_{ML}^{\mathrm{iid}}\|_{L_k}.
  \label{eq:prob-ML-trajectory-transfer}
\end{equation}
It remains to establish the desired moment bound for
$\zeta_{ML}^{\mathrm{iid}}$.

\paragraph*{Step 3: A conditional Gaussian representation of the iid process}
Let $\{\bbm\xi'_{it}\}$ be an independent copy of the iid observations,
and put
\[
  \bm u_{it}=\frac{\bbm\xi_{it}+\bbm\xi'_{it}}{\sqrt2},
  \qquad
  \bm v_{it}=\frac{\bbm\xi_{it}-\bbm\xi'_{it}}{\sqrt2}.
\]
The arrays $\bm u$ and $\bm v$ are independent standard Gaussian arrays.
Define
\begin{equation}
  \mathbb X_\theta(\bm u,\bm v)
  :=\frac1{\Ttot}
  \sum_{i=1}^N\sum_{t=1}^{T_i}
  \inner{\bm D_i\bm u_{it}}
  {\bm X_{\mathrm{LR}}\bm v_{it}}.
  \label{eq:prob-ML-conditional-process}
\end{equation}
The independent-copy argument in
\eqref{eq:prob-ML-conditional-Jensen}, followed by the rotation above,
gives
\begin{equation}
\begin{aligned}
  \|\zeta_{ML}^{\mathrm{iid}}\|_{L_k}
  &\le\left\|\sup_\theta\left|
    \frac1{\Ttot}\sum_{i,t}\left\{
    \inner{\bm D_i\bbm\xi_{it}}{\bm X_{\mathrm{LR}}\bbm\xi_{it}}
    -\inner{\bm D_i\bbm\xi'_{it}}{\bm X_{\mathrm{LR}}\bbm\xi'_{it}}
    \right\}\right|\right\|_{L_k}\\
  &=\left\|\sup_\theta
    |\mathbb X_\theta(\bm u,\bm v)
      +\mathbb X_\theta(\bm v,\bm u)|\right\|_{L_k}\\
  &\le2\left\|\sup_\theta
    |\mathbb X_\theta(\bm u,\bm v)|\right\|_{L_k}.
  \label{eq:prob-ML-Gaussian-rotation}
\end{aligned}
\end{equation}
The last inequality uses exchangeability of $\bm u$ and $\bm v$.
Conditional on $\bm u$, $\mathbb X_\theta$ is a centered Gaussian process
in $\bm v$. We next bound its conditional mean and fluctuations.

\paragraph*{Step 4: Conditional mean and moment bounds}
The dependence on $\bm u$ is controlled by its sparse empirical energies.
Define
\begin{equation}
  R(\bm u)
  :=\max_{i:\,s_i^\dagger>0}
  \sup_{\substack{J\subset[p]\\
  1\le|J|\le s_i^\dagger\wedge p}}
  \lambda_{\max}\left(
    \frac1{T_i}\sum_{t=1}^{T_i}
    \bm u_{it,J}\bm u_{it,J}^\top
  \right).
  \label{eq:prob-ML-restricted-energy}
\end{equation}
Using $\sum_iw_i\|\bm D_i\|_{\mathrm F}^2\le1$ therefore gives
\begin{equation}
\begin{aligned}
  \frac1{\Ttot}\sum_{i=1}^N\sum_{t=1}^{T_i}
  \|\bm D_i\bm u_{it}\|_2^2&=\sum_{i=1}^Nw_i\sum_{a=1}^p
  (\bm D_i)_{a,\cdot}
  \left(\frac1{T_i}\sum_{t=1}^{T_i}
  \bm u_{it}\bm u_{it}^\top\right)
  (\bm D_i)_{a,\cdot}^\top
  \\
  &\le R(\bm u)
  \sum_{i=1}^Nw_i\|\bm D_i\|_{\mathrm F}^2
  \le R(\bm u).
  \label{eq:prob-ML-restricted-energy-use}
\end{aligned}
\end{equation}

We now condition on $\bm u$ and calculate the squared increments of
\eqref{eq:prob-ML-conditional-process}.  If $\theta'$ is a second index,
let $\bm X_{\mathrm{LR}}'$ and $\bm D_i'$ be its associated matrices.
We have
\begin{equation}
\begin{aligned}
  {\mathbb E}_{\bm v}
  |\mathbb X_\theta-\mathbb X_{\theta'}|^2
  &=\frac1{\Ttot^2}{\mathbb E}_{\bm v}
  \left|
    \sum_{i=1}^N\sum_{t=1}^{T_i}
    \inner{
      \bm X_{\mathrm{LR}}^\top\bm D_i\bm u_{it}
      -(\bm X_{\mathrm{LR}}')^\top\bm D_i'\bm u_{it}}
    {\bm v_{it}}
  \right|^2
  \\
  &=\frac1{\Ttot^2}
  \sum_{i=1}^N\sum_{t=1}^{T_i}
  \left\|
    \bm X_{\mathrm{LR}}^\top\bm D_i\bm u_{it}
    -(\bm X_{\mathrm{LR}}')^\top\bm D_i'\bm u_{it}
  \right\|_2^2
  \\
  &=\frac1{\Ttot^2}
  \sum_{i=1}^N\sum_{t=1}^{T_i}
  \left\|
    (\bm X_{\mathrm{LR}}-\bm X_{\mathrm{LR}}')^\top
    \bm D_i\bm u_{it}
    +(\bm X_{\mathrm{LR}}')^\top
    (\bm D_i-\bm D_i')\bm u_{it}
  \right\|_2^2
  \\
  &\le
  \frac{2\|\bm X_{\mathrm{LR}}-\bm X_{\mathrm{LR}}'\|_{\mathrm F}^2}
  {\Ttot^2}
  \sum_{i=1}^N\sum_{t=1}^{T_i}\|\bm D_i\bm u_{it}\|_2^2
  \\
  &\quad+
  \frac{2\|\bm X_{\mathrm{LR}}'\|_{\op}^2}{\Ttot^2}
  \sum_{i=1}^N\sum_{t=1}^{T_i}
  \|(\bm D_i-\bm D_i')\bm u_{it}\|_2^2
  \\
  &\le
  C\frac{R(\bm u)}{\Ttot}
  \|\bm X_{\mathrm{LR}}-\bm X_{\mathrm{LR}}'\|_{\mathrm F}^2
  +\frac{C}{\Ttot^2}
  \sum_{i=1}^N\sum_{t=1}^{T_i}
  \|(\bm D_i-\bm D_i')\bm u_{it}\|_2^2,
\end{aligned}
  \label{eq:prob-ML-conditional-metric}
\end{equation}
where the second line follows from
$\mathbb E_{\bm v}(\bm v_{it}\bm v_{js}^\top)=\bm I_p$ when
$(i,t)=(j,s)$ and $\bm 0$ otherwise, while the last inequality uses
\eqref{eq:prob-ML-restricted-energy-use} and
$\sup_\theta\|\bm X_{\mathrm{LR}}\|_{\op}\le C$.

The two terms in \eqref{eq:prob-ML-conditional-metric} correspond to the
global low-rank and local sparse complexities.  Let $\bm G$ have independent
standard Gaussian entries, and let $\{\bm g_{it}\}$ be independent standard
Gaussian vectors.  Take $\bm G$, $\{\bm g_{it}\}$, and $\bm u$ to be
mutually independent.  For a sufficiently large absolute constant $C$,
define the comparison process
\begin{equation*}
  \mathbb Y_\theta
  :=\sqrt C\left\{
    \sqrt{\frac{R(\bm u)}{\Ttot}}
    \inner{\bm G}{\bm X_{\mathrm{LR}}}
    +\frac1{\Ttot}
    \sum_{i=1}^N\sum_{t=1}^{T_i}
    \inner{\bm g_{it}}{\bm D_i\bm u_{it}}
  \right\}.
\end{equation*}
Conditional on $\bm u$, its squared increment metric between $\theta$ and
$\theta'$ is, by definition, the conditional second moment
\begin{equation}\label{eq:lemB5_Ydiff}
  \mathbb E_{\bm G,\bm g}
  |\mathbb Y_\theta-\mathbb Y_{\theta'}|^2
  =
  C\frac{R(\bm u)}{\Ttot}
  \|\bm X_{\mathrm{LR}}-\bm X_{\mathrm{LR}}'\|_{\mathrm F}^2
  +\frac{C}{\Ttot^2}
  \sum_{i=1}^N\sum_{t=1}^{T_i}
  \|(\bm D_i-\bm D_i')\bm u_{it}\|_2^2.
\end{equation}
By \eqref{eq:prob-ML-conditional-metric}, the squared increment metric of
$\mathbb Y_\theta$ dominates that of $\mathbb X_\theta$ when $C$ is
sufficiently large, i.e.,
\[
  {\mathbb E}_{\bm v}
  |\mathbb X_\theta-\mathbb X_{\theta'}|^2\leq \mathbb E_{\bm G,\bm g}
  |\mathbb Y_\theta-\mathbb Y_{\theta'}|^2.
\]
Therefore, conditional on $\bm u$, the Sudakov--Fernique inequality
\citep[Theorem~7.2.11]{vershynin2018high} gives
\begin{equation}
\begin{aligned}
  \mathbb E_{\bm v}\sup_\theta|\mathbb X_\theta|
  &\le
  \mathbb E_{\bm G,\bm g}\sup_\theta\mathbb Y_\theta
  \\
  &\le
  \sqrt{\frac{C R(\bm u)}{\Ttot}}
  \mathbb E_{\bm G}\sup_\theta
  \inner{\bm G}{\bm X_{\mathrm{LR}}}
  +
  \frac{\sqrt C}{\Ttot}
  \mathbb E_{\bm g}\sup_\theta
  \sum_{i=1}^N\sum_{t=1}^{T_i}
  \inner{\bm g_{it}}{\bm D_i\bm u_{it}}.
\end{aligned}
  \label{eq:prob-ML-comparison-split}
\end{equation}
To bound the first term in
\eqref{eq:prob-ML-comparison-split}, note that every
$\bm X_{\mathrm{LR}}\in\mathcal X_{\mathrm{LR}}$ has unit Frobenius
norm and rank at most $2r$ by its definition in
Section~\ref{subsec:covariance_compatibility}.
Hence, the Gaussian width of $\mathcal X_{\mathrm{LR}}$ is bounded by
\begin{equation}
\begin{aligned}
  \mathbb E
  \sup_{\bm X_{\mathrm{LR}}\in\mathcal X_{\mathrm{LR}}}
  \inner{\bm G}{\bm X_{\mathrm{LR}}}
  &\le
  \mathbb E
  \left\{\sum_{j=1}^{\min\{2r,p\}}\sigma_j^2(\bm G)\right\}^{1/2}
  \le C\sqrt{pr}.
  \label{eq:prob-ML-global-width}
\end{aligned}
\end{equation}
Then, the first term on the right-hand side
of \eqref{eq:prob-ML-comparison-split} satisfies
\begin{equation}
  \sqrt{\frac{C R(\bm u)}{\Ttot}}
  \mathbb E_{\bm G}\sup_\theta
  \inner{\bm G}{\bm X_{\mathrm{LR}}}
  \le
  C\sqrt{\frac{R(\bm u)pr}{\Ttot}}.
  \label{eq:prob-ML-global-width-contribution}
\end{equation}

For the second term of \eqref{eq:prob-ML-comparison-split}, set
$\bm B_i:=T_i^{-1}\sum_{t=1}^{T_i}\bm g_{it}\bm u_{it}^\top$.
Then weighted Frobenius duality gives
\begin{equation}
  \frac1{\Ttot}\sup_\theta\sum_{i,t}
    \inner{\bm g_{it}}{\bm D_i\bm u_{it}}=\sup_{\substack{|\supp(\bm D_i)|\le s_i^\dagger\\
    \sum_iw_i\|\bm D_i\|_{\mathrm F}^2\le1}}
    \sum_iw_i\inner{\bm D_i}{\bm B_i}
  =\left\{\sum_iw_i\sup_{|\Omega_i|\le s_i^\dagger}
    \|\mathcal P_{\Omega_i}(\bm B_i)\|_{\mathrm F}^2
    \right\}^{1/2}.
  \label{eq:prob-ML-local-width-duality}
\end{equation}
When $s_i^\dagger=0$, the support $\Omega_i$ is empty,
so the corresponding summand is zero. For every other client and
every unit-Frobenius matrix $\bm M$ with at most $s_i^\dagger$ nonzeros,
\begin{equation}
  \operatorname{Var}_{\bm g}\inner{\bm M}{\bm B_i}
  =\frac1{T_i^2}\sum_t\|\bm M\bm u_{it}\|_2^2
  \le\frac{R(\bm u)}{T_i}.
  \label{eq:prob-ML-local-width-variance}
\end{equation}
Thus, for fixed $\Omega_i$, the conditional covariance of
$\operatorname{vec}(\mathcal P_{\Omega_i}(\bm B_i))$ has operator norm
at most $R(\bm u)/T_i$ and trace at most
$s_i^\dagger R(\bm u)/T_i$. Gaussian concentration gives
\[
  \mathbb P_{\bm g}\!\left\{
    \|\mathcal P_{\Omega_i}(\bm B_i)\|_{\mathrm F}
    >\sqrt{\frac{s_i^\dagger R(\bm u)}{T_i}}+t
    \,\middle|\,\bm u\right\}
  \le\exp\left\{-\frac{cT_it^2}{R(\bm u)}\right\},
  \qquad t\ge0.
\]
Every smaller support can be enlarged to size
$s_i^\dagger\wedge p^2$ without decreasing the projected norm.
A union bound over the $\binom{p^2}{s_i^\dagger\wedge p^2}$ such supports,
followed by integration of the tail, therefore gives
\[
  \mathbb E_{\bm g}\sup_{|\Omega_i|\le s_i^\dagger}
    \|\mathcal P_{\Omega_i}(\bm B_i)\|_{\mathrm F}^2
  \le\frac{CR(\bm u)}{T_i}
  \left\{s_i^\dagger+\log\binom{p^2}{s_i^\dagger\wedge p^2}\right\}\le C\frac{R(\bm u)s_i^\dagger\log p}{T_i}.
\]
Applying Jensen's inequality to
\eqref{eq:prob-ML-local-width-duality} and using
$w_i/T_i=1/\Ttot$ now yields
\begin{equation}
\begin{aligned}
  \frac1{\Ttot}\mathbb E_{\bm g}\sup_\theta
    \sum_{i,t}\inner{\bm g_{it}}{\bm D_i\bm u_{it}}
  &\le\left\{\sum_iw_i\mathbb E_{\bm g}
    \sup_{|\Omega_i|\le s_i^\dagger}
    \|\mathcal P_{\Omega_i}(\bm B_i)\|_{\mathrm F}^2\right\}^{1/2}\\
  &\le C\left\{R(\bm u)\sum_i
    \frac{w_i s_i^\dagger\log p}{T_i}\right\}^{1/2}\\
  &=C\left\{\frac{R(\bm u)\log p}{\Ttot}
    \sum_i s_i^\dagger\right\}^{1/2}.
\end{aligned}
  \label{eq:prob-ML-local-width}
\end{equation}
Substituting these two bounds into
\eqref{eq:prob-ML-comparison-split} gives the conditional mean bound
\begin{equation}
  \mathbb E_{\bm v}\sup_\theta|\mathbb X_\theta|
  \le C\sqrt{R(\bm u)}\left[
    \sqrt{\frac{pr}{\Ttot}}
    +\left\{\frac{\log p}{\Ttot}
      \sum_i s_i^\dagger\right\}^{1/2}\right].
  \label{eq:prob-ML-conditional-mean}
\end{equation}

To turn this mean bound into a conditional $L_k$ bound, we also control
the Gaussian fluctuations. The calculation in
\eqref{eq:prob-ML-conditional-metric} with
$\bm X_{\mathrm{LR}}'=\bm X_{\mathrm{LR}}$ and $\bm D_i'=\bm 0$ gives,
uniformly over the index set,
\begin{equation*}
  \operatorname{Var}_{\bm v}
  \{\mathbb X_\theta(\bm u,\bm v)\mid\bm u\}
  =\frac1{\Ttot^2}
  \sum_{i=1}^N\sum_{t=1}^{T_i}
  \|\bm X_{\mathrm{LR}}^\top\bm D_i\bm u_{it}\|_2^2
  \le C\frac{R(\bm u)}{\Ttot}.
\end{equation*}
Consequently, the supremum is a Lipschitz function of $\bm v$ with
Lipschitz constant at most $C\sqrt{R(\bm u)/\Ttot}$. Gaussian
concentration therefore gives, for $t\ge0$,
\begin{equation}
  \mathbb P_{\bm v}\!\left\{
    \sup_\theta|\mathbb X_\theta|
    >\mathbb E_{\bm v}\sup_\theta|\mathbb X_\theta|
      +C\sqrt{\frac{R(\bm u)t}{\Ttot}}
    \,\middle|\,\bm u\right\}\le e^{-t}.
  \label{eq:prob-ML-conditional-tail}
\end{equation}
Integrating this tail and using
\eqref{eq:prob-ML-conditional-mean} yields
\begin{equation}
  \left\|\sup_\theta
  |\mathbb X_\theta(\bm u,\bm v)|
  \right\|_{L_k(\bm v\mid\bm u)}
  \le C\sqrt{R(\bm u)}
  \left[
    \sqrt{\frac{pr+k}{\Ttot}}
    +\left\{
      \frac{\log p}{\Ttot}\sum_{i=1}^Ns_i^\dagger
    \right\}^{1/2}
  \right].
  \label{eq:prob-ML-conditional-process-moment}
\end{equation}
Only the random factor $\sqrt{R(\bm u)}$ remains to be controlled before
we can remove the conditioning.

\paragraph*{Step 5: Removing the conditioning and completing the moment bound}
We bound $\|R(\bm u)^{1/2}\|_{L_k}$ using the sparse covering
argument in \eqref{eq:lemB2_sparseNet}. For each client, let $\mathcal N_J$ denote
the $1/4$-net on support $J$.  The union of these nets has log-cardinality at
most $C s_i^\dagger\log p$ for client $i$.  Approximating each supported
unit sphere by its $1/4$-net gives
\begin{equation*}
  R(\bm u)^{1/2}
  \le\sqrt2
  \max_{\substack{i\in[N],\ s_i^\dagger>0,\ J\subset[p]\\
  |J|=s_i^\dagger\wedge p}}
  \max_{\overline{\bm v}\in\mathcal N_J}
  \left\{
    \frac1{T_i}\sum_{t=1}^{T_i}
    (\overline{\bm v}^\top\bm u_{it,J})^2
  \right\}^{1/2}.
\end{equation*}
For every fixed $(i,J,\overline{\bm v})$, the vector
$\{\overline{\bm v}^{\top}\bm u_{it,J}\}_{t=1}^{T_i}$ is distributed as
$N(\bm 0,\bm I_{T_i})$.  Hence, by concentration of the Euclidean norm
\citep[Theorem~3.1.1]{vershynin2018high}, for every $y\ge0$,
\begin{equation*}
  \mathbb P\left[
    \left\{
      \frac1{T_i}\sum_{t=1}^{T_i}
      (\overline{\bm v}^\top\bm u_{it,J})^2
    \right\}^{1/2}
    >1+C\sqrt{\frac{y}{T_i}}
  \right]
  \le2\exp(-cy).
\end{equation*}
Taking $y=C\{s_i^\dagger\log p+\log N+x\}$ and applying a union bound over
the supports, net points, and clients gives, for every $x\ge0$,
\begin{equation*}
  \mathbb P\left[
    R(\bm u)^{1/2}
    >
    C\left\{
      1+\max_{i:\,s_i^\dagger>0}
      \sqrt{
        \frac{s_i^\dagger\log p+\log N+x}{T_i}}
    \right\}
  \right]
  \le2\exp(-cx).
\end{equation*}
Integrating this tail bound yields, for every $k\ge1$,
\begin{equation*}
  \|R(\bm u)^{1/2}\|_{L_k}
  \le
  C\left\{
    1+\max_{i:\,s_i^\dagger>0}
    \sqrt{
      \frac{s_i^\dagger\log p+\log N+k}{T_i}}
  \right\}.
\end{equation*}
Under the sample-size condition
\eqref{eq:prob-ML-iid-moment-sample}, this proves
\begin{equation}
  \|R(\bm u)^{1/2}\|_{L_k}\le C.
  \label{eq:prob-ML-restricted-energy-moment}
\end{equation}
To pass to the unconditional moment, the tower property gives
\[
  \left\|\sup_\theta|\mathbb X_\theta|\right\|_{L_k(\bm u,\bm v)}^k
  =\mathbb E_{\bm u}\!\left[
    \left\|\sup_\theta|\mathbb X_\theta|\right\|_{L_k(\bm v\mid\bm u)}^k
  \right].
\]
Substituting \eqref{eq:prob-ML-conditional-process-moment} and
\eqref{eq:prob-ML-restricted-energy-moment}, then applying
\eqref{eq:prob-ML-Gaussian-rotation}, gives
\begin{equation}
  \|\zeta_{ML}^{\mathrm{iid}}\|_{L_k}
  \le
  C\left[
    \sqrt{\frac{pr+k}{\Ttot}}
    +\left\{
      \frac{\log p}{\Ttot}\sum_{i=1}^Ns_i^\dagger
    \right\}^{1/2}
  \right].
  \label{eq:prob-ML-iid-two-width}
\end{equation}
Combining this iid bound with the trajectory reduction
\eqref{eq:prob-ML-trajectory-transfer} proves the target moment bound
\eqref{eq:prob-ML-target-moment}.

\paragraph*{Step 6: From moments to the probability bound}
Take
\begin{equation*}
  k=1+c\min\left\{pr,\ (s_{\min}^\dagger\vee1)\log p\right\}.
\end{equation*}
Since $k\le C s_i^\dagger\log p$ whenever $s_i^\dagger>0$, the assumed sample-size
condition implies \eqref{eq:prob-ML-iid-moment-sample}.  Applying Markov's inequality to \eqref{eq:prob-ML-target-moment},
\begin{equation*}
  \mathbb P\left\{
    \zeta_{ML}>e\|\zeta_{ML}\|_{L_k}
  \right\}\le\exp(-k),
\end{equation*}
gives a failure probability bounded by
\begin{equation*}
  C\exp(-cpr)+C\exp\{-c(s_{\min}^\dagger\vee1)\log p\}.
\end{equation*}
Since $k\le Cpr$, the $k/\Ttot$ term in
\eqref{eq:prob-ML-target-moment} is absorbed by $pr/\Ttot$.
Therefore, with the probability stated in part~\textup{(i)},
\begin{equation}
  \zeta_{ML}^2
  \le CM^2
  \left\{
    \frac{pr}{\Ttot}
    +\frac{\log p}{\Ttot}\sum_{i=1}^Ns_i^\dagger
  \right\}.
  \label{eq:prob-ML-centered-two-width}
\end{equation}
By the exact identity \eqref{eq:prob-ML-zeta-target-identity}, this is precisely
\eqref{eq:prob-ML-centered-geometry}, and hence proves
part~\textup{(i)}.

\paragraph*{Shared sparse directions: part~\textup{(ii)}}
For part~\textup{(ii)}, assume $s_0^\dagger>0$, since the zero case is
immediate. Define the target supremum
\begin{equation}
\begin{aligned}
  \zeta_{0L}^2
  :={}&
  \sup_{\substack{
    \bm X_{\mathrm S}\in\mathcal X_{\mathrm S}\\
    \Omega_i\subset[p]\times[p],\ |\Omega_i|\le s_i^\dagger}}
  \sum_{i=1}^Nw_i
  \left\|
    \mathcal P_{\Omega_i}
    \{\bm X_{\mathrm S}(\widehat{\bbm\Gamma}_i-\bbm\Gamma_i)\}
  \right\|_{\mathrm F}^2.
  \label{eq:prob-0L-zeta-target-identity}
\end{aligned}
\end{equation}
The same weighted Frobenius duality gives
\eqref{eq:prob-ML-centered-block} with
$\bm X_{\mathrm{LR}}\in\mathcal X_{\mathrm{LR}}$ replaced by
$\bm X_{\mathrm S}\in\mathcal X_{\mathrm S}$ and $\zeta_{ML}$ by $\zeta_{0L}$.
Since the matrices in $\mathcal X_{\mathrm S}$ also have unit Frobenius norm,
only the Gaussian width bound needs to be changed in Steps~2--5.
By vectorization and the sparse covering argument in
\eqref{eq:lemB2_sparseNet},
\begin{equation*}
  \mathbb E
  \sup_{\bm X_{\mathrm S}\in\mathcal X_{\mathrm S}}
  \inner{\bm G}{\bm X_{\mathrm S}}
  \le
  \mathbb E
  \sup_{\substack{\Omega\subset[p]\times[p]\\
  |\Omega|\le s_0^\dagger}}
  \|\mathcal P_\Omega(\bm G)\|_{\mathrm F}
  \le C\sqrt{s_0^\dagger\log p}.
\end{equation*}
Substituting this width into the first term of
\eqref{eq:prob-ML-comparison-split} and following Steps~2--5 gives, under
\eqref{eq:prob-ML-iid-moment-sample},
\begin{equation}
  \|\zeta_{0L}\|_{L_k}
  \le CM\left[
    \sqrt{\frac{s_0^\dagger\log p+k}{\Ttot}}
    +\left\{
      \frac{\log p}{\Ttot}\sum_{i=1}^Ns_i^\dagger
    \right\}^{1/2}
  \right].
  \label{eq:prob-0L-moment}
\end{equation}

Finally, take
\begin{equation*}
  k=1+c\min\left\{
    (s_0^\dagger\vee1)\log p,\ 
    (s_{\min}^\dagger\vee1)\log p
  \right\}.
\end{equation*}
Again, $k\le C s_i^\dagger\log p$, so the assumed
clientwise sample sizes imply \eqref{eq:prob-ML-iid-moment-sample}.
The Markov argument in Step~6, applied to
\eqref{eq:prob-0L-moment}, gives a failure probability bounded by
\begin{equation*}
  C\exp\{-c(s_0^\dagger\vee1)\log p\}
  +C\exp\{-c(s_{\min}^\dagger\vee1)\log p\}.
\end{equation*}
Since $k\le C s_0^\dagger\log p$, the $k/\Ttot$ term in
\eqref{eq:prob-0L-moment} is absorbed by
$s_0^\dagger\log p/\Ttot$.  Thus, with the probability stated in
part~\textup{(ii)},
\begin{equation}
  \zeta_{0L}^2
  \le CM^2
  \left\{
    \frac{s_0^\dagger\log p}{\Ttot}
    +\frac{\log p}{\Ttot}\sum_{i=1}^Ns_i^\dagger
  \right\}.
  \label{eq:prob-0L-centered-two-width}
\end{equation}
By \eqref{eq:prob-0L-zeta-target-identity}, this is precisely
\eqref{eq:prob-0L-centered-geometry}, and hence proves
part~\textup{(ii)}.
\end{proof}

\subsection{Local approximation bounds}
\label{app:lq-local-concentration}

\begin{lemma}[Projected local approximation bound]
\label{lem:lq-local}
Under Assumption~\ref{assump:stationarity}, let
$k_i=\min\{p,s_i^\dagger+s_{i,*}\}$ and fix $u>0$.
Suppose that
\begin{equation}
 \label{eq:note-local-samples}
 T_i\ge C\frac{M}{\phi_{\max}}
 \left\{k_i\log(ep/k_i)+\log N+u\right\},
 \qquad i\in[N]\text{ with }k_i>0.
\end{equation}
Then, with probability
at least $1-2e^{-u}$, for all matrices
$\bm{\mathfrak R}_i$ satisfying
$|\supp(\bm{\mathfrak R}_i)|\le s_{i,*}$,
\begin{equation}
 \label{eq:note-local-R}
 \sup_{\substack{\Omega_i\subseteq[p]^2\, (i\in[N])\\
                         |\Omega_i|\le s_i^\dagger}}
 \sum_iw_i\|\mathcal P_{\Omega_i}
                   (\bm{\mathfrak R}_i\widehat{\bbm\Gamma}_i)\|_{\mathrm F}^2
 \le4\phi_{\max}^2\sum_iw_i\|\bm{\mathfrak R}_i\|_{\mathrm F}^2.
\end{equation}
\end{lemma}

\begin{proof}
To prove \eqref{eq:note-local-R}, we bound the contribution of each row.
Fix a row $j$ and set
$I=\{b:(\bm{\mathfrak R}_i)_{jb}\ne0\}$,
$K=\{b:(j,b)\in\Omega_i\}$, and $J=I\cup K$.
Then $|J|\le k_i$ and
\[
 \|[\mathcal P_{\Omega_i}
       (\bm{\mathfrak R}_i\widehat{\bbm\Gamma}_i)]_{j,\cdot}\|_2^2
 =\|(\bm{\mathfrak R}_i)_{j,I}
             \widehat{\bbm\Gamma}_{i,IK}\|_2^2\le\|(\bm{\mathfrak R}_i)_{j,I}\|_2^2
             \|\widehat{\bbm\Gamma}_{i,JJ}\|_{\op}^2.
\]
Thus, summing over rows and clients gives \eqref{eq:note-local-R}
on the event
\begin{equation}
\label{eq:note-local-cov-event}
 \mathcal E_{\mathrm{loc}}(u)
 :=\bigcap_{i}
 \left\{\sup_{J\subseteq[p],\ |J|\le k_i}
       \|\widehat{\bbm\Gamma}_{i,JJ}\|_{\op}
       \le2\phi_{\max}\right\}.
\end{equation}

It remains to show that this event has probability at least $1-2e^{-u}$.
Since 
\[
 \|\widehat{\bbm\Gamma}_{i,JJ}\|_{\op}
 \le\|(\widehat{\bbm\Gamma}_i-\bbm\Gamma_i)_{JJ}\|_{\op}
       +\phi_{\max},
\]
it suffices to bound the centered covariance norm  on all these coordinate sets.
For each $i$ with $k_i>0$, the clientwise version of
\eqref{eq:lemB1_hanson} and the sparse-net argument in
\eqref{eq:lemB2_sparseNet}, together with
\eqref{eq:note-local-samples}, give
\[
\begin{aligned}
 \mathbb P\left\{
   \sup_{J\subseteq[p],\ |J|\le k_i}
   \|(\widehat{\bbm\Gamma}_i-\bbm\Gamma_i)_{JJ}\|_{\op}
       >\phi_{\max}\right\}
 &\le2\exp\left\{Ck_i\log(ep/k_i)
                  -cT_i\phi_{\max}/M\right\}\\
 &\le\frac{2e^{-u}}{N}.
\end{aligned}
\]
A union bound over clients now yields
\[
 \mathbb P\{\mathcal E_{\mathrm{loc}}(u)^c\}
 \le\sum_{i:\,k_i>0}\frac{2e^{-u}}N
 \le2e^{-u}.
\]
\end{proof}

Since $k_i\le2s_i^\dagger$ and $\Xi_1\ge M/\phi_{\max}$,
\eqref{eq:app-unified-sample-conditions} implies
\eqref{eq:note-local-samples} with $u=\log(pN)$.

\section{Auxiliary Deterministic Results}
\label{sec:app-deterministic}

In this appendix, we establish the deterministic auxiliary results used in the preceding proofs. All notation is inherited from Appendix~\ref{subsec:app-main-setting}.
\subsection{Population curvature and factor geometry}
\label{subsec:det-curvature-factor-geometry}

\begin{lemma}[Population RCG inequalities]
\label{lem:det-population-rcg}
The following inequalities hold.
\begin{enumerate}[label=\textup{(\roman*)}]
\item For every $\bm H\in\mathbb R^{p\times p}$,
\begin{equation}
  \inner{\bm H\bbm\Gamma_w}{\bm H}
  \ge
  \frac{\alpha_{\mathrm{LR}}}{2}\|\bm H\|_{\mathrm F}^2
  +\frac{1}{2\beta_{\mathrm{LR}}}
  \|\bm H\bbm\Gamma_w\|_{\mathrm F}^2.
  \label{eq:app-lowrank-population-rcg}
\end{equation}

\item For every $\Omega_0\in\mathcal C_{\mathrm S}'$ and every
$\bm H$ satisfying $\supp(\bm H)\subseteq\Omega_0$,
\begin{equation}
  \inner{\bm H\bbm\Gamma_w}{\bm H}
  \ge
  \frac{\alpha_{\mathrm S}}{2}\|\bm H\|_{\mathrm F}^2
  +\frac{1}{2\beta_{\mathrm S}}
  \|\mathcal P_{\Omega_0}(\bm H\bbm\Gamma_w)\|_{\mathrm F}^2.
  \label{eq:global-sparse-population-RCG}
\end{equation}

\item For every $i\in[N]$, every
$\Omega_i\subset[p]\times[p]$ satisfying $|\Omega_i|\le s_i'$,
and every $\bm H$ satisfying $\supp(\bm H)\subseteq\Omega_i$,
\begin{equation}
  \inner{\bm H\bbm\Gamma_i}{\bm H}
  \ge
  \frac{\alpha_i}{2}\|\bm H\|_{\mathrm F}^2
  +\frac{1}{2\beta_i}
  \|\mathcal P_{\Omega_i}(\bm H\bbm\Gamma_i)\|_{\mathrm F}^2.
  \label{eq:app-local-population-rcg}
\end{equation}
\end{enumerate}
\end{lemma}

\begin{proof}
Let $\lambda_1\ge \lambda_2\ge\cdots\ge\lambda_p>0$ be the eigenvalues of $\bbm\Gamma_w$. 
For part~\textup{(i)}, write
\begin{equation*}
  \bbm\Gamma_w
  =\bm O\operatorname{diag}(\lambda_1,\ldots,\lambda_p)\bm O^\top,
  \qquad
  \alpha_{\mathrm{LR}}\le\lambda_j\le\beta_{\mathrm{LR}},\quad j\in[p],
\end{equation*}
where $\bm O$ is orthogonal, and set $\widetilde{\bm H}=\bm H\bm O$.
Then
\begin{equation*}
\begin{aligned}
  \inner{\bm H\bbm\Gamma_w}{\bm H}
  &=\sum_{j=1}^p\lambda_j
    \|\widetilde{\bm H}_{\cdot,j}\|_2^2,\\
  \|\bm H\bbm\Gamma_w\|_{\mathrm F}^2
  &=\sum_{j=1}^p\lambda_j^2
    \|\widetilde{\bm H}_{\cdot,j}\|_2^2.
\end{aligned}
\end{equation*}
For every $j\in[p]$,
\begin{equation*}
  \lambda_j
  =\frac{\lambda_j}{2}+\frac{\lambda_j}{2}
  \ge
  \frac{\alpha_{\mathrm{LR}}}{2}+\frac{\lambda_j^2}{2\beta_{\mathrm{LR}}}.
\end{equation*}
Multiplying by $\|\widetilde{\bm H}_{\cdot,j}\|_2^2$ and summing over
$j$ proves part~\textup{(i)}.

For part~\textup{(ii)}, let
\begin{equation*}
  \bm K:=\mathcal P_{\Omega_0}(\bm H\bbm\Gamma_w).
\end{equation*}
Since $\supp(\bm H),\supp(\bm K)\subseteq\Omega_0\in\mathcal C_{\mathrm S}'$, the definitions of $\alpha_{\mathrm S}$ and $\beta_{\mathrm S}$ therefore give, with the zero
cases being immediate,
\begin{equation}
\label{eq:appC_RSCRSS}
\begin{aligned}
  \inner{\bm H\bbm\Gamma_w}{\bm H}
  &=\operatorname{tr}(\bm H\bbm\Gamma_w\bm H^\top)
  \ge\alpha_{\mathrm S}\|\bm H\|_{\mathrm F}^2,\\
  \inner{\bm K\bbm\Gamma_w}{\bm K}
  &=\operatorname{tr}(\bm K\bbm\Gamma_w\bm K^\top)
  \le\beta_{\mathrm S}\|\bm K\|_{\mathrm F}^2.
\end{aligned}
\end{equation}
Using the Frobenius
Cauchy--Schwarz inequality,
\begin{equation*}
\begin{aligned}
  \|\bm K\|_{\mathrm F}^2
  &=\inner{\mathcal P_{\Omega_0}(\bm H\bbm\Gamma_w)}{\bm K}
  =\inner{\bm H\bbm\Gamma_w}{\bm K}\\
  &=\inner{\bm H\bbm\Gamma_w^{1/2}}
  {\bm K\bbm\Gamma_w^{1/2}}\\
  &\le
  \|\bm H\bbm\Gamma_w^{1/2}\|_{\mathrm F}
  \|\bm K\bbm\Gamma_w^{1/2}\|_{\mathrm F}\\
  &=
  \inner{\bm H\bbm\Gamma_w}{\bm H}^{1/2}
  \inner{\bm K\bbm\Gamma_w}{\bm K}^{1/2}\\
  &\le
  \sqrt{\beta_{\mathrm S}}
  \inner{\bm H\bbm\Gamma_w}{\bm H}^{1/2}
  \|\bm K\|_{\mathrm F}.
\end{aligned}
\end{equation*}
If $\bm K\ne\bm 0$, division and squaring give
\begin{equation*}
  \inner{\bm H\bbm\Gamma_w}{\bm H}
  \ge\frac{1}{\beta_{\mathrm S}}\|\bm K\|_{\mathrm F}^2,
\end{equation*}
and the same inequality is immediate when $\bm K=\bm 0$.  Averaging it
with the first inequality of \eqref{eq:appC_RSCRSS} yields
\begin{equation*}
  \inner{\bm H\bbm\Gamma_w}{\bm H}
  \ge
  \frac{\alpha_{\mathrm S}}{2}\|\bm H\|_{\mathrm F}^2
  +\frac{1}{2\beta_{\mathrm S}}\|\bm K\|_{\mathrm F}^2,
\end{equation*}
which proves part~\textup{(ii)}.  Replacing
$\Omega_0$, $\bbm\Gamma_w$, $\alpha_{\mathrm S}$, and $\beta_{\mathrm S}$ by $\Omega_i$,
$\bbm\Gamma_i$, $\alpha_i$, and $\beta_i$, respectively, gives the
corresponding proof of part~\textup{(iii)}.
\end{proof}

\begin{lemma}[Optimal alignment and metric comparison]
\label{lem:metric_equivalence}
The following result restates Lemmas~22, 24, and~26 of
\citet{tong2021accelerating}.
For any factor pair $(\bm L,\bm R)$, the alignment error in
\eqref{eq:scaled_factor_error} satisfies
\begin{equation}
  (\sqrt{2}-1)\left\|(\bm L,\bm R)\right\|_{\mathrm{align},*}^2
  \le
  \|\bm L\bm R^\top-\bm B_*\|_{\mathrm F}^2.
  \label{eq:det-factor-product-lower}
\end{equation}
Suppose further that $\left\|(\bm L,\bm R)\right\|_{\mathrm{align},*}\le\varepsilon\sigma_r$ for some
\(0\le\varepsilon<1\).
Then the infimum in \eqref{eq:scaled_factor_error} is attained by some
$\bm Q\in\mathrm{GL}(r)$.  For any such optimal alignment,
\begin{equation}
  \left\|(\bm L\bm Q-\bm L_*)
  \bbm\Sigma_*^{-1/2}\right\|_{\op}
  \vee
  \left\|(\bm R\bm Q^{-\top}-\bm R_*)
  \bbm\Sigma_*^{-1/2}\right\|_{\op}
  \le\varepsilon.
  \label{eq:det-aligned-factor-bounds}
\end{equation}
The reverse comparison also holds locally:
\begin{equation}
  \|\bm L\bm R^\top-\bm B_*\|_{\mathrm F}^2
  \le
  2\left(1+\frac{\varepsilon}{2}\right)^2
  \left\|(\bm L,\bm R)\right\|_{\mathrm{align},*}^2.
  \label{eq:det-factor-product-equivalence}
\end{equation}
\end{lemma}

\begin{lemma}[Preconditioned factor bounds]
\label{lem:perturbation}
Suppose that $\left\|(\bm L,\bm R)\right\|_{\mathrm{align},*}\le\varepsilon\sigma_r$ for some
$0\le\varepsilon<1$, and let $\bm Q$ be an optimal alignment from
Lemma~\ref{lem:metric_equivalence}. Define
\begin{equation}
  \bm P_L
  :=\bm L(\bm L^\top\bm L)^{-1}
  \bm Q^{-\top}\bbm\Sigma_*^{1/2},
  \qquad
  \bm P_R
  :=\bm R(\bm R^\top\bm R)^{-1}
  \bm Q\bbm\Sigma_*^{1/2}.
  \label{eq:det-aligned-preconditioners}
\end{equation}
Then
\begin{equation}
\begin{aligned}
  \|\bm P_L\|_{\op}\vee\|\bm P_R\|_{\op}
  &\le\frac{1}{1-\varepsilon},\\
  \|\bm P_L-\bm U_*\|_{\op}
  \vee\|\bm P_R-\bm V_*\|_{\op}
  &\le\frac{\sqrt{2}\,\varepsilon}{1-\varepsilon}.
\end{aligned}
  \label{eq:det-aligned-preconditioner-bounds}
\end{equation}
\end{lemma}

\begin{proof}
The result follows directly from
\citet[Lemma~25]{tong2021accelerating}, applied to the aligned factors
$\bm L\bm Q$ and $\bm R\bm Q^{-\top}$, together with
\eqref{eq:det-aligned-factor-bounds}.
\end{proof}

Recalling the classes
\(\mathcal N_{\varepsilon}\) and \(\mathcal X_{\mathrm{LR}}\) from
Section~\ref{subsec:covariance_compatibility}, with local radius
\(\varepsilon\), we introduce the following lemma.
\begin{lemma}[Low-rank directions in the local neighborhood]
\label{lem:orthonormal-lowrank-directions}
Suppose that \(\left\|(\bm L,\bm R)\right\|_{\mathrm{align},*}\le\varepsilon\sigma_r\) with
\(0<\varepsilon\le1/4\), and let \(\bm Q,\bm P_L,\bm P_R\) be as in
Lemma~\ref{lem:perturbation}.
The following holds.
\begin{enumerate}[label=\textup{(\roman*)}]
\item For any \(\bm K_L\in\mathbb R^{p\times r}\) and
\(\bm K_R\in\mathbb R^{r\times p}\) satisfying
\(\|\bm K_L\|_{\mathrm F}^2+\|\bm K_R\|_{\mathrm F}^2\le1\), 
\(\bm X_{\mathrm{LR}}:=\bm K_L\bm P_R^\top+\bm P_L\bm K_R\) satisfies
\(\|\bm X_{\mathrm{LR}}\|_{\mathrm F}\le\sqrt2/(1-\varepsilon)\).
If \(\bm X_{\mathrm{LR}}\ne\bm0\), then
\(\bm X_{\mathrm{LR}}/\|\bm X_{\mathrm{LR}}\|_{\mathrm F}
\in\mathcal X_{\mathrm{LR}}\).
\item The low-rank error
\(\bbm\Delta_{\bm B}:=\bm L\bm R^\top-\bm B_*\) satisfies
$
  \|\bbm\Delta_{\bm B}\|_{\mathrm F}
  \le\sqrt2\left(1+\varepsilon/2\right)\left\|(\bm L,\bm R)\right\|_{\mathrm{align},*}.
$
If \(\bbm\Delta_{\bm B}\ne\bm0\), then
\(\bbm\Delta_{\bm B}/\|\bbm\Delta_{\bm B}\|_{\mathrm F}
\in\mathcal X_{\mathrm{LR}}\).
\end{enumerate}
\end{lemma}

\begin{proof}
Define $\widetilde{\bm U}:=\bm L(\bm L^\top\bm L)^{-1/2}$ and
$\widetilde{\bm V}:=\bm R(\bm R^\top\bm R)^{-1/2}$.
These matrices have orthonormal columns.
The definition \eqref{eq:det-aligned-preconditioners} gives
\[
\begin{aligned}
  \bm P_L
  &=\widetilde{\bm U}(\bm L^\top\bm L)^{-1/2}
    \bm Q^{-\top}\bbm\Sigma_*^{1/2},\\
  \bm P_R
  &=\widetilde{\bm V}(\bm R^\top\bm R)^{-1/2}
    \bm Q\bbm\Sigma_*^{1/2}.
\end{aligned}
\]
Since \((\bm I-\widetilde{\bm U}\widetilde{\bm U}^{\top})\bm P_L=\bm0\),
Lemma~\ref{lem:perturbation} gives
\begin{equation}
  \|\widetilde{\bm U}\widetilde{\bm U}^{\top}
     -\bm U_*\bm U_*^{\top}\|_{\op}
  =\|(\bm I-\widetilde{\bm U}\widetilde{\bm U}^{\top})\bm U_*\|_{\op}\le\|\bm P_L-\bm U_*\|_{\op}
  \le\frac{\sqrt2\varepsilon}{1-\varepsilon}\le2\varepsilon.
\end{equation}
The analogous bound holds for \(\widetilde{\bm V}\), so
\((\widetilde{\bm U},\widetilde{\bm V})\in\mathcal N_{\varepsilon}\).

For part~\textup{(i)}, the displayed representations of
\(\bm P_L,\bm P_R\) express \(\bm X_{\mathrm{LR}}\) as a sum of the
form defining \(\mathcal X_{\mathrm{LR}}\) in
Section~\ref{subsec:covariance_compatibility}, with the orthonormal
pair \((\widetilde{\bm U},\widetilde{\bm V})\). Thus, if
\(\bm X_{\mathrm{LR}}\ne\bm0\),
\begin{equation}
  \frac{\bm X_{\mathrm{LR}}}{\|\bm X_{\mathrm{LR}}\|_{\mathrm F}}
  \in\mathcal X_{\mathrm{LR}}.
  \label{eq:app-lowrank-direction-normalization} 
\end{equation}
Moreover, Lemma~\ref{lem:perturbation} gives
\[
  \|\bm X_{\mathrm{LR}}\|_{\mathrm F}
  \le\frac{\|\bm K_L\|_{\mathrm F}+\|\bm K_R\|_{\mathrm F}}
           {1-\varepsilon}
  \le\frac{\sqrt2}{1-\varepsilon}.
\]

For part~\textup{(ii)}, the pair
\((\bm U_*,\widetilde{\bm V})\) also belongs to
\(\mathcal N_{\varepsilon}\), and
\[
  \bbm\Delta_{\bm B}
  =\bm L(\bm R^\top\bm R)^{1/2}\widetilde{\bm V}^{\top}
   -\bm U_*\bbm\Sigma_*\bm V_*^\top.
\]
For nonzero \(\bbm\Delta_{\bm B}\), dividing this representation by its
Frobenius norm yields an element of \(\mathcal X_{\mathrm{LR}}\).
The norm bound follows from
\eqref{eq:det-factor-product-equivalence}.
\end{proof}

\subsection{Sparse projection and weighted FCT}
\label{subsec:det-thresholding}

\begin{lemma}[Stability of capacity-constrained thresholding]
\label{lem:G_stability}
Let $\mathcal G_{s_0,\gamma}$ be the capacity-constrained thresholding
operator described in Section~\ref{sec:perfect_algorithm}.
Suppose $s_0\ge s_{0,*}$ and $\gamma\ge\gamma_*$. For any $\bm X\in \mathbb{R}^{p\times p}$, let
$\bm S_0^+=\mathcal G_{s_0,\gamma}(\bm X)$, and
define $\Omega_0^+
  :=\supp(\bm S_0^+)\cup\supp(\bm S_{0,*})$.
Then,
\begin{equation}
  \|\bm S_0^+-\bm S_{0,*}\|_{\mathrm F}^2
  \le
  (1+\rho_0)
  \|\mathcal P_{\Omega_0^+}(\bm X-\bm S_{0,*})\|_{\mathrm F}^2,
\end{equation}
where 
\begin{equation}
  \rho_0
  :=\frac12\left(\frac{s_{0,*}}{s_0}+2\frac{\gamma_*}{\gamma}\right)
  +\sqrt{
    \frac{s_{0,*}}{s_0}+2\frac{\gamma_*}{\gamma}
    +\frac14\left(
      \frac{s_{0,*}}{s_0}+2\frac{\gamma_*}{\gamma}
    \right)^2
  }.
\end{equation}
Furthermore, choosing $s_0\ge C \beta_{\mathrm S}^2\alpha_{\mathrm S}^{-2}s_{0,*}$ and $\gamma\ge C\beta_{\mathrm S}^2\alpha_{\mathrm S}^{-2}\gamma_*$ for some proper constant $C>0$ gives that, for a sufficiently small constant $c>0$, 
$$
  \rho_0\le c\frac{\alpha_{\mathrm S}}{\beta_{\mathrm S}}.
$$
\end{lemma}

\begin{proof}
For notational
simplicity, suppose that $\gamma p$ and $\gamma_*p$ are integers, and set
\begin{equation*}
  \vartheta_0:=\frac{s_{0,*}}{s_0}+2\frac{\gamma_*}{\gamma}.
\end{equation*}
Let $\widehat\Omega_0=\supp(\bm S_0^+)$ and
$\Omega_{0,*}=\supp(\bm S_{0,*})$. Define the missed, falsely selected, and
correctly selected supports by
\begin{equation}
  \mathcal A:=\Omega_{0,*}\setminus\widehat\Omega_0,
  \qquad
  \mathcal B:=\widehat\Omega_0\setminus\Omega_{0,*},
  \qquad
  \mathcal I:=\widehat\Omega_0\cap\Omega_{0,*}.
  \label{eq:det-capacity-greedy-ABI}
\end{equation}
We first establish the tail comparison
\begin{equation}
  \|\mathcal P_{\mathcal A}(\bm X)\|_{\mathrm F}^2
  \le
  \vartheta_0\|\mathcal P_{\mathcal B}(\bm X)\|_{\mathrm F}^2.
  \label{eq:det-capacity-greedy-tail-comparison}
\end{equation}

Classify each missed true entry by why it was excluded.  If it was examined
before termination, the scan had not yet reached $s_0$ acceptances, so its row
or column must have been saturated; place it in $\mathcal A_{\rm row}$ when
the row was saturated and otherwise in $\mathcal A_{\rm col}$.  If both were
saturated, use the row class.  If the entry was never examined, the scan had
already accepted $s_0$ entries; place it in $\mathcal A_{\rm card}$.  Hence
\begin{equation*}
  \mathcal A
  =\mathcal A_{\rm row}\cup
   \mathcal A_{\rm col}\cup
   \mathcal A_{\rm card}.
\end{equation*}
The union is disjoint.

Consider a saturated row $a$ receiving at least one missed true entry, and
write $\mathcal I_{a,\cdot}:=\{b\in[p]:(a,b)\in\mathcal I\}$.
The row contains $\gamma p-|\mathcal I_{a,\cdot}|$ false selections and at
most $\gamma_*p-|\mathcal I_{a,\cdot}|$ missed true entries.  Here
$\gamma p-|\mathcal I_{a,\cdot}|>0$.  Since every missed entry
assigned to this row has magnitude no larger than any false selection in
the row, their average squared magnitudes satisfy
\begin{equation*}
  \frac{
    \sum_{\substack{b\in[p]\\(a,b)\in\mathcal A_{\rm row}}}|X_{ab}|^2
  }{
    |\mathcal A_{\rm row}\cap(\{a\}\times[p])|
  }
  \le
  \frac{
    \sum_{\substack{b\in[p]\\(a,b)\in\mathcal B}}|X_{ab}|^2
  }{
    \gamma p-|\mathcal I_{a,\cdot}|
  }.
\end{equation*}
Then,
\begin{equation}
  \sum_{\substack{b\in[p]\\(a,b)\in\mathcal A_{\rm row}}}|X_{ab}|^2
  \le
  \frac{\gamma_*p-|\mathcal I_{a,\cdot}|}
       {\gamma p-|\mathcal I_{a,\cdot}|}
  \sum_{\substack{b\in[p]\\(a,b)\in\mathcal B}}|X_{ab}|^2
  \le
  \frac{\gamma_*}{\gamma}
  \sum_{\substack{b\in[p]\\(a,b)\in\mathcal B}}|X_{ab}|^2.
  \label{eq:det-capacity-greedy-row-charge}
\end{equation}
The last inequality uses $\gamma\ge\gamma_*$.  Summing over rows and
repeating the same argument for saturated columns yields
\begin{equation}
  \|\mathcal P_{\mathcal A_{\rm row}}(\bm X)\|_{\mathrm F}^2
  +\|\mathcal P_{\mathcal A_{\rm col}}(\bm X)\|_{\mathrm F}^2
  \le
  2\frac{\gamma_*}{\gamma}
  \|\mathcal P_{\mathcal B}(\bm X)\|_{\mathrm F}^2.
  \label{eq:det-capacity-greedy-column-charge}
\end{equation}

Next, if $\mathcal A_{\rm card}$ is empty, the desired bound   \eqref{eq:det-capacity-greedy-tail-comparison} is immediate.
Otherwise, the scan stopped after $s_0$ acceptances, so
$|\mathcal B|=s_0-|\mathcal I|$ and
$|\mathcal A_{\rm card}|\le s_{0,*}-|\mathcal I|$.  Since entries are scanned
in nonincreasing order of magnitude, every entry in $\mathcal A_{\rm card}$
has magnitude no larger than any entry in $\mathcal B$.  Their average
squared magnitudes therefore satisfy
\begin{equation*}
  \frac{
    \|\mathcal P_{\mathcal A_{\rm card}}(\bm X)\|_{\mathrm F}^2
  }{
    |\mathcal A_{\rm card}|
  }
  \le
  \frac{
    \|\mathcal P_{\mathcal B}(\bm X)\|_{\mathrm F}^2
  }{
    |\mathcal B|
  }.
\end{equation*}
Multiplying through gives
\begin{equation}
  \|\mathcal P_{\mathcal A_{\rm card}}(\bm X)\|_{\mathrm F}^2
  \le
  \frac{s_{0,*}-|\mathcal I|}{s_0-|\mathcal I|}
  \|\mathcal P_{\mathcal B}(\bm X)\|_{\mathrm F}^2
  \le
  \frac{s_{0,*}}{s_0}\|\mathcal P_{\mathcal B}(\bm X)\|_{\mathrm F}^2.
  \label{eq:det-capacity-greedy-cardinality-charge}
\end{equation}
Combining \eqref{eq:det-capacity-greedy-column-charge} and
\eqref{eq:det-capacity-greedy-cardinality-charge} proves
\eqref{eq:det-capacity-greedy-tail-comparison}.

Set $\bm H=\bm X-\bm S_{0,*}$.  Then
$\mathcal P_{\mathcal A}(\bm S_{0,*})
=\mathcal P_{\mathcal A}(\bm X)-\mathcal P_{\mathcal A}(\bm H)$ and
$\mathcal P_{\mathcal B}(\bm X)=\mathcal P_{\mathcal B}(\bm H)$.  Hence the
triangle inequality and \eqref{eq:det-capacity-greedy-tail-comparison} yield
\begin{equation*}
\begin{aligned}
  \|\mathcal P_{\mathcal A}(\bm S_{0,*})\|_{\mathrm F}
  &\le
  \|\mathcal P_{\mathcal A}(\bm H)\|_{\mathrm F}
  +\|\mathcal P_{\mathcal A}(\bm X)\|_{\mathrm F}
  \\
  &\le
  \|\mathcal P_{\mathcal A}(\bm H)\|_{\mathrm F}
  +\sqrt{\vartheta_0}\,
   \|\mathcal P_{\mathcal B}(\bm H)\|_{\mathrm F}.
\end{aligned}
\end{equation*}
Since $\mathcal A$, $\mathcal B$, and $\mathcal I$ are disjoint, we have
\begin{equation}
\begin{aligned}
  \|\bm S_0^+-\bm S_{0,*}\|_{\mathrm F}^2
  &={}
  \|\mathcal P_{\mathcal B}(\bm H)\|_{\mathrm F}^2
  +\|\mathcal P_{\mathcal I}(\bm H)\|_{\mathrm F}^2
  +\|\mathcal P_{\mathcal A}(\bm S_{0,*})\|_{\mathrm F}^2
  \\
  &\le{}
  \|\mathcal P_{\mathcal B}(\bm H)\|_{\mathrm F}^2
  +\|\mathcal P_{\mathcal I}(\bm H)\|_{\mathrm F}^2
  +\left(
    \|\mathcal P_{\mathcal A}(\bm H)\|_{\mathrm F}
    +\sqrt{\vartheta_0}\,
    \|\mathcal P_{\mathcal B}(\bm H)\|_{\mathrm F}
  \right)^2.
\end{aligned}
\end{equation}
For the second inequality above, the terms involving $\mathcal A$ and $\mathcal B$ form a quadratic form
whose coefficient matrix has largest eigenvalue
\begin{equation*}
  \lambda_{\max}\!\left(
  \begin{pmatrix}
    1 & \sqrt{\vartheta_0}\\
    \sqrt{\vartheta_0} & 1+\vartheta_0
  \end{pmatrix}
  \right)
  =
  1+\frac{\vartheta_0}{2}
  +\sqrt{\vartheta_0+\frac{\vartheta_0^2}{4}}
  =1+\rho_0.
\end{equation*}
Therefore,
\begin{equation*}
  \|\mathcal P_{\mathcal B}(\bm H)\|_{\mathrm F}^2
  +\left(
    \|\mathcal P_{\mathcal A}(\bm H)\|_{\mathrm F}
    +\sqrt{\vartheta_0}\,
     \|\mathcal P_{\mathcal B}(\bm H)\|_{\mathrm F}
  \right)^2\le
  (1+\rho_0)
  \left(
    \|\mathcal P_{\mathcal A}(\bm H)\|_{\mathrm F}^2
    +\|\mathcal P_{\mathcal B}(\bm H)\|_{\mathrm F}^2
  \right).
\end{equation*}
Adding $\|\mathcal P_{\mathcal I}(\bm H)\|_{\mathrm F}^2$ on both sides gives 
$$
  \|\bm S_0^+-\bm S_{0,*}\|_{\mathrm F}^2\le{}(1+\rho_0)
  \left(
    \|\mathcal P_{\mathcal A}(\bm H)\|_{\mathrm F}^2
    +\|\mathcal P_{\mathcal B}(\bm H)\|_{\mathrm F}^2
    +\|\mathcal P_{\mathcal I}(\bm H)\|_{\mathrm F}^2
  \right)=(1+\rho_0)\|\mathcal P_{\Omega_0^+}(\bm H)\|_{\mathrm F}^2.
$$ The last equality follows because $\mathcal A$,
$\mathcal B$, and $\mathcal I$ form a partition of $\Omega_0^+$.

Finally, the choices $s_0\ge C \beta_{\mathrm S}^2\alpha_{\mathrm S}^{-2}s_{0,*}$ and $\gamma\ge C\beta_{\mathrm S}^2\alpha_{\mathrm S}^{-2}\gamma_*$ imply
\begin{equation*}
  \vartheta_0
  \le
  \frac{3}{C}\left(\frac{\alpha_{\mathrm S}}{\beta_{\mathrm S}}\right)^2.
\end{equation*}
Since $\alpha_{\mathrm S}\le\beta_{\mathrm S}$, sufficiently large $C$ ensures
$\vartheta_0\le1$, and hence
\begin{equation*}
  \rho_0
  \le 2\sqrt{\vartheta_0}
  \le 2\sqrt{\frac{3}{C}}\frac{\alpha_{\mathrm S}}{\beta_{\mathrm S}}.
\end{equation*}
Choosing $C$ sufficiently large gives
$\rho_0\le c\alpha_{\mathrm S}/\beta_{\mathrm S}$.
\end{proof}

We next establish the stability of the weighted FCT projection
defined in \eqref{eq:FCT}.

\begin{lemma}[Stability of weighted FCT]
\label{lem:det-weighted-fcht}
Suppose the comparison collection satisfies
$(\supp(\bm S_{i,*}))_i\in\mathcal D(\{s_{i,*}\},m_*)$.
In the approximation extension, this lemma is applied with
$\bm S_{i,*}$ replaced by $\bm H_{i,*}$.
Suppose the fitted levels satisfy $m\ge m_*$ and 
$s_i\ge s_{i,*}$, for $i\in[N]$.
For arbitrary $\widetilde{\bm S}_i\in\mathbb R^{p\times p}$, let $\{\bm S_i^+\}_{i=1}^N$ denote their weighted FCT output in
\eqref{eq:FCT}, and define
  $\Omega_i^+
  :=\supp(\bm S_i^+)\cup\supp(\bm S_{i,*})$.
Set
\begin{equation}
  \rho_L
  :=2\sqrt{
    \max_{i\in[N]}\frac{s_{i,*}}{s_i}
    +\frac{m_*}{m}
  },
  \label{eq:det-fcht-amplification-factor}
\end{equation}
Then
\begin{equation}
  \sum_{i=1}^Nw_i
  \|\bm S_i^+-\bm S_{i,*}\|_{\mathrm F}^2
  \le
  (1+\rho_L)
  \sum_{i=1}^Nw_i
  \|\mathcal P_{\Omega_i^+}
  (\widetilde{\bm S}_i-\bm S_{i,*})\|_{\mathrm F}^2.
  \label{eq:det-weighted-fcht-amplification}
\end{equation}
Furthermore, for any sufficiently small absolute constant $c>0$, taking
$C>0$ sufficiently large and choosing
$s_i\ge C\beta_\mathrm{P}^2\alpha_\mathrm{P}^{-2}s_{i,*}$ and
$m\ge C\beta_\mathrm{P}^2\alpha_\mathrm{P}^{-2}m_*$ ensures
\begin{equation*}
  \rho_L\le c\alpha_\mathrm{P}\beta_\mathrm{P}^{-1}.
\end{equation*}
\end{lemma}

\begin{proof}
For brevity, set
\begin{equation*}
  \vartheta_L
  :=\max_{i\in[N]}\frac{s_{i,*}}{s_i}+\frac{m_*}{m},
  \quad
  \widehat\Omega_i:=\supp(\bm S_i^+),
  \quad
  \Omega_{i,*}:=\supp(\bm S_{i,*}).
\end{equation*}
The fitted-level conditions give $0\le\vartheta_L\le2$. For client $i$, define the missed, falsely selected, and correctly selected supports by
\begin{equation*}
  \mathcal A_i:=\Omega_{i,*}\setminus\widehat\Omega_i,
  \qquad
  \mathcal B_i:=\widehat\Omega_i\setminus\Omega_{i,*},
  \qquad
  \mathcal I_i:=\widehat\Omega_i\cap\Omega_{i,*}.
\end{equation*}
For each coordinate $(a,b)$, also let
\begin{equation*}
  \mathcal A^{ab}:=\{i:(a,b)\in\mathcal A_i\},
  \qquad
  \mathcal B^{ab}:=\{i:(a,b)\in\mathcal B_i\},
  \qquad
  \mathcal I^{ab}:=\{i:(a,b)\in\mathcal I_i\}.
\end{equation*}

Define the weighted
entrywise scores
\begin{equation*}
  \widetilde q_{i,ab}:=w_i|(\widetilde{\bm S}_i)_{ab}|^2.
\end{equation*}
For each client $i$, vectorize the weighted score matrix
$(\widetilde q_{i,ab})_{a,b=1}^p$ and stack the resulting vectors as columns
to form a $p^2\times N$ score matrix.  Its rows are indexed by coordinates
$(a,b)$ and its columns by clients $i$; we denote the entry in row $(a,b)$ and column $i$ by $(a,b,i)$. Then, FCT selects a feasible set of
entries maximizing the total retained score, with at most $m$ selected
entries in each row and at most $s_i$ in column $i$. Thus, $\mathcal A_i$, $\mathcal B_i$, and $\mathcal I_i$ are column slices
of the score table, whereas $\mathcal A^{ab}$, $\mathcal B^{ab}$, and
$\mathcal I^{ab}$ are the corresponding row slices. In view of this, we first establish the weighted tail comparison
\begin{equation}
  \sum_{i=1}^Nw_i
  \|\mathcal P_{\mathcal A_i}(\widetilde{\bm S}_i)\|_{\mathrm F}^2
  \le
  \vartheta_L
  \sum_{i=1}^Nw_i
  \|\mathcal P_{\mathcal B_i}(\widetilde{\bm S}_i)\|_{\mathrm F}^2.
  \label{eq:det-weighted-fcht-tail-dominance}
\end{equation}

Fix $(a,b)\in\mathcal A_i$.  If $\widetilde q_{i,ab}>0$, then
\begin{equation*}
  |\widehat\Omega_i|=s_i
  \quad\text{or}\quad
  \sum_{j=1}^N\mathbf 1\{(a,b)\in\widehat\Omega_j\}=m;
\end{equation*}
otherwise, adding $(a,b,i)$ would strictly increase the retained score.
Whenever the first equality holds, the column false-selection set is
nonempty; whenever the second holds, the row false-selection set is nonempty.
Indeed,
\begin{equation*}
  |\mathcal B_i|
  =s_i-|\mathcal I_i|
  \ge s_i-s_{i,*}+1>0,
  \qquad
  |\mathcal B^{ab}|
  =m-|\mathcal I^{ab}|
  \ge m-m_*+1>0,
\end{equation*}
respectively.
If only the column cap is reached, every false selection in column $i$ has
score at least $\widetilde q_{i,ab}$; otherwise, replacing it by $(a,b,i)$
would give a feasible support with larger total score.  Thus,
$\widetilde q_{i,ab}\le\widetilde q_{i,cd}$ for every
$(c,d)\in\mathcal B_i$.  If only the row cap is reached, the same argument
gives $\widetilde q_{i,ab}\le\widetilde q_{j,ab}$ for every
$j\in\mathcal B^{ab}$.  If both caps are reached, the sum of the scores of
any false selection from column $i$ and any false selection from row $(a,b)$
is at least $\widetilde q_{i,ab}$.  Indeed, the two selected cells are
distinct because $(a,b,i)$ is missed; removing them and inserting $(a,b,i)$
is feasible.  Hence, for every $(c,d)\in\mathcal B_i$ and
$j\in\mathcal B^{ab}$, optimality gives
\begin{equation*}
  \widetilde q_{i,ab}
  \le \widetilde q_{i,cd}+\widetilde q_{j,ab}.
\end{equation*}
Averaging the two single-cap comparisons gives, respectively,
\begin{equation*}
  \widetilde q_{i,ab}
  \le
  \frac1{|\mathcal B_i|}
  \sum_{(c,d)\in\mathcal B_i}\widetilde q_{i,cd}
  =:\bar q_i,
  \qquad
  \widetilde q_{i,ab}
  \le
  \frac1{|\mathcal B^{ab}|}
  \sum_{j\in\mathcal B^{ab}}\widetilde q_{j,ab}
  =:\bar q^{ab}.
\end{equation*}
When both caps are reached, averaging over
$\mathcal B_i\times\mathcal B^{ab}$ gives
\begin{equation*}
\begin{aligned}
  \widetilde q_{i,ab}
  &\le
  \frac1{|\mathcal B_i||\mathcal B^{ab}|}
  \sum_{(c,d)\in\mathcal B_i}
  \sum_{j\in\mathcal B^{ab}}
  \left(\widetilde q_{i,cd}+\widetilde q_{j,ab}\right)
  \\
  &=\bar q_i+\bar q^{ab}.
\end{aligned}
\end{equation*}
Set $\bar q_i=0$ when $|\widehat\Omega_i|<s_i$ or
$\mathcal A_i=\varnothing$, and set $\bar q^{ab}=0$ when
$\sum_{j=1}^N\mathbf 1\{(a,b)\in\widehat\Omega_j\}<m$ or
$\mathcal A^{ab}=\varnothing$.  Thus, all three cases can be written
uniformly as
\begin{equation}
  \widetilde q_{i,ab}\le\bar q_i+\bar q^{ab},
  \qquad (a,b)\in\mathcal A_i.
  \label{eq:det-weighted-fcht-pointwise}
\end{equation}
The same inequality is immediate when $\widetilde q_{i,ab}=0$.

Whenever $\bar q_i\ne0$ or $\bar q^{ab}\ne0$, respectively,
\begin{equation*}
  \frac{|\mathcal A_i|}{|\mathcal B_i|}
  \le
  \frac{s_{i,*}-|\mathcal I_i|}{s_i-|\mathcal I_i|}
  \le\frac{s_{i,*}}{s_i},
  \qquad
  \frac{|\mathcal A^{ab}|}{|\mathcal B^{ab}|}
  \le
  \frac{m_*-|\mathcal I^{ab}|}{m-|\mathcal I^{ab}|}
  \le\frac{m_*}{m}.
\end{equation*}
Summing \eqref{eq:det-weighted-fcht-pointwise} over all clients and $(a,b)\in\mathcal A_i$ yields
\begin{equation*}
\begin{aligned}
  \sum_{i=1}^N\sum_{(a,b)\in\mathcal A_i}\widetilde q_{i,ab}
  &\le
  \sum_{i=1}^N|\mathcal A_i|\bar q_i
  +\sum_{a,b=1}^p|\mathcal A^{ab}|\bar q^{ab}
  \\
  &\le
  \sum_{i=1}^N\frac{s_{i,*}}{s_i}
  \sum_{(a,b)\in\mathcal B_i}\widetilde q_{i,ab}
  +\frac{m_*}{m}
  \sum_{a,b=1}^p\sum_{j\in\mathcal B^{ab}}\widetilde q_{j,ab}
  \\
  &=
  \sum_{i=1}^N\frac{s_{i,*}}{s_i}
  \sum_{(a,b)\in\mathcal B_i}\widetilde q_{i,ab}
  +\frac{m_*}{m}
  \sum_{i=1}^N\sum_{(a,b)\in\mathcal B_i}\widetilde q_{i,ab}
  \\
  &\le
  \left(
    \max_{i\in[N]}\frac{s_{i,*}}{s_i}+\frac{m_*}{m}
  \right)
  \sum_{i=1}^N\sum_{(a,b)\in\mathcal B_i}\widetilde q_{i,ab}
  \\
  &=
  \vartheta_L
  \sum_{i=1}^N\sum_{(a,b)\in\mathcal B_i}\widetilde q_{i,ab}.
\end{aligned}
\end{equation*}
Recalling the definition of $\widetilde q_{i,ab}$ proves
\eqref{eq:det-weighted-fcht-tail-dominance}.

Equip
$(\mathbb R^{p\times p})^N$ with the weighted inner product and norm
\begin{equation*}
  \left\langle(\bm X_i),(\bm Y_i)\right\rangle_w
  :=\sum_{i=1}^Nw_i\inner{\bm X_i}{\bm Y_i},
  \qquad
  \left\|(\bm X_i)\right\|_w
  :=\left\{\sum_{i=1}^Nw_i\|\bm X_i\|_{\mathrm F}^2\right\}^{1/2}.
\end{equation*}
Setting $\bm H_i=\widetilde{\bm S}_i-\bm S_{i,*}$, we have
$\mathcal P_{\mathcal A_i}(\bm S_{i,*})
=\mathcal P_{\mathcal A_i}(\widetilde{\bm S}_i)
-\mathcal P_{\mathcal A_i}(\bm H_i)$ and
$\mathcal P_{\mathcal B_i}(\widetilde{\bm S}_i)
=\mathcal P_{\mathcal B_i}(\bm H_i)$. The triangle inequality  gives
\begin{equation*}
  \left\{
    \sum_{i=1}^Nw_i
    \|\mathcal P_{\mathcal A_i}(\bm S_{i,*})\|_{\mathrm F}^2
  \right\}^{1/2}
  \le
  \left\{
    \sum_{i=1}^Nw_i
    \|\mathcal P_{\mathcal A_i}(\bm H_i)\|_{\mathrm F}^2
  \right\}^{1/2}
  +\sqrt{\vartheta_L}
  \left\{
    \sum_{i=1}^Nw_i
    \|\mathcal P_{\mathcal B_i}(\bm H_i)\|_{\mathrm F}^2
  \right\}^{1/2}.
\end{equation*}
Moreover, since $\mathcal A_i$, $\mathcal B_i$, and $\mathcal I_i$ form a
partition of $\Omega_i^+$, the projection error decomposes as
\begin{equation*}
\begin{aligned}
  \bm S_i^+-\bm S_{i,*}
  &=
  \mathcal P_{\widehat\Omega_i}(\widetilde{\bm S}_i)
  -\mathcal P_{\Omega_{i,*}}(\bm S_{i,*})
  \\
  &=
  \mathcal P_{\mathcal B_i}(\widetilde{\bm S}_i)
  +\mathcal P_{\mathcal I_i}
    (\widetilde{\bm S}_i-\bm S_{i,*})
  -\mathcal P_{\mathcal A_i}(\bm S_{i,*})
  \\
  &=
  \mathcal P_{\mathcal B_i}(\bm H_i)
  +\mathcal P_{\mathcal I_i}(\bm H_i)
  -\mathcal P_{\mathcal A_i}(\bm S_{i,*}).
\end{aligned}
\end{equation*}
Consequently, we have
\begin{equation*}
\begin{aligned}
  \sum_{i=1}^Nw_i\|\bm S_i^+-\bm S_{i,*}\|_{\mathrm F}^2
  &=
  \sum_{i=1}^Nw_i\left\{
    \|\mathcal P_{\mathcal B_i}(\bm H_i)\|_{\mathrm F}^2
    +\|\mathcal P_{\mathcal I_i}(\bm H_i)\|_{\mathrm F}^2
    +\|\mathcal P_{\mathcal A_i}(\bm S_{i,*})\|_{\mathrm F}^2
  \right\}
  \\
  &\le
  \sum_{i=1}^Nw_i\|\mathcal P_{\mathcal I_i}(\bm H_i)\|_{\mathrm F}^2
  +\sum_{i=1}^Nw_i\|\mathcal P_{\mathcal B_i}(\bm H_i)\|_{\mathrm F}^2
  \\
  &\quad+
  \left[
    \left\{\sum_{i=1}^Nw_i
      \|\mathcal P_{\mathcal A_i}(\bm H_i)\|_{\mathrm F}^2
    \right\}^{1/2}
    +\sqrt{\vartheta_L}
    \left\{\sum_{i=1}^Nw_i
      \|\mathcal P_{\mathcal B_i}(\bm H_i)\|_{\mathrm F}^2
    \right\}^{1/2}
  \right]^2.
\end{aligned}
\end{equation*}
The last two terms involving the $\mathcal A_i$ and $\mathcal B_i$
components form the same $2\times2$ quadratic form as in
Lemma~\ref{lem:G_stability}, now in the weighted product norm.  Its largest
eigenvalue satisfies
\begin{equation*}
  \lambda_{\max}\!\left(
  \begin{pmatrix}
    1 & \sqrt{\vartheta_L}\\
    \sqrt{\vartheta_L} & 1+\vartheta_L
  \end{pmatrix}
  \right)
  =
  1+\frac{\vartheta_L}{2}
  +\sqrt{\vartheta_L+\frac{\vartheta_L^2}{4}}
  \le
  1+2\sqrt{\vartheta_L}
  =1+\rho_L,
\end{equation*}
where the inequality uses $0\le\vartheta_L\le2$. Therefore,
\begin{equation*}
  \sum_{i=1}^Nw_i\|\bm S_i^+-\bm S_{i,*}\|_{\mathrm F}^2
  \le
  (1+\rho_L)
  \sum_{i=1}^Nw_i
  \|\mathcal P_{\Omega_i^+}(\bm H_i)\|_{\mathrm F}^2.
\end{equation*}

Finally, the conditions
$s_i\ge C\beta_\mathrm{P}^2\alpha_\mathrm{P}^{-2}s_{i,*}$ and
$m\ge C\beta_\mathrm{P}^2\alpha_\mathrm{P}^{-2}m_*$ imply
\begin{equation*}
  \max_{i\in[N]}\frac{s_{i,*}}{s_i}
  \le
  \frac1C
  \left(\frac{\alpha_\mathrm{P}}{\beta_\mathrm{P}}\right)^2,
  \qquad
  \frac{m_*}{m}
  \le
  \frac1C
  \left(\frac{\alpha_\mathrm{P}}{\beta_\mathrm{P}}\right)^2.
\end{equation*}
Therefore,
\begin{equation*}
  \vartheta_L
  \le
  \frac{2}{C}
  \left(\frac{\alpha_\mathrm{P}}{\beta_\mathrm{P}}\right)^2
  \quad\text{and}\quad
  \rho_L
  =2\sqrt{\vartheta_L}
  \le2\sqrt{\frac{2}{C}}
  \frac{\alpha_\mathrm{P}}{\beta_\mathrm{P}}
  \le c\frac{\alpha_\mathrm{P}}{\beta_\mathrm{P}}
\end{equation*}
for sufficiently large $C$.
\end{proof}

\begin{lemma}[Working-cap interaction bounds]
\label{lem:working-cap-interactions}
For an admissible integer cap $m\ge1$ and
$a\in\{\mathrm{LR},\mathrm S\}$,
\begin{equation}
\label{eq:cap-energy-comparison}
 \sup_{\bm X\in\mathcal X_a}
 \sup_{\Omega\in\mathcal D(\{s_i^\dagger\},\min\{3m,N\})}
 \sum_i\|\mathcal P_{\Omega_i}(\bm X\bm E_i)\|_{\mathrm F}^2
 \le3m\Lambda_a(m)^2.
\end{equation}
\end{lemma}
\begin{proof}
Fix $\bm X$ and $\Omega$ in these classes. At each coordinate,
partition the selected clients into at most three groups of size at
most $m$. This gives disjoint support collections
$\Omega^{(b)}\in\mathcal D(\{s_i^\dagger\},m)$, $b=1,2,3$,
since $\Omega_i^{(b)}\subseteq\Omega_i$ preserves each client budget.
Frobenius duality and \eqref{eq:directional_heterogeneity_levels} give
\[
 \sum_i\|\mathcal P_{\Omega_i}(\bm X\bm E_i)\|_{\mathrm F}^2
 =\sum_{b=1}^3\sum_i
   \|\mathcal P_{\Omega_i^{(b)}}(\bm X\bm E_i)\|_{\mathrm F}^2
 \le3m\Lambda_a(m)^2.
\]
Taking the suprema proves \eqref{eq:cap-energy-comparison}.
\end{proof}

\subsection{Bounds for approximation residuals}
\label{app:lq-deterministic}

We bound the approximation residuals and the inner products used to
analyze the shared updates. Recall $\overline s_w=\sum_iw_i s_{i,*}$ and
the approximation bound $\mathcal A_q(m)$ defined in
\eqref{eq:app-unified-approximation-envelope}.
Fix an admissible integer cap $m\ge1$ and use the constant $C$ in that
definition. Also recall
\begin{equation}
 \label{eq:note-comparison-frequency}
 \underline m:=\lfloor m/(C\kappa_{\mathrm P}^2)\rfloor,\qquad
 \bm H_{i,*}:=\bm S_{i,*}^{[\underline m]},\qquad
 \bm{\mathfrak R}_i:=\bm S_{i,*}-\bm H_{i,*},
\end{equation}
where $\bm S_{i,*}^{[\underline m]}$ retains the $\underline m$ largest
magnitudes across clients at each coordinate, with fixed tie-breaking.
Write $A_R:=\sum_iw_i\|\bm{\mathfrak R}_i\|_{\mathrm F}^2$.
For $0<q\le1$, denote
\begin{equation}
 \label{eq:note-delta-profile}
 \delta_q:=\left(\frac{m_{q,*}}{\underline m+1}\right)^{1/q},
 \qquad A_{\mathrm{ref}}:=\delta_q\sum_iw_i\|\bm{\mathfrak R}_i\|_1.
\end{equation}
For $q=0$, assume $m\ge C\kappa_{\mathrm P}^2m_*$ and set $A_{\mathrm{ref}}=0$.

\begin{lemma}[Approximation residual bounds]
\label{lem:lq-tail}
Under Assumption~\ref{assump:approx_personalization_diversity},
the comparison supports belong to
$\mathcal D(\{s_{i,*}\},\underline m)$.
For $0<q\le1$, $|(\bm{\mathfrak R}_i)_{jk}|\le\delta_q$ for all $i,j,k$,
and
\begin{equation}
 \label{eq:note-tail-energy}
 A_R\le A_{\mathrm{ref}}
 \le\overline s_w\delta_q^2
 \le\mathcal A_q(m).
\end{equation}
For $q=0$, $\bm H_{i,*}=\bm S_{i,*}$ and
$A_R=A_{\mathrm{ref}}=\mathcal A_0(m)=0$.
\end{lemma}
\begin{proof}
The support claim follows from retaining at most $\underline m$
entries per coordinate and only deleting entries from each client.
For $q=0$, $m\ge C\kappa_{\mathrm P}^2m_*$ gives $\underline m\ge m_*$,
so every nonzero entry is retained.
For $0<q\le1$, an omitted magnitude exceeding $\delta_q$ would imply
that the $\underline m+1$ largest magnitudes at that coordinate have
total $q$th power greater than $m_{q,*}$.
Hence $|(\bm{\mathfrak R}_i)_{jk}|\le\delta_q$ and
\[
 |(\bm{\mathfrak R}_i)_{jk}|^2
 \le\delta_q|(\bm{\mathfrak R}_i)_{jk}|.
\]
Summing with weights gives $A_R\le A_{\mathrm{ref}}$.
Since client $i$ has at most $s_{i,*}$ residual entries,
$A_{\mathrm{ref}}\le\overline s_w\delta_q^2$.
Finally, $\underline m+1>m/(C\kappa_{\mathrm P}^2)$ gives the last inequality in
\eqref{eq:note-tail-energy}.
\end{proof}

\begin{lemma}[Bound for residual inner products]
\label{lem:lq-fractional-transfer}
Under the assumptions of Lemma~\ref{lem:lq-tail}, for arbitrary matrices
$\bm C_i\in\mathbb R^{p\times p}$,
\begin{equation}
 \label{eq:note-deterministic-transfer}
 \left|\sum_iw_i\langle\bm{\mathfrak R}_i,\bm C_i\rangle\right|^2
 \le A_{\mathrm{ref}}
 \sup_{\Omega\in\mathcal D(\{s_{i,*}\},\underline m+1)}
        \sum_iw_i\|\mathcal P_{\Omega_i}(\bm C_i)\|_{\mathrm F}^2.
\end{equation}
\end{lemma}
\begin{proof}
For $q=0$, or for $0<q\le1$ with $m_{q,*}=0$, the residuals vanish
and the inequality is immediate. For $0<q\le1$ and $m_{q,*}>0$, set
$z_{i,jk}=|(\bm{\mathfrak R}_i)_{jk}|/\delta_q$.
Lemma~\ref{lem:lq-tail} gives $0\le z\le1$ and
$\sum_{j,k}z_{i,jk}\le s_{i,*}$.
Since $z_{i,jk}\le z_{i,jk}^q$,
\[
 \sum_i z_{i,jk}
 \le\delta_q^{-q}\sum_i|(\bm{\mathfrak R}_i)_{jk}|^q
 \le\delta_q^{-q}m_{q,*}=\underline m+1.
\]
Moreover, $m\ge1$ is an integer and $C\kappa_{\mathrm P}^2>1$, so
$\underline m+1\le m$.
Weighted Cauchy--Schwarz gives
\begin{align*}
 \left|\sum_iw_i\langle\bm{\mathfrak R}_i,\bm C_i\rangle\right|^2
 &\le\delta_q^2
   \left(\sum_{i,j,k}w_i z_{i,jk}\right)
   \left(\sum_{i,j,k}w_i z_{i,jk}|C_{i,jk}|^2\right)\\
 &=A_{\mathrm{ref}}\sum_{i,j,k}w_i z_{i,jk}|C_{i,jk}|^2.
\end{align*}
The FCT support problem in Appendix~\ref{subsec:app-fct-computation}
maximizes
$\sum_iw_i\|\mathcal P_{\Omega_i}(\bm C_i)\|_{\mathrm F}^2$
over $\Omega\in\mathcal D(\{s_{i,*}\},\underline m+1)$.
By the integrality argument there, $z$ is a convex combination of
such support indicators, so its objective value cannot exceed this
maximum. Hence
\[
 \sum_{i,j,k}w_i z_{i,jk}|C_{i,jk}|^2
 \le\sup_{\Omega\in\mathcal D(\{s_{i,*}\},\underline m+1)}
       \sum_iw_i\|\mathcal P_{\Omega_i}(\bm C_i)\|_{\mathrm F}^2.
\]
Combining these bounds proves \eqref{eq:note-deterministic-transfer}.
\end{proof}

Taking $\bm C_i=\bm X\widehat{\bbm\Gamma}_i$ gives the residual bounds
used in Step~2.5 of Appendices~\ref{subsec:app-lowrank-update}
and~\ref{subsec:app-global-sparse-update}.


\section{Simulation design and implementation}
\label{app:simulation-details}

In this appendix, we detail the data-generating processes and fitting
settings for the simulations in the main article. Data follow
the model in Section~\ref{subsec:2.1},
with \(p=40\) and \(T_i=T\).
Except in the covariance design below, innovations are independent
\(N(\bm0,\bm I_{40})\). Results are averaged over 200 independent replications.
\textsc{PerFeCT-VAR} uses the true rank, 200 rounds,
\(\eta_{\mathrm{LR}}=\eta_{\mathrm S}=0.1\), \(\eta_{\mathrm P}=0.3\), and
\(\gamma =3/p\). The initialization follows Section~\ref{sec:algorithm_initialization}.
With
\(\widehat\phi_i=\|\widehat{\bbm\Gamma}_i\|_{\mathrm{op}}\),
the clientwise convex low-rank plus sparse programs use \(\varpi=p\) and
\begin{equation}
\label{eq:sim-initialization-penalties}
 \lambda_{B,i}=0.5\sqrt{\frac{\widehat\phi_i(p+\log N)}{T}},
 \qquad
 \lambda_{C,i}=0.5\sqrt{\frac{\widehat\phi_i\log(pN)}{T}},
 \qquad i\in [N].
\end{equation}
They are solved from zero using the AFNSL algorithm of \citet{basu2019low} with at most 100 iterations.

\subsection{Estimation accuracy and rate separation}
\label{subsec:sim-exp1-details}

For the true parameters, set
\(\bm B_*=\bm U_*\operatorname{diag}(1,0.5)\bm V_*^\top\), where \(\bm U_*, \bm V_*\in \mathbb R^{40\times2}\) are random matrices with orthonormal columns. We compute the incoherence parameter $\mu$ and discard draws with \(\mu>5\).
For the support of \(\bm S_{0,*}\), we randomly select
\(s_{0,*}=36\) distinct coordinates, with at most
\(\gamma_*p=2\) in each row and column. For the supports of \(\{\bm S_{i,*}\}_{i=1}^N\), we sample
\(9N\) distinct coordinates uniformly from all \(p^2\) coordinates.
Starting with a budget of 18 entries per client, we sequentially
assign each coordinate to the two clients with the largest remaining
budgets. This gives \(s_{i,*}=18\) for every client and frequency
\(m_*=2\) at each selected coordinate. All nonzero entries have independent random signs and independent magnitudes drawn from
\(U(0.1,0.2)\) for \(\bm S_{0,*}\) and
\(U(0.2,0.4)\) for \(\bm S_{i,*}\).
We retain only draws satisfying
\(\max_i\rho(\bm A_{i,*})<0.8\) at every client count.

We vary
\[
\begin{aligned}
T&\in\{100,200,400,800,1200,1600,2000\} &&\text{at }N=30,\\
\text{and}\quad N&\in\{10,20,30,40,50,60\} &&\text{at }T=800.
\end{aligned}
\]
Within each replication, keep \(\bm B_*\) and \(\bm S_{0,*}\) fixed
across \(N\) and regenerate $\{\bm S_{i,*}\}_{i=1}^N$ for each \(N\).
The samples across \(T\) are terminal segments of one trajectory
generated at \(T=2000\); samples across \(N\) are generated separately.
\textsc{PerFeCT-VAR} uses \((s_0,s_i,m)=(40,20,3)\).

Local LRS and Pooled LRS use AFNSL with at most 200 and 600 iterations, respectively.
Local LRS replaces \(p+\log N\) and \(\log(pN)\) in
\eqref{eq:sim-initialization-penalties} by \(p\) and \(\log p\).
Pooled LRS additionally replaces \(\widehat\phi_i\) by
\(\widehat\phi_w=\|N^{-1}\sum_i\widehat{\bbm\Gamma}_i\|_{\mathrm{op}}\)
and \(T\) by \(NT\).
Ditto solves its personalized quadratic problems exactly with
personalization coefficient \(1.5\).

\subsection{Personalization diversity and covariance discrepancy}
\label{subsec:sim-exp2-details}

Both designs use \(N=30\) and \(T=800\).

\paragraph*{(a) Personalization diversity}
We fix \(s_{i,*}=6\) for every client and vary only the arrangement of
personalized supports within the top-left \(15\times15\) block to obtain $\{\bm S_{i,*}\}$ with
increasing \(m_*\), corresponding to progressively weaker personalization diversity.
For the shared components, set
\(\bm B_*=0.5\bm u\bm u^\top\), where
\(u_j=15^{-1/2}\) for \(j\le15\) and zero otherwise.
Generate \(\bm S_{0,*}\) as in Appendix~\ref{subsec:sim-exp1-details}, with \(s_{0,*}=18\).

For each \(m_*\in\{1,2,3,4,6,9,12\}\), we randomly select the \(180/m_*\) coordinates in the top-left \(15\times15\) block, where \(180=30\times6\) is the total number of personalized entries. Then, we assign these coordinates to clients using the rule in Appendix~\ref{subsec:sim-exp1-details},
with exactly 6 nonzero entries per client and \(m_*\) clients per coordinate. For each selected coordinate, draw one random sign shared by all
clients using it. Draw the magnitude of each nonzero entry
independently from \(U(0.1,0.2)\). We also require \(\max_i\|\bm A_{i,*}\|_{\mathrm{op}}<0.8\). For \textsc{PerFeCT-VAR}, we use fitting levels \((s_0,s_i,m)=(20,8,m_*+1)\).

\paragraph*{(b) Directional covariance discrepancy measures}
To isolate the effect of the directional covariance discrepancy measures, we construct a data-generating process with increasing
\(\Lambda_{\mathrm{LR}}\) and \(\Lambda_{\mathrm{S}}\), while fixing all
curvature and scale coefficients
$\alpha_{\mathrm{LR}}$, $\beta_{\mathrm{LR}}$, $\alpha_{\mathrm S}$, $\beta_{\mathrm S}$, $\alpha_\mathrm{P}$, $\beta_\mathrm{P}$, $\phi_{\min}$, $\phi_{\max}$, $\tau$.

The idea of this construction is to vary directional covariance discrepancy by changing the number of client
pairs whose $\bm E_i$ align with the global components. We first generate \(\bm B_*\) and
\(\bm S_{0,*}\), choosing a sparse right singular vector for \(\bm B_*\). We then divide the clients into two equal groups and pair them across
groups, constructing opposite \(\bm S_{i,*}\) and \(\bm E_i\)
within each pair. The covariance deviations then cancel and preserve the
prespecified average covariance \(\overline{\bbm\Gamma}\).
We select a subset of client pairs and design their \(\bm E_i\)
so that \(\bm S_{i,*}\bm E_i\) aligns with \(\bm B_*\).
Since the right singular vector of \(\bm B_*\) is sparse, these
contributions can also accumulate in the shared sparse update.
Keeping all true components fixed, we then increase the number of selected pairs, thereby increasing \(\Lambda_{\mathrm{LR}}\) and \(\Lambda_{\mathrm{S}}\).

Specifically, we construct \(\bm B_*=0.3\bm u\bm v^\top\), where
\(\bm u\) has value \(1/\sqrt{18}\) at 18 randomly selected positions
among its first 39 coordinates, and \(\bm v\) has value \(1/\sqrt{8}\)
at 8 randomly selected positions among its first 38 coordinates,
with all remaining entries zero. For $\bm S_{0,*}$, randomly generate $s_{0,*}=18$ entries in the first 39 rows with $\gamma_*=2/p$,
independent random signs, and \(U(0.03,0.07)\)
magnitudes. Next, before generating \(\bm S_{i,*}\), randomly select 15 distinct positions
among the first 38 entries of \(\bm v\), restricted to positions
outside \(\supp(\bm v)\).
For \(j=1,\ldots,15\), let \(\bm z_j\in\mathbb R^p\) have value 1
at the \(j\)th selected position and 0 elsewhere.
Then, we randomly pair the clients and, after relabeling, set
\[
 \bm S_{2j-1,*}=-\bm S_{2j,*}=0.4\bm u\bm z_j^\top,
 \qquad j=1,\ldots,15.
\]
Thus \(s_{i,*}=18\) and \(m_*=2\). Accept draws with
\(\max_i\|\bm A_{i,*}\|_{\mathrm{op}}\le0.6\).

We consider four heterogeneity levels. Each level uses $k$ over 15 pairs of the relabeled clients above, to align $\bm E_i$ with the true parameters. The zero level uses \(k=0\) and \(\bm E_i=\bm0\). For low, medium, and high, let \(k=5,10,15\), respectively, and set
\begin{equation}
\label{eq:sim-aligned-covariance-deviations}
 \bm E_{2j-1}=-\bm E_{2j}
 =0.45(\bm b_j\bm z_j^\top+\bm z_j\bm b_j^\top),
 \qquad
 \bm b_j=
 \begin{cases}
  \bm v, & j\le k,\\
  \bm z_{15}, & j=k+1,\\
  \bm z_{j-1}, & j>k+1,
 \end{cases}
\end{equation}
Finally, set $\overline{\bbm\Gamma}=\operatorname{diag}(1,\ldots,1,0.55,1.45)$ and define
\[
 \bbm\Gamma_i=\overline{\bbm\Gamma}+\bm E_i,
 \qquad
 \bbm\Sigma_{\epsilon,i}
 =\bbm\Gamma_i-\bm A_{i,*}\bbm\Gamma_i\bm A_{i,*}^\top.
\]
We draw initial states from
\(N(\bm0,\bbm\Gamma_i)\) and independent innovations from
\(N(\bm0,\bbm\Sigma_{\epsilon,i})\) to generate data.
We fit \((s_0,s_i,m)=(20,20,3)\).

We compute $\Lambda_{\mathrm{LR}}$ and $\Lambda_{\mathrm{S}}$ directly from
\eqref{eq:directional_heterogeneity_levels}. When \(k=0\), all
\(\bm E_i\) are zero, so $\Lambda_{\mathrm{LR}}=\Lambda_{\mathrm{S}}=0$. In the following, consider \(k\in\{5,10,15\}\).

For the low-rank component, in \eqref{eq:directional_heterogeneity_levels}, let $\bm X$ be the normalized true shared direction and $\bm H_i$ be the normalized personalized directions from the
\(2k\) aligned clients:
\[
 \bm X=\frac{\bm B_*}{0.3}=\bm u\bm v^\top,
 \qquad
 \bm H_i=
 \begin{cases}
  \bm S_{i,*}/(0.4\sqrt{2k}), & i\le2k,\\
  \bm0, & i>2k.
 \end{cases}
\]
Then, these directions are feasible, i.e., \(\bm X\in\mathcal X_{\mathrm{LR}}\) and $(\bm H_1,...,\bm H_N) \in \mathcal{H}_{\mathrm{P}}$. Substituting them into
\eqref{eq:directional_heterogeneity_levels} gives
\[
 \Lambda_{\mathrm{LR}}
 \ge\frac{1}{\sqrt{m}}
 \left|\sum_{i=1}^{N}\langle\bm H_i,\bm X\bm E_i\rangle\right|=\frac{2}{\sqrt{m}}\sum_{j=1}^{k}
 \left\langle\frac{\bm u\bm z_j^\top}{\sqrt{2k}},
 0.45\bm u\bm z_j^\top\right\rangle
 =0.45\sqrt{\frac{2k}{m}}.
\]
The factor 2 comes from the equal contributions of the two clients
in each aligned pair.

For the shared sparse component \(\bm S_0\), take any
\(l\in\supp(\bm u)\) and choose
\[
 \bm X=\bm e_l\bm v^\top,
 \qquad
 \bm H_{2j-1}=-\bm H_{2j}
 =\begin{cases}
  \bm e_l\bm z_j^\top/\sqrt{2k}, & j\le k,\\
  \bm0, & j>k.
 \end{cases}
\]
Here \(\bm e_l\in\mathbb R^p\) has value 1 at position \(l\)
and 0 elsewhere. These are
the normalized \(l\)th-row of $\bm B_*$ and $\bm S_{i,*}$.
Since \(\bm v\) has eight nonzeros, \(\bm X\in\mathcal X_\mathrm{S}\) under \(s_0^\dagger=58\) and
\(\gamma^\dagger =8/p\). Meanwhile, \((\bm H_1,...,\bm H_N)\in\mathcal H_{\mathrm P}\). Substituting them into \eqref{eq:directional_heterogeneity_levels}
gives
\[
 \Lambda_{\mathrm{S}}
 \ge\frac{1}{\sqrt{m}}
 \left|\sum_{i=1}^{N}\langle\bm H_i,\bm X\bm E_i\rangle\right|\\
 =\frac{2}{\sqrt{m}}\sum_{j=1}^{k}
 \left\langle\frac{\bm e_l\bm z_j^\top}{\sqrt{2k}},
 0.45\bm e_l\bm z_j^\top\right\rangle
 =0.45\sqrt{\frac{2k}{m}}.
\]

Next, we show that the attained values above are also upper bounds.
For \(a\in\{\mathrm{LR},\mathrm S\}\), Cauchy--Schwarz and the normalization
of the feasible directions give
\[
\begin{aligned}
  \Lambda_a^2
  =\sup_{\substack{\bm X\in\mathcal X_a\\
  (\bm H_i)_{i=1}^N\in\mathcal H_{\mathrm P}}}
  \frac{\left|\sum_i\langle\bm H_i,\bm X\bm E_i\rangle\right|^2}{m}\le\sup_{\bm X\in\mathcal X_a}
  \frac{\sum_i\|\bm X\bm E_i\|_{\mathrm F}^2}{m}
  \le\frac{\lambda_{\max}\left(\sum_i\bm E_i\bm E_i^\top\right)}{m}.
\end{aligned}
\]
Orthogonality of \(\bm v,\bm z_1,\ldots,\bm z_{15}\) gives
\[
 \sum_{i=1}^N\bm E_i\bm E_i^\top
 =2(0.45)^2\left(
 k\bm v\bm v^\top
 +\sum_{j=1}^k\bm z_j\bm z_j^\top
 +2\sum_{j=k+1}^{15}\bm z_j\bm z_j^\top
 \right),
\]
whose largest eigenvalue is \(2k(0.45)^2\).
The upper and lower bounds coincide, giving
\[
 \Lambda_{\mathrm{LR}}=\Lambda_{\mathrm{S}}
 =0.45\sqrt{\frac{2k}{m}}
 =0.45\sqrt{2k/3},
 \qquad k\in\{0,5,10,15\}.
\]
The zero, low, medium, and high values are approximately
\(0,0.822,1.162,1.423\), respectively.

Finally, we verify that the remaining curvature and scale coefficients
are fixed across all four levels. Since $T_i\equiv T$ and \(\bm E_{2j-1}+\bm E_{2j}=\bm0\), the pooled covariance is preserved as
\[
 \bbm\Gamma_w
 =\frac1N\sum_{i=1}^N\bbm\Gamma_i
 =\overline{\bbm\Gamma}
 =\operatorname{diag}(1,\ldots,1,0.55,1.45).
\]
For \(k>0\), \(\bm E_i\) has exactly two nonzero eigenvalues, \(\pm0.45\).
Its nonzero entries lie entirely within the top-left \(38\times38\) block. Hence, in
\(\bbm\Gamma_i=\operatorname{diag}(\bm I_{38},0.55,1.45)+\bm E_i\),
the leading block has eigenvalues \(1-0.45\), \(1\), and \(1+0.45\),
while the last two diagonal entries are unchanged. At \(k=0\),
\(\bm E_i=\bm0\). Thus, at every level,
\(\lambda_{\min}(\bbm\Gamma_i)=0.55\) and
\(\lambda_{\max}(\bbm\Gamma_i)=1.45\).
For the restricted curvatures, take \(\bm H\) with \(H_{1,39}=1\)
and all other entries zero. This matrix has \(\|\bm H\|_{\mathrm F}=1\),
satisfies both shared and personalized sparsity constraints, and gives
\[
 \operatorname{tr}(\bm H\bbm\Gamma_w\bm H^\top)
 =\operatorname{tr}(\bm H\bbm\Gamma_i\bm H^\top)=0.55.
\]
Taking \(H_{1,40}=1\) instead, with all other entries zero, gives
\(1.45\) in both quadratic forms. Hence
\eqref{eq:app-full-pooled-curvature} and
\eqref{eq:app-full-population-curvature-coefficients} give
\[
 \alpha_{\mathrm{LR}}=\alpha_{\mathrm S}=\alpha_\mathrm{P}=\phi_{\min}=0.55,
 \qquad
 \beta_{\mathrm{LR}}=\beta_{\mathrm S}=\beta_\mathrm{P}=\phi_{\max}=1.45.
\]

For the innovation scale, \(\|\bm A_{i,*}\|_{\mathrm{op}}\le0.6\)
gives
\[
 0.028\bm I_p
 =(0.55-1.45\cdot0.6^2)\bm I_p
 \preceq\bbm\Sigma_{\epsilon,i}
 \preceq\bbm\Gamma_i
 \preceq1.45\bm I_p.
\]
The zero fortieth row of \(\bm A_{i,*}\) gives
\((\bm A_{i,*}\bbm\Gamma_i\bm A_{i,*}^\top)_{40,40}=0\). Hence
\[
 1.45=(\bbm\Gamma_i)_{40,40}
 =(\bbm\Sigma_{\epsilon,i})_{40,40}
 \le\|\bbm\Sigma_{\epsilon,i}\|_{\mathrm{op}}
 \le1.45,
\]
so \(\tau^2=\max_i\|\bbm\Sigma_{\epsilon,i}\|_{\mathrm{op}}=1.45\)
at all four levels.

\subsection{Working cap tradeoff in shared recovery}
\label{subsec:sim-working-cap-details}

For the experiment in Section~\ref{subsec:sim_working_cap}, set
$N=30$, $T=400$, and $\bm B_*=0.6\bm u\bm u^\top$, where
$u_j=15^{-1/2}$ for $j\le15$ and zero otherwise.
Generate $\bm S_{0,*}$ as in Appendix~\ref{subsec:sim-exp1-details},
with $s_{0,*}=12$.

For the personalized supports, we uniformly select 129 distinct coordinates from the full top-left
$15\times15$ block. Assign 128 coordinates
to 14 clients each and the remaining coordinate to 8 clients, using
the remaining-budget rule in Appendix~\ref{subsec:sim-exp1-details}
so that every client has $s_{i,*}=60$ entries.
At each coordinate with $d$ active clients, use a common random sign
and randomly assign the magnitudes $a_1,\ldots,a_d$, where
\[
 a_k=\frac{0.4}{1+\{(k-1)/2.5\}^{1/q}},\qquad q=1/4.
\]
In this construction, we obtain $m_*=14$ and $m_{q,*}=\sum_{k=1}^{14}a_k^q\approx5.57$.

Within each replication, we fix the true parameters, data, and fitting
levels $(s_0,s_i,\gamma p)=(16,75,3)$, and vary
\[
 m\in\{1,2,3,4,5,6,8,10,12,14\}.
\]

\subsection{Communication-efficient implementation}
\label{subsec:sim-communication-details}

For the experiment in Section~\ref{subsec:sim_communication},
we use Algorithm~\ref{alg:efficient_HT_FCHT} with
\(\widetilde s_0\in\{40,160\}\) and two prescribed FCT refresh
schedules. Using the round index \(l=0,\ldots,199\), the selective
schedule sets
\[
\mathcal R_{\mathrm{sel}}
=\{0,\ldots,39\}
\cup\{44,49,\ldots,99\}
\cup\{109,119,\ldots,199\},
\]
giving 62 refreshes. The every-10 schedule sets
\(\mathcal R_{10}=\{9,19,\ldots,199\}\), giving 20 refreshes.
Between refreshes, clients retain their previous supports while
continuing to update the personalized coefficients.

\end{document}